\documentclass[12pt,oneside]{book}

\usepackage[final,dvips]{graphicx} 
\usepackage[dvips,left=1.5in, right=1in, top=1in, bottom=1in, includefoot, headheight=13.6pt]{geometry}
\usepackage[all]{xy}
\usepackage{picture,graphics,epic,eepic}
\usepackage{color}
\usepackage{amssymb,amsmath,amsthm,latexsym,dsfont,verbatim}
\usepackage[title,titletoc]{appendix}

\usepackage{enumerate}
\usepackage{threeparttable}
\usepackage{tabularx}

\usepackage{supertabular}

\usepackage{setspace}
\usepackage[french,english]{babel}

\usepackage[hang, small, bf, tableposition=top]{caption}
\usepackage{fancyhdr}
\makeatletter
\renewcommand{\chaptermark}[1]{\markboth{\textsc{\@chapapp}\ \thechapter:\ #1}{}}
\makeatother
\fancypagestyle{plain}{
\fancyhf{}

\fancyfoot[LE,RO]{\thepage}
}
\makeatletter
\def\cleardoublepage{\clearpage\if@twoside \ifodd\c@page\else
    \hbox{}
    \thispagestyle{plain}
    \newpage
    \if@twocolumn\hbox{}\newpage\fi\fi\fi}
\makeatother \clearpage{\pagestyle{plain}\cleardoublepage}
\renewcommand{\chaptermark}[1]{\markboth{ \emph{#1}}{}}
\usepackage[bf,compact]{titlesec}
\titlespacing{\section}{0pt}{*1}{*1}
\titlespacing{\subsection}{0pt}{*1}{*1}
\titlespacing{\subsubsection}{0pt}{*1}{*1}
\titleformat{\chapter}[display]
{\normalfont\bfseries\filcenter}
{\LARGE\MakeUppercase{\chaptertitlename} \thechapter}
{1pc}
{\Huge}

\theoremstyle{plain}
\newtheorem{definition}{Definition}[section]
\newtheorem{problem}[definition]{Problem}

\newtheorem{conjecture}[definition]{Conjecture}

\theoremstyle{plain}
\newtheorem{proposition}[definition]{Proposition}
\newtheorem{Proposition}[definition]{Proposition}
\newtheorem{lemma}[definition]{Lemma}
\newtheorem{Lemma}[definition]{Lemma}

\newtheorem{theorem}[definition]{Theorem}
\newtheorem{Theorem}[definition]{Theorem}
\newtheorem{corollary}[definition]{Corollary}

\theoremstyle{plain}
\newtheorem{remark}[definition]{Remark}
\newtheorem{example}[definition]{Example}

\newenvironment{packed_enum}{
\begin{enumerate}
  \setlength{\itemsep}{1pt}
  \setlength{\parskip}{0pt}
  \setlength{\parsep}{0pt}
}{\end{enumerate}}

\def\>{\rangle}
\def\<{\langle}

\def\be{\begin{equation}}
\def\ee{\end{equation}}
\def\beu{\begin{equation*}}
\def\eeu{\end{equation*}}

\def\bee{\begin{eqnarray*}}
\def\eee{\end{eqnarray*}}

\newcommand{\nc}{\newcommand}
\newcommand{\ket}[1]{|#1\rangle}
\newcommand{\bra}[1]{\langle#1|}
\newcommand{\braket}[1]{|#1\rangle\langle#1|}
\newcommand{\ti}[1]{\tilde{#1}}
\newcommand{\ov}[1]{\overline{#1}}

\newcommand{\initstate}{\psi^{C_MBR}}
\newcommand{\initstateA}{\psi^{C_1C_2\ldots C_mABR}}

\newcommand{\inputstate}{\psi^{C_1C_2\ldots C_mAB}}
\newcommand{\inputGroupstate}{\psi^{{\cal T}A{\ov{\cal T}}BR}}
\newcommand{\inputRelativeA}{\psi^{{\cal T}A R_{\cK}}}
\newcommand{\inputRelativeB}{\psi^{{\overline{\cal T}} B R_{\cKbar}}}

\newcommand{\projstate}{\psi_{J_M}^{C^1_MB^0_MBR}}

\newcommand{\projstatebraket}{\braket{\psi_{J_M}}^{C^1_MB^0_MBR}}
\newcommand{\reducstate}{\psi_{{J_M}}^{C^1_M R}}

\newcommand{\mergestate}{\psi^{B_MBR}}
\newcommand{\decouplestate}{\tau^{C^1_M} \otimes \psi^{R}}
\newcommand \Tr {\mathrm{Tr}}
\newcommand \calE {{\cal E}}

\newcommand{\cK}{{\cal T}}
\newcommand{\cKbar}{{\overline{\cal T}}}

\newcommand{\cT}{\ov{\cal T}}
\newcommand{\cTo}{\ov{{\cal T}^0}}
\newcommand{\cTu}{\ov{{\cal T}^1}}
\nc{\outputrho}{\rho^{\cK^1A^1_{\cK} \cTu
B^1_{\cKbar}A_{\cK}B_{\cKbar} ABRX}}
\newcommand {\X}{{\cal X}}
\newcommand {\Y}{{\cal Y}}
\newcommand {\cX}{{\overline{\cal X}}}
\newcommand {\cY}{{\overline{\cal Y}}}

\newcommand{\cF}{{\cal F}}

\newcommand{\cH}{{\cal H}}

\newcommand{\cL}{{\cal L}}

\nc{\smfrac}[2]{\mbox{$\frac{#1}{#2}$}}
 \nc{\UA}{U^A_{j_{\cK}}}
 \nc{\VB}{V^B_{j_{\cKbar}}}
 \nc{\jK}{j_{\cK}}
 \nc{\jKbar}{j_{\cKbar}}

\def\squareforqed{\hbox{\rlap{$\sqcap$}$\sqcup$}}
\def\qed{\ifmmode\squareforqed\else{\unskip\nobreak\hfil
\penalty50\hskip1em\null\nobreak\hfil\squareforqed
\parfillskip=0pt\finalhyphendemerits=0\endgraf}\fi}

\renewenvironment{proof}{\noindent \textbf{{Proof~} }}{\qed\medskip}
\newenvironment{proof+}[1]{\noindent \textbf{{Proof #1~} }}{\qed\medskip}

\title{\large{\textbf{Multiparty quantum protocols for assisted entanglement distillation}} \\[2.5cm]
Nicolas Dutil \\[2cm]
\normalsize{School of Computer Science}\\
\normalsize{McGill University, Montr\'eal}\\
May 2011\\[4.5cm]
A thesis submitted to McGill University in partial fulfillment of the requirements of the degree of Ph.D. \\[1cm]
\normalsize{\copyright Nicolas Dutil, 2011}\author{}\date{} \\
}

\begin{document}
\frontmatter
\pagestyle{empty}
\maketitle
\fancyhead[RO]{}
\parindent=0.3in 
\pagenumbering{roman}                  
\newpage
\thispagestyle{empty}
\mbox{}
\pagestyle{plain}
\singlespacing
\onehalfspacing
\chapter*{Abstract}
\addcontentsline{toc}{section}{Abstract}
Quantum information theory is a multidisciplinary field whose objective is to understand what happens when information is stored in the state of a quantum system. Quantum mechanics provides us with a new resource, called quantum entanglement, which can be exploited to achieve novel tasks such as teleportation and superdense coding. Current technologies allow the transmission of entangled photon pairs across distances up to roughly 100 kilometers. For longer distances, noise arising from various sources degrade the transmission of entanglement to the point that it becomes impossible to use the entanglement as a resource for future tasks. One strategy for dealing with this difficulty is to employ quantum repeaters, stations intermediate between the sender and receiver that can participate in the process of entanglement distillation, thereby improving on what the sender and receiver could do on their own.

Motivated by the problem of designing quantum repeaters, we study entanglement distillation between two parties, Alice and Bob, starting from a mixed state and with the help of repeater stations. We extend the notion of entanglement of assistance to arbitrary tripartite states and exhibit a protocol, based on a random coding strategy, for extracting pure entanglement. We use these results to find achievable rates for the more general scenario, where many spatially separated repeaters help two recipients distill entanglement.

We also study multiparty quantum communication protocols in a more general context. We give a new protocol for the task of multiparty state merging. The previous multiparty state merging protocol required the use of time-sharing, an impossible strategy when a single copy of the input state is available to the parties. Our protocol does not require time-sharing for distributed compression of two senders. In the one-shot regime, we can achieve multiparty state merging with entanglement costs not restricted to corner points of the entanglement cost region. Our analysis of the entanglement cost is performed using (smooth) min- and max-entropies. We illustrate the benefits of our approach by looking at different examples.

\chapter*{R\'esum\'e}
\addcontentsline{toc}{section}{R\'esum\'e}
\foreignlanguage{french}{}
L'informatique quantique a pour objectif de comprendre les propri\'et\'es de l'information lorsque 
celle-ci est repr\'esent\'ee 
par l'\'etat d'un syst\`eme quantique. La m\'ecanique quantique nous fournit une nouvelle ressource, 
l'intrication quantique, qui 
peut \^etre exploit\'ee pour effectuer une  t\'el\'eportation quantique ou un codage superdense. 
Les technologies actuelles permettent 
la transmission de paires de photons intriqu\'es au moyen d'une fibre optique sur des distances 
maximales d'environ 100 kilom\`etres. 
Au-del\`a de cette distance, les effets d'absorption et de dispersion d\'egradent la qualit\'e de l'intrication.
Une strat\'egie pour contrer ces difficult\'es consiste en l'utilisation de 
r\'ep\'eteurs quantiques: des stations interm\'ediaires entre l'\'emetteur et le r\'ecepteur, 
qui peuvent \^etre utilis\'ees durant le processus de distillation d'intrication, d\'epassant ainsi 
ce que l'\'emetteur et le r\'ecepteur peuvent accomplir par eux-m\^emes.

Motiv\'es par le probl\`eme pr\'ec\'edent, nous \'etudions la distillation d'intrication 
entre deux parties \`a partir 
d'un \'etat mixte \`a l'aide de r\'ep\'eteurs quantiques. Nous \'etendons la notion d'intrication assist\'ee 
aux \'etats tripartites arbitraires
et pr\'esentons un protocole fond\'e sur une strat\'egie de codage al\'eatoire. 
Nous utilisons ces r\'esultats pour trouver des taux de 
distillation r\'ealisable dans le sc\'enario le plus g\'en\'eral, o\`u les deux parties ont recours 
\`a de nombreux r\'ep\'eteurs durant la distillation 
d'intrication.

En \'etroite liaison avec la distillation d'intrication, nous \'etudions \'egalement les 
protocoles de communication quantique multipartite. Nous \'etablissons un nouveau 
protocole pour effectuer un transfert d'\'etat multipartite. Une caract\'eristique 
de notre protocole est sa capacit\'e 
d'atteindre des taux qui ne correspondent pas \`a des points extr\^emes de la r\'egion r\'ealisable 
sans l'utilisation d'une strat\'egie de temps-partag\'e. 
Nous effectuons une analyse du co\^ut d'intrication en utilisant les mesures 
d'entropie minimale et maximale et illustrons les avantages 
de notre approche \`a l'aide de diff\'erents exemples. Finalement, nous proposons une variante de notre 
protocole, o\`u deux r\'ecepteurs et plusieurs \'emetteurs partagent un \'etat mixte. Notre protocole, 
qui effectue un transfert partag\'e, est appliqu\'e au probl\`eme de distillation assist\'ee.

\renewcommand{\appendixname}{}
\singlespacing
\renewcommand\contentsname{Table of Contents}
\tableofcontents
\parindent=0.3in 

\onehalfspacing
\chapter*{Notation}
\addcontentsline{toc}{section}{Notation}
\noindent
\begin{tabular}{p{3.7cm}p{10.5cm}}
\hline
\textbf{Common} &  \\
\hline
$\log$ & Binary logarithm. \\
$\ln$ & Natural logarithm. \\
$e$ & Euler's number.\\
$\mathbb{R}$ & Real numbers. \\
$\mathbb{C}$ & Complex numbers. \\
$\overline{c}$ & Complex conjugate of $c$. \\
\end{tabular} \\
\noindent
\begin{tabular}{p{3.7cm}p{10.5cm}}
\hline
\multicolumn{2}{l}{\textbf{Spaces}} \\
\hline
$A, B, C, \ldots$ & Hilbert spaces associated with the systems $A, B, C, \ldots$ \\
$\tilde{A}, \tilde{B}, \tilde{C}, \ldots$ & Typical subspaces of $A^{\otimes n}, B^{\otimes n}, C^{\otimes n}, \ldots$ \\
$d_A$ & Dimension of the space $A$. \\
$\Pi_{\tilde{A}}$ & Projector onto the typical subspace $\tilde{A}$. \\
$AB$ &  Tensor product $A \otimes B$ or composite system $AB$.\\
$A^n$ & Tensor product composed of $n$ copies of $A$. \\
$A_M$ & Tensor product $A_1 \otimes A_2 \otimes \ldots \otimes A_M$. \\
${\cal L}(A,B)$ & Space of linear operators from $A$ to $B$. \\
${\cal L}(A)$ & ${\cal L}(A,A)$. \\
\end{tabular} \\
\noindent
\begin{tabular}{p{3.7cm}p{10.5cm}}
\hline
\multicolumn{2}{l}{\textbf{Vectors}} \\
\hline
$\ket{\psi}^A, \ket{\phi}^A, \ldots$ & Vectors belonging to $A$.\\
$\braket{\psi}^A$ & Projector onto the vector $\ket{\psi}$. \\
$\langle \psi | \phi \rangle$ & Inner product of the vectors $\ket{\psi}$ and $\ket{\phi}$. \\
$\ket{\Phi^K}$ & Maximally entangled state of dimension $K$. \\
\end{tabular}\\
\noindent
\begin{tabular}{p{3.7cm}p{10.5cm}}
\hline
\multicolumn{2}{l}{\textbf{Operators}} \\
\hline
${\cal P}(A)$ & Set of positive semidefinite operators on $A$. \\
${\cal B}(A)$ & Set of density operators on $A$. \\
${\cal S}_{\leq}(A)$ & Set of sub-normalized density operators on $A$. \\
$\rho^A, \psi^A, \ldots$ & Density operators on $A$. \\
$\tau^A$ & Maximally mixed state of dimension $d_A$. \\
$\mathrm{id}_A$ & Identity map on ${\cal L}(A)$.\\
$I^{A}$ & Identity operator acting on $A$. \\
$\|X\|_1$ & Trace norm of the operator $X$. \\
$\|X\|_2$ & Hilbert-Schmidt norm of the operator $X$. \\
\end{tabular}\\
\noindent
\begin{tabular}{p{3.7cm}p{10.5cm}}
\hline
\multicolumn{2}{l}{\textbf{Distance measures for operators}} \\
\hline
$F(\rho,\bar{\rho})$ & Fidelity between $\rho$ and $\bar{\rho}$. \\
$\bar{F}(\rho,\bar{\rho})$ & Generalized fidelity between $\rho$ and $\bar{\rho}$. \\
$D(\rho,\bar{\rho})$ & Trace distance between $\rho$ and $\bar{\rho}$. \\
$\bar{D}(\rho,\bar{\rho})$ & Generalized trace distance between $\rho$ and $\bar{\rho}$. \\
$P(\rho,\bar{\rho})$ & Purified distance between $\rho$ and $\bar{\rho}$.
\end{tabular}
\\
\noindent
\begin{tabular}{p{3.7cm}p{10.5cm}}
\hline
\multicolumn{2}{l}{\textbf{Measures of information}} \\
\hline
$S(A)_{\psi}$ & von Neumann entropy of the density operator $\psi^A$. \\
$S(A|B)_{\psi}$ & Conditional von Neumann entropy of $\psi^{AB}$. \\
$H_{\min}(\rho^{AB}|\sigma^B)$ & Min-entropy of $\rho^{AB}$ relative to $\sigma^{B}$. \\
$H_{\min}(A|B)_{\rho}$ & Conditional min-entropy of $\rho^{AB}$ given $B$. \\
$H_{\max}(A|B)_{\rho}$ & Conditional max-entropy of $\rho^{AB}$ given $B$. \\
$H^{\epsilon}_{\min}(A|B)_{\rho}$ & Smooth min-entropy of $\rho^{AB}$ given $B$. \\
$H^{\epsilon}_{\max}(A|B)_{\rho}$ & Smooth max-entropy of $\rho^{AB}$ given $B$. \\
$H_2(\rho^{AB}|\sigma^{B})$ & Collision entropy of $\rho^{AB}$ relative to $\sigma^{B}$. \\
$I(A;B)_{\psi}$ & Mutual information of the density operator $\psi^{AB}$. \\
$I(A\rangle B)_{\psi} $ & Coherent information of the density operator $\psi^{AB}$. \\
$D(\psi^{AB})$ & Distillable entanglement of the density operator $\psi^{AB}$.\\
$E_A(\psi^{ABC})$ & Entanglement of assistance of the pure state $\psi^{ABC}$. \\
$D_A(\psi^{ABC})$ & Entanglement of assistance of the state $\psi^{ABC}$. \\

\hline
\end{tabular}

\onehalfspacing
\chapter*{Acknowledgements}
\addcontentsline{toc}{section}{Acknowledgements}
First, I would like to thank my two supervisors, Patrick Hayden and Claude Cr\'epeau, for their guidance and financial support throughout the years. The writing of this thesis would not have been possible without their beliefs in my success during the more difficult periods. Many thanks going to Claude Cr\'epeau for accepting to be my supervisor, which allowed me to quickly enter the Ph.D. program, for providing me with financial support for more than two years and for inviting me to a workshop in Barbados. I'm very thankful to Patrick Hayden for his quick responses to many of my questions, for deepening my understanding of the fundamental concepts of quantum information theory, for his constant optimism and enthusiasm and for offering me many opportunities to travel and meet new people.

I would like to acknowledge my coauthors Abubakr Muhammad, Kamil Br\'adler and Patrick Hayden, with whom I published my first research article. I thank also Nilanjana Datta for inviting me at a summer workshop at the University of Cambridge where a good portion of my thesis work started. I'm also grateful to Mario Berta for answering many questions I had regarding his work during my stay at Cambridge and afterwards. I'm also thankful for quick replies by Renato Renner and Marco Tomamichel, which provided me with accurate answers to questions I had regarding their work, and Mark Wilde, J\"urg Wullschleger, Andreas Winter for helpful discussions and comments regarding two papers written by me and my coauthor Patrick Hayden. I would also like to acknowledge the following members (past and present) of the CQIL: Ivan Savov, Omar Fawzi, Jan Florjanczyk, Fr\'ed\'eric Dupuis, Simon-Pierre Desrosiers, Ben Sprott, Nima Lashkari, David Avis, and Prakash Panangaden.

Finally, I would like to thank my family and particularly my parents, who always supported me throughout the years, and my soon to be wife M\'elanie Bertrand. Her support and love during this period of my life will always be remembered.
\chapter*{Contribution of authors}

Most of the work contained in this thesis appears in two papers. The material contained in Chapter 5 has been published~\cite{Dutil2} in the journal of Quantum Information and Computation. This is joint work with my supervisor Patrick Hayden. The majority of the content appearing in Chapters 3 and 4 has been submitted to the IEEE Transactions on Information Theory and is joint work with my supervisor Patrick Hayden. The current version \cite{Dutil1} of this paper is available from the e-print arXiv.

\onehalfspacing
\mainmatter

\fancyhead[LE]{\leftmark}
\fancyhead[RO]{\emph{Chapter \thechapter}}
\pagestyle{fancy}
\pagenumbering{arabic}
\setcounter{page}{1}
\titleformat{\chapter}[display]
{\normalfont\Large\filcenter\sffamily}
{\titlerule[1pt]%
\vspace{1pt}
\titlerule
\vspace{1pc}%
\LARGE\MakeUppercase{\chaptertitlename} \thechapter}
{1pc}
{\titlerule
\vspace{1pc}
\Huge}

\chapter{Introduction}
\section{Motivation}
Information is a general concept which has many meanings, but is mostly understood as knowledge communicated between two entities. The science of information has origins dating back to the 19th century, with the works of Andre\"i Markov on probability theory and Ludwig Boltzmann on statistical mechanics. The founder of the theory is usually identified as Claude E. Shannon, who formalized the notion of information through the concepts of entropy and mutual information. These measures characterize the limiting behavior of several operational quantities, such as the minimum compression length of a message or the capacity of transmitting information through a noisy channel.

Quantum information theory is a multidisciplinary field whose objective is to understand what happens when information is stored in the state of a quantum system. Quantum mechanics provides us with a new resource, called quantum entanglement, best explained from the words of Erwin Schr\"odinger, who coined the term in his 1935 seminal paper ``Discussion of probability relations between separated systems''\cite{schroedinger}:
\emph{When two systems, of which we know the states by their respective representatives, enter into temporary physical interaction due to known forces between them, and when after a time of mutual influence the systems separate again, then they can no longer be described in the same way as before, viz. by endowing each of them with a representative of its own. I would not call that one but rather the characteristic trait of quantum mechanics, the one that enforces its entire departure from classical lines of thought. By the interaction the two representatives [the quantum states] have become entangled.}

Entanglement can be measured, transformed, and purified. It is essential to performing communication tasks such as quantum teleportation \cite{teleportation} and superdense coding \cite{superdense}. It is also exploited for other computational and cryptographic tasks which are impossible for classical systems (for instance, cheating in a coin tossing challenge \cite{BB84} or winning a pseudo-telepathy game \cite{pseudo}).

\begin{figure}[t]
\begin{center}
\includegraphics{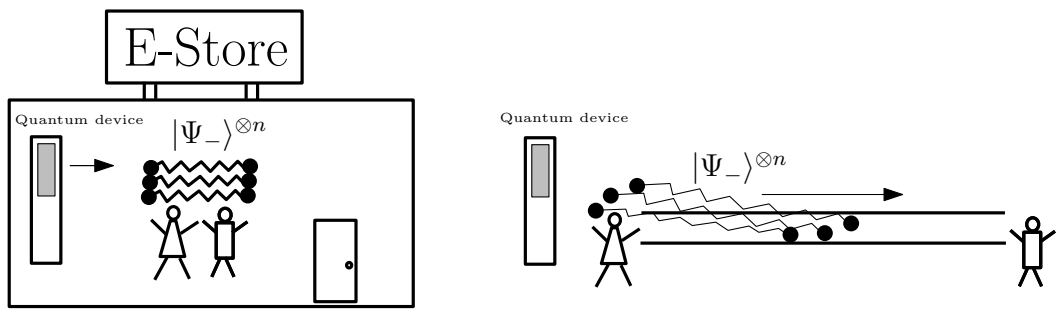}
\end{center}
\caption{Two different methods for establishing entanglement between Alice and Bob.} \label{fig:distribution}
\end{figure}

 Quantum teleportation, remote state preparation \cite{Bennett02} and device-independent cryptography ~\cite{ekert,device-indep} are examples of tasks which work on the assumption that entanglement can be shared between two spatially separated parties. To establish entanglement, the parties could meet at a common location and generate entangled pairs, with each party leaving with one half of each pair, or one of the parties could produce entanglement at his laboratory and send one half of each pair through a noiseless quantum channel (see Figure \ref{fig:distribution}) to the other party. The former strategy is currently infeasible as most quantum memories have very short storage times and are not designed to be moveable (see \cite{memories} for a review of quantum memories). As for the latter possibility, recent experiments \cite{fiber, free-space} have been successful at transmission of polarized entangled photons, with minimal loss of fidelity, over a distance of 144 kilometers in free-space. (The maximum distance is roughly 100 kilometers for transmission through a fiber.) If Alice and Bob are located further away than this distance, absorption and dispersion effects will eventually degrade entanglement fidelity to the point of making long-range entanglement-based communication impossible.

 One strategy for dealing with this difficulty is to employ quantum repeaters, stations intermediate between the sender and receiver that can participate in the process of entanglement distillation, thereby improving on what the sender and receiver could do on their own~\cite{repeaters,hierarchy,repeat-exp1,repeat-exp2}. By introducing such stations between different laboratories, and possibly interconnecting a subset of them via fiber optics, we can construct a quantum network (Figure \ref{fig:network1}).

 Each node of the network represents local physical systems which hold quantum information, stored in quantum memories. The information stored at the node can then be processed locally by using optical beam splitters \cite{Beam} and planar lightwave circuit technologies~\cite{Lightwave}, among other technologies. Entanglement between neighboring nodes can be established by locally preparing a state at one node and distributing part of it to the neighboring node using the physical medium connecting the two nodes. One of the main tasks then becomes the design of protocols that use the entanglement between the neighboring nodes to establish pure entanglement between the non adjacent nodes.
\begin{figure}[h]
\begin{center}
\includegraphics{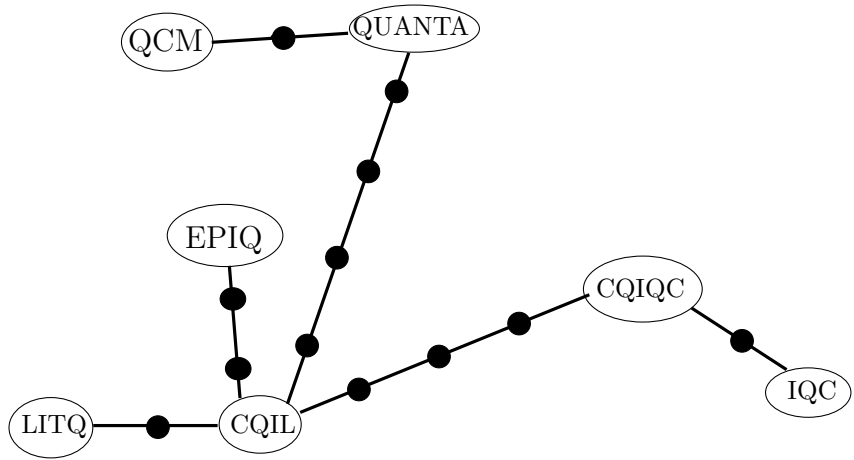}
\end{center}
\caption{A hypothetical quantum network connecting various university quantum laboratories. Repeater stations are represented by black dots.} \label{fig:network1}
\end{figure}

Let's consider the simplest non-trivial network, which was studied previously in \cite{cirac}, and consists of two laboratories separated by a repeater station (see Figure \ref{fig:swap}). At one endpoint of the network, Alice prepares an entangled system in the state $\ket{\psi}^{AC_1}=\sqrt{\lambda_1} \ket{00}^{AC_1} + \sqrt{\lambda_2}\ket{11}^{AC_1}$, and sends the $C_1$ part to the repeater station using the (noiseless) quantum channel  connecting them. Without loss of generality, we can assume that $\lambda_1 \geq \lambda_2$. The repeater prepares an entangled system in the same state $\ket{\psi}^{C_1C_2}$ and transmits the $C_2$ part to Bob. To establish entanglement between the laboratories, the repeater station performs a projective measurement on the composite system $C_1C_2$ with projectors corresponding to each of the four Bell states:
\begin{eqnarray*}
P_{00} &=& \braket{\Phi_{+}}^{C_1C_2} \\
P_{01} &=& \braket{\Psi_{+}}^{C_1C_2} \\
P_{10} &=& \braket{\Phi_{-}}^{C_1C_2} \\
P_{11} &=& \braket{\Psi_{-}}^{C_1C_2}.
\end{eqnarray*}
If the Bell measurement yields outcome $01$ or $11$, both occurring with equal probability $\lambda_1\lambda_2$, then Alice and Bob share the state $\frac{1}{\sqrt{2}}(\ket{01}^{AB}\pm \ket{10}^{AB})$. For the outcome $01$, they recover the \textit{singlet} state $\frac{1}{\sqrt{2}}(\ket{01}^{AB} - \ket{10}^{AB})$ from the state $\frac{1}{\sqrt{2}}(\ket{01}^{AB} + \ket{10}^{AB})$ if Bob applies the operator $Z$ on system, where
\begin{equation*}
 \begin{split}
 Z &= \ket{0}\bra{0}^B - \ket{1}\bra{1}^B \\
 \end{split}
\end{equation*}
 is a Pauli operator. For measurement outcomes $00$ and $10$, obtained with equal probabilities $\frac{\lambda^2_1+\lambda^2_2}{2}$, the reduced states on Alice's and Bob's systems are $\frac{1}{\sqrt{\lambda^2_1+\lambda^2_2}}(\lambda_1\ket{00}^{AB} \pm \lambda_2\ket{11}^{AB})$. These states are not \textit{maximally} entangled. ( See Chapter 2 for a precise definition.) To obtain a singlet state with optimal probability $\frac{2\lambda^2_2}{\lambda^2_1+\lambda^2_2}$, Bob performs the following generalized measurement and communicates the outcome to Alice:
\begin{figure}[t]
\begin{center}
\includegraphics{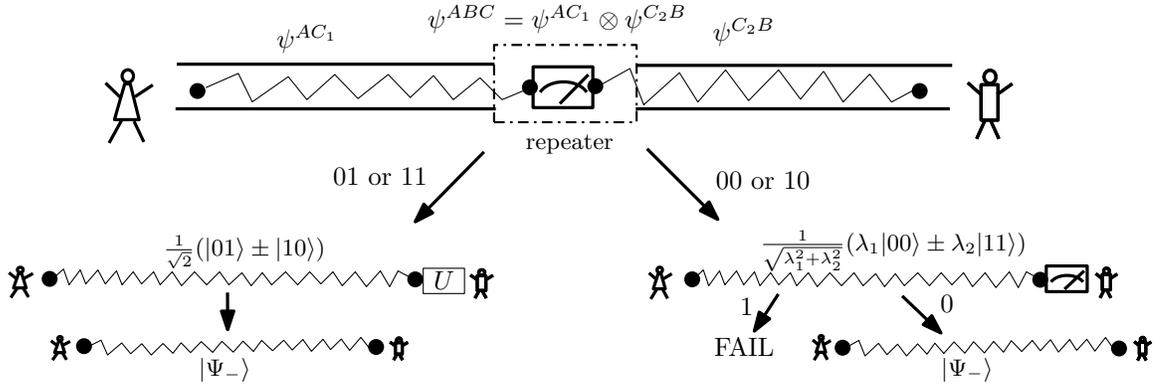}
\end{center}
\caption{A one-dimensional chain with one repeater node. To perform entanglement swapping, the repeater node performs a Bell measurement. Depending on the outcome, Bob will follow with either a decoding operation (i.e a unitary) to get back a singlet, or a generalized measurement to produce a singlet with optimal probability.} \label{fig:swap}
\end{figure}
\begin{equation}\label{eq:GenMeas}
 M_0 = \frac{\lambda_2}{\lambda_1} \braket{0}^B + \braket{1}^B, \quad
 M_1 = \sqrt{\frac{\lambda^2_1-\lambda^2_2}{\lambda^2_1}} \braket{0}^B.
\end{equation}
If outcome $0$ is obtained, Alice and Bob recover a singlet state by applying appropriate Pauli operators on Bob's share. Otherwise, a failure is declared. Thus, the singlet conversion probability for this \textit{entanglement swapping} strategy is equal to $2(\lambda_1\lambda_2 + \lambda^2_2) = 2\lambda_2$. Remarkably, as was noted in \cite{cirac}, this corresponds to the \textit{optimal singlet conversion probability} (SCP) for the state $\ket{\psi}=\sqrt{\lambda_1} \ket{00} + \sqrt{\lambda_2}\ket{11}$ (to see this, just replace $\lambda_1$ and $\lambda_2$ in eq.~(\ref{eq:GenMeas}) by $\sqrt{\lambda_1}$ and $\sqrt{\lambda_2}$). This shows that the entanglement swapping strategy maximizes singlet conversion probability between Alice and Bob, which, in a one-dimensional chain with identical pure states $\ket{\psi}$ between repeater stations, can never exceed the SCP of $\ket{\psi}$.

Unfortunately, the previous strategy cannot be extended to one dimensional chains with many repeater stations separating Alice and Bob's laboratories. In fact, as was shown in \cite{cirac}, no measurement strategy can keep the SCP between Alice and Bob from decreasing exponentially with the number of repeaters, making them useless for establishing entanglement over long distances.

One way to deal with this problem is to introduce redundancy in the network \cite{repeaters}. By preparing and distributing many copies of the state $\ket{\psi}=\sqrt{\lambda_1}\ket{00} + \sqrt{\lambda_2}\ket{11}$ across the chain, the repeater stations will be able to help Alice and Bob in producing singlets. The redundancy introduced in the network allows the stations to perform joint measurements on their shares, concentrating the entanglement found in each copy of $\ket{\psi}$ into a small number of highly entangled particles. For one-dimensional chains, the rate at which entanglement can be established between the two endpoints will approach the entropy of entanglement $S(A)_{\psi}$, no matter the number of repeaters introduced between the endpoints. The more copies of the state $\ket{\psi}$ are prepared and distributed between the nodes, the more transparent the repeaters will become, allowing us to view the entire chain as a noiseless channel for Alice and Bob.

This is an ideal situation, one unlikely to occur in real experiments, as only a finite number of copies of the state $\ket{\psi}$ will be prepared and the preparation and distribution of copies of this state across the network will be imperfect. It is also reasonable to assume that the storage of many qubits at a repeater station, or at one of the laboratories, will be more prone to errors over time than the storage of a single qubit. Hence, the global state of a quantum network will most likely be mixed. For such mixed state networks, we can ask the question: how much entanglement can we establish between Alice and Bob by performing LOCC operations on the systems part of the network ?

In the following chapters, following a brief review of the relevant concepts in information theory, we consider several variations of the previous question and look at closely related problems. Although we do not solve the assisted distillation problem completely, we give new results for a less restricted form of the problem, compared to what was considered before in the works of DiVincenzo \textit{et al}. and others \cite{dfm, Laustsen, SVW, merge, SW-Nature}, and rediscover known formulas for assisted distillation, established by Smolin et al. and Horodecki et al. in \cite{SVW, merge}, by devising new protocols. In the remainder of this chapter, we give a brief summary of each of the following chapters, and then state the contributions found in this thesis.

\section{Summary}
The thesis consists of six chapters and one appendix.
\subsubsection{Chapter 2: Preliminaries}
 This chapter is divided into three parts. First, we review relevant concepts in linear algebra. From this, we formulate the basic postulates of quantum mechanics in the language of linear algebra and discuss the density operator formalism. For our applications of quantum mechanics, this mathematical approach is more useful than standard formulations in terms of wave functions (Schr\"odinger picture) or time-dependent operators (Heisenberg picture). We then introduce the basics of quantum information theory, its formalism, and important results we will use in the following chapters. Finally, we conclude this chapter by reviewing three entanglement distillation protocols. The first two protocols discussed are examples of ``exact'' approaches to entanglement distillation: assuming the protocols can be implemented without introducing errors, they yield a number of perfect Einstein-Podolsky-Rosen (EPR) entangled pairs with high probability. The Schmidt method describes a procedure, via projective measurements, for extracting EPR pairs. The hashing method, on the other hand, hashes an unknown sequence of Bell pairs until an exact subsequence is found (with high probability). The last protocol involves a different paradigm, prevalent in information theory: the use of random coding for showing the existence of a family of protocols producing states arbitrarily close to a product of EPR pairs at near optimal rates. We discuss this protocol in an informal manner, as this approach will be studied in more detail subsequently and is central to the various tasks analyzed in this thesis.
 \subsubsection{Chapter 3: Multiparty state transfer}
 This chapter has three parts. We begin by introducing the information-processing task of transferring a system from one location to another. Previous work by Abey-esinghe et al. \cite{Hayden001} considered the problem from a ``fully quantum'' perspective: a single sender must use a minimal amount of quantum communication to transfer his entire system to the receiver. Existence of protocols achieving optimal rates was proven in this setting. Work by Horodecki et al. \cite{merge, SW-Nature}, who gave the first formulation of this problem, analyzed the task by substituting quantum communication with pre-shared entanglement and classical communication. Optimal rates were shown to be achievable by using a random measurement strategy. The problem was also extended to the multiple senders, single receiver setting, also known as distributed compression. In this chapter, we analyze new protocols for the task of multiparty state merging ($m$ senders, one receiver) and split-transfer ($m$ senders, two receivers) both in the one-shot regime and in the asymptotic setting. We also apply our split-transfer protocol to recover the formula in \cite{merge}, provided a certain conjecture holds, for the optimal assisted distillable rate when $m$ helpers and two recipients (Alice and Bob) share a pure state.
\subsubsection{Chapter 4: Entanglement cost of multiparty state transfer}
This chapter has two parts. First, we reformulate the one-shot results of Chapter 3 in terms of (smooth) min-entropies and provide protocols for one-shot multiparty state merging and one-shot split transfer. Our work extends some of the previous results by Berta \cite{Berta} and Dupuis et al. \cite{decouplingBerta}, which considered the task of one-shot state merging for the case of a single sender. In the last portion of this chapter, we compare our multiparty state merging protocols for different family of states, highlighting interesting differences between the two protocols.
\subsubsection{Chapter 5: Assisted entanglement distillation}
This chapter has four main sections. After a brief introduction, we extend the entanglement of assistance problem to the case of mixed states shared between a helper (Charlie) and two recipients (Alice and Bob). This problem was first studied in the one-shot regime by DiVincenzo et al. \cite{dfm}, and a formula was found in the asymptotic regime by Smolin et al. \cite{SVW}, which was generalized to an arbitrary number of parties by Horodecki et al. in \cite{merge}. We show an equivalence between the operational notion of assisted entanglement and a one-shot quantity maximizing the average distillable entanglement over all POVMs performed by the helper. We proceed with an asymptotic analysis of the mixed-state assisted distillation problem, deriving a bound on the achievable assisted distillable rate, which surpasses the hashing inequality in certain cases. We generalize this analysis to the multiparty scenario and compare our approach with a hierarchical distillation strategy.
\subsubsection{Chapter 6: Conclusion}
This chapter summarizes the results established in the previous chapters. We also discuss some of the open problems remaining to be solved and propose different lines of research related to the subjects touched upon in this thesis.
\subsubsection{Appendix A: Various technical results}
The appendix contains proofs of various lemmas and propositions used in the previous chapters. We give more details regarding this part of the thesis in the contribution section.
\section{Contributions}
\subsection*{Chapter 2: Preliminaries}
This chapter does not contain any original material. An effort was made, however, to present the introductory material with enough precision and substance that a reader with limited background in quantum information theory may grasp the essential ideas found in the following chapters.
\subsection*{Chapter 3: Multiparty state transfer}
\subsubsection{I. Removing time-sharing}
 The distributed compression protocol of \cite{merge}, although very intuitive and easy to understand by building upon the optimal rates achievable by the state merging primitive, must use a time-sharing argument to demonstrate achievability for rates which are not corner points. Our first contribution of this chapter is to give a protocol for achieving multiparty state merging without requiring a time-sharing strategy. More specifically, we show that distributed compression for the case of two senders is achievable without the use of time-sharing. To show this, we adapt the ideas in \cite{merge} and perform a direct technical analysis of the task of multiparty state merging, obtaining a bound on the decoupling error when each sender performs a random measurement. For the more general case of $m$ senders, there is a technical obstacle to proving that time-sharing is not required for multiparty state merging. The difficulty is more of a general quantum Shannon theory question than a problem with the analysis of multiparty state merging. We make a conjecture, which we call the multiparty typicality conjecture, and prove it is true for the case of a mixed state $\psi^{C_1C_2}$ in appendix A.
\subsubsection{II. Freedom in the distribution of catalytic entanglement}
A nice feature of our approach is to allow more freedom regarding the disposition of the catalytic entanglement, sometimes needed when performing the task of multiparty state merging. We give a simple example to illustrate the benefits of our protocol over the distributed compression protocol of \cite{merge}, which restricts catalytic entanglement to be distributed in a very specific way. This is the second contribution of our chapter. If time-sharing is not required for performing distributed compression for the case of three senders, we show that for certain states our protocol needs no catalytic entanglement, in contrast to the distributed compression protocol of \cite{merge}, which needs catalytic entanglement even if some of the entanglement rates for such states are negative.
\subsubsection{III. Split-transfer}
The last part of this chapter considers the problem of state transfer for multiple senders and two receivers. Here, the senders are split into a group ${\cal T}$ and its complement $\overline{\cal T}$. The objective is to redistribute the global state to the two receivers. More precisely, we must transfer the system $\cal T$ to one receiver while sending $\overline{\cal T}$ to the other receiver. To my knowledge, this problem has not been studied before. Two independent applications of the multiparty merging protocol will achieve a split-transfer with optimal rates. In the spirit of the previous sections of this chapter, we consider this problem directly by customizing our multiparty merging protocol for this task.
\subsubsection{IV. Answering the min-cut conjecture}
Our last contribution in this chapter is an answer to a conjecture posed by Horodecki et al. in \cite{merge} in the context of assisted distillation. The optimal multipartite entanglement of assistance rate was found to be equal to the minimum-cut bipartite entanglement $\min_{\cal T} S(A{\cal T})$, where the minimization is over all possible cuts ${\cal T}$ of the helpers. The proof in \cite{merge} is recursive: they show that, with high probability, the min-cut entanglement is preserved after one helper has finished his random measurement and apply this reasoning recursively for all other helpers. The conjecture asks if this recursive argument can be removed. More precisely, if a strategy where all the helpers performed their random measurements all at once will yield a state which preserves, with high probability, the minimum cut entanglement of the state. We show that this is true for almost all cases provided the multiparty typicality conjecture holds. Under this assumption, we show how to redistribute (many copies of) the original state using our split-transfer protocol in such a way that it preserves the min-cut entanglement. The receivers (Alice and Bob) can follow with a distillation protocol, yielding a rate of EPR pairs corresponding to the min-cut entanglement of the original state.
\subsection*{Chapter 4: Entanglement cost of multiparty state transfer}
\subsubsection{I. Entanglement cost region of multiparty merging}
Our first contribution of this chapter is to reformulate the upper bound derived in Chapter 3 for the decoupling error as a function of various min-entropy quantities. With this result in hand, we give a partial characterization in terms of min-entropies of the entanglement cost region achievable for multiparty state merging when a single copy of the state is available. For any point of this region, we show the existence of multiparty merging protocols of the kind described in the previous chapter, where all the senders measure their systems simultaneously and the decoder implemented by the receiver is not restricted to recovering the systems one at a time. We derive analogous results for the task of split-transfer by applying the same proof technique.
\subsubsection{II. Smooth min-entropy characterization}
Using the approach of Horodecki et al. \cite{merge} for achieving a distributed compression of a multipartite state $\psi^{C_1C_2\ldots C_mR}$, we analyze the entanglement cost associated with multiparty merging when a single-shot state merging protocol is applied iteratively, according to some ordering $\pi: \{1,2,\ldots, m\} \rightarrow \{1,2,\ldots, m\}$ on the senders. By building upon the results of Berta \cite{Berta} and Dupuis et al. \cite{decouplingBerta}, we show the existence of multiparty merging protocols with arbitrarily small error and entanglement cost characterized by the smooth min-entropies of the reduced states $\psi^{C_i \tilde{R}_{\pi^{-1}(i)}}$, where $\tilde{R}_{\pi^{-1}(i)}$ is the relative reference for the sender $C_i$ with respect to an ordering $\pi$ of the senders. This is the second contribution of this chapter.
\subsubsection{III. Examples of one-shot distributed compression}
The remainder of this chapter is devoted to examples. We compare the protocols described in this chapter for the task of distribution compression. We give three examples, two of them being closely related to the second distributed compression example of Chapter 3, and look at the entanglement costs required for merging the states. We find, once again, that our direct approach to the task of multiparty merging yields better results: our protocol outperforms an application of many single-shot two-party state merging protocols by allowing some of the senders to transfer their systems for free.
\subsection*{Chapter 5: Assisted entanglement distillation}
\subsubsection{I. Generalizing the entanglement of assistance}
The first contribution of this chapter is to extend the one-shot entanglement of assistance quantity, first defined in \cite{dfm}, to handle mixed states $\psi^{ABC}$ shared between two recipients (Alice and Bob) and a helper Charlie. This quantity reduces to the original entanglement of assistance when the state is pure. We give an operational definition of assisted distillation for mixed states $\psi^{ABC}$ and show an equivalence between the optimal distillable rate and the regularization of the entanglement of assistance quantity. This equivalence is used in the following section for proving achievable rates on the optimal assisted distillable rate when the parties share many copies of a mixed state $\psi^{ABC}$. We give two upper bounds to the entanglement of assistance for mixed states, and provide an example which saturates one of the upper bounds.
\subsubsection{II. Achievable rates for assisted distillation}
Using the equivalence between the optimal distillable rate and the regularization of the entanglement of assistance quantity, we give a lower bound on the optimal rate for assisted distillation of mixed states for the case of one helper. We prove the existence of a measurement for the helper Charlie which will preserve, with arbitrarily high probability, the minimum cut coherent information $L(\psi):=\{I(AC \rangle B)_{\psi}, I(A\rangle BC)_{\psi} \}$ of the input state. This is the second contribution of this chapter. Using this measurement in a double blocking strategy, Alice and Bob can recover singlets at the rate $L(\psi)$ by applying standard distillation protocols as in \cite{DW}. If Charlie preprocesses his share of the state to optimize the minimum cut coherent information, higher rates can potentially be achieved. If the state $\psi^{ABC}$ does not saturate strong subadditivity, and the coherent information $I(C \rangle AB)_{\psi}$ is positive, the achievable rate is higher than what the hashing inequality guarantees when performing a one-way distillation protocol.
\subsubsection{III. Optimality}
 Achievability of the min-cut coherent information has a surprising consequence: we can achieve a rate close to what could be obtained if Charlie were allowed to send his system to either Alice or Bob, whichever minimizes the minimum cut coherent information. We give a specific example where Charlie is not capable of transferring his system to Alice for free, but the assisted rates achievable are nonetheless close to $I(AC \rangle B)_{\psi}$. When $L(\psi)$ is the coherent information $I(A \rangle BC)_{\psi}$, however, Charlie can merge his system to Bob. For such a case, we can achieve an optimal rate for assisted distillation by applying a merging protocol before engaging in a distillation protocol. This is the third contribution of this chapter.
\subsubsection{IV. Fault-tolerance}
We compare our assisted distillation protocol to a hierarchical strategy consisting of entanglement distillation followed by entanglement swapping. The first example we analyze considers a one-dimensional chain where the Alice to Charlie's channel is noiseless but the Charlie to Bob channel is noisy. For a state in a product form, we find that the rate achieved by our protocol is the same as the rate obtained by using a hierarchical strategy. We modify our setup by introducing a CNOT error affecting Charlie's systems. We show that our random measurement strategy is fault-tolerant against such error: the assisted distillation rate remains the same, even in the absence of error correction by Charlie. On the other hand, the rate obtained by a hierarchical strategy becomes null. Thus, we identify a major weakness to using hierarchical strategies: it is not fault-tolerant against errors arising at Charlie's laboratory.
\subsubsection{V. Multipartite entanglement of assistance}
The last part of this chapter generalizes the multipartite entanglement of assistance of \cite{SVW, merge} to allow an arbitrary multipartite mixed state shared between $m$ helpers and two receivers. Our one-shot quantity reduces to the original multipartite entanglement of assistance quantity when the state is pure. We derive an upper bound to this quantity, and then perform an asymptotic analysis, proving the existence of protocols achieving a rate which is at least the minimum cut coherent information $I(A\cK \rangle B\cKbar)_{\psi}$, where $\cK$ is a cut of the helpers. Our proof relies on a multiple blocking strategy and suggests the possibility of a simpler protocol for achieving the minimum cut coherent information. This is the fifth contribution of this chapter.
\subsection*{Appendix A: Various technical results}
 \subsubsection{I. A different proof of the twirling average}
 We give a detailed calculation of the twirling average, a key result (see \cite{merge} and \cite{Hayden001} for the original proof) used in Chapter 3 for proving one of the important results of this thesis. Our proof does not rely on Schur's lemma, a fundamental result in representation theory, but instead relies on the invariance property of the Haar measure with respect to permutations, sign-flip operators and Hadamard transformations.
\subsubsection{II. Convexity of the entanglement of assistance}
 We give a proof of the convexity of the entanglement of assistance for pure ensembles $\{p_i, \psi_i^{ABC}\}$. This result is used in Chapter 5 for proving an upper bound to the one-shot entanglement of assistance.
\subsubsection{III. Lower bound to the smooth max entropy}
By removing the smallest eigenvalues of a state $\rho$, without disturbing the state too much, we get a useful lower bound to the smooth max entropy $H^{\epsilon}_{\max}(\rho)$. We use this bound in Chapter 4 for the various examples we analyze.
\subsubsection{IV. Multiparty typicality conjecture}
We give a proof that the multiparty typicality conjecture is true for the case of a mixed state $\psi^{C_1C_2}$. Our proof relies on a well-known inequality of probability theory and uses a double blocking strategy for constructing a state which satisfies the typicality conjecture (see Section \ref{sec:iid}).

\chapter{Preliminaries}
 \section{Representation of physical systems}
  \subsection{Hilbert spaces and linear operators}
A set $V$ is a vector space over a field ${\cal F}$ if given two operations, vector addition and scalar multiplication, it satisfies certain axioms (see table \ref{fig:axiomVec}). Examples of commonly used fields are the field of real numbers $\mathbb{R}$, the field of complex numbers $\mathbb{C}$, and the Galois field $F_2$ consisting of two elements, $0$ and $1$, for which addition and multiplication correspond to XOR and AND operations.
\begin{table}
\begin{tabular}{|p{5cm} p{9cm}|}
\hline
\textbf{closure} & If $u$ and $v$ are in $V$, then $u+v$ is in $V$. If $a \in \mathbb{C}$ and $u \in V$, then $au \in V$. \\
\textbf{associativity} & $u +(v+w) = (u+v)+w$ for all $u, v$, and $w$ in $V$. \\
\textbf{compatibility} & $a(bv)=(ab)v$ for all $v \in V$ and all $a,b \in \mathbb{C}$.\\
\textbf{commutativity} & $u+v = v+u$ for all $u,v \in V$.\\
\textbf{zero element} & An element $0$ in $V$ exists such that $v+0 = v = 0 +v$ for all $v \in V$. \\
\textbf{inverse} & For each $v \in V$, an element $-v$ exists in $V$ such that $-v + v = 0 = v + (-v)$.\\
\textbf{distributivity} & $a(v+w) = av + aw$ and $(a+b)v = av + bv$ for all $v,w \in V$ and $a,b \in \mathbb{C}$. \\
\textbf{identity} & $1v = v$ for all $v \in V$. \\
\hline
\end{tabular}
\caption{Axioms for a complex vector space.}
\label{fig:axiomVec}
\end{table}
Examples of vector spaces are the Euclidean $n$-space $\mathbb{R}^n$, the complex vector space $\mathbb{C}^n$, and the space of all functions $f: X \rightarrow \cF$ for any fixed set $X$. For the space $\mathbb{C}^n$, the vectors are the $n$-tuples $z=(z_1, z_2, \ldots, z_n)$ with $z_i \in \mathbb{C}$, and the addition and scalar multiplication operations are defined in a pointwise fashion: for vectors $x = (x_1, x_2, \ldots, x_n)$ and $y =(y_1, y_2, \ldots, y_n)$ in $\mathbb{C}^n$, and scalars $a \in \mathbb{C}$, we have
  \begin{equation*}
     x+y = (x_1+y_1, x_2+y_2,\ldots, x_n+y_n) \quad \text{and} \quad ax = (ax_1, ax_2, \ldots, ax_n).
  \end{equation*}
 A set $B = \{v_1, v_2, \ldots, v_n\}$ of vectors in $V$ is called a basis of the vector space $V$ if it is a linearly independent set which generates the whole space $V$. That is, no vector in $B$ can be written as a linear combination of finitely many other vectors in $B$, and the set of all linear combinations of the vectors in $B$ correspond to the whole space $V$. A vector space with basis $\{v_1, v_2,\ldots, v_n\}$ is said to have dimension $d_V = n$. 

To add notions of length and distance to a vector space, we introduce a third operation, called the inner product $\langle u,v \rangle: V \times V \rightarrow \cF$. Here, the field is usually taken to be either $\mathbb{R}$ or $\mathbb{C}$. An inner product must satisfy the three properties described in table \ref{tab:inner}.
 \begin{table}
\begin{tabular}{|p{5cm}p{9cm}|}
\hline
\textbf{conjugate symmetry} & $\langle u,v\rangle=\overline{\langle v,u\rangle}$ for all $u,v \in V$. \\
\textbf{linearity} & $\langle v+w,u\rangle=\langle v,u\rangle + \langle w,u\rangle$ and $\langle rv,w\rangle =r\langle v,w\rangle$ for all $u, v, w$ in $V$ and $r \in \mathbb{C}$.\\
\textbf{positive-definiteness} & $\langle v,v \rangle \geq 0$ for all $v \in V$ with equality iff $v = 0$. \\
\hline
\end{tabular}
\caption{Axioms for the inner product when ${\cal F}=\mathbb{C}$.}
\label{tab:inner}
\end{table}
A vector space $V$ with an inner product $\langle, \rangle$ is called an inner product space. We can define an inner product for the space $\mathbb{C}^n$ as follows:
\begin{equation*}
  \langle x,y \rangle := \overline{x}_1 y_1 + \overline{x}_2 y_2 + \ldots + \overline{x}_n y_n.
\end{equation*}
We have $\langle x,x \rangle = |x_1|^2 + |x_2|^2 + \ldots + |x_n|^2 \geq 0$, and the other two axioms can be verified just as easily. For an inner product space $V$, we assign a ``length'' to a vector $v$ via the norm
\begin{equation*}
\|v\| := \sqrt{\langle v,v\rangle}.
\end{equation*}
A vector space $V$ on which a norm is defined is called a normed vector space. For two vectors $x, y$ of a normed space $V$, we can add a notion of distance between two vectors $x$ and $y$ by using the norm:
\begin{equation*}
d(x,y) := \|x - y\|.
\end{equation*}
Symmetry and positivity of $d(x,y)$ follow easily from the above definitions. The triangle inequality $d(x,z) \leq d(x,y) + d(y,z)$ can be recovered using the Cauchy-Schwarz inequality:
\begin{equation*}
 |\langle x,y \rangle | \leq \|x\| \|y\|.
\end{equation*}
A space $V$ for which a distance function $d(x,y)$ is defined is called a metric space.  A metric space is complete if and only if every sequence $x_1, x_2, x_3, \ldots$ of vectors in $V$ for which $d(x_n, x_m) \rightarrow 0$, as both $n$ and $m$ independently tends toward infinity, converges in $V$. That is, for every such sequence $x_1, x_2, x_3, \ldots$ there exists a $y \in V$ such that $d(x_n, y) \rightarrow 0$ as $n \rightarrow \infty$.

A Hilbert space $\cH$ is a real or complex inner product space which is also a complete metric space with respect to the distance function induced by the inner product. For finite dimensional Hilbert spaces, the completeness criterion is automatically met and, thus, any real or complex inner product space is also a Hilbert space. As we will see shortly, Hilbert spaces arise in quantum mechanics to model the state space of a physical system. The tasks analyzed in this thesis involve quantum systems which can be adequately described using finite dimensional complex Hilbert spaces. Henceforth, we assume the Hilbert spaces to be of finite dimension. Vectors for a complex Hilbert space ${\cal H}_A$ associated with a physical system $A$ are written using the Dirac notation, also known as bra-ket notation, in the form $\ket{\psi}^A, \ket{\phi}^A, \ldots$ These vectors are called \textit{kets}, and for every ket $\ket{\psi}^A$ of the Hilbert space ${\cal H}_A$, henceforth written simply as $A$, there is an associated linear functional $\bra{\psi}^A: A \rightarrow \mathbb{C}$ called a \textit{bra}:
\begin{equation*}
   \bra{\psi} (\ket{\phi}) := \langle \psi, \phi \rangle,
\end{equation*}
where the right hand side is the inner product of the two vectors $\ket{\psi}^A$ and $\ket{\phi}^A$. The motivation for the bra-ket notation comes from this last definition, where we see that by removing parentheses around the vector $\ket{\phi}^A$ and fusing the bars together on the left hand side of the definition, we obtain a complex number $\langle \psi | \phi \rangle$ called a bra-ket or bracket.

A basis for the space $\mathbb{C}^2$ is given by $\{(0,1), (1,0) \}$, which can be rewritten in braket notation as $\{ \ket{0}, \ket{1} \}$. This is known as the \textit{computational basis} for the space $\mathbb{C}^2$. For the general space $\mathbb{C}^n$, the computational basis will be written as $\{\ket{1}, \ket{2}, \ket{3}, \ldots, \ket{n}\}$.  Any vector $\ket{\psi} \in \mathbb{C}^n$ can then be written as
\begin{equation*}
\ket{\psi} = \sum^n_{i=1}\alpha_i \ket{i} \quad \alpha_i \in \mathbb{C}.
\end{equation*}

Given two Hilbert spaces $A$ and $B$, we can construct a larger Hilbert space of dimension $d_Ad_B$ by taking the tensor product $A \otimes B$. Given two orthonormal bases $\{\ket{v_i}^A\}_{i=1}^{d_A}$ and $\{\ket{w_j}^B\}^{d_B}_{j=1}$ of $A$ and $B$ (i.e $\langle v_i | v_j \rangle = 0$ and $\langle w_i | w_j \rangle = 0$ for any $i \neq j$ ), the tensor product $A \otimes B$ is the space generated by the basis elements $\{\ket{v_i}^A \otimes \ket{w_j}^B\}$. How tensor products $\ket{v_i}^A \otimes \ket{w_j}^B$ are formed for two vectors $\ket{v_i}^A$ and $\ket{w_j}^B$ is a bit more technical, and we refer to \cite{Halmos} for more information on this subject. The tensor product for complex vector spaces satisfies the following three properties:
\singlespacing
\begin{enumerate}
  \item For any $z \in \mathbb{C}$ and arbitrary vectors $\ket{v}^A$ of $A$ and $\ket{w}^B$ of $B$, \[z(\ket{v}^A \otimes \ket{w}^B) = (z \ket{v}^A) \otimes \ket{w}^B = \ket{v}^A \otimes (z\ket{w}^B).\]
  \item For arbitrary vectors $\ket{v_1}^A$ and $\ket{v_2}^A$ in $A$ and $\ket{w}^B$ in $B$, \[(\ket{v_1}^A + \ket{v_2}^A) \otimes \ket{w}^B = \ket{v_1}^A \otimes \ket{w}^B + \ket{v_2}^A \otimes \ket{w}^B. \]
 \item For arbitrary vectors $\ket{v}^A$ in $A$ and $\ket{w_1}^B$ and $\ket{w_2}^B$ in $B$, \[ \ket{v}^A \otimes (\ket{w_1}^B  + \ket{w_2}^B ) = \ket{v}^A \otimes \ket{w_1}^B + \ket{v}^A \otimes \ket{w_2}^B. \]
\end{enumerate}
\onehalfspacing
As an example, for the two Hilbert spaces $A:=\mathbb{C}^n$ and $B:=\mathbb{C}^m$, the tensor product of the two vectors $\ket{\psi}^A = \sum^{n}_{i=1} \alpha_i \ket{i}^A$ and $\ket{\phi}^B = \sum^{m}_{j=1} \beta_j \ket{j}^B$ is given by
 \begin{equation*}
   \ket{\psi}^A \otimes \ket{\phi}^B = \sum^{n}_{i=1}\sum^{m}_{j=1} \alpha_i \beta_j \ket{ij}^{AB},
 \end{equation*}
where we have written $\ket{ij}^{AB}$ for the tensor product $\ket{i}^A \otimes \ket{j}^B$. This shorthand notation will often be used in the following chapters.

Vectors of a Hilbert space $A$ can be transformed via linear operators $L : A \rightarrow B$. The image of $L$ is defined as
\begin{equation*}
\mbox{im } L := \{L \ket{v} : \ket{v} \in A\}.
\end{equation*}
It is a subspace of $B$ and its dimension is called the \textit{rank} of $L$. 
The set of all linear operators $L: A \rightarrow B$ is denoted by $\cL(A,B)$. 
For linear operators acting from $A$ to itself, we use the shorthand notation $\cL(A)$. Given any basis $\{v_1,v_2,\ldots, v_n\}$ of a Hilbert space $A$, the trace of an operator $L \in \cL(A)$ is defined as
\begin{equation*}
\Tr(L) := \sum^n_{i=1} \bra{v_i} A \ket{v_i}.
\end{equation*}

Several classes of linear operators will be of interest to us. The first one is the set of hermitian operators acting on the Hilbert space $A$.
Given a linear operator $L : A \rightarrow A$, the hermitian conjugate (adjoint) $L^{\dag}$ of $L$ is the unique operator such that 
for all vectors $\ket{v}^A, \ket{w}^A \in A$,
\begin{equation*}
 \langle v | Lw\rangle = \langle L^{\dag} v | w\rangle.
\end{equation*}
An operator $H$ whose hermitian conjugate is $H$ is known as hermitian or self-adjoint. In general, for two operators $A$ and $B$, 
we have $(AB)^{\dag} = B^{\dag} A^{\dag}$. By convention, we also define $\ket{v}^{\dag} := \bra{v}$. 
General hermitian operators can be written elegantly via the spectral decomposition theorem.

\begin{theorem}[Spectral decomposition]
Let $H$ be an hermitian operator acting on a Hilbert space $A$. Then, there exists an orthonormal basis $\{\ket{e_i}^A\}^{d_A}_{i=1}$ of 
$A$ such that $H$ is diagonal with respect to this basis:
\begin{equation*}
   H  = \sum^{d_A}_{i=1} \lambda_i \braket{e_i}^A,
\end{equation*}
where all the eigenvalues $\lambda_i$ of $H$ are real numbers.
\end{theorem}

The image of $H$ is spanned by all the eigenvectors $\ket{e_i}^A$ with non-zero eigenvalues. It is also called the \textit{support} of $H$.

An important subclass of hermitian operators are projection operators $P: A \rightarrow W$. These are hermitian operators which are also idempotent:
 \begin{equation*}
    P^{\dag} = P \quad \text{and} \quad P^2 = P.
 \end{equation*}
Given a $d_W$-dimensional subspace $W$ of a Hilbert space $A$ and an orthonormal basis $\{\ket{i}^{d_W}_{i=1}\}$ of $W$, 
the projector onto the subspace $W$ is defined as
\begin{equation*}
  P = \sum^{d_W}_{i=1} \braket{i}^W.
\end{equation*}
The orthogonal complement of $P$ is given by $Q:=I^A-P$, where $I^A$ is the identity operator on $A$. 
For a vector $\ket{\psi}^A$ of $A$, the operator $\ket{\psi}\bra{\psi}^A$, often written simply as $\psi^A$, 
is the projector onto the 1-dimensional subspace spanned by the vector $\ket{\psi}^A$. 
Projectors will be used later on to describe the process of measuring a physical system.

Another subclass of hermitian operators that we will frequently use are the positive semidefinite operators. 
An operator $X$ on $A$ is positive-semidefinite if for all vectors $\ket{v}^A \in A$, the inner product $\bra{v} X \ket{v}$ is a real and non-negative number. We will often drop the word ``semidefinite'' and refer to $X$ simply as a positive operator. 
From the spectral decomposition theorem, the eigenvalues and the trace of a positive operator $X$ must be non-negative real numbers. 
Given hermitian operators $H$ and $K$ acting on the space $A$, we say that $H \leq K$ if $K-H$ is positive. 
This defines a partial ordering on the set of hermitian operators.
The class of positive semidefinite operators of trace one have a special importance in quantum mechanics, and we refer to them as \textit{density operators}. We will see shortly that they capture the statistical behavior of a quantum system.

One very important class of linear operators we will be concerned with are the unitary operators. 
An operator $U_A$ in ${\cal L}(A)$ is said to be unitary if $U_A^{\dag}U_A=I_A$. 
This also implies $U_AU_A^{\dag}= I^A$. Unitary operators preserve lengths and angles between vectors. For any pair of vectors $\ket{v_1}^A$ and $\ket{v_2}^A$ of $A$, we have
\begin{equation*}
  \begin{split}
     \| U_A \ket{v_1}^A \|_1 &= \| \ket{v_1}^A \|_1. \\
     \langle U_A\ket{v_1}^A, U_A\ket{v_2}^A \rangle &= \langle \ket{v_1}^A, \ket{v_2}^A \rangle. \\
  \end{split}
\end{equation*}
The last line can be easily seen to hold by rewriting the inner product $\langle U_A\ket{v_1}^A, U_A\ket{v_2}^A \rangle$ in braket notation as $\bra{v_1} U_A^{\dag} U_A \ket{v_2}$. The result then follows by substituting $U_A^{\dag}U_A$ with the identity operator. Unitary operators can be used to construct new orthonormal bases: given an orthonormal basis $\{\ket{v_i}^A\}^{d_A}_{i=1}$ of $A$, let $\ket{w_i}^A = U_A \ket{v_i}^A$. Then $\{\ket{w_i}^A\}^{d_A}_{i=1}$ is an orthonormal basis of $A$.

We can generalize the class of unitary operators by considering input and output spaces of different dimensions. 
An \textit{isometry} $F: A \rightarrow B$ for two Hilbert spaces $A$ and $B$ is a linear operator which satisfies
\begin{equation*}
    d_B(F\ket{a},F\ket{b}) = d_A(\ket{a},\ket{b}),
\end{equation*}
for any two vectors $\ket{a},\ket{b}$ in $A$. A unitary operator $U_A$ on a Hilbert space $A$ 
is a special case of an isometry where the function $F$ is also surjective (i.e the image of $F$ is $A$). 
For two Hilbert spaces $A$ and $B$ of different sizes, with $d_A \leq d_B$, we can extend any unitary operator $U_A$ 
to an isometry $V: A \rightarrow B$ by identifying a subspace $B'$ of dimension $d_{B'}=d_A$ with $A$ and 
letting $V \ket{v}^A = \pi(U\ket{v}^A)$, where $\pi: A \rightarrow B'$ is an isomorphic map from $A$ to $B$. 
Define the kernel ($\mathrm{ker}(F)$) of $F$ to be the subspace of all vectors in $A$ which map to the zero element of $B$ 
under the function $F$. A function $W:A \rightarrow B$ is a \textit{partial isometry} if, for any two vectors $\ket{a},\ket{b}$ 
in the orthogonal complement of $\mathrm{ker}(W)$, we have $d_B(W\ket{a},W\ket{b}) = d_A(\ket{a},\ket{b})$. 
Partial isometries appear in chapter 3 to model a \textit{random coding} strategy in the context of state merging.

Finally, given two linear operators $R$ and $S$ acting on the spaces $A$ and $B$ respectively, the Kronecker product $R \otimes S$ is the matrix
\begin{equation}
 R \otimes S := \left ( \begin{array}{ccc}
                  R_{11}S & \cdots & R_{1d_A}S \\
                   \vdots & \ddots & \vdots \\
                  R_{d_A1}S & \cdots & R_{d_Ad_A}S 
                \end{array} \right ),
\end{equation}
where $R$ and $S$ are the matrix representations of the operators $R$ and $S$. 

\subsection{Quantum mechanics}
Unlike the theory of relativity, which was the work of a single individual \cite{Einstein}, the theory of quantum mechanics as we know it today was the culmination of years of work from various physicists during the first half of the twentieth century. The failure of classical physics to explain observed phenomena such as the ultraviolet catastrophe and the photoelectric effect forced physicists to reconsider the nature of the physical world. A new set of rules was required for making accurate predictions on the outcome of any scientific experiment. After a relatively long process of trial and error, a mathematical formulation of quantum mechanics was made precise and found to successfully predict all scientific experiments known at the time. Since then, no known experiment has contradicted the predictions of quantum mechanics. Any physical theory based on the structure of quantum mechanics must obey the following four basic postulates:
\begin{description}
  \item[Postulate 1] Associated with any physical system $A$ is a Hilbert space $A$ called the state space. The system is completely described by its density operator $\psi^A$, which acts on the state space of the system $A$.

  \item[Postulate 2] The evolution of a closed quantum system $A$ is described by a unitary transformation $U$. That is, if $\psi^A_{t_1}$ and $\psi^A_{t_2}$ are the density operators of the system $A$ at times $t_1$ and $t_2$, they are related by a unitary operator $U_A$ which depends only on $t_1$ and $t_2$:
          \begin{equation*}
              \psi^A_{t_2} = U_A \psi^A_{t_1} U_A^{\dag}.
          \end{equation*}
  \item[Postulate 3] Quantum measurements realized on a physical system $A$ are described by a set of linear operators $\{M_m\}$ acting on the state space of $A$. The probability of obtaining outcome $m$ is given by
      \begin{equation*}
         p(m) = \Tr(M^{\dag}_m M_m \psi^A),
      \end{equation*}
where $\psi^A$ is the density operator describing the system $A$. After obtaining outcome $m$, the system is described by the density operator
 \begin{equation*}
 \frac{M_m \psi^A M_m^{\dag}}{\Tr(M_m \psi^A M_m^{\dag})}.
 \end{equation*}
 The operators $\{M_m\}$ satisfy the completeness equation,
    \begin{equation}\label{eq:Post3}
        \sum_m M^{\dag}_m M_m = I^A.
    \end{equation}
  \item[Postulate 4] The state space of a composite physical system is the tensor product $A_1 \otimes A_2 \otimes \ldots \otimes A_n$ of the state spaces $A_1, A_2, \ldots, A_n$ of the component physical systems. Moreover, if each system $A_i$ is described by the density operator $\psi^{A_i}$, the density operator of the system $A_1A_2\ldots A_n$ is given by $\psi^{A_1}\otimes \psi^{A_2} \otimes \ldots \otimes \psi^{A_n}$.
\end{description}
Other equivalent formulations of quantum mechanics exist (see, for instance, \cite{Cohen,Peres,Griffith}). In the context of quantum information theory, however, the previous formulation in terms of density operators will be very useful as we will often deal with composite systems in an unknown state. The density operator gives a complete mathematical description of the statistical behavior of its associated system. 
From the spectral decomposition, any density operator $\psi^A$ can be written as a convex combination of normalized eigenstates:
\begin{equation*}
  \psi^A = \sum^d_{i=1} \lambda_i \braket{\psi_i}^A,
\end{equation*}
where $\lambda_i > 0$ for $1 \leq i \leq d$. If $d=1$, the system is in the \textit{pure} state $\ket{\psi}^A$, often written simply as $\psi^A$. We will often use the term ``state'' to refer to the density operator $\braket{\psi}^A$ as opposed to the vector $\ket{\psi}^A$ of the state space. If $d > 1$, the system is said to be in the \textit{mixed} state $\psi^A$. In such a case, different \textit{ensembles} $\{p_i, \ket{\psi_i}^A\}$ of pure states may realize the density operator $\psi^A = \sum_i p_i \braket{\psi_i}^A$. As an example, consider the ensembles $\{1/2, \ket{0}^A, 1/2 , \ket{1}^A\}$ and $\{1/2, H\ket{0}^A, 1/2 , H\ket{1}^A \}$, where $H$ is the Hadamard operation:
\begin{equation*}
 \begin{split}
H\ket{0}^A &:= \frac{1}{\sqrt{2}} (\ket{0}^A+\ket{1}^A), \\
H\ket{1}^A &:= \frac{1}{\sqrt{2}} (\ket{0}^A-\ket{1}^A). \\
\end{split}
\end{equation*}
Both ensembles realize the same density operator
\begin{equation*}
\psi^A = 1/2 \braket{0}^A + 1/2\braket{1}^A = I^A/2 = 1/2 ( H \braket{0}^A H^{\dag} + H\braket{1}^AH^{\dag}),
\end{equation*}
which is called a \textit{maximally mixed} state $\tau^A := I^A/2$ of dimension $d_A$.

When realizing a quantum measurement on a system $A$, we may only be interested in the outcome of this measurement (for instance, to distinguish between two possible states $\psi_1^A$ and $\psi_2^A$ of the system). For a set of measurement operators $\{M_m\}$, let
\begin{equation*}
E_m = M_m^{\dag} M_m \quad \forall m.
\end{equation*}
Then, $\{E_m\}$ are positive operators as for any $\ket{v}^A$, we have $\bra{v}E_m\ket{v} \geq 0$. According to Postulate 3, the probability of obtaining outcome $m$ is given by $\Tr(M^{\dag}_m M_m \psi^A)$. Replacing $M^{\dag}_mM_m$ by $E_m$, the probability $p(m)$ is equal to $\Tr(E_m \psi^A)$, and from the completeness equation eq.~(\ref{eq:Post3}), we have
\begin{equation*}
\sum_m E_m = I^A.
\end{equation*}
The set $\{E_m\}$ is the \textit{POVM} (Positive Operator Valued Measurement) associated with the measurement. Conversely, let $\{E_m\}$ be a set of positive operators acting on $A$ which satisfy $\sum_m E_m = I^A$. Writing $E_m = \sum^d_{i=1} \mu_i \braket{e_i}^A$ using the spectral decomposition, we can define measurement operators $M_m := \sqrt{E_m}$, where \[\sqrt{E_m} := \sum^d_{i=1} \sqrt{\mu_i} \braket{e_i}^A.\] Then, the set $\{E_m\}$ is the POVM associated with the measurement described by the operators $\{M_m\}$. It is important to understand that measurements described via POVMs will generally not allow us to known the state of the system afterwards. This is due to the fact that, for a given POVM $\{E_m\}$, we can choose any set of unitaries $\{U_m\}$ and construct measurement operators $M_m = U_m \sqrt{E_m}$ which will describe a measurement with POVM $\{E_m\}$.

Suppose we have a composite system $AB$, whose state is described by its density operator $\psi^{AB}$. We can prescribe a ``state'' to the subsystem $A$ via the \textit{reduced} density operator:
\begin{equation*}
   \psi^A := \Tr_B \psi^{AB},
\end{equation*}
where $\Tr_B : A\otimes B \rightarrow A$ is a linear operator known as the \textit{partial trace}, and is defined by
\begin{equation*}
  \Tr_B (\ket{x_1}\bra{x_2}^A \otimes \ket{y_1}\bra{y_2}^B) := \ket{x_1}\bra{x_2}^A \Tr(\ket{y_1}\bra{y_2}^B),
\end{equation*}
where $\ket{x_1}^A$ and $\ket{x_2}^A$ are any two vectors in the state space of $A$ and $\ket{y_1}^B$ and $\ket{y_2}^B$ are any two vectors in the state space of $B$. This gives a correct description of the statistical behavior of the subsystem $A$ in the sense that for any measurement ${M_m}$ on $A$, the outcome probabilities $p(m)$ computed using $\psi^A$ equal the probabilities $q(m)$ computed using the density operator $\psi^{AB}$ for the measurement ${M_m \otimes I_B}$. The partial trace is the unique function satisfying this property. For the tensor product state $\psi^{AB} = \sigma^A \otimes \phi^B$, the reduced state $\psi^A$ is equal to $\Tr_B \psi^{AB} = \sigma^A$.

\subsection{Schmidt decomposition and purifications}

For any pure state $\psi^{AB}$ of a composite system $AB$, we can always find orthonormal states $\ket{e_i}^A$ for the system $A$ and orthonormal states $\ket{f_i}^B$ for the system $B$ such that $\psi^{AB}$ can be written as a superposition of the states $\ket{e_if_i}^{AB}$. This is the Schmidt decomposition theorem:
\begin{theorem}[Schmidt decomposition]\label{thm:schmidtdecomp}
Suppose $\ket{\psi}^{AB}$ is a pure state of a composite system $AB$. Then, there exist orthonormal states $\{\ket{e_i}^A\}^d_{i=1}$ for the system $A$ and orthonormal states $\{\ket{f_i}^B\}^d_{i=1}$ for the system $B$ such that
 \begin{equation*}
   \ket{\psi}^{AB} = \sum^d_{i=1} \lambda_i \ket{e_i}^A\ket{f_i}^B,
 \end{equation*}
 where the $\lambda_i$ are non-negative real numbers satisfying $\sum_i \lambda^2_i = 1$ known as Schmidt coefficients. The number of non-zero values $\lambda_i$ is called the Schmidt rank for the state $\ket{\psi}^{AB}$.
\end{theorem}
As a consequence of the Schmidt decomposition, the reduced density operators $\psi^A$ and $\psi^B$ of the state $\psi^{AB}$ share the same spectrum:
\begin{equation*}
\begin{split}
 \psi^A &= \sum^d_{i=1} \lambda^2_i \braket{e_i}^A \\
 \psi^B &= \sum^d_{i=1} \lambda^2_i \braket{f_i}^B. \\
\end{split}
\end{equation*}
For a proof of the Schmidt decomposition theorem, see Nielsen and Chuang~\cite{Nielsen}.

For a mixed state $\psi^A$, it is always possible to introduce another system $R$, called a purification system, and find a pure state $\ket{\psi}^{AR}$ of the composite system $AR$ such that $\psi^{A} = \Tr_R \psi^{AR}$. To see this, write $\psi^{A}$ as $\sum^{d_A}_{i=1} \lambda_i \braket{e_i}^A$ using the spectral decomposition theorem. Let $R$ have dimension $d_R := d_A$, with orthonormal basis states $\ket{e_i}^R$, and define the pure state
\begin{equation*}
  \ket{\psi}^{AR} = \sum^{d_A}_{i=1} \sqrt{\lambda_i} \ket{e_i}^A \ket{e_i}^R.
\end{equation*}
Then, we have
\begin{equation*}
 \begin{split}
   \Tr_R(\braket{\psi}^{AR}) &= \sum_{ij} \sqrt{\lambda_i \lambda_j} \ket{e_i}\bra{e_j}^A \Tr(\ket{e_i}\bra{e_j}^R) \\
     &= \sum_{ij} \sqrt{\lambda_i \lambda_j} \ket{e_i}\bra{e_j}^A \delta_{i,j} \\
    &= \sum^{d_A}_{i=1} \lambda_i \braket{e_i}^A \\
    &= \psi^A, \\
 \end{split}
\end{equation*}
where $\delta_{i,j}$ is the Kronecker symbol. It is always possible to purify a state $\psi^A$ in more than one way. However, for any two purifications $\ket{\psi}^{AR_1}$ and $\ket{\phi}^{AR_2}$ of a state $\psi^A$, with purification systems $R_1$ and $R_2$, there exists a partial isometry $U : R_1 \rightarrow R_2$ taking the $R_1$ system to the $R_2$ system such that
\begin{equation} \label{eq:isometrylink}
  \ket{\phi}^{AR_2} = (I^{A} \otimes U^{R_1}) \ket{\psi}^{AR_1} .
\end{equation}

\subsection{Separable states and maximally entangled states}
A state $\psi^{AB}$ of a composite system $AB$ is called separable if it can be written as a convex combination of tensor products of density operators  $\{\rho^A_i\}$ of the subsystem $A$ and density operators $\{\sigma^B_i\}$ of the subsystem $B$:
\begin{equation*}
 \psi^{AB} = \sum_i p_i \rho^A_i \otimes \sigma^B_i,
\end{equation*}
where $\sum_i p_i = 1$. Since the partial trace is a linear operator, we get from the previous equation that $\psi^A = \sum_i p_i \rho^A_i$ and $\psi^B = \sum_i p_i \sigma^B_i$. A state $\psi^{AB}$ which is not separable is called \textit{entangled}. For a pure state $\ket{\psi}^{AB}$, the previous definition of separability implies that $\ket{\psi}^{AB}$ is separable if and only if there exist vectors $\ket{x}^A$ and $\ket{y}^B$ of $A$ and $B$ such that $\ket{\psi}^{AB} = \ket{x}^A\ket{y}^B$. Alternatively, a pure state $\ket{\psi}^{AB}$ is a product state if and only if its Schmidt rank is $1$. Therefore, any pure entangled state for a composite system $AB$ whose subsystems $A$ and $B$ both have dimensions two must have a Schmidt decomposition with two non-zero values $\lambda_1$ and $\lambda_2$.
Examples of pure entangled states for such a composite system are the Bell states:
\begin{eqnarray*}
  \ket{\Phi_{+}}^{AB} &:=& \frac{1}{\sqrt{2}}(\ket{00} + \ket{11}), \\
  \ket{\Phi_{-}}^{AB} &:=& \frac{1}{\sqrt{2}}(\ket{00} - \ket{11}), \\
  \ket{\Psi_{+}}^{AB} &:=& \frac{1}{\sqrt{2}}(\ket{01} + \ket{10}), \\
  \ket{\Psi_{-}}^{AB} &:=& \frac{1}{\sqrt{2}}(\ket{01} - \ket{10}),
\end{eqnarray*}
where the last state is known as the \textit{singlet} state or an \textit{EPR pair} \cite{EPR}. These states form a basis of the tensor product space $A \otimes B$. In general, for two systems $A$ and $B$ with $d_A \leq d_B$, a \textit{maximally entangled} state of dimension $d_A$ is defined as:
\begin{equation*}
 \ket{\Phi^d} := \frac{1}{\sqrt{d_A}} \sum^{d_A}_{m=1} \ket{m}^A \otimes \ket{m}^B,
\end{equation*}
where $\{ \ket{m}^A\}^{d_A}_{i=1}$ is an orthonormal basis for $A$ and $\{\ket{m}^B\}^{d_A}_{i=1}$ is a family of orthonormal vectors on $B$. The amount of bipartite entanglement in a state $\psi^{AB}$ is measured in \textit{ebits}, with Bell states having an amount of entanglement equal to one \textit{ebit}. A maximally entangled state of dimension $d_A$ is said to have $\log(d_A)$ ebits.

Other examples of entangled states are obtained via mixtures of Bell states. The family of Werner states \cite{Werner} is defined as
\begin{equation*}
  W_F = F \braket{\Psi_{-}} + \frac{(1-F)}{3} (\braket{\Phi_{+}} + \braket{\Phi_{-}} + \braket{\Psi_{+}}),
\end{equation*}
where $0 \leq F \leq 1$. The value of $F$ for the Werner state $W_F$ is equal to $\bra{\Psi_{-}}W_F\ket{\Psi_{-}}$, which is the \textit{entanglement fidelity} of $W_F$ relative to the singlet state. For an arbitrary bipartite state $\psi^{AB}$, the entanglement fidelity of $\psi^{AB}$ relative to the singlet state is defined as:
\begin{equation*}
F^2(\psi^{AB},\braket{\Psi_{-}}) := \bra{\Psi_{-}}\psi^{AB}\ket{\Psi_{-}}.
\end{equation*}
The Werner state $W_F$ is separable for $F \leq 1/2$, and entangled for $F > 1/2$. Werner states facilitate the analysis (see, for instance, Bennett et al. \cite{Bennett}) and construction of entanglement distillation protocols: the process of converting a large number of copies of an entangled state $\psi^{AB}$ to a smaller number of highly entangled states such as EPR pairs.
\section{Quantum information}
\subsection{von Neumann entropy}
The Shannon entropy $H(X)$ \cite{Shannon} of a random variable $X$ yielding outcome $x$ with probability $p_x$ is defined as:
\begin{equation*}
  H(p_x) = H(X) := - \sum_x p_x \log p_x,
\end{equation*}
where the logarithm is taken base 2. This quantity is always non-negative, with $H(X)=0$ if and only if the random variable $X$ yields a definite outcome $x$ with $p_x=1$. It takes a maximal value of $\log d$ for a random variable $X$ generating $d$ possible outcomes with equal probabilities.

For data communication, the Shannon entropy is the theoretical limit at which information produced by a source can be compressed, transmitted and recovered in a lossless way. Its basic unit is the bit. A binary random variable $X$ taking values 0 and 1 with probability one-half is said to have 1 bit of entropy.

For a density operator $\psi^A$, with spectral decomposition $\psi^A = \sum^d_{i=1} \lambda_i \braket{e_i}^A$, we define its \textit{von Neumann entropy} $S(A)_{\psi}$ as
\begin{equation*}
  S(A)_{\psi} = - \Tr( \psi^A \log \psi^A),
\end{equation*}
where $\log \psi^{A} = \sum^d_{i=1} \log(\lambda_i) \braket{e_i}^A$ and the logarithm is taken base two (we define $0\log(0) := 0$). The spectral decomposition of $\psi^A$ allows us to relate the von Neumann entropy to the Shannon entropy:
\begin{equation*}
  S(A)_{\psi} = -\sum^d_{i=1} \lambda_i \log \lambda_i = H(X),
\end{equation*}
where $X$ is a random variable yielding outcome $i$ with probability $\lambda_i$. The von Neumann entropy is non-negative, and is zero if and only if the state is pure. For the maximally mixed state $I^A/d$, we have $S(A)_{I^A/d} = \log d$.

The original motivation for the von Neumann entropy did not come from an information-theoretical context, unlike the Shannon entropy. It was actually an attempt to extend a thermodynamical concept, the Gibbs entropy (see \cite{Fermi} for an introduction to thermodynamics), to the quantum setting. The extension of Shannon's work to the quantum regime happened several years later, beginning with the work of Ohya and Petz \cite{Ohya}. The quantum version of Shannon's noiseless coding theorem was obtained by Benjamin Schumacher \cite{Coding}, who coined the term \textit{qubit}, the basic unit of quantum information, and characterized the von Neumann entropy as the optimal rate at which quantum information produced by a source can be compressed, transmitted and recovered by a receiver in a lossless way.

A qubit is a 2-dimensional quantum system $A$. A composite system in any of the Bell states constitute two qubits. The reduced state of either subsystem is in the maximally mixed state, with entropy $S(A)_{I^A/2} = S(B)_{I^B/2} = 1$. If a state $\psi^{AB}$ of a composite system $AB$ is pure, we have $S(A)_{\psi} = S(B)_{\psi}$. Given a state $\psi^A = \sum_i p_i \psi^A_i$ written as a convex combination of other states $\psi^A_i$, the von Neumann entropy is a \textit{concave} function of its inputs $\psi^A_i$:
\begin{equation*}
  S(A)_{\psi} \geq \sum_i p_i S(A)_{\psi^A_i},
\end{equation*}
where equality holds iff all the states $\psi^A_i$, for which $p_i > 0$, are identical.  The von Neumann entropy is invariant under unitary transformations on the state $\psi^A$:
\begin{equation*}
  S(A)_{\psi} = S(A)_{U\psi U^{\dag}}.
\end{equation*}
For a state $\psi^{XA}$ of the form $\psi^{XA} = \sum_i p_i \braket{i}^X \otimes \psi^A_i$, we have
\begin{equation*}
  S(XA)_{\psi} = H(X) + \sum_i p_i S(A)_{\psi^A_i},
\end{equation*}
where $H(X) = -\sum_i p_i \log p_i$. We refer to the state $\psi^{XA}$ as a \textit{classical-quantum} (cq-)state with classical system $X$.

For a tensor product state $\psi^{AB} = \rho^A \otimes \sigma^B$, we have $S(AB)_{\psi} = S(A)_{\rho} + S(B)_{\sigma}$. The von Neumann entropy of a joint state $\psi^{AB}$ satisfies the following inequality, known as subadditivity (see \cite{Nielsen} for a proof):
\begin{equation}\label{eq:subadditivity}
  S(AB)_{\psi} \leq S(A)_{\psi} + S(B)_{\psi},
\end{equation}
where $S(A)_{\psi}$ and $S(B)_{\psi}$ are the von Neumann entropies for the corresponding reduced density operators $\psi^A$ and $\psi^B$. Equality holds if and only if the state $\psi^{AB}$ can be written in the product form $\psi^{A} \otimes \psi^{B}$. Given a state $\psi^{AB}$ of a composite system $AB$, the \textit{conditional} von Neumann entropy $S(A|B)_{\psi}$ is defined as
\begin{equation*}
  S(A|B)_{\psi} = S(AB)_{\psi} - S(B)_{\psi}.
\end{equation*}
 Unlike the Shannon entropy $H(X|Y) = H(XY) - H(Y)$, the conditional von Neumann entropy can be negative. As an example, consider the singlet state $\ket{\Psi_{-}}^{AB}$. We have $S(AB)_{\psi}=0$ since the state is pure, and $S(B)_{\psi}=1$ since its reduced density operator is the maximally mixed state $I_B/2$. For a tripartite system $A \otimes B \otimes C$ in the state $\psi^{ABC}$, the conditional von Neumann entropy $S(A|BC)$ is bounded above by $S(A|B)$:
\begin{equation*}
   S(A|BC)_{\psi} \leq S(A|B)_{\psi}.
\end{equation*}
This is known as the strong subadditivity property of the von Neumann entropy. First conjectured by Lanford and Robinson \cite{StrongSub}, a proof of this inequality was obtained by Lieb and Ruskai in \cite{Lieb3}. A simple operational proof of this inequality also follows from quantum state merging \cite{merge,SW-Nature}.
\subsection{Quantum operations, instruments, LOCC}
A \textit{quantum channel} is a medium for carrying quantum information from one location to another. To motivate its mathematical description, let's consider a collection of linearly polarized photons, each prepared in some polarization state $\psi_i^A$. As part of a protocol implemented by two spatially separated parties (for instance, the BB84 cryptographic protocol of Bennett and Brassard \cite{BB84}), we need to send the photons through a fiber optic channel to another laboratory $B$. Fiber optics, unfortunately, are not a perfect medium for information transmission. Photons sent through a fiber are subject to attenuation, also known as transmission loss, and dispersion effects. These phenomena will have highly undesirable consequences on the polarization state of the photons, and may prevent the detection of photons by the receiver's apparatus. We can model the process of photons passing through a fiber as two systems which interact for some period of time. If we assume the photons and the fiber form a closed system, their interaction can be described by a unitary operator $U: AE \rightarrow BE$, where $E$ represents the fiber optic system (also called the environment). The exact specification of the unitary will depend on the characteristics of the fiber. We can assume, prior to transmission, that the $AE$ system is in a product state $\psi^A \otimes \braket{0}^E$. A fiber optic channel is then represented by a map ${\cal N}:A \rightarrow B$ such that
     \begin{equation}\label{eq:qop}
   \tilde{\psi}^B = {\cal N}(\psi^A) := \Tr_E (U (\psi^A \otimes \braket{0}^E) U^{\dag}).
   \end{equation}
This is known as the Stinespring form for the channel ${\cal N}$. We can re-express the previous formula by introducing linear operators $E_i : A \rightarrow B$, defined as:
\begin{equation*}
  E_i := \bra{i}^E U \ket{0}^E,
\end{equation*}
where $\{\ket{i}^E\}$ is a basis of the environment system $E$. Putting these into eq.~(\ref{eq:qop}), we have
\begin{equation*}
  \tilde{\psi}^B = {\cal N}(\psi^{A}) = \sum_i E_i \psi^A E^{\dag}_i.
\end{equation*}
Alternatively, one could start from a set of linear operators $E_i: A \rightarrow B$ which satisfy
\begin{equation}\label{eq:compl}
\sum_i E^{\dag}_i E_i \leq I^A,
\end{equation}
known as the completeness relation for the operators $\{E_i\}$, and define the \textit{quantum operation}
\begin{equation}\label{eq:osr}
  {\cal E}(\psi^A) := \sum_i E_i \psi^A E^{\dag}_i.
\end{equation}
That the output of this operation is a sub-normalized (i.e $\Tr \rho \leq 1$) density operator follows from the completeness relation. Quantum channels can be regarded as a quantum operation whose intent is to carry quantum information. As an example, consider the following channel ${\cal E}:A \rightarrow B$ for transmitting a qubit in the state $\ket{\psi}^{A}=\alpha_0 \ket{0}^A+\alpha_1\ket{1}^A$:
\begin{equation}\label{eq:photonloss}
{\cal E}(\braket{\psi}^A) := (1-p) \braket{\psi}^B + p\braket{2}^B, 
\end{equation}
where $0 \leq p \leq 1$ and $\braket{2}^B$ is a state orthogonal to $\ket{0}^B$ and $\ket{1}^B$. This is a simple model for photon loss. The transmitted photon is either perfectly detected with probability $1-p$ or replaced by some ``erasure'' state with probability $p$.  

Eq.~(\ref{eq:osr}) is known as the operator-sum representation of a quantum operation. It allows for more general forms of quantum operations than the previous formulation of quantum operations in terms of interacting systems. The operators $E_i$ are called \textit{Kraus} operators. A quantum operation which has a non-trace-preserving output corresponds to a process which occurs with probability $\Tr({\cal E}(\psi^A))$ (for instance, a specific measurement outcome).

For a composite system $AR$, a quantum operation acting on the density operator $\psi^{AR}$ should leave the composite system $AR$ in a density operator (up to some normalization) after the operation is performed. This is called the completely positive requirement:
\begin{equation*}
 \tilde{\psi}^{AR} := (I^R \otimes {\cal E})(\psi^{AR}) \geq 0
\end{equation*}
for any extra system $R$ of arbitrary dimension. Quantum operations defined via the operator-sum representation satisfy this property (again, see \cite{Nielsen} for a proof of this fact).

We can describe measurements as a set of non trace-preserving quantum operations $\{{\cal E}_i\}$. To illustrate this, let's consider a special kind of generalized measurement, called a projective measurement, described by a set of orthogonal projectors $\{P_i\}$ (i.e $P_i P_j = \delta_{i,j} P_i$) satisfying:
\begin{equation}\label{eq:proj}
   \sum_i P_i = I^A.
\end{equation} The probability of obtaining outcome $i$ for a system $A$ in the state $\psi^A$ is then given by $\Tr (P_i \psi^A)$. Alternatively, we could have described projective measurement as the set of quantum operations ${\cal E}_i(\psi^A) := P_i \psi^A P^{\dag}_i$. That these are valid operations follows from the completeness equation eq.~(\ref{eq:proj}). We obtain the measurement outcome $i$ with probability $\Tr ({\cal E}_i(\psi^A))$ since
\begin{equation*}
\begin{split}
\Tr ({\cal E}_i(\psi^A)) &= \Tr (P_i \psi^A P_i) \\
&= \Tr (P_i \psi^A ).
\end{split}
\end{equation*}
The normalized state after obtaining outcome $i$ is given by $\frac{{\cal E}_i(\psi^A)}{\Tr ({\cal E}_i(\psi^A))}$. Notice that each quantum operation ${\cal E}_i$ is described using only one Kraus operator and that the sum $\sum_i {\cal E}_i$ is a trace-preserving quantum operation.

Generalized measurements can be described similarly using a set of quantum operations $\{{\cal E}_i\}$, with ${\cal E}_i(\psi^A)=M_i \psi^A M^{\dag}_i$. We can generalize the previous examples by considering quantum operations with operator-sum representations containing more than one Kraus operator. We call an \textit{instrument} ${\cal I} := \{{\cal E}_i\}$ \cite{EBDavies} a set of completely positive maps (i.e non trace-preserving quantum operation) which sums to a completely positive and trace preserving map. The elements of the set $\{{\cal E}_i\}$ are the instrument components. Instruments can be used in protocols when one party needs to perform a measurement followed by an isometry conditioned on the classical outcome of the measurement (see Section \ref{sec:random-meas} in Chapter 3).

Suppose two parties share a bipartite system $AB$ in the state $\psi^{AB}$, but have access only to a classical communication channel (i.e., a channel which transmits only classical data). The parties can send information by performing a finite number of rounds of local measurements (or other local processing such as instruments) and classical communication of the outcomes between them. These types of operations are a special class of quantum operations known as \textit{LOCC} (Local Operations and Classical Communication). They can be written elegantly in the operator sum representation as:
\begin{equation}
  {\cal E}(\psi^{AB}) = \sum_i (X^A_i \otimes Y^B_i )(\psi^{AB})(X^A_i \otimes Y_i^B)^{\dag}. \label{eq:locc}
\end{equation}
Note, however, that the class of operations which can be written in the previous form includes operations which are not in the LOCC class. Quantum operations satisfying eq.~(\ref{eq:locc}) are called separable. Teleportation and distillation protocols are examples of tasks which are performed using LOCC operations.

\subsection{Distance measures}
Given two probability distributions $p(x)$ and $q(x)$ over the same index set ${\cal X}$, the total variation distance between $p(x)$ and $q(x)$ is defined as
\begin{equation*}
  D(p, q) = \frac{1}{2} \sum_{x \in {\cal X}} |p(x) - q(x)|,
\end{equation*}
where $|x|$ is the absolute value of $x$. The total variation distance $D(p,q)$ is a metric: it is non-negative for any distributions $p,q$, it is symmetric in its arguments ($D(p,q) = D(q,p)$), and it satisfies the triangle inequality:
\begin{equation*}
 D(p,q) \leq D(p,r) + D(r,q),
\end{equation*}
where $p,q,$ and $r$ are arbitrary probability distributions over the same index set.

The \textit{trace distance} between two density operators $\rho^A$ and $\sigma^A$ is given by
\begin{equation*}
 D(\rho^A, \sigma^A) := \frac{1}{2} \| \rho^A - \sigma^A \|_1,
\end{equation*}
with the trace norm $\|X\|_1$ of an operator $X$ defined as
\begin{equation*}
\|X\|_1 = \Tr \sqrt{X^{\dag}X}.
\end{equation*}
Here, the $\sqrt{X}$ function for a positive operator $X$ is defined via the spectral decomposition of $X$:
\begin{equation*}
\sqrt{X} := \sum_i \sqrt{\lambda_i} \braket{e_i}^A.
\end{equation*}
If $\rho^A = \sum_i p_i \braket{i}^A$ and $\sigma^A =\sum_i q_i \braket{i}^A$, it is easy to see that the trace distance $D(\rho^A, \sigma^A)$ reduces to the total variation distance $D(p,q)$. The trace distance $D(\rho^A, \sigma^A)$ extends the total variation distance by providing a measure of closeness for states which are not simultaneously diagonalizable. (Two states $\rho^A$ and $\sigma^A$ which are simultaneously diagonalizable can be written as $\rho^A =\sum_i \lambda_i \braket{e_i}^A$ and $\sigma^A = \sum_i \mu_i \braket{e_i}^A$ for a common set of eigenvectors $\{e_i\}.$) Symmetry and non-negativity of the trace distance follows easily from the definition of the trace norm. The triangle inequality
\begin{equation*}
D(\rho^A, \sigma^A) \leq D(\rho^A, \phi^A) + D(\phi^A, \sigma^A),
\end{equation*}
is also satisfied for any state $\phi^A$.

For any two orthogonal states $\rho^A$ and $\sigma^A$, the trace distance is maximized and is equal to $1$. The trace distance $D(\rho^A, \sigma^A)$ is equal to zero if and only if the states are the same. The trace distance $D(\rho^A,\sigma^A)$ is invariant under unitary operations performed on $\rho^A$ and $\sigma^A$:
\begin{equation*}
 D(\rho^A, \sigma^A) := D(U\rho^A U^{\dag}, U \sigma^A U^{\dag}).
\end{equation*}
The trace distance can only decrease under trace-preserving quantum operations (a property also known as monotonicity):
\begin{equation}\label{eq:distOp}
 D({\cal E}(\rho^A),{\cal E}(\sigma^A)) \leq D(\rho^A,\sigma^A).
\end{equation}

Another measure of closeness between two states $\rho^A$ and $\sigma^A$ is obtained via the \textit{fidelity} \cite{uhlmann:fid,jozsa:fid}:
\begin{equation*}
  F(\rho^A, \sigma^A) = \Tr \sqrt{ \sqrt{\rho^A} \sigma^A \sqrt{\rho^A}}.
\end{equation*}
This can be re-expressed using the trace norm as
\begin{equation*}
  F(\rho^A, \sigma^A) = \| \sqrt{\rho^A} \sqrt{\sigma^A} \|_1.
\end{equation*}
The fidelity between two states $\rho^A$ and $\sigma^A$ is equal to one if and only if the states are the same. It is always non-negative and is zero for any two orthogonal states $\rho^A$ and $\sigma^A$.  It is also invariant under unitary operations performed on $\rho^A$ and $\sigma^A$:
\begin{equation*}
  F(\rho^A, \sigma^A) = F(U\sigma^A U^{\dag}, U \sigma^A U^{\dag}).
\end{equation*}
The trace distance is bounded by the fidelity (see Fuchs and van de Graaf \cite{FuchsVandegraaf:fidelity} for a proof) in the following way:
\begin{equation}
 \label{Lemma:relation}
  1- F(\rho, \sigma) \leq D(\rho,\sigma) \leq \sqrt{1-F^2(\rho,\sigma)}.
\end{equation}
The fidelity is increasing under trace-preserving quantum operations:
\begin{equation*}
 F({\cal E}(\rho^A),{\cal E}(\sigma^A)) \geq F(\rho^A,\sigma^A).
\end{equation*}
An incredible theorem, known as Uhlmann's theorem, relates the fidelity $F(\rho^A,\sigma^A)$ to a maximization over purifications of $\rho^A$ and $\sigma^A$:
\begin{theorem}[Ulhmann's theorem \cite{uhlmann:fid}]
Let $\rho^A$ and $\sigma^A$ be states of a system $A$. Introduce a second system $R$ which is a ``copy'' of $A$. Then,
\begin{equation*}
  F(\rho^A, \sigma^A) = \max_{\ket{\psi},\ket{\phi}} |\langle \psi | \phi \rangle |,
\end{equation*}
where the maximization is over all purifications $\ket{\psi}^{AR}$ of $\rho^A$ and $\ket{\phi}^{AR}$ of $\sigma^A$.
\end{theorem}
A proof of this theorem can be found in \cite{Nielsen}. We will use in the next chapter a very useful corollary to Ulhmann's theorem:
\begin{corollary}\label{cor:Ulhmann}
Let $\rho^A$ and $\sigma^A$ be states of a system $A$. Introduce a second system $R$ which is a ``copy'' of $A$. Then,
\begin{equation*}
  F(\rho^A, \sigma^A) = \max_{\ket{\phi}} |\langle \psi | \phi \rangle |,
\end{equation*}
where $\ket{\psi}^{AR}$ is any fixed purification of $\rho^A$, and the maximization is over all purifications $\ket{\phi}^{AR}$ of $\sigma^A$.
\end{corollary}
Another important result we will frequently use is the Fannes inequality \cite{Fannes}, which bounds the difference in the von Neumann entropies of two states $\rho^A$ and $\sigma^A$ as a function of their trace distance.
\begin{Lemma}[Fannes Inequality]\label{lem:Fannes}
     Let $\rho^A$ and $\sigma^A$ be states on a $d$-dimensional Hilbert space $A$. Let $\epsilon > 0$ be such that $\|\rho^A - \sigma^A\|_1 \leq \epsilon$. Then \begin{equation*}
     \bigg |S(A)_{\rho} - S(A)_{\sigma} \bigg | \leq \eta(\epsilon) \log{d},
     \end{equation*}where $\eta(x) = x - x \log x$ for $x \leq \frac{1}{e}$. When $x > \frac{1}{e}$, we set $\eta(x)=x + \frac{\log(e)}{e}$.
\end{Lemma}
\subsection{Typicality}

Suppose an information source emits a sequence of letters $x_1x_2x_3\ldots$ taken from an alphabet ${\cal X}$ according to some probability distribution $p(x)$. If the source is memoryless (i.e., each letter in the sequence is an i.i.d. random variable $X$ with probability distribution $p(x)$), we expect the frequency at which each letter $x$ will appear in the sequence to depend on the probability weight $p(x)$ associated with the letter $x$. For a sequence $x^n$ of $n$ letters $x_1x_2x_3\ldots x_n$, let $N(x|x^n)$ be the number of times the letter $x$ appears in the sequence $x^n$.  Given a probability distribution $P(x)$ over the alphabet $\cal X$, we say that the sequence $x^n:=x_1x_2x_3\ldots x_n$ is of type $P$ if $N(x|x^n) = nP(x)$ for every letter $x$. That is, the letter $x$ appears exactly $nP(x)$ times in the sequence $x^n$. For a number $\delta > 0$, the set ${\cal T}^n_{p,\delta}$ of $\delta$-\textit{typical sequences} of length $n$ for the probability distribution $p(x)$ is defined as
\begin{equation*}
   {\cal T}^n_{p,\delta} := \bigg \{ x^n : \forall x \in {\cal X}, \bigg | \frac{N(x|x^n)}{n} - p(x) \bigg | \leq \delta \bigg \}.
\end{equation*}
Typicality can be exploited to prove theoretical bounds on the achievable rates for discrete memoryless sources. (See, for example, the proofs of the source coding theorem and the noisy channel coding theorem \cite{Shannon,Cover}.) For any $\epsilon, \delta > 0$, we have for sufficiently large values of $n$:
\begin{eqnarray}\label{eq:typicality}
  p^n({\cal T}^n_{p,\delta}) &\geq& 1 - \epsilon \\
  2^{-n(H(X) + c\delta)} \leq p^n(x^n) &\leq& 2^{-n(H(X)-c\delta)} \quad \forall x^n \in {\cal T}^n_{p,\delta} \label{eq:typ1}\\
  (1-\epsilon)2^{n(H(X)-c\delta)} \leq |{\cal T}^n_{p,\delta}| &\leq& 2^{n(H(X)+c\delta)} \label{eq:typ2},
\end{eqnarray}
where $p^n(x^n) = \prod^n_{i=1} p(x_i)$ and $c$ is some positive constant. We refer to~\cite{InfTheory} for a proof of these statements.

The classical notion of typicality extends to the quantum setting by considering a memoryless quantum source emitting a sequence of unknown states $\ket{x_1}\ket{x_2}\ket{x_3}\ldots..$ with known density operator $\psi^A = \sum_x p_x \braket{x}^A$. The source transmits a state $\ket{x}^A$ with probability $p(x)$. We can alternatively think of the source as emitting many copies of the density operator $\psi^A$. For $n$ copies of the state $\psi^A$, its spectral decomposition is written as:
\begin{equation*}
 \psi^{\otimes n}_A := (\psi^A)^{\otimes n} = \sum_{x^n} p^x(x^n) \braket{x^n}.
\end{equation*}
If we perform a measurement in the basis $\ket{x^n}$, it follows from the classical notion of typicality that, for sufficiently large values of $n$, we will obtain with very high probability a measurement outcome $x^n$ belonging to the set of typical sequences ${\cal T}^n_{p,\delta}$. For $\delta > 0$, we define the $\delta-$\textit{typical subspace} $\tilde{A}^n_{\psi,\delta}$ for the density operator $\psi_A^{\otimes n}$ as
\begin{equation*}
  \tilde{A}^n_{\psi,\delta} := \mathrm{span} \bigg \{ \ket{x^n} \bigg | x^n \in {\cal T}^n_{p,\delta} \bigg \}.
\end{equation*}
The projector into the typical subspace $\tilde{A}^n_{\psi,\delta}$ is given by:
\begin{equation*}
  \Pi^n_{\psi,\delta} := \sum_{x^n \in {\cal T}^n_{p,\delta}} \braket{x^n}.
\end{equation*}
In later chapters, we will often abbreviate the typical projector $\Pi^n_{\psi,\delta}$ associated with the state $\psi^{\otimes n}_A$ as $\Pi_{\tilde{A}}$, where $\tilde{A}$ is shorthand notation for the typical subspace $\tilde{A}^n_{\psi,\delta}$.

The probability of obtaining a measurement outcome $x^n \in {\cal T}^n_{p,\delta}$ is equal to $\Tr (\psi_A^{\otimes n} \Pi^n_{\psi,\delta})$, and since $p^n({\cal T}^n_{p,\delta}) \geq 1-\epsilon$ for any $\delta, \epsilon > 0$ and sufficiently large $n$, we have
\begin{equation}\label{eq:unn}
  \Tr (\psi_A^{\otimes n} \Pi^n_{\psi,\delta}) \geq 1- \epsilon.
\end{equation}
We also have the following properties (see Abeyesinghe et al. \cite{Hayden001} for proofs of these facts), analogous to eqs.~(\ref{eq:typ1}) and ~(\ref{eq:typ2}), for any $\epsilon, \delta > 0$ and sufficiently large values of $n$:
\begin{eqnarray}\label{eq:typicXX}
  2^{-n(S(A)_{\psi} + c\delta)} \Pi^n_{\psi,\delta} &\leq& \Pi^n_{\psi,\delta}\psi_A^{\otimes n} \Pi^n_{\psi,\delta} \leq 2^{-n(S(A)_{\psi}-c\delta)} \Pi^n_{\psi,\delta} \label{eq:un}\\
  (1-\epsilon) 2^{n(S(A)_{\psi} -c\delta)} &\leq& \Tr(\Pi^n_{\psi,\delta}) \leq 2^{n(S(A)_{\psi}+c\delta)}\label{eq:deux} \\
  \Tr[\Psi^2_{\tilde{A}}] := \Tr[ (\Psi^{\tilde{A}})^2] &\leq& (1-\epsilon)^{-2} 2^{-n(S(A)_{\psi}-3c\delta)}, \label{eq:trois}
\end{eqnarray}
where $c$ is some constant and $\Psi^{\tilde{A}}$ is the normalized state obtained after projecting the state $\psi^{\otimes n}_A$ into the typical subspace $\Pi^n_{\psi,\delta}$:
\begin{equation*}
   \Psi^{\tilde{A}} := \frac{\Pi^n_{\psi,\delta} \psi_A^{\otimes n} \Pi^n_{\psi,\delta}}{\Tr (\Pi^n_{\psi,\delta} \psi_A^{\otimes n})}.
 \end{equation*}
The last line is obtained by combining eqs.~(\ref{eq:unn}), (\ref{eq:un}) and (\ref{eq:deux}) with the fact that $\Tr[A^2] \leq \Tr[B^2]$ for any two positive operators such that $A \leq B$. Typical subspaces are a helpful tool when performing asymptotic analysis of quantum protocols (see, for example, \cite{Coding,merge,Hayden001}). This is due, in part, to the following lemma:
\begin{Lemma}[Gentle Measurement Lemma \cite{Winter02}]
  Let $\rho^A$ be a sub-normalized state (i.e $\rho^A \geq 0$ and $\Tr[\rho^A] \leq 1$). For any operator $0 \leq X \leq I$ such that $\Tr[X \rho^A] \geq 1 -\epsilon$, we have
\[
  \bigg \| \sqrt{X} \rho^A \sqrt{X} - \rho^A \bigg \|_1 \leq 2 \sqrt{\epsilon}.
\]
\end{Lemma}
For a proof of the Gentle Measurement Lemma, see \cite{Winter02}. The better constant obtained above is from Ogawa and Nagaoka \cite{Ogawa}.
An application of the previous lemma combined with the triangle inequality yields
\begin{equation*}
   \bigg \| \Psi^{\tilde{A}} - \psi_A^{\otimes n} \bigg \|_1 \leq 2\sqrt{\epsilon} + \epsilon \leq 4\sqrt{\epsilon}
\end{equation*}
for $0 \leq \epsilon \leq 1$. Projecting to the typical subspace preserves the information (up to an arbitrarily small loss) contained in the state $\psi^{\otimes n}_A$ and allows the analysis of the asymptotic behavior of most information-processing tasks to become much simpler to perform.

\section{(Smooth) min- and max-entropies}
For a discrete random variable $X$ taking values in a set $\{x_1, x_2, \ldots, x_n\}$ with probability $p(x_i)$, the \textit{R\'enyi entropy} \cite{Renyi} of order $\alpha$, where $\alpha \neq 1$, is defined as
\begin{equation*}
H_{\alpha}(X) = \frac{1}{1-\alpha}\log \biggl (\sum^n_{i=1} p(x_i)^{\alpha} \biggr ).
\end{equation*}
When $\alpha \rightarrow 1$, we recover the Shannon entropy $H(X) = -\sum^n_{i=1} p(x_i) \log p(x_i)$:
\begin{equation*}
\begin{split}
 \lim_{\alpha \rightarrow 1} H_{\alpha}(X) &= \lim_{\alpha \rightarrow 1} \frac{\frac{d}{d\alpha}\log \left(\sum^n_{i=1} p(x_i)^{\alpha}\right)}{\frac{d}{d\alpha}(1-\alpha)}  \\
 &= \lim_{\alpha \rightarrow 1} \frac{-\log(e)\sum^n_{i=1}\frac{d}{d\alpha}p(x_i)^{\alpha}}{\sum_{i=1}^n p(x_i)^{\alpha}} \\
 &= \lim_{\alpha \rightarrow 1} \frac{-\log(e)\sum^n_{i=1} p(x_i)^{\alpha}\ln p(x_i)}{\sum_{i=1}^n p(x_i)^{\alpha}} \\
 &= -\sum^n_{i=1} p(x_i) \log p(x_i).
\end{split}
\end{equation*}
The first line is l'H\^opital's rule and the third line follows from $\frac{d}{d\alpha} p(x_i)^{\alpha} = p(x_i)^{\alpha} \ln p(x_i)$.
Taking the limit of $H_{\alpha}(X)$ as $\alpha \rightarrow \infty$, we obtain the classical \textit{min-entropy}:
\begin{equation*}
H_{\min}(X) := \lim_{\alpha \rightarrow \infty}H_{\alpha}(X) = - \log \max_{i} p(x_i).
\end{equation*}
The R\'enyi entropies were introduced by R\'enyi in 1961 as alternatives to the
Shannon entropy as measures of information. The Shannon entropy, viewed more abstractly, is also the unique function of a probability distribution which satisfies a precise set of postulates (see \cite{Renyi} for a detailed description of these postulates). R\'enyi extended the notion of entropy to more general random variables, often called ``incomplete'' by R\'enyi because their observation could occur with probability less than one. By generalizing the postulates characterizing the Shannon entropy, other information-theoretic quantities were obtained, such as the family of R\'enyi entropies. Applications of the R\'enyi entropies abound in areas such as cryptography \cite{Renner05simpleand,Cachin97smoothentropy} and statistics \cite{Order,Diversity}.

Quantum min- and max-entropies are adaptations of the classical R\'enyi entropies of order $\alpha$ when $\alpha \rightarrow \infty$ and $\alpha=1/2$ respectively. Let ${\cal S}_{\leq}(AR)$ be the set of sub-normalized density operators (i.e $\Tr (\bar{\rho}^{AR}) \leq 1$) acting on the space $AR$.
The \textit{quantum min-entropy} \cite{Renner02} of an
operator $\rho^{AR} \in {\cal S}_{\leq}(AR)$ relative to a density operator $\sigma^{R}$ is given by
 \begin{equation*}
     H_{\min}(\rho^{AR}|\sigma^{R}):= -\log \lambda,
 \end{equation*}
where $\lambda$ is the minimum positive number such that $\lambda (I^A \otimes \sigma^{R}) - \rho^{AR}$ is a
positive operator. The \textit{conditional min-entropy} $H_{\min}(A|R)_{\rho}$ is obtained by maximizing the previous quantity over
the set ${\cal B}(R)$ of density operators $\sigma^{R}$ for the system $R$:
 \begin{equation*}
    H_{\min}(A|R)_{\rho} := \max_{\sigma^{R} \in {\cal B}(R)} H_{\min}(\rho^{AR}|\sigma^{R}).
  \end{equation*}

 For two sub-normalized states $\rho$ and $\bar{\rho}$, we define the \textit{purified distance} \cite{Renner01} between $\rho$ and $\bar{\rho}$ as
 \begin{equation*}
    P(\rho,\bar{\rho}) := \sqrt{1- \overline{F}(\rho,\bar{\rho})^2},
 \end{equation*} 
 where $\overline{F}(\rho,\bar{\rho})$ is the generalized fidelity between $\rho$ and $\bar{\rho}$:
 \begin{equation*}
   \overline{F}(\rho,\bar{\rho}) := F(\rho, \bar{\rho}) + \sqrt{(1-\Tr \rho)(1- \Tr\bar{\rho})}.
 \end{equation*}
The purified distance is related to the trace distance $D(\rho, \bar{\rho}) := \frac{1}{2}\| \rho - \bar{\rho} \|_1$ as follows
 \begin{equation}
    D(\rho, \bar{\rho}) \leq P(\rho, \bar{\rho}) \leq 2 \sqrt{D(\rho,\bar{\rho})}.\label{eq:purified}
 \end{equation}
A proof of this fact follows directly from Lemma 6 of \cite{Renner01}. (Lemma 6 actually relates the purified distance to the generalized distance $\bar{D}(\rho, \bar{\rho})$. However, $\bar{D}(\rho,\bar{\rho})$ is bounded above by $2D(\rho,\bar{\rho})$ and bounded below by $D(\rho,\bar{\rho})$.)

Using the purified distance as our measure of closeness, we obtain the family of \textit{smooth min-entropies} $\{H^{\epsilon}_{\min}(A|R)_{\rho}\}$ by optimizing over all sub-normalized density operators close to $\rho^{AB}$ with respect to $P(\bar{\rho},\rho)$:
 \begin{equation*}
    H^{\epsilon}_{\min}(A|R)_{\rho} := \max_{\bar{\rho}^{AR} \in {\cal S}_{\leq}(AR)} H_{\min}(A|R)_{\bar{\rho}},
 \end{equation*}
where the maximization is taken over all $\bar{\rho}^{AR}$ such that $P(\bar{\rho}^{AR},\rho^{AR}) \leq \epsilon$.
Given a purification $\rho^{ABR}$ of $\rho^{AR}$, with purifying
system $B$, the family of \textit{smooth max-entropies} $\{H^{\epsilon}_{\max}(A|B)_{\rho}\}$ is defined as
 \begin{equation}\label{eq:duality}
    H^{\epsilon}_{\max}(A|B)_{\rho} := - H^{\epsilon}_{\min}(A|R)_{\rho}
 \end{equation}
for any $\epsilon \geq 0$. 
When $\epsilon=0$, an alternative expression for the \textit{max-entropy} $H_{\max}(A|B)_{\rho}$ was obtained by Koenig et al. \cite{Renner03}:
\begin{equation}\label{eq:dualityy}
  H_{\max}(A|B)_{\rho} = \max_{\sigma^B \in {\cal B}(B)} \log F^2(\rho^{AB}, I^A \otimes \sigma^B),
\end{equation}
where the maximization is taken over all density operators $\sigma^B$ on the space $B$. The smooth max-entropy can also be expressed as
 \begin{equation}\label{eq:smoothmax}
    H^{\epsilon}_{\max}(A|B)_{\rho} = \min_{\bar{\rho}^{AB}\in {\cal S}_{\leq}(AB)} H_{\max}(A|B)_{\bar{\rho}},
 \end{equation}
where the minimum is taken over all sub-normalized $\bar{\rho}^{AB}$ such that $P(\bar{\rho}^{AB},\rho^{AB}) \leq \epsilon$. We refer to \cite{Renner01} for a proof of this fact. From eq.~(\ref{eq:dualityy}), the smooth max-entropy $H^{\epsilon}_{\max}(\rho^{A})$ of a sub-normalized operator $\rho^A \in {\cal S}_{\leq}(A)$ reduces to
\begin{equation}\label{eq:hmax}
  H^{\epsilon}_{\max}(\rho^{A}) = 2\log \sum_x \sqrt{\bar{r}_x},
\end{equation}
where $\bar{r}_x$ are the eigenvalues of the sub-normalized density operator $\bar{\rho}^A$ which optimizes the right hand side of eq.~(\ref{eq:smoothmax}).

The smooth min- and max-entropies are also known to satisfy other useful properties such as quantum data processing
inequalities and concavity of the max-entropy (see~\cite{Renner01}). These measures were introduced to characterize information-theoretic tasks which tolerate a small error on the desired outcome. In Chapter 4, we describe a protocol for the task of multiparty state merging and analyze the entanglement cost using smooth min-entropies.

We also need, for technical reasons (see eq.~(\ref{eq:trace}) found in Chapter 4), another entropic
quantity called the \textit{conditional collision entropy} \cite{Renner02}:
\begin{equation*}
   H_2(\rho^{AB}|\sigma^{B}):=  -\log \Tr \bigg [ \bigg ((I_A \otimes \sigma_{B}^{-1/4}) \rho^{AB} (I_A \otimes \sigma_B^{-1/4}) \bigg )^2 \bigg ],
\end{equation*}
 where $\rho^{AB}$ is a density operator for the system $AB$. It is a quantum adaptation of the classical conditional collision entropy. The following lemma, proven in \cite{Renner02}, relates the quantum min-entropy to the collision entropy:
\begin{Lemma}\cite{Renner02}\label{lem:rener4}
For density operators $\rho^{AB}$ and $\sigma^{B}$ with $\mathrm{supp}\{\Tr_A(\rho^{AB})\} \subseteq \mathrm{supp}\{\sigma^{B}\}$, we have
  \[ H_{\min}(\rho^{AB}|\sigma^{B}) \leq H_2(\rho^{AB}|\sigma^{B}). \]
\end{Lemma}

The last two results we will need are the additivity of the min-entropy and the following lemma which relates the trace norm of an
hermitian operator $X$ to its Hilbert-Schmidt norm $\|X\|_2 := \sqrt{\Tr(X^{\dag}X)}$, with respect
to a positive operator $\sigma$:
 \begin{Lemma}\label{Lemma:TrSchmidt}
 Let $X$ be an hermitian operator acting on a space $A$ and $\sigma$ be a positive operator on $A$. We have
     \begin{equation*}
         \|X\|_1 \leq \sqrt{\Tr(\sigma)}  \|\sigma^{-1/4} X \sigma^{-1/4} \|_2.
     \end{equation*}
 \end{Lemma}
\begin{proof}
Rewrite the right hand side as $\sqrt{\Tr(\sigma)\Tr(X \sigma^{-1/2} X \sigma^{-1/2})}$ and apply Lemma 5.1.3 of \cite{Renner02}.
\end{proof}
\begin{Lemma}[Additivity]\label{lem:addi}
Let $\rho^{AB}$ and $\rho^{A'B'}$ be sub-normalized density operators for the systems $AB$ and $A'B'$ respectively. For density operators $\sigma^{B}$ and $\sigma^{B'}$, we have
\begin{equation*}
  H_{\min}(\rho^{AB} \otimes \rho^{A'B'}|\sigma^{B} \otimes \sigma^{B'}) = H_{\min}(\rho^{AB}|\sigma^{B}) + H_{\min}(\rho^{A'B'}|\sigma^{B'}).
\end{equation*}
\end{Lemma}
Additivity follows straightforwardly from the definition of the quantum min-entropy.

\section{Previous distillation protocols}

\subsection{The Schmidt projection method}
 One of the first protocols for extracting pure entanglement was devised by Bennett et al. \cite{Concentrate} and works on a supply of partly entangled pure states $\psi^{AB}$ with \textit{entropy of entanglement}
 \begin{equation*}
 E(\psi^{AB}) := S(A)_{\psi}.
 \end{equation*} First, let's assume the states being shared are qubits. The extension of the protocol to higher dimensional systems will be straightforward. Using the Schmidt decomposition, we can write the pure state $\psi^{AB}$ as:
 \begin{equation*}
  \psi^{AB} = \cos(\theta) \ket{e_0 f_0 } + \sin(\theta) \ket{e_1f_1},
 \end{equation*}
 with $0 \leq \theta \leq \pi/2$ and $\sin{\theta}^2 + \cos{\theta}^2 = 1$. For the tensor product state $\psi^{\otimes n}_{AB}$, we have
 \begin{equation}\label{eq:tensor}
  \psi^{\otimes n}_{AB} = \bigotimes^n_{i=1} \bigg (\cos(\theta) \ket{e^i_0 f^i_0} + \sin(\theta)\ket{e^i_1f^i_1} \bigg ). \\
 \end{equation}
 By expanding the right hand side of eq.~(\ref{eq:tensor}), we get coefficients of the form \[ \lambda_k := \cos^{n-k}\theta \sin^{k}\theta, \] for $0 \leq k \leq n$. Let $P^A_k$ be the associated projector onto the subspace of dimension $\dbinom{n}{k}$ spanned by the vectors $\{\ket{e_{i_1}^1e^2_{i_2}\ldots e^n_{i_n}}\}$ having coefficient $\lambda_k$ in eq.~(\ref{eq:tensor}). Alice performs a projective measurement with projectors $\{P^A_k\}$, yielding the outcome $k$ with probability
 \begin{equation*}
    p_k = \dbinom{n}{k} (\cos^2\theta)^{n-k} (\sin^2\theta)^{k}.
 \end{equation*}
 By virtue of the original entanglement, Bob will obtain the same value for $k$ if he wants to perform his measurement. Hence, the measurement produces a maximally entangled state $\psi^{AB}_k$ in a $\dbinom{n}{k}^2$-dimensional subspace of the original $2^{2n}$ dimensional space.

The efficiency of the above procedure can be understood as follows. The expected entropy of entanglement $\sum_k p_k E(\psi^{AB}_k)$ for the residual states is non increasing under local operations (see \cite{Bennett} for a proof). Hence, we must have
\begin{equation}\label{eq:expectedE}
   \sum_k p_k S(A)_{\psi_k} \leq E(\psi_{AB}^{\otimes n}) = n E(\psi^{AB}) = n S(A)_{\psi}.
\end{equation}
The von Neumann entropy of the reduced state $\psi^{\otimes n}_{A}$ can only increase under projective measurements:
\begin{equation}\label{eq:Schmidtlocal}
 S(A)_{\psi^{\otimes n}} = nS(A)_{\psi} \leq S(A)_{\psi_{\mathrm{out}}},
\end{equation}
where $\psi^{AB}_{\mathrm{out}} := \sum_k p_k \psi^{AB}_k$.
Since we have (see \cite{Nielsen} for a proof)
\begin{equation*}
S(A)_{\psi_{\mathrm{out}}} \leq H(p_k) + \sum_k p_k S(A)_{\psi_k},
\end{equation*}
we combine with eqs.~(\ref{eq:expectedE}) and ~(\ref{eq:Schmidtlocal}) to obtain
\begin{equation*}
\sum_k p_k S(A)_{\psi_k} \leq n E(\psi^{AB}) \leq H(p_k) + \sum_k p_k S(A)_{\psi_k}.
\end{equation*}
But $H(p_k)$ is the entropy of a binomial distribution of $n$ trials with success probability $p = \sin^2 \theta$, which increases only logarithmically with $n$. Hence, as $n\rightarrow \infty$, the expected entropy of entanglement converges to $nE$, and so the original entanglement is preserved.

To convert the residual states into a standard form such as EPR pairs, however, we must use a \textit{double blocking} strategy: partitioning the original input into different subsets and applying the above procedure on each of those subsets. If we have $m$ batches of tensor states $\psi^{\otimes n}_{AB}$, we obtain a sequence of $k$ values $k_1,k_2,\ldots,k_m$ by applying the previous strategy of projective measurements on each batch. To obtain a desired number of singlet states, we need to have an adequate amount of batches $m$ at our disposal. Let
\begin{equation*}
 D_m = \dbinom{n}{k_1}\dbinom{n}{k_2}\ldots \dbinom{n}{k_m},
 \end{equation*}
be the product of the binomial combinations $\dbinom{n}{k}$ for the first $m$ batches and fix some $\epsilon > 0$. If $D_m$ lies between $2^l$ and $2^l(1+\epsilon)$, we have enough batches so that we can recover $l$ singlets by projecting (Alice) the residual states $\psi^A_{k_1}\otimes \psi^A_{k_2} \ldots \psi^A_{k_m}$ onto a large space of dimension $2^{l}$ (followed by some Pauli operators). Otherwise, we will need more supply of initial entanglement.

With probability greater than $1-\epsilon$, the projection onto the large space will succeed and, by virtue of the entanglement of the residual states, Alice and Bob will now share $l$ singlets. With probability less than $\epsilon$, the residual states are projected in a smaller subspace of dimension $D_m - 2^l < 2^l \epsilon$. In such case, a failure is declared and the protocol is aborted. As discussed in \cite{Concentrate}, the converting into product of singlets can also be shown to preserve the original entanglement in the limit of large $n$.

\subsection{The hashing method}\label{sec:hashing}
The first distillation protocols working on a supply of mixed states $\psi^{AB}$ appeared in \cite{Bennett}.  If Alice and Bob can communicate classical information, Bennett et al. constructed a strategy, called the hashing method, for producing a non-zero yield of pure entanglement if the pairs are drawn from an ensemble of Bell states with known density operator:
\begin{equation*}
\psi^{AB} = p_0 \braket{\Phi_{+}} + p_1 \braket{\Psi_{+}} + p_2 \braket{\Phi_{-}} + p_3 \braket{\Psi_{-}},
\end{equation*}
with von Neumann entropy $S(AB)_{\psi}$ given by the Shannon entropy $H(p)$ of the probability distribution $\{p_0,p_1,p_2,p_3\}$.
In this section, we give a brief description of this method. Each of the four Bell states can be encoded using 2 classical bits in the following way:
\begin{equation}\label{eq:Bell}
\begin{split}
\ket{\Phi_{+}} &=  00 \\
\ket{\Psi_{+}} &=  01 \\
\ket{\Phi_{-}} &=  10 \\
\ket{\Psi_{-}} &=  11.
\end{split}
\end{equation}
An unknown sequence of $n$ Bell states can then be represented as a bit string of length $2n$. For instance, the sequence $\Phi_{+}\Psi_{-}\Phi_{-}$ is encoded as $001110$. The parity of a subset $s$ of the bits in a string $x$ is equal to the modulo-2 sum of the bitwise AND between $s$ and $x$, or equivalently, the Boolean inner product $s\cdot x$. For instance, if we have $x=001110$ and $s=100100$, then the parity of the selected subset $000100$ of $x$ is equal to 1.

At the start of the protocol, Alice and Bob share an unknown sequence of $n$ Bell states. Let $x_0$ be the bit sequence of length $2n$ corresponding to this unknown sequence of Bell states. The hashing method consists of $n-m$ rounds of the following procedure: At the start of round $k+1$, $k=0,1,\ldots,n-m-1$, Alice chooses a random subset $s$ of the unknown bit sequence $x_k$ of length $2(n-k)$ and sends it to Bob. Alice and Bob then determine the parity of $s$ by performing local operations on their share of the remaining $n-k$ pairs. In \cite{Bennett}, it is shown how to obtain the parity of a subset $s$ by measuring a single pair of qubits, while preserving the Bell-state nature of the remaining pairs. The unmeasured pairs, represented by the classical string $x_{k+1}$, will change according to some deterministic hash function $x_{k+1}=f_s(x_k)$.  At the end of the round, Alice and Bob share an unknown sequence $x_{k+1}=f_s(x_k)$ of $n-k-1$ Bell states.

For any two possible candidates $y\neq z$ at any given round in the above protocol, the probability that they will agree on a random subset $s$ is equal to $1/2$. This can be seen from the fact that
\begin{equation*}
  (s \cdot y) \oplus (s \cdot z) = s \cdot (y \oplus z),
\end{equation*}
where $\oplus$ denotes addition modulo-2. Thus, after each parity measurement is performed, we can expect on average at least half of the remaining candidates to be eliminated. The remaining candidates will be mapped into a set of possible output sequences according to the hash function $f_s(x_k)$. After $r$ rounds, the probability that two distinct candidates $y$ and $z$ remain distinct is therefore bounded by $2^{-r}$, combining the previous two facts. From typicality, we know that with arbitrarily good probability, the unknown string $x$ will be in a set of $2^{n(S(AB)_{\psi}+\delta)}$ typical sequences, where $\delta > 0$ can be made arbitrarily small by choosing $n$ large enough. Hence, if we perform approximately $n-m=n(S(AB)_{\psi}+2\delta)$ rounds, we can expect to identify a single good candidate $x$ with high probability. Failure occurs if after $n-m$ rounds, we have more than one candidate or no candidate if we restrict our search within the set of typical sequences.

If a single candidate is found after $n-m=n(S(AB)_{\psi}+2\delta)$ rounds, Alice and Bob share $n (1 -(S(AB)_{\psi}+2\delta))$ copies of impure Bell states which can be transformed into maximally entangled states by applying suitable Pauli operators:
\begin{equation*}
X = \left( \begin{array}{cc}
0 & 1  \\
1 & 0  \\
 \end{array} \right)\quad
 Y = \left( \begin{array}{cc}
0 & -i  \\
i & 0  \\
 \end{array} \right)\quad
 Z = \left( \begin{array}{cc}
1 & 0  \\
0 & -1  \\
 \end{array} \right).
\end{equation*}
For instance, if the sequence of unknown pairs is $\ket{\Phi_{+}}\ket{\Psi_{-}}\ket{\Phi_{-}}$, Alice and Bob can recover three singlet states by applying the Pauli operators $Z\otimes I \otimes I$ on Alice's share and $X \otimes I \otimes X$ on Bob's share. As the number of copies $n$ grows larger, this strategy will produce a yield approaching $1-S(AB)_{\psi}$ maximally entangled pairs per copy of the input state $\psi^{AB}$.

A more recent protocol by Devetak and Winter \cite{DW}, relying on different ideas, extends the previous result to arbitrary mixed entangled states $\psi^{AB}$, and proves the following hashing inequality, named after the hashing method, for \textit{one-way distillable entanglement} $D_{\rightarrow}(\psi^{AB})$ (see \cite{DW} for a formal definition):
\begin{lemma}[Hashing inequality \cite{Bennett,DW}]
\label{hashinginequality} Let $\psi^{AB}$ be an arbitrary bipartite mixed state. Then,
\begin{equation}\label{eq:hashing}
D_{\rightarrow}(\psi^{AB}) \geq S(B)_{\psi} - S(AB)_{\psi} =: I(A\rangle B)_{\psi}.
\end{equation}
\end{lemma}
The quantity on the right hand side is known as the \textit{coherent information} from $A$ to $B$ of the state $\psi^{AB}$.
For the case of mixtures of Bell states, the coherent information evaluates to $1-S(AB)_{\psi}$, which is exactly the yield attained by the hashing method. Devetak and Winter showed how to achieve coherent information for arbitrary states~\cite{DW} by exploiting the connection between entanglement distillation and quantum data transmission, for which the coherent information had already been demonstrated to be achievable~\cite{Lloyd-cap,igor-cap,shor-cap}.

\subsection{Entanglement of assistance} \label{sec:task}

Suppose a memoryless quantum source emits an unknown sequence of Bell states $\ket{\Psi_{-}}$,$\ket{\Psi_{+}}$,$\ket{\Phi_{-}}$,$\ket{\Psi_{+}}$,$\ldots$ with known density operator $\psi^{AB}$ and gives one qubit of each pair to a laboratory $A$ (Alice), while the other halves are sent to another laboratory $B$ (Bob). Suppose each Bell state is produced with equal probability. From Alice and Bob's point of view, they share many copies of a maximally mixed state $I^{AB}/4$. The hashing inequality suggests that it is impossible to distill entanglement for this state. If Alice and Bob have no additional information about the source, this assumption is correct (see \cite{Bennett} for a proof). If a third party (Charlie) has access to the source, however, and tells Alice and Bob the exact sequence of states produced, they can recover a singlet state $\Psi_{-}$ for each pair of qubits shared by applying appropriate Pauli operators.

An equivalent view of the previous example is to assume that Alice, Bob and Charlie share the state
\begin{equation*}
\braket{\psi}^{ABC} = \sum_{i_1i_2\ldots i_n} \frac{1}{4^n} \braket{\Psi}_{i_1i_2\ldots i_n}^{AB} \otimes \braket{i_1i_2\ldots i_n}^C,
\end{equation*}
where each state $\Psi_{i_1i_2\ldots i_n}$ is a product of Bell states. Each index $i_j$ corresponds to one of the four Bell states using the encoding of eq.~(\ref{eq:Bell}). If Charlie measures his system in the basis $\{\ket{i_1i_2\ldots i_n}^C\}$, and sends the measurement outcome to Alice and Bob, they can apply Pauli operators to recover $n$ singlet states $\Psi_{-}$. This is in sharp contrast with the results of the hashing method, which gives a zero yield for the reduced state $\psi^{AB}$.

Motivated by this simple example, a natural question to ask is how much additional entanglement can be distilled from a tripartite state $\psi^{ABC}$ if third party assistance is available.  This has been studied and solved completely \cite{dfm, Laustsen, SVW, merge} when the parties share an arbitrary pure tripartite state $\psi^{ABC}$.  The \emph{entanglement of assistance} \cite{dfm} for the pure state $\psi^{ABC}$ is defined as:
\begin{equation}
\begin{split}
\label{eq:EofA} E_A(\psi^{AB}):=E_A(\psi^{ABC})
      &:= \max_{{\cal E}}  \sum_i p_i S(AB)_{\psi_i},
\end{split}
\end{equation}
where the maximum is over all decompositions ${\cal E} = \{p_i,
\psi_i^{AB}\}$ of $\psi^{AB}=\Tr_C \psi^{ABC}$ into
a convex combination of pure states . By applying a POVM on the system $C$, the
helper Charlie can effect any such pure state convex decomposition
$\rho^{AB} = \sum_i p_i\psi_i^{AB}$ for Alice and Bob's
state \cite{Hughston}, and so the quantity $E_A$
maximizes the average amount of entanglement that Alice and Bob can
distill with help from Charlie. Since $E_A$ is not, in general,
additive under tensor products \cite{dfm}, it will often be the
case that collective measurements performed by Charlie on the joint
state $\psi_{C}^{\otimes n}$ will be more beneficial to Alice
and Bob than individual measurements on individual copies of
$\psi^{C}$.

Define the quantity $E^{\infty}_A$ as the optimal EPR rate distillable between Alice and Bob with the help of Charlie under LOCC quantum operations. This was shown by Smolin et al. \cite{SVW} to be equal to the regularization of $E_A$:
\be \label{eq:EofAreg} E_A^{\infty}(\psi^{ABC})=\lim_{n \rightarrow \infty}
\frac{1}{n} E_A(\psi_{ABC}^{\otimes n}).
\ee
Furthermore, a nice simple expression for $E_A^{\infty}$ in terms of entropic quantities was also obtained:
\begin{equation}
\label{thm:EofA}
E_A^{\infty}(\psi^{ABC})=\min\{S(A)_\psi,S(B)_\psi\}.
\end{equation}
The proof of this result is based on a variant of the Holevo-Schumacher-Westmoreland (HSW) theorem \cite{Holevo,HSW}\footnote{The version of the HSW theorem relevant for proving eq.~(\ref{thm:EofA}) appears in \cite{DW}.}, and makes use of random coding in a very original way. Write the state $\psi_{ABC}$ in Schmidt form as $\sum_j \sqrt{q_j} \ket{\psi_j}^{AB} \ket{j}^C$ and consider $n$ copies of it:
 \begin{equation*}
  (\ket{\psi}^{ABC})^{\otimes n} = \sum_{J} \sqrt{q_J} \ket{\psi_J}^{A^nB^n} \ket{J}^{C^n},
 \end{equation*}
where $\ket{\psi_J}^{A^nB^n}=\ket{\psi_{j_1}}^{A^nB^n}\ket{\psi_{j_2}}^{A^nB^n}\ldots \ket{\psi_{j_n}}^{A^nB^n}$.
 After projecting the system $C^n$ into a subspace of constant type $P$, Charlie selects a random code ${\cal J}=(J^1, J^2, \ldots, J^N)$, where each $J^i$ is a sequence $j_1j_2\ldots j_n$ of type $P$, by performing an appropriate POVM $\{\smfrac{N}{c}\braket{t_{\cal J}(\alpha)}\}_{{\cal J},\alpha}$ on his system, where
\begin{equation*}
  \ket{t_{\cal J}(\alpha)} := \frac{1}{\sqrt{N}}\sum^{N}_{\beta=1} e^{2\pi i \alpha \beta} \ket{J^{\beta}}^{C^n}
\end{equation*}
and $c$ is a constant chosen so that $\sum_{{\cal J},\alpha} \smfrac{N}{c}\braket{t_{\cal J}(\alpha)} = \Pi_P$, where $\Pi_P$ is the projector onto the type $P$ subspace of $C^n$.  The HSW theorem guarantees that if the number of codewords $N$ is a bit less than $2^{n\chi}$, where $\chi$ is the Holevo information
  \begin{equation*}
  \chi = \chi(\{q_i,\psi^A_i\}) := S(A)_{\psi^A} - \sum_i q_i S(A)_{\psi^A_i}
  \end{equation*}
of the channel $i \rightarrow \psi^A_i$  (assuming w.l.o.g. that $S(A)_{\psi} \leq S(B)_{\psi}$), then the code will be good. That is, with very high probability, both Alice and Bob have good decoders for their respective channels $i \rightarrow \psi^A_i$ and $i \rightarrow \psi^B_i$. Once Alice and Bob know the code selected, they each apply the decoders associated with this code in a \textit{coherent} way: if $\{D^A_m\}^N_{m=1}$ is a decoder for Alice, she applies the isometry
\begin{equation*}
  V_A = \sum_m \sqrt{D^A_m} \otimes \ket{m}^{A'}
\end{equation*}
to her system and similarly for Bob. This will produce residual states $\nu^{AA'BB'}$ such that
\begin{equation*}
 \ket{\nu}^{AA'BB'} \approx \frac{1}{\sqrt{N}}\sum^N_{m=1} e^{-2\pi i \alpha m /N} \ket{\psi_{J^m}}^{A^nB^n}\ket{m}^{A'}\ket{m}^{B'},
\end{equation*}
where $e^{-2\pi i \alpha m /N}$ are phases introduced by Charlie's POVM. The entropy of entanglement $S(AA')_{\nu}$ of the residual states can be shown to be arbitrarily close to the rate of eq.~(\ref{thm:EofA}). A standard distillation protocol, within in a double blocking scheme, can then be applied to recover EPR pairs at this rate.

A simpler proof of eq.~(\ref{thm:EofA}) appears in \cite{merge} and uses state merging to reveal a property at the core of third party assisted distillation: If a third party holds the purifying system $C$ of a bipartite state $\psi^{AB}$, then he can always transfer his system to either Alice or Bob, whichever will result in the least bipartite entanglement. Both eqs.~(\ref{eq:EofAreg}) and (\ref{thm:EofA}) will also follow from our more general Theorems \ref{thm:equivalence} and \ref{thm:lowerbound} in Chapter 5.

The previous scenario can be generalized to include multiple helpers who will assist Alice and Bob in distilling entanglement. In the \textit{multipartite entanglement of assistance} problem \cite{SVW}, the task is to distill EPR pairs from a $(m+2)$-partite pure state $\inputstate$ shared between two recipients (Alice and Bob) and $m$ other helpers $C_1,C_2,\ldots,C_m$. If many copies of the input state are available, the optimal EPR rate was shown in \cite{merge} to be
equal to
\begin{equation}\label{eq:mincutX}
  E^{\infty}_A(\inputstate) := \min_{\cal T} S(A{\cal T})_{\psi} =: E_{min-cut}(\inputstate),
\end{equation}
where ${\cal T} \subseteq \{C_1,C_2,\ldots,C_m\}$ is a subset (i.e a bipartite cut) of the helpers. We
denote the complement by $\cT := \{C_1C_2\ldots C_m\} \setminus {\cal T}$. We call $\min_{\cal T}\{S(A{\cal T})_{\psi}\}$ the \textit{minimum cut entanglement} (min-cut entanglement) of the state $\inputstate$. We will recover eq.~(\ref{eq:mincutX}) from our more general Theorem \ref{thm:gen}. 
\chapter{Multiparty State Transfer}
\section{Introduction}
Suppose two parties share $n$ copies of a bipartite mixed state $\psi^{AB}$ and one of the parties (Alice) wants to transfer her system $A^n$
to the other party (Bob) using as little quantum communication as possible. Consider a purification $\psi^{ABR}$ of this state, where $R$ is the purifying (reference) system. An appropriate measure for the correlation between Alice's system and the purification system is the \textit{quantum mutual information}:
\begin{equation*}
 I(A;R)_{\psi} := S(A)_{\psi} + S(R)_{\psi} - S(AR)_{\psi}.
\end{equation*}
For the systems $A^n$ and $R^n$, we have $I(A^n;R^n)_{\psi^{\otimes n}}=nI(A;R)_{\psi}$ bits of correlation. Superdense coding suggests a strategy for transferring the system $A^n$ to Bob: encode $nI(A;R)_{\psi}$ bits of mutual information in a subsystem $A_1$ of dimension $\frac{n}{2}I(A;R)_{\psi}$ and send this subsystem to Bob. This will transfer the correlation between Alice and the reference to Bob, leaving Alice with a system $A_2$ \textit{decoupled} (i.e decorrelated) from the reference $R$. Using his knowledge of the density operator $\psi^{AB}$, Bob can then recover the entire system $A^n$ via local operations on the systems $A_1B^n$.

Abeyesinghe et al. \cite{Hayden001} showed that by applying a random (Haar distributed) unitary $U:\tilde{A} \rightarrow A_1A_2$ to a subspace $\tilde{A}$ of $A^n$, transmission rates arbitrarily close to $\frac{1}{2}I(A;R)_{\psi}$ are achievable. When Bob receives the $A_1$ system, he holds a purification of the system $A_2R^n$, which can be taken, by means of an isometry $V: A_1B^n \rightarrow \tilde{B} B^n$, to the original state $\psi^{\otimes n}_{ABR}$ with high fidelity. The decoding also distills entanglement, in the form of maximally entangled states shared between the sender and the receiver, at an ebit rate approaching $\frac{1}{2}I(A;B)_{\psi}$. This is known as the \textit{fully quantum Slepian-Wolf} (FQSW) protocol.

If a quantum channel is not available between Alice and Bob, but they share enough entanglement, they can substitute quantum data transmission with teleportation to achieve the previous task. This is known as \textit{quantum state merging} and is the original formulation of the state transfer problem. The main result obtained by Horodecki et al. \cite{SW-Nature,merge} is a proof of the existence of protocols achieving near optimal ebit rates, arbitrarily close to the conditional von Neumann entropy $S(A|B)_{\psi}$. This result is easily derived by modifying an FQSW protocol: teleport the qubits using $\frac{1}{2}I(A;R)_{\psi}$ ebits and recycle the distilled entanglement at the end of the protocol. In the limit of many copies of the state $\psi^{ABR}$, only a net rate of \[\frac{1}{2}I(A;R)_{\psi} - \frac{1}{2}I(A;B)_{\psi} = S(A|B)_{\psi}\] ebits is needed to transmit the system $A^n$ with high fidelity. If the conditional entropy is negative, the protocol returns ebits for future communication, but still requires catalytic entanglement to achieve the transfer.

Horodecki et al., however, took a more direct approach, similar to what is found in Abeyesinghe et al. \cite{Hayden001}, by considering random (Haar distributed) measurements on the $A^n$ system. The benefit of this approach is to remove the catalytic number of ebits needed when adapting an FQSW protocol to perform state merging. In other words, when the conditional von Neumann entropy is negative, the existence of protocols achieving a state transfer by LOCC only (no injected entanglement) was shown. This result has important consequences in the context of entanglement distillation. For positive coherent information, there exist state merging protocols which are also one-way entanglement distillation protocols. By preprocessing Alice's system to optimize the coherent information as much as possible, and applying a state merging protocol, near optimal rates are achievable (see Theorem 13 of Devetak and Winter \cite{DW}).
\begin{figure}[t]
  \centering
    \includegraphics{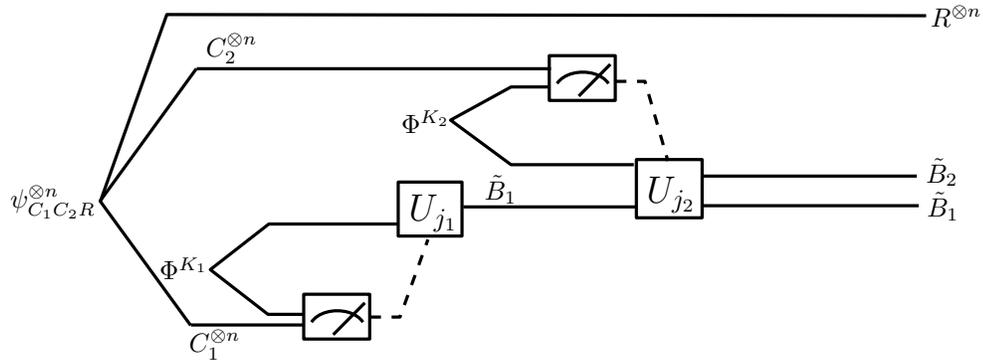}
\caption{Quantum circuit representing a distributed compression protocol involving two senders. Solid black lines indicate quantum information and dashed lines classical information. The protocol depicted involves the combination of two state merging protocols for recovering the systems $C^{\otimes n}_1$ and $C^{\otimes n}_2$ at the receiver's location.}\label{fig:distcompression}
\end{figure}

In this chapter, we analyze two extensions of the state transfer problem. First, we consider $m$ senders and a decoder/receiver sharing a state $\psi^{C_1C_2\ldots C_mBR}$, and look at the number of ebits needed for transferring the systems $C_1, C_2, \ldots, C_m$ to the receiver, sometimes equivalently referred to as ``merging'' the state $\psi^{C_1C_2\ldots C_mBR}$ to the receiver. This extension of the state merging task to a more general multiparty setting was also analyzed by Horodecki et al. \cite{merge} under the name of \textit{distributed compression}. A combination of state merging protocols, initiated by each of the $m$ senders, was shown to yield optimal ebit rates for distributed compression. For instance (see Figure \ref{fig:distcompression}), if many copies of a state $\psi^{C_1C_2R}$ are distributed to two parties $C_1$ and $C_2$, the sender $C_1$ can first transfer his system to the receiver at a compression rate arbitrarily close to the entropy $S(C_1)_{\psi}$. The second sender follows by merging his system with the receiver at an ebit rate approaching $S(C_2|C_1)_{\psi}$. If the second sender goes first instead, rates close to $S(C_2)_{\psi}$ and $S(C_1|C_2)_{\psi}$ are achieved instead.

 These distributed compression protocols, although optimal in the rates, require the use of \textit{time-sharing} for achieving rates which are not corner points of the rate region. Time-sharing consists of partitioning a large supply of states and applying different protocols to each subset. If a single copy of the state is available, this approach becomes impossible. We remedy this problem by showing the existence of multiparty merging protocols which work even if the parties share a single copy of the input state. That is, they don't require the use of time-sharing for performing the task of multiparty merging. This entails the existence of decoders of a more general form (see Figure \ref{fig:multipartymerging}). A side-effect of time-sharing for the distributed compression protocol is to restrict the decoding implemented by the receiver to have a composition form $U_{j_m}U_{j_{m-1}}\ldots U_{j_1}$.  A more general form $U_{j_1j_2\ldots j_m}$ for the decoder allows the distribution of the pre-shared entanglement to be different, while achieving the same rates. To illustrate this fact, we construct a specific example where our multiparty merging protocol allows a different distribution of the catalytic entanglement required for merging the state, compared to a distributed compression protocol as discussed in \cite{merge}. We prove that time-sharing is not required for the case of distributed compression involving two senders. For the more general task of multiparty state merging of $m$ senders with side information at the receiver, we need to extend some of the well-known properties of typicality to the multiparty setting in order to show that time-sharing is not required. We discuss some of the difficulties in proving such results. In the next chapter, we characterize the entanglement cost of multiparty merging when a single copy is available to the parties by using the relevant entropic quantities for this regime.
 \begin{figure}[t]
  \centering
    \includegraphics{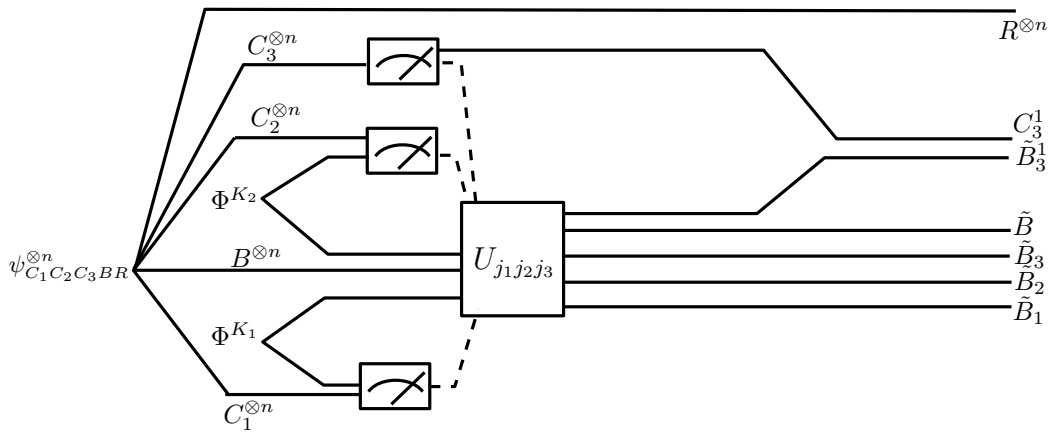}
\caption{Quantum circuit for a multiparty state merging protocol involving three senders. The senders perform simultaneous measurements and send their measurement outcomes to the receiver. Upon reception, the receiver applies a unitary to recover the sender's systems with good fidelity. In the situation depicted above, there is also additional entanglement distilled between the receiver and the sender $C_3$. }\label{fig:multipartymerging}
\end{figure}

The second part of this chapter describes a slightly more general kind of state transfer. We introduce a second receiver $A$, and for a bipartite partition ${\cal T} \subseteq \{1,2,\ldots,m\}$ of the senders, we consider the problem of redistributing the state $\psi^{C_1C_2\ldots C_m ABR}$ to the two receivers: The systems which are part of ${\cal T}$ are sent to the receiver $A$ while the systems which are part of $\overline{T}$ go to the receiver $B$. We call this task a \textit{split-transfer} for the state $\psi^{C_1C_2\ldots C_mABR}$. It has an important use in the context of assisted distillation: we show how to distill entanglement at the min-cut entanglement rate of eq.~(\ref{eq:mincutX}) by combining a split-transfer protocol with a standard distillation protocol.

Analyzing multiparty protocols requires the labeling of numerous systems, dimension quantities, measurement outcomes, etc\ldots To avoid losing the reader
with cumbersome notation, we make the following abbreviations: For a protocol involving $m$ senders, we denote by $C_M$ the composite system $C_1C_2\ldots C_m$. The capital letter $M$, written as a subscript next to a label, will always denote the composition of $m$ objects. For instance, if each of the $m$ senders have an extra system $C^0_i$, the label $C^0_M$ denotes the composite system $C^0_1C^0_2\ldots C^0_m$. The symbol ${\cal T}$ will denote either a subset of the $m$ senders or the composite system $\bigotimes_{i \in {\cal T}}C_i$. It will usually be clear which definition applies from the context. The complement of the set ${\cal T}$ is denoted by $\overline{{\cal T}}$, and may also denote the composite system $\bigotimes_{i \in {\overline{\cal T}}} C_i$.

\section{Multiparty state merging}
\subsection{Definitions and main theorem}
Let $\Lambda^m_{\rightarrow} : C_MC^0_M \otimes BB^0_M \rightarrow
C^1_M \otimes B^1_MBB_M$ be an LOCC quantum channel implemented by
$m$ senders $C_1,C_2,\ldots,C_m$ and a decoder/receiver $B$. Initially, each sender has a system $C_i$ and also
an \textit{ancilla} $C^0_i$: an extra system of dimension $K_i:=d_{C^0_i}$. The receiver also has
ancillas $B^0_M$, with $d_{B^0_i} =
d_{C^0_i}$, and $B^1_M$ with $d_{B^1_i} = d_{C^1_i}$. The systems $C^0_M$ and $B^0_M$ are in the maximally entangled state $\Phi^{K_1} \otimes \Phi^{K_2} \otimes \ldots \otimes \Phi^{K_m}$, with the state $\Phi^{K_i}$ consisting of $\log(K_i)$ ebits shared between the sender $C_i$ and the receiver $B$. After applying the channel, the senders have subsystems $C^1_1, C^1_2, \ldots C^1_m$ of $C_M$, and the receiver holds three systems:
$B$, $B^1_M$ and $B_M$, with $B_M$ being an ancillary system of dimension $d_{B_M} = d_{C_M}$.
\begin{figure}[t]
  \centering
    \includegraphics{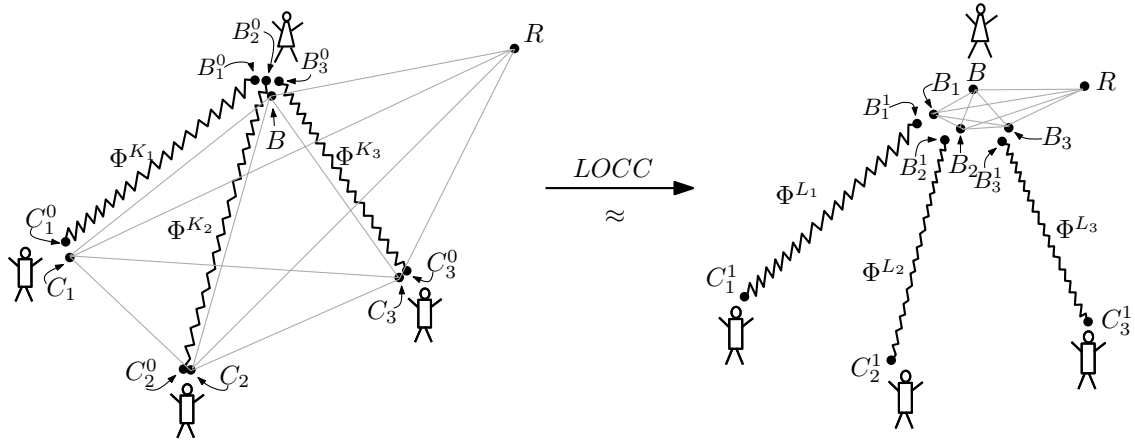}
\caption{Picture of the initial and final steps of a multiparty
state merging protocol involving three senders. The jagged lines indicate maximally entangled pairs shared between the receiver and the senders. The solid lines indicate correlation between the parties and the reference. At the end of the protocol, the systems $C_1, C_2$ and $C_3$ are transferred to the receiver. }\label{fig:statemerging}
\end{figure}

This channel implements multiparty state merging (see Figure \ref{fig:statemerging}) for the state $\psi^{C_MBR}$ if the output state $(\mathrm{id}_R \otimes \Lambda^m_{\rightarrow})(\psi^{C_MBR} \otimes \Phi^{K_M})$ is approximately a tensor product of the initial state $\initstate$ and maximally
entangled states $\Phi^{L_1} \otimes \Phi^{L_2} \otimes
\ldots \otimes \Phi^{L_m}$ shared between the senders and the
decoder. Each state $\Phi^{L_i}$ is $\log(L_i)$ ebits shared between the sender $C_i$ and the receiver. In more formal terms, we have the following definition:

\begin{definition}[$m$-Party Quantum State Merging]
Let $\Lambda^m_{\rightarrow}$ be defined as in the previous
paragraphs. We say that $\Lambda^m_{\rightarrow}$ is an
\emph{$m$-party state merging protocol} for the state $\initstate$ with error $\epsilon$ and
entanglement cost $\overrightarrow{E}:= (\log K_1 -\log L_1,\log K_2 -\log L_2, \ldots, \log K_m - \log L_m)$ if
\begin{equation*}
 \biggl \| (\mathrm{id}_{R} \otimes \Lambda^m_{\rightarrow})(\initstate \otimes \Phi^{K_M}) -
\mergestate \otimes \Phi^{L_M} \biggr \|_1 \leq
\epsilon,
\end{equation*} where the state $\mergestate$ corresponds to the initial state
$\initstate$ with the system $B_M$ substituted for $C_M$. If we
are given $n$ copies of the same state, $\psi = (\sigma)^{\otimes
n}$, the entanglement rate $\overrightarrow{R}(\sigma)$ is defined
as $\overrightarrow{R}(\sigma) :=
\frac{1}{n}\overrightarrow{E}(\psi)$.
\end{definition}

Before stating the main theorem, we need to define what it means
for a rate-tuple $\overrightarrow{R}$ to be achievable for
multiparty merging using LOCC operations.

\begin{definition}[The Rate Region]\label{def:rateregion}
We say that the rate-tuple $\overrightarrow{R}:=(R_1,R_2,\ldots,R_m)$ is \textit{achievable} for
multiparty merging of the state $\psi^{C_M BR}$ if,
for all $\epsilon > 0$, we can find an $N(\epsilon)$ such that for every $n \geq
N(\epsilon)$ there exists an $m$-party state merging
protocol $\Lambda^{m}_{n,\rightarrow}$ acting on $\psi^{\otimes n} \otimes \Phi^{K_M^n}$ with error
$\epsilon$ and entanglement rate approaching $\overrightarrow{R}$. We call the closure of the set of achievable rate-tuples the \textit{rate region}.
\end{definition}

For the task of distributed compression (i.e., no side information at the decoder), the rate region was characterized in \cite{merge} by the inequalities
\begin{equation}\label{eq:boundsRRR}
\sum_{i \in {\cal T}} R_i \geq S({\cal T}|\overline{{\cal T}})_{\psi} \phantom{==} \quad \text{for all nonempty subsets } {\cal T}\subseteq \{1,2,\ldots,m\}.
\end{equation}
If a rate $R_i$ is negative for an achievable rate-tuple $(R_1, R_2, \ldots, R_m)$, close to $-nR_i$ ebits shared between the sender $C_i$ and the receiver are also returned by the protocol. Allowing the receiver to have side information $B$ as well leads to a similar set of equations describing the rate region associated with the task of multiparty state merging.

\begin{Theorem}[$m$-Party Quantum State Merging \cite{merge}] \label{thm:statemerging}
Let $\initstate$ be a pure state shared between $m$ senders
$C_1,C_2,\ldots,C_m$ and a receiver $B$, with purifying system
$R$. Then, the rate-tuple $\overrightarrow{R}:=(R_1,R_2,\ldots,R_m)$ is part of the rate region for
multiparty merging if and only if the inequality \be \label{eq:Mainthm3} \sum_{i \in {\cal
T}} R_i \geq S({\cal T}|\overline{{\cal T}}B)_{\psi} \ee holds for
all non empty subsets ${\cal T} \subseteq \{1,2,...,m\}$.
\end{Theorem}

The theorem was proved in \cite{merge} by showing that the corner points of the region are achievable and then using time-sharing to interpolate between them. We will extend the previous result to the one-shot setting in Chapter 4. In Section \ref{sec:iid}, we show that for distributed compression of two senders, time-sharing is not required for transferring the systems $C^n_1$ and $C^n_2$ at any rate satisfying the inequalities of eq.~(\ref{eq:boundsRRR}).

\subsection{The ``fully'' decoupled approach} \label{sec:merging-many}
 \begin{figure}[t]
  \centering
    \includegraphics{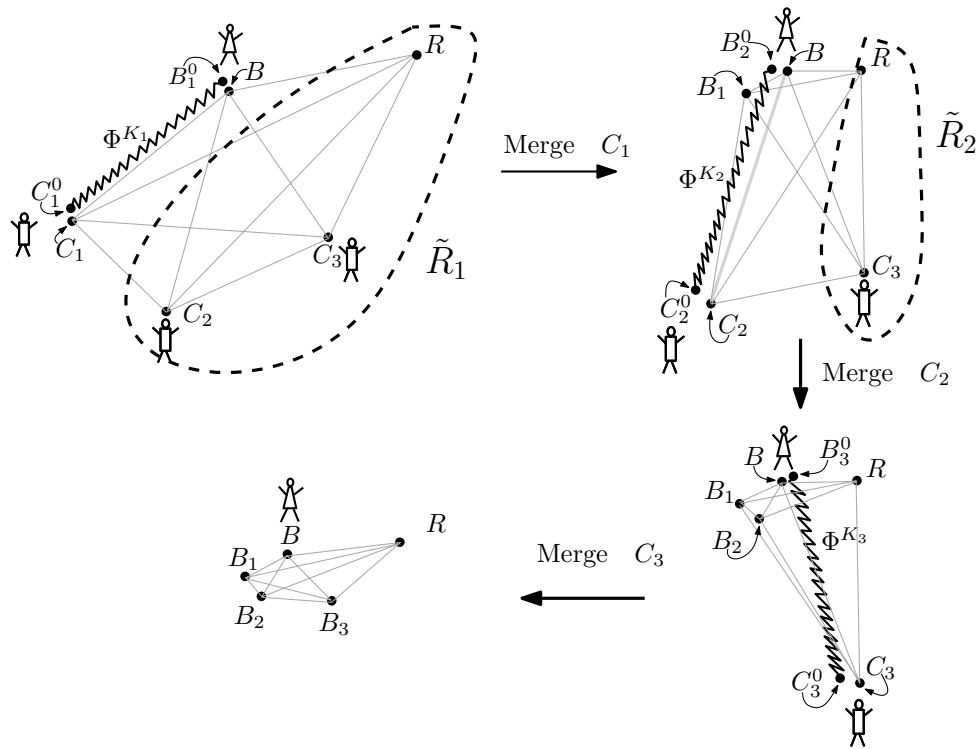}
\caption{The distributed compression protocol of \cite{merge} for the case of three senders. Systems part of a relative reference are found within the dotted region. For this scenario, we assume the senders consume just enough entanglement to transfer their systems. No maximally entangled states are returned by the protocol. }\label{fig:relativereference}
\end{figure}

 The distributed compression protocol described in \cite{merge} achieves multiparty merging by transferring the systems one at a time using the two-party state merging protocol. This has the consequence of decoupling the senders from the reference only at the end of the protocol. To better illustrate this, suppose the first $i-1$ senders have transferred their systems to the receiver. The original state $\psi^{C_MBR}$ at this point in the protocol can be written as $\psi^{C_{i} \tilde{B} \tilde{R}}$, where $\tilde{B}$ is a system of the same dimension as the composite system $BC_1C_2\ldots C_{i-1}$ and $\tilde{R}:=C_{i+2}\ldots C_mR$ is the \textit{relative} reference with respect to the sender $C_{i}$. To transfer the system $C_{i}$, the sender performs an \textit{incomplete measurement} on his composite system $C^0_i \otimes C_i$, destroying most of the correlation with the relative reference $\tilde{R}$:
 \begin{equation}
 \label{eq:decoupleR}
  \psi^{C^1_i \tilde{R}}_{j_i} \approx \tau^{C^1_i} \otimes \psi^{\tilde{R}},
 \end{equation}
 where $\psi^{C^1_i \tilde{R}}_{j_i}$ is the reduced state of $\psi^{C^1_i \tilde{B} \tilde{R}}_{j_i}$ for an outcome $j_i$. After this measurement, the senders are not entirely decoupled with the reference as correlation may exist between the senders $C_{i+1}, C_{i+2}, \ldots, C_{m}$ and the reference. Since $\Phi^{L_i} \otimes \psi^{B_i\tilde{B}\tilde{R}}$ purifies the state on the right hand side of eq.~(\ref{eq:decoupleR}), with $B_i$ being a system of the same dimension as $C_i$, there exists an isometry (Corollary \ref{cor:Ulhmann}) $V : \tilde{B}B^0_i \rightarrow B_i \tilde{B} B^1_i$ implementable by the receiver which allows the recovery of the system $C_i$ and $\log(L_i)$ ebits. The protocol continues in this fashion, with $C_{i+1}$ being the next sender to measure his systems (see Figure \ref{fig:relativereference}).

 To achieve multiparty state merging when the senders simultaneously measure their systems, as in Figure \ref{fig:multipartymerging}, we require the measurements to ``fully'' decouple each sender from the reference and all other senders. More precisely, suppose each of the senders $C_1, C_2, \ldots, C_m$ performs an incomplete measurement, described by Kraus operators $P^j_i: C^0_iC_i \rightarrow C^1_i$ mapping $C^0_iC_i$ to a subspace $C^1_i$. Then, the reduced post-measurement states $\psi^{C^1_MR}_{J_M}$, where $J_M:= (j_1, j_2, \ldots, j_m)$ are the measurement outcomes, must satisfy the stronger requirement that for all outcomes $J_M$
 \begin{equation*}
    \psi^{C^1_M R}_{J_M} \approx \tau^{C^1_1} \otimes \tau^{C^1_2} \otimes \ldots \otimes \tau^{C^1_m} \otimes \psi^{R},
 \end{equation*}
where $\tau^{C^1_i}$ is a maximally mixed state of dimension $L_i$.

Let's consider the case where the state $\psi^{C^1_MR}_{J_M}$ is exactly in the
product form $\tau^{C^1_M} \otimes \psi^{R}$.  The state
$\projstate$ purifies $\tau^{C^1_M} \otimes \psi^R$, with
purification systems $B^0_MB$. Another purification of
$\tau^{C^1_M} \otimes \psi^R$ is also given by $\Phi^{L_M} \otimes
\mergestate$, where the state $\mergestate$ corresponds to the
original state $\initstate$. It follows from the Schmidt decomposition (Theorem \ref{thm:schmidtdecomp}) that these two
purifications are related by a partial isometry $U_{J_M}: B^0_MB \rightarrow
B^1_MBB_M$ on the receiver's side such that
\begin{equation*}
(I^{C^1_MR} \otimes U_{J_M}) \projstate (I^{C^1_MR} \otimes U_{J_M})^{\dag} = \Phi^{L_M} \otimes \mergestate.
\end{equation*}
Hence, if each sender can perfectly decouple his system from the other senders and the reference, the sender's systems can be recovered at the receiver's location by applying a partial isometry $U_{J_M}$ on the systems of the receiver, which will also distill $\log(L_M)$ ebits.

The previous scenario was ideal, and in general, will not be feasible for most states $\initstate$. Hence, we relax our decoupling requirement and accept that the measurements performed by the senders will perturb the reference $\psi^{R}$ up to some tolerable disturbance, and that a small dose of correlation between the senders' shares might still be present. In more formal terms, a multiparty state merging protocol consists of the following steps (also depicted in Figure \ref{fig:multipartymerging}):

\begin{enumerate}
 \item Each sender $C_{i}$ applies a quantum instrument ${{\cal I}_i} :=\{\calE^i_j\}^X_{j=1}$ to his share of the state $\initstate \otimes \Phi^{K_M}$. The instrument components $\calE^i_j$ map the space $C_iC^0_i$ to a subspace $C^1_i$ of dimension $L_i$.
 \item The senders $C_1,C_2,\ldots,C_m$ send their classical outputs $J_M:=(j_1,j_2,\ldots,j_m)$ to the decoder $B$.
 \item The decoder uses his side information $\psi^{B}$, his share of the maximally entangled states $\Phi^{K_M}$, and the classical information $J_M$ to perform a decoding operation ${\cal D}_{J_M}:BB^0_M\rightarrow B^1_MBB_M$ (i.e a trace-preserving completely positive map (TP-CPM)) and recover the state $\initstate \otimes \Phi^{L_M}$.
\end{enumerate}
The state of the systems $C^1_MB^0_MBR$ after steps 1 and 2 are performed can be written as:
 \begin{equation}\label{eq:LOCC}
   \begin{split}
     \psi^{C^1_MB^0_MBR} &:= \sum_{J_M} (\mathrm{id}^{B^0_MBR} \otimes \calE_{J_M})(\initstate \otimes \Phi^{K_M}) \\
   &=\sum_{J_M} p_{J_M} \psi_{J_M}^{C^1_MB^0_MBR}\\
   \end{split}
 \end{equation} where $\calE_{J_M} := {\calE}^1_{j_1} \otimes {\calE}^2_{j_2} \otimes \ldots \otimes {\calE}^m_{j_m}$ and $\psi_{J_M}^{C^1_MB^0_MBR}$ is the normalized version of the state $(\mathrm{id}^{B^0_MBR} \otimes \calE_{J_M})(\initstate \otimes \Phi^{K_M})$. If we restrict the operators $\calE^i_j$ to consist of only one Kraus operator
  \begin{equation*}
   \calE^i_j(\rho)=A^i_j \rho (A^i_j)^{\dag} \quad \text{for all $i,j$}
   \end{equation*} and to satisfy
  \begin{equation*}
  \sum_j (A^i_j)^{\dag}A^i_j = I^{C_i},
   \end{equation*} the outcome states $\psi^{C^1_MB^0_MBR}_{J_M}$ are pure and are the result of performing $m$ incomplete measurements, one for each sender $C_i$.

\begin{Proposition}[Compare to Proposition 4 of \cite{merge}] \label{prop:mergeCond}
 Let $\psi_{J_M}^{C^1_MB^0_MBR}$ be defined as in eq.~(\ref{eq:LOCC}), with reduced density operator
 $\psi_{J_M}^{C^1_MR}$. Define the following quantity:
  \begin{equation*}\label{quanterror:eq}
    Q_{{\cal I}}(\initstate \otimes \Phi^{K_M}) := \sum_{J_M} p_{J_M}  \| \psi_{J_M}^{C^1_MR} -
    \tau^{C^1_M} \otimes \psi^{R} \|_1,
  \end{equation*}where $p_{J_M}$ is the probability of obtaining the state $\psi_{J_M}^{C^1_MB^0_MBR}$ after all the senders have performed their instruments.
If $Q_{\cal I}(\initstate \otimes \Phi^{K_M}) \leq \epsilon$, then there exists an $m$-party state merging protocol for the state $\initstate$ with error $2\sqrt{\epsilon}$ and entanglement cost $\overrightarrow{E} = (\log K_1-\log L_1, \log K_2 -\log L_2, \ldots, \log K_m - \log L_m)$, where $L_i := d_{C_i^1}$ for all $1\leq i \leq m$.
\end{Proposition}

\begin{proof} The proof of the above statement is very similar to the proof of Proposition 4 in \cite{merge}. We give the full proof here for completeness. Using the relation between trace distance and fidelity (see eq.~(\ref{Lemma:relation})), we have
\begin{equation*}
\sum_{J_M} p_{J_M} F(\reducstate,\decouplestate) \geq 1 - \frac{\epsilon}{2}.
\end{equation*}
From Corollary \ref{cor:Ulhmann} of Ulhmann's theorem, there exists a partial isometry (i.e., a
decoding) $U_{J_M}:B^0_MB \rightarrow B^1_MBB_M$ implementable by the receiver
such that
\begin{equation*}
F(\reducstate,\decouplestate) = F\biggl ((I^{C^1_MR} \otimes
U_{J_M})\projstate (I^{C^1_MR} \otimes U_{J_M})^{\dag},\Phi^{L_M} \otimes \mergestate \biggr ).
\end{equation*}
Using the concavity of $F$ (see \cite{Nielsen} for a proof) in its first argument, we have
\begin{equation*}
  \begin{split}
  F(&\psi_{out}^{C^1_MB^1_MB_MBR}, \Phi^{L_M} \otimes \mergestate)\\
  & \phantom{======} \geq  \sum_{J_M}p_{J_M} F\biggl ((I^{C^1_MR} \otimes
U_{J_M})\projstate (I^{C^1_MR} \otimes U_{J_M})^{\dag},\Phi^{L_M} \otimes \mergestate \biggr )  \\
  & \phantom{======} \geq 1- \frac{\epsilon}{2}, \\
   \end{split}
\end{equation*} where
  \begin{equation*}
   \psi_{out}^{C^1_MB^1_MBB_MR} := \sum_{J_M}p_{J_M} (I^{C^1_MR} \otimes U_{J_M})\projstatebraket (I^{C^1_MR} \otimes U_{J_M})^{\dag}
  \end{equation*}
is the output state of the protocol. Using the relation between fidelity and trace distance
once more, we arrive at
 \begin{equation*}
 \biggl \| \psi_{out}^{C^1_M B^1_MBB_MR} -  \Phi^{L_M} \otimes \mergestate \biggr \|_1 \leq 2\sqrt{\epsilon-\epsilon^2/4} \leq 2\sqrt{\epsilon}.
 \end{equation*}
\end{proof}

\subsection{Merging by random measurements} \label{sec:random-meas}

In this section, we give a bound on the \textit{decoupling} error when a measurement-based random coding strategy is used to achieve a multiparty state merging of the state $\psi^{C_MBR}$. The senders $C_i$ will simultaneously measure their systems $C_iC^0_i$ using $N_i = \lfloor \frac{d_{C_i}K_i}{L_i} \rfloor$ projectors of rank $L_i$, and a little remainder, followed by a unitary mapping the outcome state to a subspace $C^1_i$.

\begin{Proposition}[One-Shot Multiparty State Merging] \label{prop:isometry}
Let $\initstate \otimes \Phi^{K_M}$ be a multipartite state shared between $m$ senders and a receiver $B$. For each
sender $C_i$, there exists an instrument ${\cal I}_i=\{{\cal E}^i_j\}$ consisting of
$N_i := \lfloor \frac{d_{C_i}K_i}{L_i} \rfloor$ partial isometries of rank $L_i$ and one of rank $L_i'= d_{C_i}K_i - N_i L_i < L_i$
such that the overall decoupling error $Q_{{\cal I}}(\initstate \otimes \Phi^{K_M})$ is bounded
by
\begin{equation} \label{eq:upperbound}
  \begin{split}
  Q_{\cal I} &\leq 2 \sum_{\substack{{\cal T} \subseteq \{1,2,...,m\} \\ {\cal T} \notin \emptyset}} \prod_{i \in {\cal T}}\frac{L_i}{d_{C_i}K_i} + 2\sqrt{d_{R} \sum_{\substack{{\cal T} \subseteq \{1,2,\ldots,m\} \\ {\cal T}
\neq \emptyset}}  \prod_{i \in {\cal T}} \frac{L_i}{K_i} \Tr \bigg [ \psi^2_{R{\cal T}} \bigg ]} =: \Delta_{\cal I}, \\
  \end{split}
\end{equation}
and there is a merging protocol with error at most $2\sqrt{\Delta_{\cal I}}$.
\end{Proposition}
\newcommand \omeg {\omega^{C^1_MR}(U_M)}
\newcommand \omegJ {\omega^{C^1_MR}_{J_M}(U_M)}

\newcommand \omegapos {\omega^{\tilde{C}^1_M\tilde{R}}(U_M)}
\newcommand \omegaposa {\omega^{\widetilde{C^1_MR}}(U_M)}
\newcommand \omegsquare {\omega^2_{C^1_MR}(U_M)}
\newcommand \decouple {\frac{L_M}{d_{C_M}} \tau^{C^1_M} \otimes \psi^{R}}
\newcommand \decoupled {\frac{L_M}{d_{C_M}K_M} \tau^{C^1_M} \otimes \psi^{R}}

\newcommand {\expec}[1]{E \bigg [ #1 \bigg ]}
\newcommand {\smexpec}[1]{E[#1]}

To prove this proposition, we show that the average decoupling error when the senders perform random instruments using the Haar measure is bounded from above by the right hand side of eq.~(\ref{eq:upperbound}). We will need the following technical
lemma, which generalizes Lemma 6 in \cite{merge} to the case of
$m$ senders. The proof follows a similar line of reasoning.
\begin{Lemma}[Compare to Lemma 6 in \cite{merge}] \label{Lemma:rdmisometry}
For each sender $C_i$, let $Q_i$ be a projector of dimension $L_i$ onto a subspace $C^1_i$ of $C_i$ and $U_i$ a unitary acting on $C_i$. Define the sub-normalized density operator
\begin{equation*}
 \omeg  := (Q_1U_1 \otimes Q_2U_2 \otimes \ldots \otimes Q_mU_m \otimes I^{R}) \psi^{C_MR}
(Q_1U_1 \otimes Q_2U_2  \otimes \ldots \otimes Q_mU_m \otimes I^{R})^{\dag},
\end{equation*}
where $U_M:=U_1 \otimes U_2 \otimes \ldots \otimes U_m$. We have
\begin{equation}\label{eq:decoupling}
\begin{split}
\displaystyle
\int_{\mathbb{U}(C_1)}\int_{\mathbb{U}(C_2)}\cdots\int_{\mathbb{U}(C_m)}
\bigg \| &\omeg  - \decouple \bigg \|_1 dU_M  \\
&\leq \frac{L_M}{d_{C_M}} \sqrt{d_{R} \sum_{\substack{{\cal T} \subseteq
\{1,2,\ldots,m\}\\ {\cal T} \neq \emptyset}}  \prod_{i \in {\cal
T}} L_i \Tr \bigg [ \psi^2_{R{\cal T}} \bigg ]},
\end{split}
\end{equation}
where the average is taken over the unitary groups $\mathbb{U}(C_1), \mathbb{U}(C_2), \ldots, \mathbb{U}(C_m)$ using the Haar measure. Here $dU_M:=dU_1dU_2\ldots dU_m$ and $\int_{\mathbb{U}(C_i)} dU_i=1$ for all $1 \leq i \leq m$.
\end{Lemma}

\newcommand{\inte}{\int_{\mathbb{U}(C_M)}}
\begin{proof} For the remainder of this proof, write the integral $\int_{\mathbb{U}(C_1)}\int_{\mathbb{U}(C_2)}\cdots\int_{\mathbb{U}(C_m)}dU_M$ as $\int_{\mathbb{U}(C_M)}dU_M$. First, we simplify the integral using Lemma \ref{lem:tensorlemma}, found in Appendix~A:
\begin{equation}
  \begin{split} \label{variance:eq}
 \inte \bigg \| &\omeg - \decouple \bigg \|^2_2 dU_M \\
  &=  \inte \bigg \| \omeg -\inte \omeg dU_M \bigg \|^2_2 dU_M \\
  &=  \inte \Tr\bigg [\omegsquare\bigg ] dU_M - \Tr\biggl [ \bigg (\inte \omeg dU_M\bigg )^2 \biggr ],
 \end{split}
\end{equation}
To evaluate the average of $\Tr[\omegsquare]$, we use the swap ``trick'' (Lemma \ref{lem:swaptrick}):
\begin{equation*}
\begin{split}
 \Tr[\omegsquare] &= \Tr \bigg [ \bigg (\omeg \otimes \omegaposa\bigg ) F^{C^1_MR,\widetilde{C^1_MR}}  \bigg ] \\
 &= \Tr \bigg [ \bigg(\omeg \otimes \omegapos\bigg ) F^{C^1_M\tilde{C}^1_M} \otimes F^{R\tilde{R}} \bigg]. \\
 \end{split}
\end{equation*}
The second line was obtained using Lemma \ref{lem:swaptensor}. By expanding the right hand side of this equality, the average of $\Tr[\omegsquare]$ becomes equal to:
\newcommand {\swapC} {F^{C^1_M\tilde{C}^1_M}}
\newcommand {\swapCs} {F^{\scriptscriptstyle C^1_M\tilde{C}^1_M}}
\newcommand {\swapRs} {F^{\scriptscriptstyle R\tilde{R}}}
\newcommand \swapR {F^{R\tilde{R}}}
\newcommand \UU {UU_{C_M\tilde{C}_M}}
\begin{equation}
  \begin{split}
 &\inte \Tr\bigg [\omegsquare \bigg ] dU_M \\
 &= \inte \Tr \bigg [\bigg (\omeg \otimes \omegapos \bigg )\swapC \otimes \swapR \bigg ] dU_M \\
&= \inte \Tr \bigg [(U_{\scriptscriptstyle M} \otimes \tilde{U}_{\scriptscriptstyle M}
\otimes I^{\scriptscriptstyle R\tilde{R}} ) ( \psi^{\scriptscriptstyle C_MR} \otimes
\psi^{\scriptscriptstyle \tilde{C}_M\tilde{R}}) (U_{\scriptscriptstyle M} \otimes \tilde{U}_{\scriptscriptstyle M} \otimes I^{\scriptscriptstyle R\tilde{R}})^{\dag}(\swapCs \otimes \swapRs) \bigg ] dU_M \\
&= \Tr \bigg [ (\psi^{C_MR} \otimes \psi^{\tilde{C}_M\tilde{R}} )\bigg (\inte (U_M \otimes \tilde{U}_M)^{\dag} \swapC (U_M \otimes
\tilde{U}_M)dU_M \bigg ) \otimes \swapR \bigg ], \label{avgomega:eq}
  \end{split}
\end{equation}
where the unitary $\tilde{U}_M$ is a ``copy'' of $U_{M}$ which
acts on $\tilde{C}_M$. Observe that the projections $Q_1, Q_2, \ldots, Q_m$ from the
state $\omega^{C^1_MR}$ were absorbed by the swap operators:
\begin{equation*}
F^{C^1_i\tilde{C}^1_i} = Q_i \otimes Q_i F^{C_i\tilde{C}_i} Q_i \otimes Q_i.
\end{equation*} The average $\inte (U_{M} \otimes \tilde{U}_M)^{\dag} \swapC(U_{M} \otimes
\tilde{U}_{M})dU_M$ is expanded using Lemma \ref{lem:swaptensor}:
\begin{equation} \label{eq:expand1}
 \inte (U_M \otimes \tilde{U}_M)^{\dag} (\swapC) (U_M \otimes \tilde{U}_M) dU_M =
   \bigotimes^m_{i=1} \inte (U_i \otimes \tilde{U}_i)^{\dag} F^{C^1_i\tilde{C}^1_i} (U_i \otimes \tilde{U}_i) dU_i. \\
\end{equation}
Using Proposition \ref{prop:twirl}, we have
\begin{equation}\label{eq:expand2}
\inte (U_i \otimes \tilde{U}_i)^{\dag} F^{C^1_i\tilde{C}^1_i} (U_i \otimes \tilde{U}_i) dU_i =r_i I^{C_i\tilde{C}_i} + s_i
F^{C_i\tilde{C}_i},
\end{equation}
where the coefficients $r_i$ and $s_i$ are given by
\begin{equation}
 \begin{split}\label{coeffs:eq}
   r_i = \frac{L_i(d_{C_i} - L_i)}{d_{C_i}(d^2_{C_i} - 1)} \leq \frac{L_i}{d^2_{C_i}}, \\
 s_i = \frac{L_i(L_id_{C_i} - 1)}{d_{C_i}(d^2_{C_i} -1)} \leq \frac{L^2_i}{d^2_{C_i}}. \\
\end{split}
\end{equation} Substituting eqs.~(\ref{eq:expand1}), (\ref{eq:expand2}) into eq.~(\ref{avgomega:eq}), we get
\begin{equation}
  \begin{split}\label{set:eq}
    \inte \Tr \bigg[ \omegsquare \bigg ]dU_M &= \Tr \bigg [ ( \psi^{C_MR} \otimes \psi^{\tilde{C}_M\tilde{R}}) \bigotimes_{i=1}^m
    \bigg (r_i I^{C_i\tilde{C}_i} + s_i F^{C_i\tilde{C}_i} \bigg ) \otimes F^{R\tilde{R}} \bigg ] \\
   &= \sum_{{\cal T} \subseteq \{1,2,\ldots,m\}} \prod_{i \notin {\cal T}} r_i \prod_{i \in {\cal T}} s_i \Tr \bigg
    [ \psi^2_{R{\cal T}} \bigg ], \\
  \end{split}
\end{equation}
where the symbol ${\cal T}$ appearing in $\psi^{R{\cal T}}$ denotes the
composite system $\otimes_{i \in {\cal T}}C_i$. When ${\cal T}$ is the empty
set, the second line of the previous equation reduces to
$\prod_{i=1}^m r_i \Tr[\psi^2_R]$. From eq.~(\ref{coeffs:eq}),
we can bound this quantity from above by:
\begin{equation*}
\begin{split}
\prod_{i=1}^m r_i \Tr[\psi^2_R] &\leq \frac{L_M}{d^2_{C_M}} \Tr[\psi^2_R] \\
& = \Tr \bigg [\frac{L_M^2}{d^2_{C_M}} \tau^2_{C^1_M} \otimes \psi^2_R \bigg ] \\
& = \Tr \bigg [ \bigg ( \frac{L_M}{d_{C_M}} \tau^{C^1_M} \otimes \psi^{R} \bigg )^2 \bigg ] \\
& = \Tr \biggl [ \bigg ( \inte \omeg dU_M \bigg )^2 \biggr ]. \\
\end{split}
\end{equation*}
Hence, using eqs.~(\ref{variance:eq}), (\ref{coeffs:eq}), (\ref{set:eq}) and the previous bound, we have
\begin{equation}\label{eq:decouple2}
  \inte \bigg \| \omeg - \decouple \bigg \|^2_2 dU_M \leq \sum_{\substack{{\cal T} \subseteq \{1,2,\ldots,m\}\\ {\cal T} \notin \emptyset}} \prod_{i \notin {\cal T}} \frac{L_i}{d^2_{C_i}} \prod_{i \in {\cal T}} \frac{L^2_i}{d^2_{C_i}} \Tr
  \bigg [ \psi^2_{R{\cal T}} \bigg ].
\end{equation}
To obtain a bound on the average of eq.~(\ref{eq:decoupling}), we use Lemma \ref{lem:relHS}:
\begin{equation*}
  \begin{split}
    \inte \bigg \| \omeg  - &\decouple \bigg \|^2_1 dU_M  \\
    & \leq L_M d_{R} \inte  \bigg \| \omeg - \decouple \bigg \|^2_2 dU_M \\ & \leq L_M d_{R} \sum_{\substack{{\cal T} \subseteq \{1,2,\ldots,m\}\\ {\cal T} \neq \emptyset}} \prod_{i \notin {\cal T}} \frac{L_i}{d^2_{C_i}} \prod_{i \in {\cal T}} \frac{L^2_i}{d^2_{C_i}} \Tr \bigg [ \psi^2_{R{\cal T}} \bigg ] \\ & \leq L_M^2 \frac{d_{R}}{d^2_{C_M}} \sum_{\substack{{\cal T} \subseteq \{1,2,\ldots,m\}\\ {\cal T} \neq \emptyset}}  \prod_{i \in {\cal T}} L_i  \Tr \bigg [ \psi^2_{R{\cal T}} \bigg ]. \\
   \end{split}
\end{equation*} Finally, using the concavity of the square root function, we have
\begin{equation*}
  \inte \bigg \| \omeg  - \decouple \bigg \|_1 dU_M \leq \frac{L_M}{d_{C_M}}
\sqrt{d_{R} \sum_{\substack{{\cal T} \subseteq \{1,2,\ldots,m\}\\
{\cal T} \neq \emptyset}}  \prod_{i \in {\cal T}} L_i \Tr \bigg [
\psi^2_{R{\cal T}} \bigg ]}.
\end{equation*}
\end{proof}
\begin{proof+}{of Proposition \ref{prop:isometry}} For each sender $C_i$, fix $N_i:= \lfloor \frac{d_{C_i}K_i}{L_i} \rfloor$ orthogonal subspaces $W^1_i, W^2_i, \ldots, W^{N_i}_i$ of dimensions $L_i$ and one subspace $W^0_i$ of dimension $L'_i = d_{C_i}K_i - N_iL_i < L_i$. For each subspace $W^j_i$, let $V^j_i$ be an isometry from the subspace $W^j_i$ to a subspace $C^1_i$ of dimension $L_i$. Let $Q^j_i := V^j_i P^j_i$ be partial isometries, where $P^j_i$ is the projector onto the subspace $W^j_i$. Note that $Q^0_i$ maps to a subspace of $C^1_i$ of dimension $L'_i < L_i$. Choose $m$ unitaries $U_1, U_2, \ldots, U_m$ using the Haar distribution, with $U_i$ acting on $C^0_iC_i$. Set the instrument ${\cal I}_i$ for the sender $C_i$ to have components ${\cal E}^j_i(\rho) := (Q^j_i U_i) \rho (Q^j_i U_i)^{\dag}$ for $0 \leq j \leq N_i$ (See Figure \ref{fig:instrument}).
 \begin{figure}[t]
  \centering
    \includegraphics{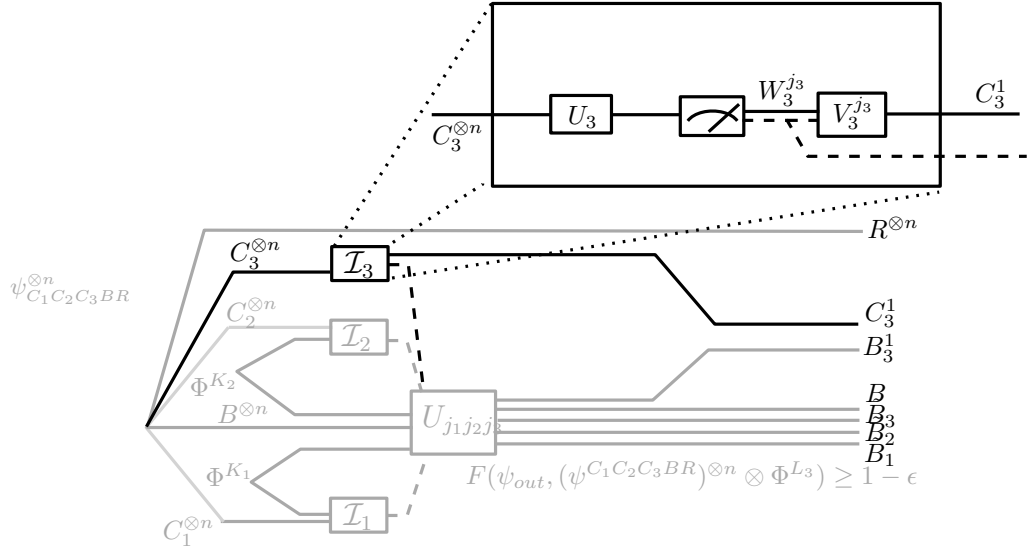}
\caption{Quantum circuit of the inner workings of an instrument ${\cal I}_3$ performed by the sender $C_3$ as described in the proof of Proposition \ref{prop:isometry}.}\label{fig:instrument}
\end{figure}

Define the state \[\omegJ := (Q^{j_1}_1U_1 \otimes \ldots \otimes Q^{j_m}_mU_m \otimes I^{R})\psi^{C_MR}\otimes \tau^{C^0_M}(Q^{j_1}_1U_1 \otimes \ldots \otimes Q^{j_m}_mU_m \otimes I^{R})^{\dag}.\] For an outcome $J_M=(j_1, j_2,\ldots, j_m)$, with $j_i \neq 0$ for all $1 \leq i \leq m$, the trace norm $\bigg \| \omegJ - \decoupled \bigg \|_1$ is bounded from above on average by:
\begin{equation}
\begin{split}
  \inte \sum_{j_{1}=1}^{N_1} \sum_{j_2=1}^{N_2} \cdots &\sum_{j_m=1}^{N_m} \bigg \| \omegJ
- \decoupled \bigg \|_1 dU_M \\
&\leq \bigg ( \prod_{i=1}^m N_i
\bigg ) \frac{L_M}{d_{C_M}K_M} \sqrt{d_{R} \sum_{\substack{{\cal T}
\subseteq \{1,2,...,m\} \\ {\cal T} \neq \emptyset}}  \prod_{i \in
{\cal T}} \frac{L_i}{K_i}  \Tr
\bigg [ \psi^2_{R{\cal T}} \bigg ]} \\
&\leq \sqrt{d_{R} \sum_{\substack{{\cal T} \subseteq \{1,2,...,m\}
\\ {\cal T} \neq \emptyset}}  \prod_{i \in {\cal T}}
\frac{L_i}{K_i}
\Tr \bigg [ \psi^2_{R{\cal T}} \bigg ]}, \\
\end{split}\label{eq:measurement}
\end{equation}
where we have used Lemma \ref{Lemma:rdmisometry} to get the bound, using the fact that
\begin{equation*}
\begin{split}
\bigg \| \omegJ - &\decoupled \bigg \|_1 \\
&= \\
\bigg \| (V_{J_M} \otimes I^R)\omega^{W}_{J_M}(U_M)(V_{J_M} \otimes I^R)^{\dag} - &(V_{J_M} \otimes I^R)(\frac{L_M}{d_{C_M}K_M}\tau^{W}_{J_M}\otimes \psi^R)(V_{J_M} \otimes I^R)^{\dag} \bigg \|_1, \\
\end{split}
\end{equation*}
where $V_{J_M} := V^{j_1}_1 \otimes V^{j_2}_2 \otimes \ldots \otimes V^{j_m}_m$ is the tensor product of isometries mapping the projected subspaces $W^{j_1}_1, W^{j_2}_2, \ldots, W^{j_m}_m$ to $C^1_M$, and the state $\omega^{W}_{J_M}$ is equal to:
\[
\omega^{W}_{J_M}(U_M) := (P^{j_1}_1U_1 \otimes \ldots \otimes P^{j_m}_mU_m \otimes I^{R})\psi^{C_MR}\otimes \tau^{C^0_M}(P^{j_1}_1U_1 \otimes \ldots \otimes P^{j_m}_mU_m \otimes I^{R})^{\dag}.\]
Taking the normalisation into account, with $p_{J_M}(U_M)=\Tr(\omega^{C^1_MR}_{J_M}(U_M))$ and $\psi^{C^1_MR}_{J_M}=\frac{1}{p_{J_M}}\omega^{C^1_MR}_{J_M}(U_M)$, we need to show that on average the $p_{J_M}(U_M)$ are close to $\frac{L_M}{d_{C_M}K_M}$. Looking at eq.~(\ref{eq:measurement}) and tracing out, we get
\begin{equation*}
  \inte \sum_{j_{1}=1}^{N_1} \sum_{j_2=1}^{N_2} \cdots \sum_{j_m=1}^{N_m} \bigg | p_{J_M}(U_M) -
\frac{L_M}{d_{C_M}K_M} \bigg | dU_M \leq \sqrt{d_{R}
\sum_{\substack{{\cal T} \subseteq \{1,2,\ldots,m\}
\\ {\cal T} \neq \emptyset}}  \prod_{i \in {\cal T}} \frac{L_i}{K_i} \Tr \bigg [ \psi^2_{R{\cal T}} \bigg ]}. \\
\end{equation*} Hence we obtain, using the triangle inequality,
\begin{equation*}
\begin{split}
 \inte \sum_{j_{1}=1}^{N_1} \sum_{j_2=1}^{N_2} \cdots \sum_{j_m=1}^{N_m}& p_{J_M}(U_M) \bigg \|
\psi^{C^1_MR}_{J_M} - \tau^{C^1_M} \otimes \psi^{R} \bigg \|_1 dU_M \\
&\leq 2\sqrt{d_{R} \sum_{\substack{{\cal T} \subseteq
\{1,2,\ldots,m\}
\\ {\cal T} \neq \emptyset}}  \prod_{i \in {\cal T}}
\frac{L_i}{K_i} \Tr \bigg [ \psi^2_{R{\cal T}} \bigg ]} =:
\Gamma_{\psi \otimes \Phi^{K_M}}.
\end{split}
\end{equation*}
Lastly, we need to consider what happens when at least one sender $C_i$ obtains a measurement outcome $j_i$ equal to 0. For an outcome
$J_M$, define the subset ${\cal T}(J_M) \subseteq \{1,2,...m\}$ such that $i \in {\cal T}(J_M)$ if and only if $j_i = 0$. Also,
define the set ${\cal Z} = \{J_M:|{\cal T}(J_M)| > 0\}$. Then, it is easy to show that the cardinality of the set ${\cal Z}$ is
 \begin{equation*}
 |{\cal Z}| = \sum_{\substack{{\cal T} \subseteq \{1,2,...,m\} \\ {\cal T} \neq \emptyset}} \prod_{i \notin {\cal T}} N_i.
 \end{equation*}
For an outcome $J_M \in {\cal Z}$, the average probability of the state $\omega^{C^1_MR}_{J_M}(U_M)$ is given by
\begin{equation*}
 \begin{split}
 \inte p_{J_M}(U_M) dU_M &= \Tr \bigg [ \inte \omega^{C^1_MR}_{J_M}(U_M) dU_M \bigg ]\\
 &= \Tr \bigg [ \bigotimes_{i \in {\cal T}(J_M)} Q^0_i \tau^{C_iC^0_i} (Q^0_i)^{\dag} \bigotimes_{i \notin {\cal T}(J_M)} Q^{j_i}_i \tau^{C_iC^0_i} (Q^{j_i}_i)^{\dag} \bigg ] \\
 &=\frac{\prod_{i \in {\cal T}(J_M)} L_i' \prod_{i \notin {\cal T}(J_M)} L_i}{d_{C_M}K_M}.
 \end{split}
\end{equation*}
With this formula in hand, and the fact that the trace norm between two states is at most two, we recover the left hand side of eq.~(\ref{eq:upperbound}) for the average value of the decoupling error $Q_{\cal I}(\initstate \otimes \Phi^{K_M})$:
\begin{equation} \label{eq:final}
  \begin{split}
   \inte \sum_{j_{1}=0}^{N_1} \sum_{j_2=0}^{N_2} \cdots \sum_{j_m=0}^{N_m} p_{J_M} \bigg \| \psi^{C^1_MR}_{J_M} - &\tau^{C^1_M} \otimes \psi^R \bigg \|_1 dU_M \\ &\leq 2\sum_{\substack{{{\cal T} \subseteq \{1,2,...,m\}} \\ {\cal T} \neq \emptyset}} \frac{\prod_{i \in {\cal T}} L_i' \prod_{i \notin {\cal T}} N_iL_i}{d_{C_M}K_M} + \Gamma_{\psi \otimes \Phi^{K_M}}\\
&\leq 2 \sum_{\substack{{\cal T} \subseteq \{1,2,...,m\} \\ {\cal T} \neq \emptyset}}\prod_{i \in {\cal T}} \frac{L_i'}{d_{C_i}K_i} + \Gamma_{\psi \otimes \Phi^{K_M}} \\
&\leq 2 \sum_{\substack{{\cal T} \subseteq \{1,2,...,m\} \\ {\cal T} \neq \emptyset}}\prod_{i \in {\cal T}} \frac{L_i}{d_{C_i}K_i} + \Gamma_{\psi \otimes \Phi^{K_M}}.
  \end{split}
\end{equation}
Hence, there exist instruments ${\cal I}_1, {\cal I}_2, \ldots,{\cal I}_m$ implementable by the senders with decoupling error at most the right hand side of this inequality. From Proposition \ref{prop:mergeCond}, we recover the second statement of the proposition, and so, we are done.
\end{proof+}

\subsection{Asymptotic analysis} \label{sec:iid}
In this section, we analyze the case where the parties have at
their disposal arbitrarily many copies of the state $\initstate$. Theorem \ref{thm:statemerging} was proved in \cite{merge} by relying on a time-sharing strategy. In this section, we give a proof that our protocol requires no time-sharing for the special case of distributed compression involving two senders. We then discuss the main difficulty when attempting to generalize our proof technique for the general task of multiparty state merging. We also give examples to illustrate the benefits of our protocol over the distributed compression protocol of~\cite{merge}.
\subsubsection{Proof of Theorem \ref{thm:statemerging} for two senders}

To prove the direct statement of the theorem for two senders and no side information at the receiver's location, we use Proposition \ref{prop:isometry} in combination with Schumacher compression \cite{Coding}.
For $n$ copies of the state $\psi^{C_1C_2R}$, consider the Schumacher
compressed state \be \ket{\Omega}^{\tilde{C}_1\tilde{C}_2\tilde{R}} := (\Pi_{\ti{C}_1} \otimes
\Pi_{\ti{C}_2} \otimes \Pi_{\ti{R}}) (\ket{\psi}^{C_1C_2R})^{\otimes n},\ee
and its normalized version $\ket{\Psi}^{\tilde{C}_1\tilde{C}_2\tilde{R}} := \frac{1}{\sqrt{\langle
\Omega | \Omega \rangle}}\ket{\Omega}^{\tilde{C}_1\tilde{C}_2\tilde{R}}$. The systems
$\ti{C}_1,\ti{C}_2$ and $\ti{R}$ are the $\delta$-typical subspaces of $C^n_1,C^n_2$ and $R^n$. The projectors onto these subspaces are denoted by $\Pi_{\ti{C}_1},\Pi_{\ti{C}_2}$ and $\Pi_{\ti{R}}$. For any $\epsilon > 0$ and $n$ large enough, we have
\begin{equation*}
\begin{split}
 \Tr(\psi^{\otimes n}_{R}\Pi_{\ti{R}}) \geq 1-\epsilon, \quad \Tr(\psi^{\otimes n}_{C_1}\Pi_{\ti{C}_1}) \geq 1-\epsilon \quad \Tr(\psi^{\otimes n}_{C_2}\Pi_{\ti{C}_2}) \geq 1-\epsilon.\\
\end{split}
\end{equation*}
Using Lemma \ref{Lem:operatorineq}, which generalizes the ``union'' bound of Abeyesinghe et al. \cite{Hayden001} to an arbitrary number of typical projectors, we can bound the norm of $\Omega^{\tilde{C}_1\tilde{C}_2\tilde{R}}$:
\begin{equation}
 \langle \Omega | \Omega \rangle = \bra{\psi}^{\otimes n} \Pi_{\ti{C}_1} \otimes \Pi_{\ti{C}_2} \otimes \Pi_{\ti{R}} \ket{\psi}^{\otimes n} \geq 1-3\epsilon.
\end{equation}
The properties for the typical projectors
$\Pi_{\ti{C}_1},\Pi_{\ti{C}_2}$ and $\Pi_{\ti{R}}$ allow us to
tightly bound the various dimensions and purities appearing in
Proposition \ref{prop:isometry} by appropriate ``entropic''
formulas. In particular, we have (see Chapter 2) for $n$ large enough and any \textit{single} system $F=C_1, C_2, R$:
  \begin{equation} \label{eq:typproperties}
        \begin{split}
         (1-\epsilon)2^{n(S(F)_{\psi} - \delta)} &\leq \Tr[\Pi_{\ti{F}}] \leq 2^{n(S(F)_{\psi}+\delta)} \\
              &\Tr [ \Psi_{\ti{F}}^2 ] \leq (1-\epsilon)^{-2}2^{-n(S(F)_{\psi} - 3\delta)}, \\
         \end{split}
  \end{equation}
where $\delta > 0$ is a typicality parameter. Let $\overrightarrow{R} = (R_1, R_2)$ be any rate-tuple which satisfies the inequalities
\begin{equation} \label{eq:constraints}
\begin{split}
R_1 &> S(C_1|C_2)_{\psi}\\
R_2 &> S(C_2|C_1)_{\psi}\\
R_1+R_2 &> S(C_1C_2)_{\psi}.
\end{split}
\end{equation}
We construct a family of multiparty merging protocols on the state $\Psi^{\tilde{C}_1\tilde{C}_2\tilde{R}}$ with vanishing error as follows: If $R_i \leq 0$, the sender $C_i$ performs a random instrument using projectors of rank $L_i := \lfloor 2^{-nR_i} \rfloor$ (and possibly one of rank $L'_i \leq L_i$). No maximally entangled state $\Phi^{K_i}$ is shared with the receiver (i.e $K_i := 1$). If $R_i > 0$, the sender $C_i$ shares a maximally entangled state $\Phi^{K_i}$ of rank $K_i := \lceil 2^{nR_i} \rceil$ with the receiver and performs a random instrument with rank one projectors ($L_i := 1$).
From Proposition \ref{prop:isometry}, the average decoupling error $Q_{\cal I}(\Psi^{\tilde{C}_1\tilde{C}_2\tilde{R}} \otimes
\Phi^{K_1}\otimes \Phi^{K_2})$ is then bounded from above by
\begin{equation} \label{eq:qerrorY}
 \begin{split}
\int_{\mathbb{U}(\tilde{C}_1)}\int_{\mathbb{U}(\tilde{C}_2)} &Q_{\cal I}(\Psi^{\tilde{C}_1\tilde{C}_2\tilde{R}} \otimes \Phi^{K_1}\otimes \Phi^{K_2}) dU_1dU_2 \\ &\leq 2\sum_{\substack{{\cal T} \subseteq \{1,2\} \\ {\cal T} \neq \emptyset}} \prod_{i \in {\cal T}} \frac{L_i}{d_{\ti{C}_i}K_i} + 2\sqrt{d_{\ti{R}} \sum_{\substack{{\cal T} \subseteq \{1,2\} \\ {\cal T} \neq \emptyset}}  \prod_{i \in {\cal T}} \frac{L_i}{K_i} \Tr \bigg [ \Psi^2_{\ti{R}\ti{{\cal T}}} \bigg ]} \end{split}
\end{equation}
Since the purities $\Tr[\Psi_{\tilde{R}\tilde{C}_1}^2], \Tr[\Psi_{\tilde{R}\tilde{C}_2}^2]$ and $\Tr[\Psi_{\tilde{R}\tilde{C}_1\tilde{C}_2}^2]$ are equal to $\Tr[\Psi_{\tilde{C}_2}^2], \Tr[\Psi_{\tilde{C}_1}^2]$ and one respectively, the previous inequality simplifies to:
\begin{equation}
\int_{\mathbb{U}(\tilde{C}_1)}\int_{\mathbb{U}(\tilde{C}_2)} Q_{\cal I}(\Psi^{\tilde{C}_1\tilde{C}_2\tilde{R}} \otimes \Phi^{K_1}\otimes \Phi^{K_2}) dU_1dU_2\\
\leq 2\Gamma_{\Psi} + \frac{2}{1-\epsilon} \Upsilon_{\Psi}, \\
\end{equation}
where
\begin{equation}
\begin{split}
\Gamma_{\Psi} &:= \frac{2^{-n(R_1 +S(C_1)_{\psi}-\delta)}}{(1-\epsilon)} + \frac{2^{-n(R_2 +S(C_2)_{\psi}-\delta)}}{(1-\epsilon)} + \frac{2^{-n(R_1 + R_2 + S(C_1)_{\psi} + S(C_2)_{\psi}-2\delta)}}{(1-\epsilon)^2} \\
\Upsilon_{\Psi}&:= \sqrt{2^{-n (R_1 + S(C_2)_{\psi} - S(R)_{\psi} - 4\delta)} +2^{-n (R_2 + S(C_1)_{\psi} -S(R)_{\psi} -4\delta)} + 2^{-n (R_1 + R_2-S(R)_{\psi} -2\delta)}}.
\end{split}
\end{equation}
From the rate constraints of eq.~(\ref{eq:constraints}) and the subadditivity of the von Neumann entropy, we can set the typicality parameter $\delta$ by choosing $n$ large enough that
\begin{equation}\label{eq:typical}
\begin{split}
  R_1 + S(C_2)_{\psi}-S(R)_{\psi} -4\delta &> 0 \\
  R_2 + S(C_1)_{\psi}-S(R)_{\psi} -4\delta &> 0 \\
  R_1 + R_2 -S(R)_{\psi} -2\delta &> 0 \\
\end{split}
\end{equation}
Hence, for these values of $\epsilon$ and $\delta$, the bound of eq.~(\ref{eq:qerrorY}) vanishes for large values of $n$. By the Gentle Measurement Lemma and the triangle inequality, we have
 \begin{equation} \bigg \| (\psi^{C_1C_2R})^{\otimes n} - \Psi^{\tilde{C}_1\tilde{C}_2\tilde{R}} \bigg \|_1 \leq 4 \sqrt{3\epsilon},
 \end{equation}
and so, if we apply the same protocol on the state $(\psi^{C_1C_2R})^{\otimes n}$, we get an error of $\textit{O}(\sqrt{\epsilon})+ \textit{O}(2^{-n\delta})$. Since $\epsilon$ can be made arbitrarily small, there exists for $n$ large enough, a family of multiparty merging protocols with arbitrarily small error and entanglement rate approaching $\overrightarrow{R}$. Hence, the rate-tuple $\overrightarrow{R}$ is achievable. To recover the full rate region, we take the closure of the set of rate-tuples satisfying the constraints of eq.~(\ref{eq:constraints}). This proves the direct part of Theorem \ref{thm:statemerging} for the case of distributed compression involving two senders.

The converse part of the theorem is easily established for an arbitrary number of senders by using the converse statement of the state merging theorem (i.e $m=1$) of \cite{merge}. Suppose the receiver has obtained the systems $\overline{{\cal T}}$ and a sender holds the entire remaining system ${\cal T}$ to be transferred. Then, an ebit rate of at least $S({\cal T}|\overline{{\cal T}}B)_{\psi}$ is required to transfer $\cK$ by the converse of the merging theorem of \cite{merge}. Obviously, if the system ${\cal T}$ is distributed across $|{\cal T}|$ senders and only LOCC operations are allowed, a total rate $\sum_{i \in {\cal T}}R_i$ of at least $S({\cal T}|\overline{{\cal T}}B)_{\psi}$ is also needed.  \qed\medskip

To understand why the previous approach fails to generalize to more than two senders, consider the Schumacher compressed state
\be\label{eq:statee} \ket{\Omega_1}^{\tilde{C}_M\tilde{B}\tilde{R}} := (\Pi_{\ti{C}_1} \otimes \Pi_{\ti{C}_2} \otimes \ldots \otimes \Pi_{\ti{C}_m} \otimes \Pi_{\tilde{B}} \otimes \Pi_{\ti{R}}) (\ket{\psi}^{C_MBR})^{\otimes n}.\ee
We would like to bound the various purities and dimensions for this state by appropriate ``entropic'' formulas as we did earlier for the case of two senders. However, it is most likely that the purities $\Tr[(\Omega_1^{\ti{R}\cK})^2]$ are not bounded by $2^{-n(S(R\cK)_{\psi}-3\delta)}$ for all non empty subsets $\cK \subseteq \{1,2,\ldots, m\}$, and so, we cannot conclude that the decoupling error appearing in eq.~(\ref{eq:qerrorY}) vanishes as $n \rightarrow \infty$. For $m=3$, a state of the form
\begin{equation*}
\ket{\Theta}^{\ti{C}_1\ti{C}_2\ti{C}_3}:= \Pi_{\ti{C}_1\ti{C}_2\ti{C}_3}\Pi_{\ti{C}_2\ti{C}_3}\Pi_{\ti{C}_1\ti{C}_3}\Pi_{\ti{C}_1\ti{C}_2}\Pi_{\ti{C}_1}\Pi_{\ti{C}_2}\Pi_{\ti{C}_3}(\ket{\psi}^{C_1C_2C_3})^{\otimes n}
\end{equation*}
is probably a better candidate for satisfying the typicality bounds on the purities of the various reduced states of $\Theta^{\ti{C}_1\ti{C}_2\ti{C}_3}$. Note that $\Pi_{\ti{C}_1}$ in the previous equation is a shorthand for $I^{\ti{C}_2\ti{C}_3}\otimes \Pi_{\ti{C}_1}$, where $\Pi_{\ti{C}_1}$ is the projector onto the $\delta-$typical subspace for $(\psi^{C_1})^{\otimes n}$. We think this is an issue that will come up often when performing asymptotic analysis of multiparty quantum communication protocols. 
We make the following conjecture about typicality in a multiparty scenario:
\begin{conjecture}[Multiparty typicality conjecture]\label{eq:notimesharing}
 Consider $n$ copies of an arbitrary multipartite state $\psi^{C_1C_2\ldots C_m}$. For any fixed $\epsilon > 0, \delta_{\cK} > 0$ and $n$ large enough, there exists a state $\Psi^{C_1C_2\ldots C_m}$ which satisfies
 \begin{equation*}
 \begin{split}
  \| \Psi - \psi^{\otimes n} \|_1 &\leq \nu(\epsilon) \\
  \Tr[ (\Psi^{{\cal T}})^2] &\leq (1-\mu(\epsilon))^{-2} 2^{-n(S({\cal T})_{\psi} - \delta_{{\cal T}})} \\
 \end{split}
  \end{equation*}
for all non empty subsets ${\cal T} \subseteq \{1,2,\ldots, m\}$. Here, $\nu(\epsilon)$ and $\mu(\epsilon)$ are functions of $\epsilon$ which vanish by choosing arbitrarily small values for $\epsilon$.
\end{conjecture}
The conjecture is true for $m=2$. Readers interested will find the proof in the typicality section of Appendix A. (see Proposition \ref{prop:mixedstate}.)

\subsubsection{A simple example of distributed compression for two senders}
To illustrate some of the key differences between the protocols shown to exist by the proof of Theorem \ref{thm:statemerging} and distributed compression protocols as discussed in \cite{merge}, let's
consider again the task of distributed compression for two senders
sharing a state $\psi^{C_1C_2R}$, with purifying system $R$.
Recall that in distributed compression, the receiver has no prior information about the state. For the case of two senders, the rate region is described by the inequalities:
\begin{equation} \label{eq:condDistCompressV}
  \begin{split}
   R_1 &\geq S(C_1 | C_2)_{\psi} \\
   R_2 &\geq S(C_2 | C_1)_{\psi} \\
   R_1 + R_2 &\geq S(C_1C_2)_{\psi} \\
  \end{split}
\end{equation}
Let's consider a very simple state:
\begin{equation*}
 \ket{\psi}^{C_1C_2R}:=\ket{\psi^{C_1C_2^1C_2^2R}} := \ket{\Psi_{-}}^{C_1C_2^1} \otimes \ket{\Psi_{-}}^{C_2^2R},
\end{equation*}
where $\ket{\Psi_{-}}^{C_1C_2^1}$ is an EPR pair shared between the senders $C_1$ and $C_2$ and $\ket{\Psi_{-}}^{C_2^2R}$ is an EPR pair shared between the reference $R$ and $C_2$. Let's compute the entropies of eq.~(\ref{eq:condDistCompressV}) related to the rates $R_1$ and $R_2$:
\begin{equation*}
\begin{split}
S(C_1|C_2)_{\psi} &= -1 \\
S(C_2|C_1)_{\psi} &= 0 \\
S(C_1C_2)_{\psi} &= 1 \\
\end{split}
\end{equation*}
The total entanglement cost for merging is at least one ebit per copy of the input state, no matter which protocol is used to perform distributed compression. For achieving the rates $R_1 = 0$ and $R_2 = 1$, we can use our multiparty state merging protocol if we inject, prior to performing measurements on the senders, an ebit per copy of the input state between the sender $C_2$ and the receiver. A distributed compression protocol as in \cite{merge}, however, will need to inject 2 ebits per copy of the input state between the sender $C_2$ and the receiver if the system $C_2$ is first transferred, followed by the system $C_1$. If the system $C_1$ is transferred first instead, then one ebit per copy is needed between the sender $C_1$ and the receiver. Thus, the distribution of the catalytic entanglement for the distributed compression of \cite{merge} is more restricted than our multiparty state merging protocol. Time-sharing can be used to achieve the rate-tuple $(R_1, R_2)$ with $R_1=0$ and $R_2=1$, but it may require many more copies of the input state to achieve this.

\subsubsection{Distributed compression for three senders}
Our proof of Theorem \ref{thm:statemerging} is for the case of two senders.  We suspect that Conjecture~\ref{eq:notimesharing} will hold for the case of three senders. Under this assumption, we can show a real advantage to using our protocol over a distributed compression protocol relying on multiple applications of two-party state merging.
For the case of three senders, the rate region is described by the inequalities:
\begin{equation} \label{eq:condDistCompress}
  \begin{split}
   R_1 &\geq S(C_1 | C_2C_3)_{\psi} \\
   R_2 &\geq S(C_2 | C_1C_3)_{\psi} \\
   R_1 + R_2 &\geq S(C_1C_2|C_3)_{\psi} \\
   R_1 + R_3 &\geq S(C_1C_3|C_2)_{\psi} \\
   R_2 + R_3 &\geq S(C_2C_3|C_1)_{\psi} \\
   R_1 + R_2 + R_3 &\geq S(C_1C_2C_3)_{\psi} \\
   \end{split}
\end{equation}
Let's consider the state
\begin{equation*}
 \ket{\psi}^{C_1C_2C_3R}:=\ket{\psi^{C_1C_2^1C_2^2C_3^1C_3^2R}} := \ket{\Psi_{-}}^{C_1C_2^1} \otimes \ket{\Psi_{-}}^{C_3^1R} \otimes \ket{\phi}^{C_2^2C_3^2},
\end{equation*}
where $\ket{\phi}^{C^2_2C_3^2}:=\sqrt{\lambda} \ket{00}^{C^2_2C^2_3} + \sqrt{1-\lambda}\ket{11}^{C^2_2C^2_3}$, with $0 < \lambda < 1$, is a pure bipartite entangled state with entropy of entanglement:
\begin{equation*}
 E(\phi) = S(C^3_2)_{\phi} = -\lambda \log{\lambda} - (1-\lambda)\log(1-\lambda) > 0.
\end{equation*}
Let's compute the entropies of eq.~(\ref{eq:condDistCompress}) related to the rates $R_1$ and $R_2$ as a function of $S(C^3_2)_{\phi}$:
\begin{equation*}
\begin{split}
S(C_1|C_2C_3)_{\psi} &= -1 \\
S(C_2|C_1C_3)_{\psi} &= -S(C^2_3)_{\phi} -1 \\
S(C_1C_2|C_3)_{\psi} &= -S(C^2_3)_{\phi} \\
\end{split}
\end{equation*}
\begin{figure}[t!]
\centering
    \includegraphics[width=0.8\textwidth]{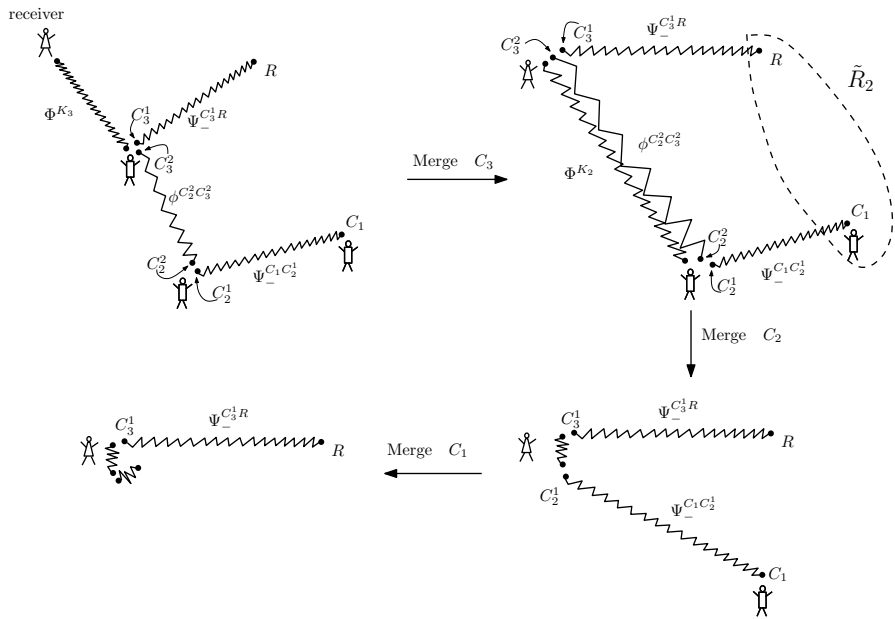}
\caption{Entanglement structure of the state $\psi^{C_1C_2C_3R}$. The need for catalytic entanglement using a distribution compression protocol is depicted.}\label{fig:rateregion}
\end{figure}

The state $\phi$ is entangled ($S(C^2_3)_{\phi} > 0$), hence there exists an achievable rate-tuple $(R_1, R_2, R_3)$, by choosing $R_3$ big enough, satisfying the inequalities of eq.~(\ref{eq:condDistCompress}) with $R_1 < 0, R_2 < 0$ and $R_1 + R_2 < 0$. Under our initial assumption, there exist multiparty merging protocols with arbitrarily small error achieving this rate-tuple if we inject $(\log K_3) \approx nR_3$ ebits between the sender $C_3$ and the receiver. The protocols will return around $-nR_1$ ebits between the sender $C_1$ and the receiver and approximately $-nR_2$ ebits between the sender $C_2$ and the receiver. No catalytic entanglement is required for both of these senders.

Consider any distributed compression protocol on the other hand. If no catalytic entanglement is used, the protocol must first merge the system $C_3$ to the receiver, even if time-sharing is used. Assuming the receiver has recovered perfectly the system $C_3$, the protocol can then merge either $C_1$ or $C_2$. A time-sharing strategy will choose to merge $C_2$ on a subset of the input copies, while merging $C_1$ on the other copies. The ebit rate must be at least $S(C_2|C_3)_{\psi}$ for merging $C_2$, and at least $S(C_1|C_3)_{\psi}$ for merging $C_1$. We have
\begin{equation*}
\begin{split}
 S(C_1|C_3)_{\psi} &= S(C_1)_{\Psi_{-}} = 1, \\
S(C_2|C_3)_{\psi} &= 1 - S(C^2_3)_{\phi}, \\
\end{split}
\end{equation*}
and so, as long as $\ket{\phi}^{C^2_2C^2_3}$ is not maximally entangled (i.e., $S(C^2_3)_{\phi} \neq 1$), both of these entropies are positive. Hence, transferring $C_1$ or $C_2$ to the receiver requires the injection of catalytic entanglement.
\begin{figure}[t!]
\centering
    \includegraphics[width=0.9\textwidth]{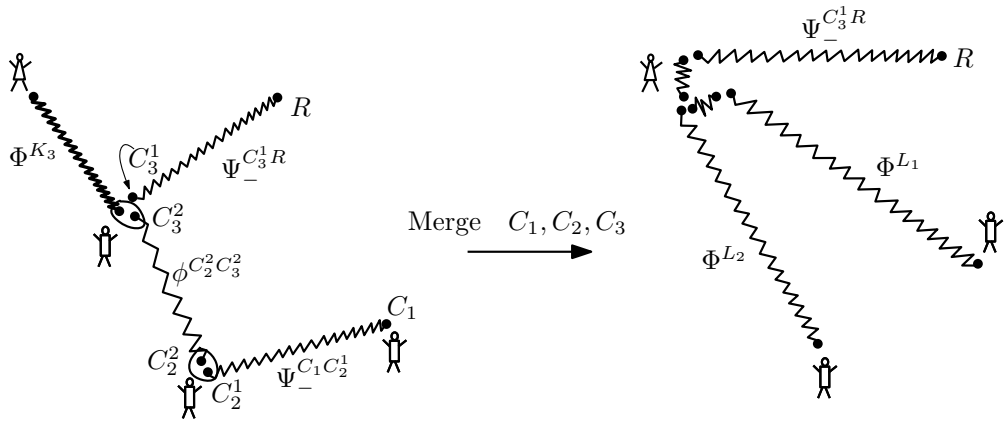}
\caption{Using an ``entanglement swapping'' trick, the multiparty merging protocol needs no catalytic entanglement for the senders $C_1$ and $C_2$ to transfer the state to the receiver. Some of the labels were removed for clarity. }\label{fig:rateregionVR}
\end{figure}

This example is better understood by looking at the entanglement structure of the state $\psi^{C_1C_2C_3R}$ (see Figure \ref{fig:rateregionVR}). After merging $C_3$, a distributed compression protocol will decouple either the system $C_2$ from its relative reference $\tilde{R}_2:=C_1R$ or $C_1$ from its relative reference $\tilde{R}_1:=C_2R$. There are 2 bits of mutual information between the system $C_2$ and the relative reference $\tilde{R}_2$. To transfer this correlation to the receiver using anything less than a perfectly entangled pair is impossible as superdense coding is optimal \cite{superdense}. The same reasoning applies for the system $C_1$.

To grasp why catalytic entanglement is not necessary for the senders $C_1$ and $C_2$, assuming Conjecture \ref{eq:notimesharing} holds, for our multiparty merging protocol, observe that the mutual information between the systems $C_1C_2C_3$ and the reference $R$ is entirely concentrated between the systems $C_3$ and $R$. If we boost the number of ebits shared between the receiver and the sender $C_3$, we can decouple the systems $C_1C_2C_3$ from the reference, and generate ebits between the receiver and the systems $C_1$ and $C_2$ through an ``entanglement swapping'' effect. This highlights a fundamental difference between protocols working on two parties and multipartite protocols: entanglement can be produced between two parties by other means than entanglement distillation \cite{Concentrate,Bennett,DW} or entanglement gambling \cite{vidal}. This example also exhibits a natural trade-off between injecting more entanglement than needed at one place and being able to produce entanglement or at the very least transfer other systems without requiring additional entanglement.

\section{Split transfer} \label{sec:split-transfer}
In the previous sections, we have analyzed and characterized the
entanglement cost for merging the state $\initstate$ to a single
receiver (Bob) in the asymptotic setting and in the one-shot regime.
Here, we modify our initial setup by introducing a second decoder
$A$ (Alice), who is spatially separated from Bob and also has side
information about the input state. That is, the senders $C_1,
C_2,\ldots, C_m$ and the two receivers Alice and Bob share a
global state $\psi^{C_1C_2\ldots C_mABR}$ and the objective is
then to redistribute the state $\initstateA$ to Alice and Bob. The motivation for this problem comes from the
multipartite entanglement of assistance problem \cite{SVW,merge},
where the task is to distill entanglement in the form of EPR pairs
from a $(m+2)$-partite pure state $\inputstate$ shared between two
recipients (Alice and Bob) and $m$ other \textit{helpers} $C_1,C_2,\ldots,C_m$. Recall the formula for the optimal assisted EPR rate:
\begin{equation}\label{eq:mincut1}
  E^{\infty}_A(\inputstate) := \min_{\cal T} S(A{\cal T})_{\psi} =: E_{\mathrm{min-cut}}(\inputstate, A:B),
\end{equation}
where ${\cal T} \subseteq \{C_1,C_2,\ldots,C_m\}$ is a subset (i.e a bipartite cut) of the helpers.

The proof that the rate given by eq.~(\ref{eq:mincut1}) is
achievable using LOCC operations (i.e., no pre-shared entanglement allowed) consists of showing that the
min-cut entanglement of the state $\inputstate$ is arbitrarily well preserved after each sender has finished
performing a random measurement on his system. The procedure
described in the proof of \cite{merge} makes use of a
\textit{multi-blocking} strategy. That is, given $n$ copies of the input
state $\inputstate$, the first helper will perform $d=n/r$ random
measurements, each acting on $r$ copies of $\inputstate$ and
generating a number of possible outcomes. For a sequence of measurement outcomes
$j_1,j_2,\ldots,j_d$, we group together the residual states corresponding to outcome $1$, then group the ones
corresponding to outcome $2$, etc... When this is done, the next
helper will perform random measurements for each of these groups in
the same way the first sender proceeded. That is, for each group,
you need to divide into blocks, and so on. Needless to say, this approach fails if few copies are available to the parties.

\begin{figure}[t!]
\centering
    \includegraphics[width=\textwidth]{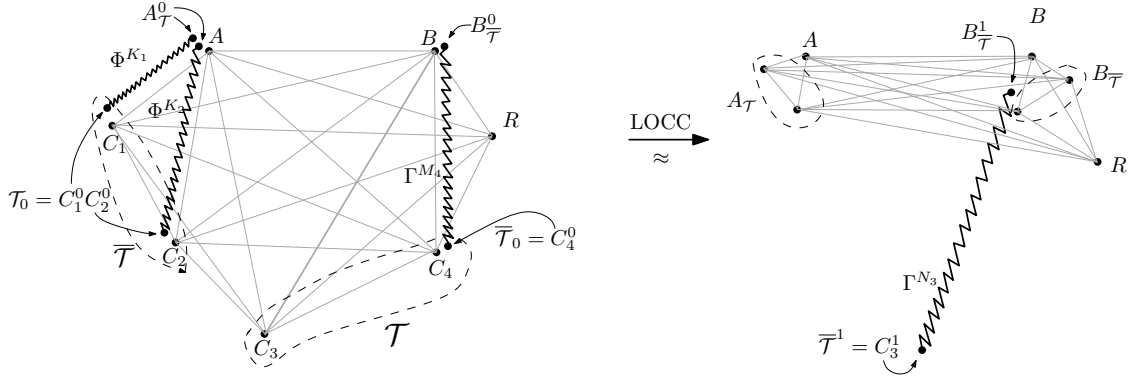}
\caption{Picture of the initial and final steps of a split-transfer protocol involving four senders. Jagged lines represent maximally entangled states shared between the receivers and the senders and solid lines represent correlation between the parties. The senders $C_3$ and $C_4$ are transferred to $A$, while the senders $C_1$ and $C_2$ are recovered by the receiver $B$. At the end of the split-transfer protocol, we have $\log(N_3)$ ebits shared between $C_3$ and $B$ for this particular example. }\label{fig:splittransfer}
\end{figure}

It was conjectured by Horodecki \textit{et al.} \cite{merge} that these layers of blocking
could be removed by letting all the helpers perform simultaneous
measurements on their respective typical subspaces. Such a
strategy would still produce states which preserve the min-cut
entanglement, thereby providing a way to prove
eq.~(\ref{eq:mincut1}) without the need for a recursive argument.
In this section, we show that if Conjecture \ref{eq:notimesharing} is true, there exists
an LOCC protocol acting on the pure state $\psi^{C_1C_2\ldots C_mAB}$ which will send a cut
${\cal T}_{\min}$ which minimizes $S(A\cK)_{\psi}$ to Alice and its complement to Bob. The protocol
consists of two parts: First, all the helpers will perform
random measurements on their typical subspaces and broadcast their
outcomes to both decoders. Then, Alice will use the classical
information coming from the helpers which are part of the cut
${\cal T}_{\min}$ and apply an isometry $U$, while Bob will apply
an isometry $V$ depending on the outcomes of the helpers belonging
to $\cT_{\min}$. This will redistribute the initial state to
Alice and Bob. Standard distillation protocols, as described in chapter 2,
on the recovered state will yield
EPR pairs at the min-cut rate of eq.~($\ref{eq:mincut1}$).

The following definition of a \textit{split-transfer} adapts the multiparty merging definition to the case of two receivers. We follow the notational convention established in the previous sections for labeling the various systems, dimensions, etc\ldots Figure \ref{fig:splittransfer} helps making sense of all the labels.
\begin{definition}[Split-Transfer]
Let $\psi^{{\cal T}A{\ov{\cal T}}BR}$ be an $(m+2)$-partite state,
and assume the senders and the decoders share maximally entangled states $\Phi^{K_{\cK}}:=\bigotimes_{i
\in \cal T} \Phi^{K_i}$ and $\Gamma^{M_{\cKbar}}:=\bigotimes_{i \in \cT}
\Gamma^{M_i}$. We call the LOCC operation ${\cal M}: {\cal T}\cK^0 \ov{\cal
T}\cTo \otimes AA^0_{{\cal T}} \otimes BB^0_{\ov{{\cal T}}}
\rightarrow {\cal T}^1A^1_{{\cal T}}AA_{{\cal T}} \otimes
{\ov{\cal T}}^1B^1_{\ov{{\cal T}}}BB_{\ov{{\cal T}}} $ a split transfer for the state $\inputGroupstate$ with error
$\epsilon$ and entanglement costs
$\overrightarrow{E_{\cK}}(\psi):=\bigoplus_{i \in \cK}(\log K_i
-\log L_i)$ and $\overrightarrow{E_{\cKbar}}(\psi):=\bigoplus_{i
\in \cKbar}(\log M_i - \log N_i)$ if
 \begin{equation}\label{eq:splittransfer}
 \bigg \| (\mathrm{id}_R \otimes {\cal M})(\inputGroupstate \otimes \Phi^{K_{\cK}}  \otimes \Gamma^{M_{\cKbar}}) -  \psi_{A_{{\cal T}}AB_{\ov{\cal T}}BR} \otimes \Phi^{L_{\cK}} \otimes \Gamma^{N_{\cKbar}} \bigg \|_1 \leq \epsilon,
\end{equation}
where $\Phi^{L_{\cK}}:=\bigotimes_{i \in \cK}\Phi^{L_i},\Gamma^{N_{\cKbar}}:=\bigotimes_{j \in \cKbar}\Gamma^{N_j}$ are maximally entangled states distributed appropriately between the senders and the receivers. The systems $A_{\cal T}$ and $B_{\cT}$ are ancillary systems of the same size as ${\cal T}$ and $\cKbar$ and are held by Alice and Bob respectively. For the state $\Psi:=(\inputGroupstate)^{\otimes n}$, the entanglement rates $\overrightarrow{R_{\cK}}(\psi)$ and $\overrightarrow{R_{\cKbar}}(\psi)$ are defined as $\frac{1}{n}\overrightarrow{E_{\cKbar}}(\Psi)$ and
$\frac{1}{n}\overrightarrow{E_{\cKbar}}(\Psi)$.
\end{definition}

In the above definition, we have denoted by $\bigoplus_{i \in \cK} (\log K_i - \log L_i)$ a vector of length $|\cK|$ whose components are given by $\log K_i -\log L_i$ for $i \in \cK$.

The rate region of a split-transfer for the state $\inputGroupstate$ is defined in a manner analogous to definition
\ref{def:rateregion}. We omit the details here, but whenever we
will say that a rate is achievable for a split-transfer of the
state $\inputGroupstate$, it means that it is strictly contained in the
rate region (i.e., not on the boundary).

\subsection{Decoupling relative references}
We saw in the previous sections how the distributed compression protocol of \cite{merge} achieves a multiparty merging for the state $\psi^{C_MBR}$ by decoupling each sender from its relative reference one at a time. We use this approach here to show the existence of good decoders for the receivers when the senders perform simultaneous measurements. The main technical difficulty is to formally prove that simultaneous measurements by the senders still produce a state allowing a good (i.e., high fidelity) recovery of the initial state by the receivers. We extend Proposition \ref{prop:mergeCond} to our present scenario by following a similar route to that of section \ref{sec:merging-many}. We begin by analyzing an ideal situation.

For a pure state $\psi^{\cK A \cKbar BR}$, suppose each sender performs
an incomplete measurement on their respective
shares of the state. For a measurement outcome $J_M := (j_1,j_2,\ldots,j_m)$, define the state
\begin{equation}
 \begin{split}
 \ket{\psi_{J_M}^{\cK^1 A \cTu BR}} &:= \frac{1}{\sqrt{p_{J_M}}}(P^1_{j_1} \otimes P^2_{j_2} \otimes \ldots \otimes P^m_{j_m} \otimes I^{ABR}) \ket{\psi^{\cK A \cKbar BR}} \\
   & =: \frac{1}{\sqrt{p_{J_M}}}(P^{\cK}_{j_{\cK}} \otimes P^{\cKbar}_{j_{\cKbar}} \otimes I^{ABR}) \ket{\psi^{\cK A \cKbar BR}}, \\
 \end{split}
\end{equation}
where the Kraus operators $P^{\cK}_{j_{\cK}} = \bigotimes_{i \in \cK} P^i_{j_i}$ and $P^{\cKbar}_{j_{\cKbar}} = \bigotimes_{i \in \cKbar} P^i_{j_i}$ map the spaces $\cK$ and $\cKbar$ to the subspaces $\cK^1$ and $\cTu$. Related to this state is the $R_{\cK}$-\textit{relative} state
\[
\ket{\varphi_{j_{\cK}}^{\cK^1 A R_{\cK}}} :=
\frac{1}{\sqrt{p_{j_{\cK}}}}(P^{\cK}_{j_{\cK}} \otimes I^{\cKbar
ABR})\ket{\psi^{\cK A \cKbar BR}}
,\] where $R_{\cK}:=\cKbar B R$ and $p_{j_{\cK}}$ is the
probability of getting the outcome $j_{\cK}$. If each sender in ${\cal T}$ perfectly decouples his system from the relative reference $R_{\cK}$ and the other senders in ${\cK}$, we have
\begin{equation}
 \varphi_{j_{\cK}}^{\cK^1 R_{\cK}} = \tau^{\cK^1} \otimes \psi^{R_{\cK}},
\end{equation}
where $\tau^{\cK^1} = \bigotimes_{i \in \cK} \tau^{C^1_i}$ is
the maximally mixed state of dimension $L_{\cK}:=\prod_{i \in \cK}d_{C^1_i}$ on the system $\cK^1$.
From the Schmidt decomposition, there exists an isometry
$U^A_{j_{\cK}}: A \rightarrow A^1_{\cK} A_{\cK} A$ which Alice can
implement such that
\begin{equation}
(I^{\cK^1 R_{\cK}} \otimes U^A_{j_{\cK}}) \ket{\varphi_{j_{\cK}}^{\cK^1A R_{\cK}}} = \ket{\Phi^{L_{\cK}}} \otimes \ket{\psi^{ A_{\cK} A R_{\cK}}},
\end{equation}
where the state $\ket{\psi^{A_{\cK} A R_{\cK}}}$ is the same as
the original state $\ket{\psi^{\cK A \cKbar BR}}$ with the
ancillary system $A_{\cK}$ substituted for $\cK$.

 After merging the systems ${\cal T}$ to the receiver $A$, the senders $\cKbar$ follow with their measurements. Define the $R_{\cKbar}$-\textit{relative} state
 \[\ket{\upsilon_{j_{\cKbar}}^{\cTu B R_{\cKbar}}}
:=\frac{1}{\sqrt{p_{j_{\cKbar}}}}(P^{\cKbar}_{j_{\cKbar}} \otimes
I_{A_{\cK}ABR})\ket{\psi^{A_{\cK}A \cKbar BR}},\]
where $R_{\cKbar}:= A_{\cK}AR$. Assume once again that each sender perfectly decouples his system from the relative reference and the other senders in $\cKbar$. We have
\begin{equation}
 \upsilon_{j_{\cKbar}}^{\cTu R_{\cKbar}} = \tau^{\cTu} \otimes \psi^{R_{\cKbar}},
\end{equation}
where $\tau^{\cTu} = \bigotimes_{i \in \cKbar} \tau^{C^1_i}$ is
the maximally mixed state of dimension $N_{\cKbar}:=\prod_{i\in \cKbar}d_{C^1_i}$ on the system $\cTu$.
From the Schmidt decomposition, there exists an isometry $V^B_{j_{\cKbar}}: B \rightarrow B^1_{\cKbar} B_{\cKbar}
B$ implementable by Bob such that
\begin{equation}
(I^{\cTu R_{\cKbar}} \otimes V^B_{j_{\cKbar}}) \ket{\upsilon_{j_{\cKbar}}^{\cTu B R_{\cKbar}}} = \ket{\Gamma^{N_{\cKbar}}} \otimes \ket{\psi^{A_{\cK} A B_{\cKbar} BR}},
\end{equation}
where the state $\ket{\psi^{A_{\cK} A B_{\cKbar} BR}}$ is the same
as the original state $\inputGroupstate$ with the ancillary
systems $A_{\cK}$ and $B_{\cKbar}$ substituted for $\cK$ and
$\cKbar$.

If the senders perform their measurements simultaneously instead, the initial state is recovered by applying the isometries $U^A_{j_{\cK}}$ and $V^B_{j_{\cKbar}}$ to the the outcome state $\ket{\psi_{J_M}^{\cK^1 A \cTu BR}}$:
\begin{equation}
 \begin{split}
   &(I^{\cK^1 \cTu R} \otimes U^A_{j_{\cK}} \otimes V^B_{j_{\cKbar}}) \ket{\psi_{J_M}^{\cK^1 A \cTu BR}}\\
   &=\frac{1}{\sqrt{p_{J_M}}}(I^{\cK^1 \cTu R} \otimes U^A_{j_{\cK}} \otimes V^B_{j_{\cKbar}})( P^{\cKbar}_{j_{\cKbar}} \otimes I^{\cK^1 ABR})(P^{\cK}_{j_{\cK}} \otimes I^{A R_{\cK}})\ket{\psi^{\cK A \cKbar B R}} \\
   &=\frac{1}{\sqrt{p_{J_M}}}(I^{\scriptscriptstyle \cK^1 A^1_{\cK} \cTu R_{\cKbar} } \otimes V^B_{j_{\cKbar}})(P^{\cKbar}_{j_{\cKbar}} \otimes I^{\scriptscriptstyle \cK^1 A^1_{\cK} B R_{\cKbar} } )(I^{\scriptscriptstyle \cK^1 BR} \otimes U^A_{j_{\cK}})(P^{\cK}_{j_{\cK}} \otimes I^{\scriptscriptstyle A R_{\cK}})\ket{\psi^{\cK A \cKbar B R}} \\
   &=\frac{1}{\sqrt{p_{J_M}}}(I^{\cK^1 A^1_{\cK} \cTu R_{\cKbar}} \otimes V^B_{j_{\cKbar}})(P^{\cKbar}_{j_{\cKbar}} \otimes I^{\cK^1 A^1_{\cK}BR_{\cKbar}} )(I^{\cK^1 BR} \otimes U^A_{j_{\cK}}) \sqrt{p_{j_{\cK}}}\ket{\varphi_{j_{\cK}}^{\cK^1 A R_{\cK} }}\\
   &=\sqrt{\frac{p_{j_{\cK}}}{p_{J_M}}}(I^{\cK^1 A^1_{\cK} \cTu R_{\cKbar}} \otimes V^B_{j_{\cKbar}})(P^{\cKbar}_{j_{\cKbar}} \otimes I^{\cK^1 A^1_{\cK}BR_{\cKbar}}) \ket{\Phi^{L_{\cK}}} \otimes \ket{\psi^{A_{\cK} A \cKbar B R}} \\
   &=\sqrt{\frac{p_{j_{\cK}}}{p_{J_M}}}(I^{\cK^1 A^1_{\cK} \cTu R_{\cKbar}} \otimes V^B_{j_{\cKbar}}) \ket{\Phi^{L_{\cK}}} \otimes \sqrt{p_{j_{\cKbar}}}\ket{\upsilon_{j_{\cKbar}}^{\cTu B R_{\cKbar} }} \\
   &=\sqrt{\frac{p_{j_{\cK}}p_{j_{\cKbar}}}{p_{J_M}}}\ket{\Phi^{L_{\cK}}} \otimes \ket{\Gamma^{N_{\cKbar}}} \otimes \ket{\psi^{A_{\cK} A B_{\cKbar} BR}}. \\
 \end{split}
\end{equation}
Since the states $(I^{\cK^1 \cTu R} \otimes U^A_{j_{\cK}} \otimes
V^B_{j_{\cKbar}}) \ket{\psi_{J_M}^{\cK^1 A \cTu BR}}$ and $\ket{\Phi^{L_{\cK}}}
\otimes \ket{\Gamma^{N_{\cKbar}}} \otimes \ket{\psi^{A_{\cK} A B_{\cKbar}
BR}}$ are both normalized, we have $p_{J_M} =
p_{j_{\cK}}p_{j_{\cKbar}}$. Hence, for this ideal scenario, the decodings implemented by the receivers for a split-transfer protocol which decouples the senders from their relative references one at a time can also be used to recover the initial state if the measurements performed by the senders are done at the same time. This agrees with our intuition that the choice of the decoder applied by the receiver $A$ (resp. $B$) should not depend on the measurement outcomes of the systems ${\cKbar}$ (resp. $\cK$). We resume the previous ideas in the following adaptation of Proposition \ref{prop:mergeCond}:

\begin{Proposition}[Conditions for a Split-Transfer] \label{prop:mergeConds}
Let $\inputGroupstate$ be a multipartite state shared between $m$ senders and two receivers. Suppose the senders simultaneously perform incomplete measurements on their systems, yielding a state $\ket{\psi_{J_M}^{\cK^1 A \cTu B R}}:=\frac{1}{\sqrt{p_{J_M}}}(P^{\cK}_{j_{\cK}} \otimes P^{\cKbar}_{j_{\cKbar}} \otimes I^{ABR})\ket{\psi^{\cK A \cKbar BR}}$ for an outcome $J_M$.

Using the notation of the previous paragraphs, define the decoupling errors $Q^1_{\cal I}(\inputRelativeA)$ and $Q^2_{\cal I}(\inputRelativeB)$:
 \begin{equation}
    \begin{split}
        Q^1_{\cal I}(\inputRelativeA) &:= \sum_{j_{\cK}} p_{j_{\cK}}\|\varphi_{j_{\cK}}^{\cK^1 R_{\cK} } - \tau^{\cK^1} \otimes \psi^{R_{\cK} }\|_1, \\
        Q^2_{\cal I}(\inputRelativeB) &:= \sum_{j_{\cKbar}} p_{j_{\cKbar}}\|\upsilon_{j_{\cKbar}}^{\cTu R_{\cKbar}} - \tau^{\cTu} \otimes \psi^{R_{\cKbar}}\|_1.
    \end{split}
 \end{equation} If $Q^1_{\cal I} \leq \epsilon$ and $Q^2_{\cal I} \leq \epsilon'$, then there exists a split-transfer for the state $\inputGroupstate$ with error $2\sqrt{\epsilon}+2\sqrt{\epsilon'}$ and entanglement costs $\overrightarrow{E_{\cK}}=\bigoplus_{i \in \cK}(-\log L_i), \overrightarrow{E_{\cKbar}}=\bigoplus_{i \in \cKbar}(-\log N_i)$.
 \end{Proposition}

\begin{proof}
Apply Proposition \ref{prop:mergeCond} for each of the decoupling errors $Q^1_{\cal I}$ and $Q^2_{\cal I}$. Since these are bounded by $\epsilon$ and $\epsilon'$, there exist isometries $U^A_{j_{\cK}}$ and $V^B_{j_{\cKbar}}$ such that
 \begin{eqnarray}
     \bigg \| \sum_{j_{\cK}} p_{\scriptscriptstyle j_{\cK}} (I^{\scriptscriptstyle R_{\cK}} \otimes U^{\scriptscriptstyle A}_{\scriptscriptstyle j_{\cK}}) \varphi^{\scriptscriptstyle \cK^1 A R_{\cK} }_{j_{\cK}}(I^{\scriptscriptstyle R_{\cK} } \otimes U^{\scriptscriptstyle A}_{\scriptscriptstyle j_{\cK}})^{\dag}-\psi^{\scriptscriptstyle A_{\cK} A \cKbar B R} \otimes \Phi^{\scriptscriptstyle L_{\cK}} \bigg \|_1\leq2\sqrt{\epsilon}&& \label{eq:qerror1}\\
     \bigg \| \sum_{\scriptscriptstyle j_{\cKbar}} p_{\scriptscriptstyle j_{\cKbar}} (I^{\scriptscriptstyle R_{\cKbar} } \otimes V^{\scriptscriptstyle B}_{\scriptscriptstyle j_{\cKbar}}) \upsilon^{\scriptscriptstyle \cTu B R_{\cKbar}}_{j_{\cKbar}}(I^{\scriptscriptstyle R_{\cKbar} } \otimes V^{\scriptscriptstyle B}_{j_{\cKbar}})^{\dag} - \psi^{\scriptscriptstyle A_{\cK} A B_{\cKbar} B R} \otimes \Gamma^{\scriptscriptstyle N_{\cKbar}} \bigg \|_1\leq2\sqrt{\epsilon'}.&& \label{eq:qerror}
 \end{eqnarray}
  If we apply the isometries $U^A_{j_{\cK}}$ and $V^B_{j_{\cKbar}}$ to the outcome state $\ket{\psi^{\cK^1 A \cTu BR}_{J_M}}$, we have
 \begin{equation*}
  \begin{split}
    &\rho := \sum_{J_M} p_{J_M} \bigg ( (I^{\cK^1 \cTu R} \otimes \UA \otimes \VB) \psi^{\cK^1 A \cTu BR}_{J_M} (I^{\cK^1 \cTu R} \otimes \UA \otimes \VB)^{\dag} \bigg )\\
     &= \sum_{J_M} p_{j_{\cK}}\bigg ((I \otimes \VB)(P^{\cKbar}_{\jKbar}\otimes I)(I \otimes \UA)\varphi_{j_{\cK}}^{\cK^1 A R_{\cK}}(I \otimes \UA)^{\dag}(P^{\cKbar}_{\jKbar} \otimes I)^{\dag}(I \otimes \VB)^{\dag} \bigg )\\
     &=\sum_{\jKbar} (I \otimes \VB)(P^{\cKbar}_{\jKbar}\otimes I) \zeta (P^{\cKbar}_{\jKbar} \otimes I)^{\dag}(I \otimes \VB)^{\dag}  \\
     &=: {\cal M}(\zeta)\\
  \end{split}
 \end{equation*}
where $\zeta := \sum_{\jK} p_{j_{\cK}} (I \otimes \UA) \varphi_{j_{\cK}}^{\cK^1 A R_{\cK} }(I \otimes \UA)^{\dag}$. It can be seen as the output state we would get if only the senders in $\cK$ wanted to transfer their systems to the receiver~$A$. The map ${\cal M}$, as defined above, is a trace-preserving LOCC  operation (i.e., measurements by the senders in ${\cKbar}$ followed by an isometry on $B$). 
Note that we removed some of the superscript notation for the sake of clarity.

To bound the trace distance between the output state $\rho$ and the state $\psi^{A_{\cK} A B_{\cKbar} BR} \otimes \Phi^{L_{\cK}}
\otimes \Gamma^{N_{\cKbar}}$, we introduce the following intermediate state
 \begin{equation}
  \begin{split}
   \sigma &:= \sum_{\jKbar}(I \otimes \VB)(P^{\cKbar}_{\jKbar} \otimes I)(\psi^{A_{\cK} A \cKbar B R} \otimes \Phi^{L_{\cK}})(P^{\cKbar}_{\jKbar} \otimes I)^{\dag}(I \otimes \VB)^{\dag} \\
   &= {\cal M}(\psi^{A_{\cK} A \cKbar BR} \otimes \Phi^{L_{\cK}}) \\
   & = \Phi^{L_{\cK}} \otimes \sum_{\jKbar} p_{\scriptscriptstyle j_{\cKbar}} (I^{\scriptscriptstyle R_{\cKbar} } \otimes V^{\scriptscriptstyle B}_{\scriptscriptstyle j_{\cKbar}}) \upsilon^{\scriptscriptstyle \cTu B R_{\cKbar}}_{j_{\cKbar}}(I^{\scriptscriptstyle R_{\cKbar} } \otimes V^{\scriptscriptstyle B}_{j_{\cKbar}})^{\dag}, \\
  \end{split}
 \end{equation}
and apply the triangle inequality
\begin{equation}
 \begin{split}
   \bigg \| \rho - &\psi^{A_{\cK} A B_{\cKbar} BR} \otimes \Phi^{L_{\cK}} \otimes \Gamma^{N_{\cKbar}} \bigg \|_1 \leq \bigg \|\rho - \sigma \bigg \|_1 +\bigg \| \sigma - \psi^{A_{\cK} A B_{\cKbar} BR} \otimes \Phi^{L_{\cK}} \otimes \Gamma^{N_{\cKbar}}\bigg\|_1. \\
 \end{split}
\end{equation}
From eq.~(\ref{eq:qerror}), the trace norm $\bigg \| \sigma - \psi^{A_{\cK} A B_{\cKbar} BR}
\otimes \Phi^{L_{\cK}} \otimes \Gamma^{N_{\cKbar}} \bigg \|_1$ is bounded from above by $2\sqrt{\epsilon'}$. To bound $\|\rho - \sigma \|_1$, we have
 \begin{equation}
  \begin{split}
  \| \rho - \sigma \|_1 &= \bigg \| {\cal M}(\zeta) - {\cal M}(\psi^{A_{\cK} A \cKbar BR} \otimes \Phi^{L_{\cK}}) \bigg \|_1 \\
    &\leq  \bigg \| \zeta - \psi^{A_{\cK} A \cKbar BR} \otimes \Phi^{L_{\cK}} \bigg \|_1 \\
    &\leq 2\sqrt{\epsilon}.
  \end{split}
 \end{equation}
The first inequality holds since the trace distance is non-increasing under quantum operations, and the second inequality is just eq.~(\ref{eq:qerror1}).
Thus, we have a split-transfer for the state $\initstateA$ with error $2\sqrt{\epsilon}+2\sqrt{\epsilon'}$.
\end{proof}

\subsection{Split-transfer by random measurements}
With this result in hand, a one-shot split-transfer protocol for the state $\inputGroupstate$ is obtained by two independent applications of Proposition \ref{prop:isometry}, followed by an application of Proposition \ref{prop:mergeConds}. We state the result here.

\begin{Proposition}[One-Shot Split-Transfer] \label{prop:rdnSplit}
Let $\inputGroupstate$ be a multipartite state for $m$ senders and two receivers.  For each sender $C_i$ in the cut $\cK$, there exists an instrument ${\cal I}_i=\{\calE^i_j\}_{j=0}^{F_i}$ consisting of $F_i
= \lfloor \frac{d_{C_i}K_i}{L_i} \rfloor$ partial isometries of rank $L_i$ and one of rank $L_i'= d_{C_i}K_i - F_i
L_i < L_i$ such that the decoupling error $Q^1_{{\cal I}}(\inputRelativeA \otimes \Phi^{K_{\cK}})$ is bounded by
\begin{equation} \label{eq:upperbound1}
  \begin{split}
   Q^1_{{\cal I}}(\inputRelativeA \otimes \Phi^{K_{\cK}}) &\leq 2 \sum_{\substack{{\cal X} \subseteq \cK \\ {\cal X} \neq \emptyset}} \prod_{i \in {\cal X}}\frac{L_i}{d_{C_i}K_i} + 2\sqrt{d_{R_{\cK}} \sum_{\substack{{\cal X} \subseteq \cK \\ {\cal X} \neq \emptyset}}  \prod_{i \in {\cal X}} \frac{L_i}{K_i} \Tr \bigg [ \psi^2_{R_{\cK}{\cal X} } \bigg ]} =: \Delta^1_{\cal I}. \\
  \end{split}
\end{equation}
Similarly, for each sender $C_i$ in the cut $\cKbar$, there exists an instrument ${\cal I}_i=\{\calE^i_j\}_{j=0}^{G_i}$ consisting of $G_i
= \lfloor \frac{d_{C_i}M_i}{N_i} \rfloor$ partial isometries of rank $N_i$ and one of rank $N_i'= d_{C_i}M_i - G_i
N_i < N_i$ such that the decoupling error $Q^2_{{\cal I}}(\inputRelativeB \otimes \Gamma^{M_{\cKbar}})$ is bounded by
\begin{equation} \label{eq:upperbound2}
  \begin{split}
   Q^2_{{\cal I}}(\inputRelativeB \otimes \Gamma^{M_{\cKbar}}) &\leq 2 \sum_{\substack{{\cal Y} \subseteq \cKbar \\ {\cal Y} \neq \emptyset}} \prod_{i \in {\cal Y}}\frac{N_i}{d_{C_i}M_i} + 2\sqrt{d_{R_{\cKbar}}\sum_{\substack{{\cal Y} \subseteq \cKbar \\ {\cal Y} \neq \emptyset}}  \prod_{i \in {\cal Y}} \frac{N_i}{M_i} \Tr \bigg [ \psi^2_{R_{\cKbar} {\cal Y}} \bigg ]} =: \Delta^2_{\cal I}. \\
  \end{split}
\end{equation}
Finally, there exists a split-transfer for the state $\inputGroupstate$
with error $ 2\sqrt{\Delta^1_{\cal I}}+2\sqrt{\Delta^2_{\cal I}}$.
The left hand sides of eqs.~(\ref{eq:upperbound1}) and (\ref{eq:upperbound2}) are bounded
from above on average by their right hand sides if we perform random
instruments on all the senders using the Haar measure.
\end{Proposition}
\begin{proof}
The bounds on the decoupling errors $Q_{\cal I}^1$ and $Q_{\cal I}^2$ are obtained by two independent applications of Proposition
\ref{prop:isometry}. We leave the details to the
reader. The existence of a split-transfer with error
$2\sqrt{\Delta^1_{\cal I}}+2\sqrt{\Delta^2_{\cal I}}$ follows from Proposition \ref{prop:mergeConds}. Note here that
since the senders have additional entanglement at their disposal,
the partial isometries $P^{\cK}_{\jK}$ and $P^{\cKbar}_{\jKbar}$
in Proposition \ref{prop:mergeConds} are replaced by
$P^{\cK\cK^0}_{\jK}$ and $P^{\cKbar\cTo}_{\jKbar}$. These will act
on the spaces $\cK\cK^0$ and $\cKbar\cTo$ respectively, with
output spaces corresponding to $\cK^1$ and $\cTu$.
\end{proof}

\subsection{Asymptotic analysis} \label{sec:mergeprotocol}
The asymptotic analysis for the split-transfer problem is done using the approach of Section~\ref{sec:iid} by treating each decoupling error separately. We arrive at a variation on Theorem \ref{thm:statemerging}:

\begin{Theorem} \label{thm:splittransfer}
Let $\inputGroupstate$ be a multipartite state shared between $m$ senders and two receivers. The rates $\overrightarrow{R_{\cK}}(\psi):=\bigoplus_{i \in \cK}(R_i)$ and $\overrightarrow{R_{\cKbar}}(\psi):=\bigoplus_{i \in \cKbar}(R_i)$ are in the rate region for a split-transfer of the state $\inputGroupstate$ iff the inequalities
\begin{align}
   \label{eq:splitcond1}\sum_{i \in {\cal X}} R_i &\geq S({\cal X}|{\overline{\cal X}}A)_{\psi}\\
   \label{eq:splitcond2}\sum_{i \in {\cal Y}} R_i &\geq S({\cal Y}|{\overline{\cal Y}}B)_{\psi}
\end{align}
hold for all non-empty subsets ${\cal X} \subseteq {\cK}$ and ${\cal Y} \subseteq {\cKbar}$. \end{Theorem}

\begin{proof}
By Proposition \ref{prop:mergeConds}, the task of split-transfer is equivalent to performing simultaneously two multiparty state merging protocols, each one involving different parties. If we view the state $\inputGroupstate$ as a multipartite state consisting of senders $C_i$ belonging to the set $\cK$, a receiver $A$ and a reference $RB \cKbar$, we can apply
Theorem \ref{thm:statemerging} to show the existence of multiparty merging protocols with arbitrarily small error and an entanglement rate satisfying the inequalities of eq.~(\ref{eq:splitcond1}). Viewing instead the state $\inputGroupstate$ as a multipartite state consisting of senders $C_i$ belonging to the set $\cKbar$, a receiver $B$ and a reference $RA\cK$, Theorem \ref{thm:statemerging} tells us of the existence of multiparty merging protocols with arbitrarily small error and an entanglement rate satisfying eq.~(\ref{eq:splitcond1}). For corner points, we can applying these protocols simultaneously, and we have, by Proposition \ref{prop:mergeConds}, a family of split-transfer protocols with vanishing error as $n$ grows larger. The other rates of the regions described by eqs.~(\ref{eq:splitcond1}) and (\ref{eq:splitcond2}) are achieved by using a time-sharing strategy. This proves the direct part.

To prove the converse, consider any cut ${\cal X}$ of the senders in $\cK$ and look at the preservation of the entanglement
across the cut $A\cX$ vs $\X \cKbar B R$.  Assume, for technical reasons, that $L_i \leq 2^{O(n)}$ for all $i \in \cK$. The initial entropy of entanglement across the cut $A\cX$ vs $\X \cKbar B R$ is
 \begin{equation}\label{eq:Ein}
    E_{in} := n S(A\cX)_{\psi} + \sum_{i \in \X} \log K_i.
 \end{equation}
At the end of any LOCC operation on the state $(\inputGroupstate)^{\otimes n}$, the output state is an ensemble $\{q_k,\psi^k_{\cK^1A^1_{\cK}A^nA^n_{\cK}\cKbar^1B^1_{\cKbar}B^nB^n_{\cKbar}R^n}\}$ of pure states. Using monotonicity of the entropy of entanglement under LOCC \cite{Bennett}, we have
\begin{equation}\label{eq:EqIn}
n S(A\cX)_{\psi} + \sum_{i \in \X} \log K_i \geq \sum_k q_k S(\cX^1A^1_{\cK} A^n A^n_{\cK})_{\psi^k},
\end{equation}
where $\cX^1:=\bigotimes_{i \in \cX} C^1_i$. For any split-transfer of the state $(\inputGroupstate)^{\otimes n}$ with error $\epsilon$, we have
\begin{equation}
 \sum_k q_k F^2(\psi^k_{\cK^1A^1_{\cK}A^nA^n_{\cK}\cKbar^1B^1_{\cKbar}B^nB^n_{\cKbar}R^n}, \psi^{\otimes n}_{AA_{{\cal T}}BB_{\ov{\cal T}}R} \otimes \Phi^{L_{\cK}} \otimes \Gamma^{N_{\cKbar}}) \geq (1-\epsilon/2)^2.
\end{equation}
This follows from the definition of a split-transfer (eq.~(\ref{eq:splittransfer})) and the fact that $F^2$ is linear when one argument is pure. Using the relation between trace distance and fidelity (eq.~(\ref{Lemma:relation})), and the convexity of the $x^2$ function, we rewrite this as
\begin{equation}
 \sum_k q_k \bigg \|\psi^k_{\cK^1A^1_{\cK}A^nA^n_{\cK}\cKbar^1B^1_{\cKbar}B^nB^n_{\cKbar}R^n} - \psi^{\otimes n}_{AA_{{\cal T}}BB_{\ov{\cal T}}R} \otimes \Phi^{L_{\cK}} \otimes \Gamma^{N_{\cKbar}} \bigg \| \leq 2\sqrt{\epsilon(1-\epsilon/4)}.
\end{equation}
By monotonicity of the trace norm under partial tracing, we get
\begin{equation}
 \sum_k q_k \bigg \|\psi^k_{\cX^1A^1_{\cK}A^nA^n_{\cK}} - \psi^{\otimes n}_{AA_{{\cal T}}} \otimes \tau^{A_{\X}^1} \otimes \bigotimes_{i \in \cX} \Phi^{L_i} \bigg \| \leq 2\sqrt{\epsilon(1-\epsilon/4)}.
\end{equation}
Using the Fannes inequality (Lemma \ref{lem:Fannes}) and the concavity of the $\eta$-function, we have
\begin{equation*}
 \begin{split}
 \sum_k q_k \bigg |S(\cX^1A^1_{\cK}A^nA^n_{\cK})_{\psi^k} - \sum_{i \in \X} &\log L_i - n S(A\X\cX)_{\psi} \bigg |  \\ &\leq (2\sum_{i \in \cK}\log L_i + n\log d_A + n\log d_{A_{\cK}}) \eta(2\sqrt{\epsilon(1-\epsilon/4)})\\  &\leq O(n)\eta(2\sqrt{\epsilon(1-\epsilon/4)}).\\
 \end{split}
\end{equation*}
Finally, using eq.~(\ref{eq:EqIn}), we have
\begin{equation}
\sum_{i \in \X} R_i = \sum_{i \in \X} \frac{1}{n}(\log K_i - \log L_i) \geq S(\X|\cX A)_{\psi} - O(1)\eta(2\sqrt{\epsilon(1-\epsilon/4)})
\end{equation}
for any non empty subset $\X \subseteq \cK$. Using a similar argumentation, we can show that
\begin{equation}
\sum_{i \in Y} R_i = \sum_{i \in \Y} \frac{1}{n}(\log M_i - \log N_i) \geq S(\Y|\cY B)_{\psi} - O(1)\eta(2\sqrt{\epsilon(1-\epsilon/4)})
\end{equation}
holds for any non empty subset $\Y \subseteq \cKbar$. By letting $n \rightarrow \infty$ and $\epsilon \rightarrow 0$, we get the converse.
\end{proof}

\subsection{An application: entanglement of assistance} \label{sec:appl}
We are now ready to give a protocol for distilling entanglement at the min-cut rate $E_{min-cut}(\inputstate)$. If the multiparty typicality conjecture holds, which we recall is a statement on the extension of well-known typicality properties to a multiparty scenario, we can answer the conjecture of \cite{merge}: \textit{With high probability, the min-cut entanglement $E_{min-cut}(\inputstate)$ of the state $\inputstate$ is arbitrarily well-preserved after the helpers $C_1,C_2,\ldots,C_m$ perform simultaneous random measurements on their typical subspaces.}

\begin{Theorem}\label{cor:LOCC}
Let $\inputstate$ be a multipartite pure state with $m$ helpers and two receivers. Given arbitrarily many copies of $\inputstate$, there exists an LOCC protocol achieving the optimal ``assisted'' EPR rate:
\begin{equation}\label{eq:mincut}
  E^{\infty}_A(\psi,A:B) = \min_{\cK} \{S(A\cK)_{\psi}  \}
\end{equation}
\end{Theorem}

\begin{proof}
Denote by $\cK_{\min} \subseteq \{1,2,\ldots, m\}$ the min-cut of the smallest possible size such that:
\be
\forall \cK \subseteq \{C_1,C_2,\ldots,C_m\}:
S(A\cK_{\min})_{\psi} \leq S(A\cK)_{\psi}.
\ee
When $\cK_{\min}$ is not the empty set, we have, for any non-empty subset ${\X} \subseteq {\cal T}_{min}$:
 \begin{equation}\label{eq:zero}
\begin{split}
S(\X|\cX A)_{\psi} &= S(\cK_{\min}A)_{\psi} - S(\X \cKbar_{\min}B)_{\psi} \\
& < S(\cX A)_{\psi} - S(\X \cKbar_{\min} B)_{\psi}\\
&= S(\X \cKbar_{\min} B)_{\psi} - S(\X \cKbar_{\min} B)_{\psi} \\
&= 0,\\
\end{split}
\end{equation}
where in the second line we have used the fact that $S(A \cK_{\min})_{\psi} < S(A \cK)$ when $\cK$ is a smaller cut than $\cK_{\min}$. If $\cK_{\min} \neq \{1,2,\ldots, m\}$, we have, for any non-empty subset $\Y \subseteq \cKbar_{\min}$:
\begin{equation}
  \begin{split}
S(\Y|\cY B)_{\psi} &= S(\cKbar_{\min}B)_{\psi} - S(\Y \cK_{\min} A)_{\psi} \\
&= S(\cK_{\min}A)_{\psi} - S(\Y \cK_{\min} A)_{\psi} \\
&\leq 0. \\
  \end{split}
\end{equation}

Since $S(\X|\cX A)$ is negative for all non empty subsets $\X \subseteq \cK_{\min}$, there are negative rates (i.e $\overrightarrow{R_{\cK_{\min}}}=\bigoplus_{i \in \cK_{\min}}(R_i)$ where $R_i < 0$ for all $i \in \cK_{\min}$) achievable for the helpers in $\cK_{\min}$ by Theorem \ref{thm:splittransfer}. Hence, by Conjecture \ref{eq:notimesharing}, there is no need for additional injection of entanglement for these helpers. To distill entanglement, we can set their random instruments with projectors of rank $L_i = \lfloor 2^{-nR_i} \rfloor$\footnote{The previous protocol of \cite{merge} for the multipartite entanglement of assistance restricted the helpers to measurements with rank one projectors.}. If the conditional entropies $S(\Y|\cY B)_{\psi}$ are negative for all non empty subsets $\Y \subseteq \cK_{\min}$, we proceed similarly for the helpers in $\cK_{\min}$. Then, the average decoupling errors $Q^1_{\cal I}$ and $Q^2_{\cal I}$ will vanish as $n$ grows larger, and this without injecting additional entanglement. Using Markov's inequality, we have, for any $\epsilon_1, \epsilon_2 > 0$:
\begin{equation*}
\begin{split}
 \int P_{J_{\cK}}( \| \varphi_{J_{\cK}}^{\cK^1_{\min} R_{\cK_{\min}}} - \tau^{\cK^1_{\min}} \otimes \psi^{R_{\cK_{\min}}} \|_1\geq \epsilon_1)dU_{\cK_{\min}}  &\leq  \frac{ \int Q^1_{\cal I}(\psi^{\cK_{\min}A R_{\cK_{\min}}})dU_{\cK_{\min}}}{\epsilon_1} \\
 \int P_{J_{\cKbar}}( \| \upsilon_{J_{\cKbar}}^{\cKbar^1_{\min} R_{\cKbar^1_{\min}}} - \tau^{\cKbar_{\min}} \otimes \psi^{R_{\cKbar_{\min}}} \|_1 \geq \epsilon_2) dU_{\cKbar_{\min}} &\leq \frac{\int Q^2_{\cal I}(\psi^{\cKbar_{\min}B R_{\cKbar_{\min}}})dU_{\cKbar_{\min}}}{\epsilon_2}, \\
\end{split}
\end{equation*}
where the averages are respectively taken over the unitary groups $\mathbb{U}(\tilde{C}_i), i \in \cK_{\min}$ and $\mathbb{U}(\tilde{C}_i), i \in \cKbar_{\min}$. The states  $\varphi_{J_{\cK}}^{\cK^1_{\min} R_{\cK_{\min}}}$ and $\upsilon_{J_{\cKbar}}^{\cKbar^1_{\min} R_{\cKbar_{\min}}}$ are the outcome states following instruments performed by the helpers. Since the decoupling errors can be made arbitrarily small as $n$ grows larger, these probabilities will also vanish on average. Hence, there exist instruments performed by the senders which allow, with arbitrarily high probability, the redistribution of the original state to the two receivers, with $\cK_{\min}$ going to $A$ and its complement to $B$. This in turn implies that the min-cut entanglement must be arbitrarily well-preserved when the senders perform random instruments on their typical subspaces.

When some of the conditional entropies $S(\Y|\cY B)_{\psi}$ are zero, there is no split-transfer protocol achieving negative rates for all the helpers in $\cKbar_{\min}$. We leave the conjecture open for these cases. However, it is still possible to redistribute the original state and preserve the min-cut entanglement by injecting an arbitrarily small number of singlets between the cut $\Y$ vs $B \cY$. Another alternative is to produce entanglement using a \textit{subset} of the initial number of copies available to the parties. The procedure works as follows: First, the helpers in $\cK_{\min}$ transfer their systems to the receiver $A$ using a multiparty merging protocol. As the conditional entropies $S(\X|\cX A)_{\psi}$ are negative for the minimum cut, this can be achieved using only LOCC operations. The receiver $A$ then performs a measurement on $AA_{\cK_{\min}}$ of the kind described in the entanglement of assistance protocol of \cite{SVW} (see Chapter 2 also). The helper is $AA_{\cK_{\min}}$ for our setting and the recipients of the entanglement are $\Y$ and $B\cY$. Since $\Y$ and $B\cY$ are not individual parties, the decoding part of the protocol cannot be implemented. The amount of entropy of entanglement for an outcome state $\psi_{j}$, however, should be arbitrarily close to
 \[ E(\psi_j) \approx \min\{ S(\Y)_{\psi} , S(B\cY)_{\psi} \} \]
by the formula of eq.~(\ref{thm:EofA}). If the minimum is zero, the systems in $\Y$ are in a known pure state $\psi^{\Y}$ (i.e $S(Y)_{\psi} = 0$) as, by hypothesis, the conditional entropy $S(\Y | \cY B)_{\psi}$ is zero. The receiver can then locally prepare a system $A_{\Y}$ in the state $\psi^{\Y}$.

If the entropy of entanglement is positive for the outcome state, the helpers part of $\cKbar_{\min}$ can exploit this entanglement to merge their state by LOCC. An entangled state shared between $\Y$ and $B\cY$, with entropy of entanglement $E(\psi):=\Delta$, contributes $-\Delta$ to $S(\Y|\cY B)_{\psi}$, making negative rates achievable for the helpers in $\cKbar_{\min}$. Thus, the overall protocol still uses only LOCC operations and produces a state $\varphi^{A^n_{\cK_{\min}}A^n B^n_{\cKbar_{\min}}B^n}$ such that
\begin{equation}\label{eq:dist}
 \bigg \| \varphi^{A^n_{\cK_{\min}}A^n} - (\psi^{A_{\cK_{\min}}A})^{\otimes n} \bigg \|_1 \leq \bigg \| \varphi^{A^n_{\cK_{\min}}A^n B^n_{\cKbar_{\min}}B^n} - (\psi^{A_{\cK_{\min}}A B_{\cKbar_{\min}}B})^{\otimes n} \bigg \|_1 \leq \epsilon,
\end{equation}
where $\psi^{A_{\cK_{\min}}AB_{\cKbar_{\min}}B}$ is the original
state $\psi^{\cK_{\min}A \cKbar_{\min}B}$ with the systems
$A_{\cK_{\min}}$ and $B_{\cKbar_{\min}}$ substituted for the
systems $\cK_{\min}$ and $\cKbar_{\min}$. Applying the Fannes
inequality, we have
 \begin{equation}
   \bigg | S(A^n_{\cK_{\min}}A^n)_{\varphi} - nS(A_{\cK_{\min}}A^n)_{\psi} \bigg | \leq n \log (d_{A_{\cK_{\min}}}d_A) \eta(\epsilon)
 \end{equation}
which implies that
 \begin{equation}
   S(A^n_{\cK_{\min}}A^n)_{\varphi} = n(S(A_{\cK_{\min}}A)_{\psi} \pm \delta) = n(S(A\cK_{\min})_{\psi} \pm \delta),
 \end{equation}
 where $\delta$ can be made arbitrarily small by letting $\epsilon \rightarrow 0$.  Alice and Bob can distill arbitrarily close to the min-cut rate by applying an entanglement distillation protocol, as in \cite{DW}, on the state $\varphi^{A^n_{\cK_{\min}} A^n B^n_{\cKbar_{\min}} B^n}$.
\end{proof}

\chapter{Entanglement Cost of Multiparty State Transfer}

Most information processing tasks are analyzed under the assumptions of independence and repeatability. For instance, the noiseless channel coding theorem of Shannon starts from the premise that a source emits an infinite sequence of i.i.d. random variables according to a known distribution $p(x)$. In the quantum regime, Schumacher compression assumes that the quantum source is producing many identical copies of a state $\rho$. For such scenarios, the relevant measure of information is the von Neumann entropy. If we drop either of these two assumptions, the well-known formulas for operational quantities such as the minimum compression length of a message or the capacity of a noisy channel are not applicable anymore, and other measures of information such as the spectral entropy rates of \cite{spectral} and the smooth min- and max-entropies of Renner \cite{Renner02} become more appropriate.

When a single copy of a state $\psi^{ABR}$ is available, it was found in \cite{Berta} that the smooth max-entropy is the information theoretic measure which characterizes the entanglement cost associated with the task of state merging. The minimal entanglement cost $\log K_1 - \log L_1$ is bounded from below in the one-shot regime by
 \begin{equation}\label{eq:lowerboundCost}
    \log K_1 - \log L_1 \geq \hat{H}^{\sqrt{\epsilon}}_{\max}(A|B)_{\psi},
 \end{equation}
where $\hat{H}^{\epsilon}_{\max}(A|B)_{\psi}$ is an alternative version of the smooth max-entropy, which optimizes the max-entropy over all density operators $\bar{\psi}^{AB}$ close in the trace distance (i.e $D(\bar{\psi}^{AB},\psi^{AB}) \leq \epsilon$). Whenever the entanglement cost\footnote{The numbers $K_1, L_1$ are natural numbers, and so we must choose values for $K_1$ and $L_1$ such that $\log K_1 - \log L_1$ is minimal, but greater or equal than the right hand side of eq.~(\ref{eq:berta}).} satisfies
 \begin{equation}\label{eq:berta}
   \log K_1 - \log L_1 \geq \hat{H}_{\max}^{\frac{\epsilon^2}{64}}(A|B)_{\psi} + 4\log \left(\frac{1}{\epsilon} \right ) + 12,
 \end{equation}
there exists a state merging protocol for the state $\psi^{ABR}$ with error $\epsilon$. This result was derived by re-expressing the upper bound of Proposition \ref{prop:isometry} (when $m=1$) as a function of the
smooth min-entropy.

In this chapter, we extend some of the results of Berta \cite{Berta} to the multiparty setting, where $m$ senders and a receiver $B$ share a state $\initstate$, with purifying system $R$. We give a partial description of the entanglement cost region for multiparty merging of the state $\initstate$. Our main contribution is to characterize a subset of the entanglement cost region where multiparty merging for the state $\initstate$ is achievable with error $\epsilon$. For any entanglement cost tuple $(\log K_1-\log L_1, \log K_2-\log L_2, \ldots, \log K_m - \log L_m)$ satisfying
 \begin{equation}\label{eq:cost1}
  \begin{split}
   \log K_{\cK} - \log L_{\cK} := \sum_{i \in \cK} \log \left ( \frac{K_i}{L_i} \right ) &\geq -H_{\min}(\psi^{\cK R}|\psi^R) + 4\log \left(\frac{1}{\epsilon} \right )+ 2m + 8,  \\
  \end{split}
 \end{equation}
for all subsets $\cK \subseteq \{1,2,\ldots, m\}$, we show the existence of multiparty merging protocols with error $\epsilon$, where the senders measure their systems simultaneously and the decoding is not restricted to a composition form $U_{j_m}U_{j_{m-1}}\ldots U_1$. We can also show a similar result when we have $m$ senders and two receivers (split-transfer).

We also consider a different protocol where each sender merges his system one at a time as in the distributed compression protocol of \cite{merge}. The advantage in this approach is that it allows us to use the better entanglement costs proven in Dupuis et al. \cite{decouplingBerta} for one-shot state merging. We can achieve multiparty merging of the state $\initstate$ with error $\epsilon$ whenever the entanglement costs satisfy
\begin{equation}\label{eq:xxxx}
\log \bigg (\frac{K_{i}}{L_{i}} \bigg ) \geq  -H^{\frac{\epsilon^2}{52m^2}}_{\min}(C_i|\tilde{R}_{\pi^{-1}(i)})_{\psi} + 4 \log \left(\frac{2m}{\epsilon} \right ) + 2\log(13) \quad \mbox{ for all }1\leq i \leq m,
\end{equation}
where $\tilde{R}_{\pi^{-1}(i)}$ is the relative reference for the sender $C_i$ with respect to an ordering $\pi: \{1,2, \ldots, m\}$ of the senders.

The last part of this chapter is devoted to examples for one-shot distributed compression. We compute bounds on the entanglement cost for our two protocols and discuss some of the shortcomings to using one-shot two-party state merging protocols for achieving the task of multiparty state merging. The last example considers a family of states for which smoothing has little effect on the min-entropies appearing in eq.~(\ref{eq:xxxx}). The entanglement contained in such states does not reduce to bipartite entanglement between some of the subsystems, making it harder to analyze the entanglement costs of merging. We can nonetheless get interesting bounds on the entanglement cost by using the Gershgorin Circle Theorem \cite{Horn,wiki}, a standard result in matrix theory.
\section{Achieving one-shot multiparty state transfer}
Our first result is a reformulation of Lemma \ref{Lemma:rdmisometry} in Chapter 3 as a function of min-entropies:
\begin{Lemma}[Compare to Lemma 4.5 of \cite{Berta}] \label{Lemma:oneshotdecouple}
For each sender $C_i$, let $P_i: C_i \rightarrow C_i^1$ be a projector of dimension $L_i$ onto
a subspace $C_i^1$ of $C_i$ and $U_i$ a unitary acting on $C_i$. Define the sub-normalized state
 \begin{equation*}
   \omeg := (P_1U_1 \otimes P_2U_2 \otimes \ldots
   \otimes P_mU_m \otimes I^R) \initstate (P_1U_1 \otimes P_2U_2
   \otimes \ldots \otimes P_mU_m \otimes I^R)^{\dag}.
 \end{equation*}
If $U_1, U_2,\ldots, U_m$ are Haar distributed unitaries, then for any state $\sigma^R$ of the system~$R$, we have
 \begin{equation}\label{eq:weak}
   \int_{\mathbb{U}(C_M)} \biggl \|  \omeg -
   \frac{L_M}{d_{C_M}} \tau^{C^1_M} \otimes \psi^{R} \biggr \|_1 dU_M \leq
   \frac{L_M}{d_{C_M}}\sqrt{\sum_{\substack{\cK \subseteq \{1,2,\ldots,m\} \\ \cK
   \neq \emptyset}}
   2^{-(H_{\mathrm{min}}(\psi^{\cK R}|\sigma^R) -
   \log L_{\cK})}}.
 \end{equation}
\end{Lemma}

\begin{proof} By Lemma \ref{Lemma:TrSchmidt}, we have, for
any state $\sigma^R$ of the reference $R$,
\begin{equation}
\begin{split}
\label{eq:A.6} \biggl \| \omeg -
&\frac{L_M}{d_{C_M}}\tau^{C^1_M} \otimes \psi^{R} \biggr \|_1 \\ &\leq \sqrt{L_M} \biggl \|
(I^{C^1_M} \otimes \sigma_R^{-\frac{1}{4}})(
\omeg - \frac{L_M}{d_{C_M}}\tau^{C^1_M} \otimes \psi^R)(I^{C^1_M} \otimes
\sigma_R^{-\frac{1}{4}}) \biggr \|_2.
\end{split}
\end{equation}
\newcommand{\tomeg}{\tilde{\omega}^{C^1_MR}(U_M)}
Define the states
\begin{equation*}
\begin{split}
 \tilde{\psi}^{C_MR} &:= (I^{C_M} \otimes \sigma_R^{-\frac{1}{4}})
 \psi^{C_MR}(I^{C_M} \otimes \sigma_R^{-\frac{1}{4}}) \\
 \tomeg &:= (P_1U_1 \otimes \ldots
 \otimes P_mU_m \otimes I^R) \tilde{\psi}^{C_MR} (P_1U_1 \otimes \ldots \otimes P_mU_m \otimes I^R)^{\dag} \\
 \end{split}
\end{equation*}
and rewrite eq.~(\ref{eq:A.6}) as
\begin{equation*}
\biggl \| \omeg -
\frac{L_M}{d_{C_M}}\tau^{C^1_M} \otimes \psi^{R} \biggr \|_1 \leq \sqrt{L_M}\biggl \|\tomeg -
\frac{L_M}{d_{C_M}}\tau^{C^1_M} \otimes \tilde{\psi}^R \biggr \|_2.
\end{equation*}
Using eq.~(\ref{eq:decouple2}) in the proof of Lemma \ref{Lemma:rdmisometry}, we have
\begin{equation}\label{eq:tilde}
 \begin{split}
\int_{\mathbb{U}(C_M)} \biggl \| \tomeg -
\frac{L_M}{d_{C_M}}\tau^{C^1_M} \otimes \tilde{\psi}^R \biggr \|_2^2
dU_M & \leq \sum_{\substack{ \cK \subseteq \{1,2,\ldots,m\} \\
\cK \neq \emptyset}} \prod_{i \notin \cK} \frac{L_i}{d^2_{C_i}}
\prod_{i \in \cK} \frac{L_i^2}{d^2_{C_i}}
\Tr \biggl [ \tilde{\psi}_{\cK R}^2 \biggr ] \\
 & \leq \frac{L_M}{d^2_{C_M}} \sum_{\substack{ \cK \subseteq
\{1,2,\ldots,m\} \\ \cK \neq \emptyset}}L_{\cK} \Tr
\biggl [ \tilde{\psi}_{\cK R}^2 \biggr ]. \\
 \end{split}
\end{equation}
The quantity $\Tr[\tilde{\psi}_{\cK R}^2]$ is rewritten
as:
\begin{equation}\label{eq:trace}
 \begin{split}
   \Tr [\tilde{\psi}_{\cK R}^2] & = \Tr \biggl [ \biggl (
   \Tr_{\overline{\cK}}[\tilde{\psi}^{C_MR}] \biggr)^2 \biggr]  \\
    &= \Tr \biggl [ \biggl ( (I^{\cK} \otimes
    \sigma_R^{-\frac{1}{4}})\psi^{\cK R}(I^{\cK} \otimes
    \sigma_R^{-\frac{1}{4}}) \biggr )^2 \biggr ] \\
   &= 2^{-H_2(\psi^{\cK R} | \sigma^R)}, \\
 \end{split}\
\end{equation}
where $H_2(\psi^{\cK R}|\sigma^R)$ is the conditional collision
entropy of $\psi^{\cK R}$ relative to $\sigma^R$. Combining
eqs.~(\ref{eq:A.6}), (\ref{eq:tilde}), (\ref{eq:trace})
and using the fact that $H_{\mathrm{min}}(\psi^{\cK
R}|\sigma^{R}) \leq H_2(\psi^{\cK R}|\sigma^R)$ (Lemma \ref{lem:rener4}) and the square root function is concave, we have
\begin{equation}
 \begin{split}
   \int_{\mathbb{U}(C_M)} \biggl \| \omeg -
   \frac{L_M}{d_{C_M}}\tau^{C^1_M} \otimes \psi^{R} \biggr \|_1 dU_M & \leq \frac{L_M}{d_{C_M}}
   \sqrt{ \sum_{\substack{ \cK \subseteq
\{1,2,\ldots,m\} \\ \cK \neq \emptyset}}L_{\cK}
2^{-H_2(\psi^{\cK R}|\sigma^R)}} \\
 & \leq \frac{L_M}{d_{C_M}} \sqrt{\sum_{\substack{\cK \subseteq \{1,2,\ldots,m\} \\ \cK
   \neq \emptyset}}
   2^{-(H_{\mathrm{min}}(\psi^{\cK R}|\sigma^R) -
   \log L_{\cK})}}, \\
 \end{split}
\end{equation}
and so we are done.
\end{proof}
Using this result, we can rephrase Proposition \ref{prop:isometry} in the language of min-entropies. This result extends Lemma 4.6 in \cite{Berta} to the multiparty setting.

\begin{Theorem}[Compare to Lemma 4.6 of \cite{Berta}]\label{thm:cost1}
Let $\initstate$ be a multipartite state shared by $m$ senders and a receiver, with purifying system $R$. If, for any $\epsilon > 0$, the entanglement cost
$\overrightarrow{E} = (\log K_1 - \log L_1,\log K_2 - \log
L_2,\ldots, \log K_m - \log L_m)$ satisfies
 \begin{equation}\label{eq:cost1}
  \begin{split}
   \log K_{\cK} - \log L_{\cK} := \sum_{i \in \cK} \log \left ( \frac{K_i}{L_i} \right ) &\geq -H_{\min}(\psi^{\cK R}|\psi^R) + 4\log \left(\frac{1}{\epsilon} \right )+ 2m + 8  \\
  \end{split}
 \end{equation}
for all non-empty subsets $\cK \subseteq \{1,2,\ldots, m\}$, there
exists a state merging protocol for the state $\initstate$ with error
$\epsilon$.
\end{Theorem}

\begin{proof}
We fix a random instrument for each sender $C_i$ as in Proposition $\ref{prop:isometry}$:
each sender $C_i$ has an instrument with $N_i:=\lfloor
\frac{d_{C_i}K_i}{L_i} \rfloor$ partial isometries $P^j_i =
Q^j_iU_i$, where $U_i$ is a Haar distributed unitary acting
on the system $C_iC^0_i$. Recall that $Q^j_i$ is a partial isometry mapping to a subspace $C_i^1$ of $C_iC^0_i$. If $d_{C_i}K_i > N_i L_i$, there is an
extra partial isometry $P^0_i$ of rank $L'_i:=d_{C_i}K_i -
N_iL_i < L_i$. Applying Lemma \ref{Lemma:oneshotdecouple} for the
state $\psi^{C_MR} \otimes \tau^{K_M}$, with $\sigma^R = \psi^R$, and using additivity of the min-entropy (Lemma \ref{lem:addi}), we have
 \begin{equation}\label{eq:deriv}
   \begin{split}
  \int_{\mathbb{U}(C_M)}\sum_{j_{1}=1}^{N_1} \sum_{j_2=1}^{N_2} \cdots \sum_{j_m=1}^{N_m} \bigg \| &\omegJ
- \decouple \bigg \|_1 dU_M \\ &\leq \frac{\prod_{i=1}^m N_i
L_i}{d_{C_M}K_M} \sqrt{\sum_{\substack{\cK \subseteq
\{1,2,\ldots,m\} \\ \cK \neq \emptyset}}
2^{-(H_{\mathrm{min}}(\psi^{\cK R}|\psi^R) + \log K_{\cK} -
   \log L_{\cK})}} \\
  &\leq \sqrt{\sum_{\substack{\cK \subseteq \{1,2,\ldots,m\} \\ \cK
   \neq \emptyset}}
   2^{-(H_{\mathrm{min}}(\psi^{\cK R}|\psi^R) + \log K_{\cK} -
   \log L_{\cK})}}.
\end{split}
\end{equation}
Using the constraint of eq.~(\ref{eq:cost1}) for the entanglement cost $\log{K_{\cK}} - \log{L_{\cK}}$, we simplify the previous
inequality to
 \begin{equation}\label{eq:lhs}
   \begin{split}
  \int_{\mathbb{U}(C_M)}\sum_{j_{1}=1}^{N_1} \sum_{j_2=1}^{N_2} \cdots \sum_{j_m=1}^{N_m} \bigg \| &\omegJ
- \decouple \bigg \|_1 dU_M \\ &\leq \sqrt{\sum_{\substack{\cK
\subseteq \{1,2,\ldots,m\} \\ \cK
   \neq \emptyset}}
   2^{-(H_{\mathrm{min}}(\psi^{\cK R}|\psi^R) + \log K_{\cK} -
   \log L_{\cK})}} \\
   &\leq \frac{\epsilon^2}{2^{\frac{m+8}{2}}} \leq \frac{\epsilon^2}{16}.
\end{split}
\end{equation}
Taking normalisation into account, with $p_{J_M}(U_M) =
\Tr(\omega_{J_M}^{C^1_MR}(U_M))$ and $\psi_{J_M}^{C^1_MR} =
\frac{\omega_{J_M}^{C^1_MR}(U_M)}{p_{J_M}(U_M)}$, we trace out the left hand
side of eq.~(\ref{eq:lhs}) and obtain
\begin{equation*}
\int_{\mathbb{U}(C_M)}\sum_{j_{1}=1}^{N_1} \sum_{j_2=1}^{N_2} \cdots \sum_{j_m=1}^{N_m} \bigg | p_{J_M}(U_M) - \frac{L_M}{d_{C_M}} \bigg | dU_M \leq \frac{\epsilon^2}{16}. \\
\end{equation*}
Applying the triangle inequality, we have
\begin{equation*}
  \int_{\mathbb{U}(C_M)}\sum_{j_{1}=1}^{N_1} \sum_{j_2=1}^{N_2} \cdots \sum_{j_m=1}^{N_m} p_{J_M}(U_M) \bigg \| \psi^{C^1_MR}_{J_M}
- \tau^{C^1_M} \otimes \psi^R \bigg \|_1 dU_M \leq \frac{\epsilon^2}{8}.
\end{equation*}
Combining this with eq.~(\ref{eq:final}), found in the proof of Proposition \ref{prop:isometry}, we have an upper bound to the decoupling
error $Q_{\cal I}(\psi^{C_MBR} \otimes \Phi^{K_M})$ as a function of the parameter~$\epsilon$:
\begin{equation*}
\begin{split}
 &\int_{\mathbb{U}(C_M)}\sum_{j_{1}=0}^{N_1} \sum_{j_2=0}^{N_2} \cdots \sum_{j_m=0}^{N_m} p_{J_M}(U_M) \bigg \| \psi^{C^1_MR}_{J_M} - \tau^{C^1_M} \otimes \psi^R \bigg \|_1 dU_M \\
&\leq 2 \sum_{\substack{{\cal T} \subseteq \{1,2,...,m\} \\ {\cal
T} \neq \emptyset}}\prod_{i \in {\cal T}} \frac{L_i}{d_{C_i}K_i} +
\int_{\mathbb{U}(C_M)} \sum_{j_{1}=1}^{N_1}\cdots
\sum_{j_m=1}^{N_m} p_{J_M} \bigg \| \psi^{C^1_MR}_{J_M}
- \tau^{C^1_M} \otimes \psi^R \bigg \|_1 dU_M \\
 &\leq  \sum_{\substack{{\cal T} \subseteq \{1,2,...,m\}}} \frac{2\epsilon^4 2^{H_{\min}(\psi^{\cK R}|\psi^R)}}{2^{2m+8}d_{C_{\cK}}} + \frac{\epsilon^2}{8} \\
 &\leq  \sum_{\substack{{\cal T} \subseteq \{1,2,...,m\}}} \frac{2\epsilon^4  2^{H_{\min}(\psi^{\cK})}}{2^{2m+8}d_{C_{\cK}}} + \frac{\epsilon^2}{8}\\
&\leq  \frac{\epsilon^4}{2^{m+7}} + \frac{\epsilon^2}{8} \leq \frac{\epsilon^2}{4}.
   \end{split}
 \end{equation*}
The fourth line holds by the strong subadditivity of the min-entropy \cite{Renner02}:
\[ H_{\min}(\psi^{\cK R}|\psi^R) \leq H_{\min}(\psi^{\cK}), \]
and the last line holds since
\[H_{\min}(\psi^{\cK}) = -\log{\lambda_{\max}(\psi^{\cK})} \leq \log{d_{C_{\cK}}}.\]
By Proposition \ref{prop:mergeCond}, there exists a state merging protocol for the state $\initstate$ with error
$2\sqrt{\epsilon^2/4}=\epsilon$, and so we are done.
\end{proof}

Lemma 4.6 of Berta \cite{Berta} suggest an improvement of our previous Theorem: replacing the min-entropies $H_{\min}(\psi^{\cK R}|\psi^R)$ appearing in eq.~(\ref{eq:cost1}) by their conditional versions $H_{\min}(\cK|R)_{\psi}$. The theorem would remain valid for these weaker constraints on the entanglement cost provided we can prove a more general version of Lemma~\ref{Lemma:oneshotdecouple}:
\begin{equation}\label{eq:x1}
\int_{\mathbb{U}(C_M)} \biggl \|  \omeg -
   \frac{L_M}{d_{C_M}} \tau^{C^1_M} \otimes \psi^{R} \biggr \|_1 dU_M \leq
   \frac{L_M}{d_{C_M}}\sqrt{\sum_{\substack{\cK \subseteq \{1,2,\ldots,m\} \\ \cK
   \neq \emptyset}}
   2^{-(H_{\mathrm{min}}(\psi^{\cK R}|\sigma_{\cK}^R) -
   \log L_{\cK})}},
\end{equation}
where $\sigma_{\cK}^R$ are $2^m - 1$ possibly different density operators. Assuming this inequality to be true, we can adapt the previous proof by setting $\sigma_{\cK}^R := \bar{\sigma}^R_{\cK}$, where $H_{\min}(\psi^{\cK
R}|\bar{\sigma}^R_{\cK})=H_{\min}(\cK|R)_{\psi}$, and $\log K_{\cK} -
\log L_{\cK} \geq -H_{\min}(\cK|R)_{\psi} +
4\log \left(\frac{1}{\epsilon} \right ) + 2m +8$.

The results of Berta also suggest another improvement of our Theorem \ref{thm:cost1}: smoothing the min-entropies $H_{\min}(\cK|R)_{\psi}$ around sub-normalized density operators $\bar{\psi}^{\cK R}$ close in distance to the state $\psi^{\cK R}$. To satisfy these looser requirements on the entanglement cost, we would need to adjust the proof of Theorem \ref{thm:cost1} by updating eq.~(\ref{eq:x1}) to an even stronger version, where
\begin{equation}\label{eq:x}
\int_{\mathbb{U}(C_M)} \biggl \|  \omeg -
   \frac{L_M}{d_{C_M}} \tau^{C^1_M} \otimes \psi^{R} \biggr \|_1 dU_M \leq
   \frac{L_M}{d_{C_M}}\sqrt{\sum_{\substack{\cK \subseteq \{1,2,\ldots,m\} \\ \cK
   \neq \emptyset}}
   2^{-(H^{\epsilon}_{\mathrm{min}}(\cK|R)_{\psi} -
   \log L_{\cK})}},
\end{equation}
and this inequality holds for any fixed $\epsilon \geq 0$. Presently, it is unclear if these improvements of Lemma $\ref{Lemma:oneshotdecouple}$ hold. We leave it as an interesting open problem:

\begin{conjecture}\label{conj:cost1}
Let $\psi^{C_MBR}$ be a multipartite state shared between $m$ senders and a receiver $B$, with purifying system $R$. For
any $\epsilon > 0$, there exist multiparty state merging protocols for the state $\initstate$ with
error $\epsilon$ whenever the entanglement cost
$\overrightarrow{E}:=(\log K_1 - \log L_1, \log K_2 - \log L_2,
\ldots , \log K_m - \log L_m)$ satisfies
\begin{equation}\label{eq:cost2}
   \log K_{\cK} - \log L_{\cK} := \sum_{i \in \cK} (\log K_i - \log L_i) \geq   H^{\epsilon}_{\max}(\cKbar B|B)_{\psi} +
   O(\log1 / \epsilon)+O(m)
 \end{equation}
for all non-empty subsets $\cK \subseteq \{1,2,\ldots, m\}$.
\end{conjecture}
The main difficulty in proving the conjecture is that it allows independent smoothing of each of the min-entropies. It is straightforward to modify our proof to allow smoothing using a common state for all the min-entropies, but the
monolithic nature of the protocol does not naturally permit tailoring
the smoothing state term-by-term. We can, however, give a partial characterization of the entanglement cost in terms of smooth min-entropies if we apply the single-shot state merging protocol of Berta \cite{Berta} on one sender at a time. We will actually use a more recent result by Dupuis et al. \cite{decouplingBerta}, which characterizes the entanglement cost of merging using the current definition of the smooth max-entropies in terms of the purified distance.

\begin{Proposition}\label{prop:epsilonCost}
For a multipartite state $\initstate$ shared between $m$ senders and a receiver $B$, let $\pi:\{1,2,\ldots,m\}\rightarrow \{1,2,\ldots,m\}$ be any ordering of the $m$ senders $C_1,C_2,\ldots,C_m$. For any entanglement cost $\overrightarrow{E}=(\log K_1 - \log L_1, \ldots , \log K_m - \log
L_m)$ satisfying
 \begin{equation} \label{eq:entanglementcost}
     \log \bigg (\frac{K_{i}}{L_{i}} \bigg ) \geq  -H^{\frac{\epsilon^2}{52m^2}}_{\min}(C_i|\tilde{R}_{\pi^{-1}(i)})_{\psi} + 4 \log \left(\frac{2m}{\epsilon} \right ) + 2\log(13) \quad \mbox{ for all }1\leq i \leq m,
 \end{equation}
where $\tilde{R}_i:=R \bigotimes^m_{j=i+1}C_{\pi(j)}$ is the relative reference for the sender $C_{\pi(i)}$, there exists a multiparty state merging protocol for the state $\initstate$ with error $\epsilon$.
\end{Proposition}
\begin{proof}
Our multiparty state merging protocol for the state
$\initstate$ consists of transferring the sender's systems one at a time according to the
ordering $\pi$: The sender $C_{\pi(1)}$ merges his part of the
state first, followed by $C_{\pi(2)}$, $C_{\pi(3)}$, etc. Write
the input state $\psi^{C_M BR}$ as $\psi^{C_{\pi(1)}\tilde{R}_1B}$, where $\tilde{R}_1 = R C_{\pi(2)}C_{\pi(3)}\ldots
C_{\pi(m)}$ is the relative reference for the sender $C_{\pi(1)}$. By Theorem 5.2 of \cite{decouplingBerta}, there exists a
state merging protocol of error $\epsilon/m$ and entanglement cost\footnote{The merging error in \cite{decouplingBerta} is defined in terms of the purified distance, which is lower bounded by the trace distance. We adjusted the entanglement cost proved in \cite{decouplingBerta} to meet our definition of merging, which is expressed in terms of the trace norm.}
\begin{equation}
 \begin{split} \log K'_1 - \log L'_1 &:=
-H^{\frac{\epsilon^2}{52m^2}}_{\min}(C_{\pi(1)}|\tilde{R}_1)_{\psi} +
4 \log \left(\frac{2m}{\epsilon} \right ) + 2\log(13) \\
 &\leq \log K_{\pi(1)} - \log L_{\pi(1)}, \\
  \end{split}
\end{equation} producing an output
state $\rho^{C^1_{\pi(1)}B^1_{\pi(1)}B_{\pi(1)}B\tilde{R}_1}$ satisfying
\begin{equation}
  \bigg \| \rho_1^{C^1_{\pi(1)}B^1_{\pi(1)}B_{\pi(1)}B \tilde{R}_1} - \psi^{B_{\pi(1)}B\tilde{R}_1} \otimes \Phi^{L_1} \bigg \|_1 \leq \frac{\epsilon}{m},
\end{equation}
where the system $B_{\pi(1)}$ is substituted for the system
$C_{\pi(1)}$.

After $C_{\pi(1)}$ has merged his share, the next
sender $C_{\pi(2)}$ performs a random instrument on his systems
and send the measurement outcome to the receiver.
Assume the parties share the state $\psi^{B_{\pi(1)}B\tilde{R}_1} \otimes \Phi^{L_1}$ instead of the output
state $\rho_1$. Write the state $\psi^{B_{\pi(1)}B\tilde{R}_1}$ as $\psi^{C_{\pi(2)}B_2\tilde{R}_2}$, with
$B_2:=B_{\pi(1)}B$ and $\tilde{R}_2:=RC_{\pi(3)}C_{\pi(4)}\ldots
C_{\pi(m)}$. Using Theorem 5.2 of Dupuis et al. \cite{decouplingBerta} for the second time, there exists
a state merging protocol of error $\epsilon/m$ and entanglement
cost \begin{equation}
\begin{split}
 \log K'_2 - \log L'_2 &=
-H^{\frac{\epsilon^2}{52m^2}}_{\min}(C_{\pi(2)}|\tilde{R}_2)_{\psi} +
4\log \left(\frac{2m}{\epsilon} \right )+2\log(13)  \\
& \leq \log K_{\pi(2)} - \log L_{\pi(2)}, \\
\end{split}
\end{equation}
producing the output state $\rho_2^{C^1_{\pi(1)}C^1_{\pi(2)}B^1_{\pi(1)}B^1_{\pi(2)}B_{\pi(2)}B_2 \tilde{R}_2}$, which satisfies
\begin{equation}
  \bigg \| \rho_2^{C^1_{\pi(1)}C^1_{\pi(2)}B^1_{\pi(1)}B^1_{\pi(2)}B_{\pi(2)}B_2 \tilde{R}_2} - \psi^{B_{\pi(2)}B_2\tilde{R}_2} \otimes \Phi^{L_1} \otimes \Phi^{L_2} \bigg \|_1 \leq \frac{\epsilon}{m},
\end{equation}
where the system $B_{\pi(2)}$ is substituted for the system
$C_{\pi(2)}$. If we apply the same protocol on the state
$\rho_1^{C^1_{\pi(1)}B^1_{\pi(1)}B_{\pi(1)}B \tilde{R}_1}$, we have
an output state $\rho_3$ satisfying
\begin{equation}
 \begin{split}
 \bigg \|& \rho_3^{C^1_{\pi(1)}C^1_{\pi(2)}B^1_{\pi(1)}B^1_{\pi(2)}B_{\pi(2)}B_2 \tilde{R}_2} - \psi^{B_{\pi(1)}B_{\pi(2)}B \tilde{R}_2} \otimes \Phi^{L_1} \otimes \Phi^{L_2} \bigg \|_1 \\ &\leq \|\rho_3 - \rho_2 \|_1 + \|\rho_2 - \psi^{B_{\pi(1)}B_{\pi(2)}B \tilde{R}_2} \otimes \Phi^{L_1} \otimes \Phi^{L_2} \|_1 \\
   &\leq \| \rho_1 - \psi^{B_{\pi(1)}B\tilde{R}_1} \otimes \Phi^{L_1} \|_1 + \|\rho_2 - \psi^{B_{\pi(1)}B_{\pi(2)}B \tilde{R}_2} \otimes \Phi^{L_1} \otimes \Phi^{L_2} \|_1 \\
  &\leq \frac{2\epsilon}{m}. \\
  \end{split}
\end{equation}
The second line is an application of the triangle inequality and the third line was obtained using monotonicity of the
trace distance under quantum operations.
\begin{figure}[t]
\centering
    \includegraphics[width=0.7\textwidth]{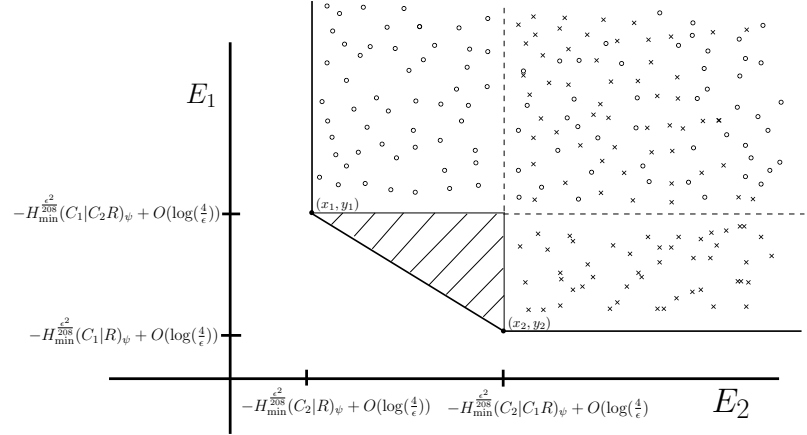}
\caption{Entanglement cost region prescribed by Proposition~\ref{prop:epsilonCost} for multiparty state merging for the case of two senders. The axes correspond to the entanglement cost
$E_1:=\log K_1-\log L_1$ and $E_2:=\log K_2 - \log L_2$. We have two possible orderings for the senders, and
according to Proposition~\ref{prop:epsilonCost}, two intersecting
regions where existence of a $2$-party state
merging of error $\epsilon$ can be shown. The region covered with circles (resp. crosses) is the entanglement cost region associated with merging the sender $C_1$ (resp. $C_2$) first. }\label{fig:rateregion1}
\end{figure}
The analysis for the other senders $C_{\pi(3)},C_{\pi(4)},\ldots,C_{\pi(m)}$ is performed similarly, and so, the final output state $\rho_{m}$ satisfies
  \begin{equation}
  \bigg \| \rho_m - \psi^{B_{\pi(1)}B_{\pi(2)}\ldots B_{\pi(m)} B R} \otimes \Phi^{L_1} \otimes \Phi^{L_2} \otimes \ldots \otimes \Phi^{L_m}\bigg \|_1 \leq \epsilon.
  \end{equation}
Hence, there exists a multiparty state
merging protocol of error $\epsilon$ for the state $\initstate$, with entanglement cost $\overrightarrow{E}:=(\log K_1 - \log L_1, \log K_2 - \log L_2, \ldots, \log K_m - \log L_m)$ satisfying
eq.~(\ref{eq:entanglementcost}).
\end{proof}

Figure \ref{fig:rateregion1} depicts the entanglement cost regions prescribed by Proposition \ref{prop:epsilonCost} for multiparty state merging for the case of two senders. Note that the hatched area is not part of the cost region characterized by  eq.~(\ref{eq:cost1}).

If only a single copy of $\inputGroupstate$ is available to the involved parties, we can adapt the argument of Theorem \ref{thm:cost1} and prove the following result concerning the existence of split-transfer protocols with error $\epsilon$:
\begin{Proposition}
\label{prop:1ShotSplit}
Given a partition $\cK \subseteq \{1,2,\ldots,m\}$ of the senders $C_1, C_2, \ldots, C_m$, let $\inputGroupstate$ be a multipartite state shared between $m$ senders and two receivers $A$ and $B$, with purifying system $R$. For any $\epsilon_1, \epsilon_2> 0$, if the entanglement costs
$\overrightarrow{E_{\cK}} = \bigoplus_{i \in \cK}(\log K_i - \log L_i)$ and $\overrightarrow{E_{\cKbar}} = \bigoplus_{i \in \cKbar} (\log M_i - \log N_i)$ satisfy
 \begin{equation*}
   \begin{split}
   \sum_{i \in {\cal X}} (\log{K_i}- \log{L_i}) &\geq  -H_{\min}(\psi^{{\cal X} R_{\cK}}|\psi^{R_{\cK}}) + 4\log \left (\frac{1}{\epsilon_1} \right) + 2|\cK| + 8  \\
   \sum_{i \in {\cal Y}}( \log{M_i}-\log{N_i} ) &\geq  -H_{\min}(\psi^{{\cal Y} R_{\cKbar}}|\psi^{R_{\cKbar}}) + 4\log \left(\frac{1}{\epsilon_2} \right) + 2|\cKbar| + 8 \\
     \end{split}
 \end{equation*}
for all non-empty subsets ${\cal X} \subseteq \cK$ and ${\cal Y} \subseteq \cKbar$, there
exists a split-transfer protocol for the state $\inputGroupstate$ with error
$\epsilon_1+\epsilon_2$.
\end{Proposition}

\begin{proof}
The proof is very similar to the proof of Theorem \ref{thm:cost1}. First, we fix random instruments for each helper $C_i$ in a manner analogous to Proposition \ref{prop:isometry}. For each helper $C_i$ in $\cK$, we have $F_i = \lfloor \frac{d_{C_i}K_i}{L_i}\rfloor$ partial isometries $Q^j_i U_i$ of rank $L_i$, where $Q^j_i$ is defined as in Proposition \ref{prop:isometry} and $U_i$ is a Haar distributed unitary acting on $C_iC^0_i$. If $F_iL_i < d_{C_i}K_i$, we also have a partial isometry of rank $L'_i < L_i$. Similarly, for each helper $C_i$ in $\cKbar$, we have $G_i = \lfloor \frac{d_{C_i}M_i}{N_i}\rfloor$ partial isometries $Q^j_i U_i$ of rank $N_i$, and one of rank $N'_i$ if $G_iN_i < d_{C_i} M_i$. For a measurement outcome $J_M:=(j_1,j_2,\ldots, j_m)$, let $J_{\cK} = \bigoplus_{i \in \cK} j_i$ be the vector of length $t=|\cK|$ whose components correspond to the measurement outcomes for the senders belonging to the cut $\cK$. The $i$-th element of $J_{\cK}$ will be denoted by $j_{{\cK}(i)}$. Define
 \begin{equation}
   \omega_{J_{\cK}}^{\cK^1 R_{\cK}} := (Q_{\cK}^{J_{\cK}}U_{\cK} \otimes I_{R_{\cK}}) \psi^{\cK R_{\cK}} (Q_J^{\cK}U_{\cK} \otimes I_{R_{\cK}})^{\dag},
 \end{equation}
where $Q^{J_{\cK}}_{\cK} := \bigotimes_{i \in \cK} Q_i^{j_i}$. Applying Lemma \ref{Lemma:oneshotdecouple} to the state $\psi^{\cK R_{\cK}} \otimes \tau^{K_{\cK}}$, we have
 \begin{equation}\label{eq:deriv2}
   \begin{split}
  \int_{\mathbb{U}(C_{\cK})}\sum_{j_{{\cK}(1)}=1}^{F_1} \sum_{j_{{\cK}(2)}=1}^{F_2} \cdots \sum_{j_{{\cK}(t)}}^{F_t} \bigg \| &\omega^{\cK^1 R_{\cK}}_{J_{\cK}}
- \frac{L_{\cK}}{d_{C_{\cK}}} \tau^{\cK^1} \otimes \psi^{R_{\cK}} \bigg \|_1 dU_{\cK} \\ &\leq \frac{\prod_{i \in \cK} F_i
L_i}{d_{C_{\cK}}K_{\cK}} \sqrt{\sum_{\substack{{\cal X} \subseteq
\cK \\ {\cal X} \neq \emptyset}}
2^{-(H_{\mathrm{min}}(\psi^{{\cal X} R_{\cK}}|\psi^{R_{\cK}}) + \log K_{{\cal X}} -
   \log L_{{\cal X}})}} \\
  &\leq \sqrt{\sum_{\substack{{\cal X} \subseteq \cK \\ {\cal X}
   \neq \emptyset}}
   2^{-(H_{\mathrm{min}}(\psi^{{\cal X} R_{\cK}}|\psi^{R_{\cK}}) + \log K_{{\cal X}} -
   \log L_{{\cal X}})}}, \\
\end{split}
\end{equation}
where $K_{{\cal X}} := \prod_{i \in {\cal X}} K_i$. By hypothesis, we have
 \[
 \sum_{i \in \X} (\log K_{i} - \log L_{i}) \geq -H_{\min}(\psi^{{\cal X} R_{\cK}}|\psi^{R_{\cK}}) + 4\log \left(\frac{1}{\epsilon_1} \right ) + 2|\cK| + 8
 \]
 for all non empty subsets $\X \subseteq \cK$. By the proof of Theorem \ref{thm:cost1}, we have the following bound on the average quantum error $Q^1_{\cal I}(\psi^{\cK A R_{\cK}} \otimes \Phi^{K_{\cK}})$:
\begin{equation*}
\begin{split}
 &\int_{\mathbb{U}(C_{\cK})}\sum_{j_{{\cK}(1)}=0}^{F_1} \sum_{j_{{\cK}(2)}=0}^{F_2} \cdots \sum_{j_{{\cK}(t)}=0}^{F_t} p_{J_{\cK}} \bigg \| \psi^{\cK^1 R_{\cK}}_{J_{\cK}}
- \tau^{\cK^1} \otimes \psi^{R_{\cK}} \bigg \|_1 dU_{\cK} \\
&\leq 2 \sum_{\substack{{\cal X} \subseteq \cK \\ {\cal
X} \neq \emptyset}}\prod_{i \in {\cal X}} \frac{L_i}{d_{C_i}K_i} +
{\sum_{j_{{\cK}(1)}=1}^{F_1} \sum_{j_{{\cK}(2)}=1}^{F_2} \cdots
\sum_{j_{{\cK}(t)}=1}^{F_t} \int_{\mathbb{U}(C_{\cK})}p_{J_{\cK}} \bigg \| \psi^{\cK^1 R_{\cK}}_{J_{\cK}}
- \tau^{\cK^1} \otimes \psi^{R_{\cK}} \bigg \|_1 } dU_{\cK}\\
 &\leq  \sum_{\substack{{\cal X} \subseteq \cK \\ {\cal
X} \neq \emptyset}}\frac{2\epsilon_1^4 2^{H_{\min}(\psi^{\cal X})}}{2^{2t+8}d_{C_{{\cal X}}}} + \frac{\epsilon^2_1}{8}\\
& \leq  \frac{\epsilon_1^4}{2^{t+7}} + \frac{\epsilon^2_1}{8} \leq \frac{\epsilon_1^2}{4}, \\
 \end{split}
 \end{equation*}
 where $t=|\cK|, p_{J_{\cK}} = \Tr(\omega^{\cK^1 R_{\cK}}_{J_{\cK}})$ and $\psi^{\cK^1 R_{\cK}}_{J_{\cK}} = \frac{1}{p_{J_{\cK}}}\omega^{\cK^1 R_{\cK}}_{J_{\cK}}$. In a similar way, we bound the average quantum error $Q^2_{\cal I}(\psi^{\cKbar B R_{\cKbar}} \otimes \Gamma^{M_{\cKbar}})$ as follows:
\begin{equation*}
\begin{split}
 &\int_{\mathbb{U}(C_{\cKbar})}\sum_{j_{{\cKbar}(1)}=0}^{G_1} \sum_{j_{{\cKbar}(2)}=0}^{G_2} \cdots \sum_{j_{\cKbar(m-t)}=0}^{G_{m-t}} p_{J_{\cKbar}} \bigg \| \psi^{\cKbar^1 R_{\cKbar}}_{J_{\cKbar}}
- \tau^{\cKbar^1} \otimes \psi^{R_{\cKbar}} \bigg \|_1 dU_{\cKbar} \\
 &\leq  \sum_{\substack{{\cal Y} \subseteq \cKbar \\ {\cal
Y} \neq \emptyset}}\frac{2\epsilon_2^4 2^{H_{\min}(\psi^{{\cal Y}})}}{2^{2(m-t)+8}d_{C_{{\cal Y}}}} + \frac{\epsilon^2_2}{8}\\
 &\leq  \frac{\epsilon^4_2}{2^{m-t+7}} + \frac{\epsilon^2_2}{8} \leq \frac{\epsilon^2_2}{4}. \\
 \end{split}
 \end{equation*}
By Proposition \ref{prop:mergeConds}, there exists a split-transfer protocol of error $\epsilon_1+\epsilon_2$, and so we are done. \end{proof}

\renewcommand{\min}{{\operatorname{min}}}
\renewcommand{\max}{{\operatorname{max}}}

\section{One-shot distributed compression} \label{sec:example}
Both Theorem \ref{thm:cost1} and Proposition \ref{prop:epsilonCost} describe entanglement cost regions where multiparty merging is achievable for any fixed $\epsilon > 0$. The proof of Theorem \ref{thm:cost1} is significantly more complicated than that of Proposition \ref{prop:epsilonCost}. To illustrate the benefits accruing from the additional effort, we modify our second example of Chapter~3 for distributed compression and show that Theorem \ref{thm:cost1} ``beats'' Proposition \ref{prop:epsilonCost}. That is, it shows the existence of protocols which allow the senders $C_1$ and $C_2$ to transfer their systems for free, a task impossible, for this example, for protocols of the kind described in the proof of Proposition \ref{prop:epsilonCost}. Our last example considers a family of states for which smoothing has little effect on the min-entropies appearing in eq.~(\ref{eq:entanglementcost}).

\subsection{Example I}
Recall the state $\psi^{C_1C_2C_3R}$ for the second example of Section \ref{sec:iid}:
\begin{equation*}
\ket{\psi}^{C_1C_2C_3R}:=\ket{\psi^{C_1C_2^1C_2^2C_3^1C_3^2R}} := \ket{\Psi_{-}}^{C_1C_2^1} \otimes \ket{\Psi_{-}}^{C_3^1R} \otimes \ket{\phi}^{C_2^2C_3^2},
\end{equation*}
where $\ket{\phi}^{C^2_2C_3^2}:=\sqrt{\lambda} \ket{00}^{C^2_2C^2_3} + \sqrt{1-\lambda}\ket{11}^{C^2_2C^2_3}$ is a pure bipartite entangled state with entropy of entanglement:
\begin{equation*}
 E(\phi) = S(C^2_2)_{\phi} = -\lambda \log{\lambda} - (1-\lambda)\log(1-\lambda) > 0.
\end{equation*}
 Let's replace the EPR pairs $\ket{\Psi_{-}}^{C_1C_2^1}$ and $\ket{\Psi_{-}}^{C_3^1R}$ by maximally entangled states of dimension $d$ and the state $\phi$ by a $d_1$-dimensional pure bipartite entangled state $\ket{\vartheta}$:
\begin{equation*}
\ket{\psi}^{C_1C_2C_3R}:=\ket{\psi^{C_1C_2^1C_2^2C_3^1C_3^2R}} := \ket{\Phi^{d}}^{C_1C_2^1} \otimes \ket{\Phi^d}^{C_3^1R} \otimes \ket{\vartheta}^{C_2^2C_3^2}.
\end{equation*}
where $\ket{\vartheta}^{C^2_2C_3^2}:=\sum^{d_1}_{i=1} \sqrt{\lambda_i} \ket{ii}^{C^2_2C^2_3}$, with $\lambda_i \geq \lambda_{i+1}$ and entropy of entanglement:
\begin{equation*}
 E(\vartheta) = S(C^2_2)_{\vartheta} = -\sum^{d_1}_{i=1}\lambda_i \log{\lambda_i} > 0.
\end{equation*}
\newcommand{\errort}{4\log(1/\epsilon) + 14}

Let $(E_1, E_2, E_3)$ be an entanglement cost-tuple satisfying the requirements of Theorem \ref{thm:cost1}:
\begin{equation} \label{eq:condDistCompress2}
  \begin{split}
   E_1 &\geq -H_{\min}(\psi^{C_1R}|\psi^R) +  \errort \\
   E_2 &\geq -H_{\min}(\psi^{C_2R}|\psi^R) + \errort  \\
   E_1 + E_2 &\geq -H_{\min}(\psi^{C_1C_2R}|\psi^R) + \errort\\
   E_1 + E_3 &\geq -H_{\min}(\psi^{C_1C_3R}|\psi^R) + \errort \\
   E_2 + E_3 &\geq -H_{\min}(\psi^{C_2C_3R}|\psi^R) + \errort \\
   E_1 + E_2 + E_3 &\geq -H_{\min}(\psi^{C_1C_2C_3R}|\psi^R) + \errort.
   \end{split}
\end{equation}
As in our previous analysis of such states in Section \ref{sec:iid}, we focus on the costs $E_1$ and $E_2$ sufficient for merging with fixed error $\epsilon$. By additivity of the min-entropy, we can simplify the min-entropies appearing in eq.~(\ref{eq:condDistCompress2}) for the costs $E_1$ and $E_2$:
\begin{equation*}
\begin{split}
 H_{\min}(\psi^{C_1R}|\psi^R)  &= H_{\min}(\psi^{C_1}) = \log(d) \\
 H_{\min}(\psi^{C_2R}|\psi^R) &= H_{\min}(\psi^{C_2}) = \log(d) -\log \lambda_1(\vartheta^{C_2^2}) \\
 H_{\min}(\psi^{C_1C_2R}|\psi^{R}) &= H_{\min}(\psi^{C_2^2}) = -\log \lambda_1(\vartheta^{C_2^2}).
 \end{split}
\end{equation*}
The constraints on the costs $(E_1, E_2)$ become
\begin{equation} \label{eq:condDistCompress3}
  \begin{split}
   E_1 &\geq -\log(d) +  \errort \\
   E_2 &\geq \log \lambda_1(\vartheta^{C_2^2}) -\log(d)  + \errort  \\
  E_1 + E_2 &\geq \log \lambda_1(\vartheta^{C_2^2}) + \errort.
  \end{split}
\end{equation}
For a fixed error $\epsilon$, the right hand side of the cost sum $E_1 + E_2$ will dominate the coefficients $\errort$ if we adjust the min-entropy of the reduced state $\vartheta^{C^2_2}$ to be an increasing function of $d$. For the purpose of showing the superiority of Theorem \ref{thm:cost1} over Proposition \ref{prop:epsilonCost}, set $\vartheta^{C^2_2C^3_2}$ to be the maximally entangled state of dimension $d^{\epsilon}$. This gives the following constraint on the cost sum $E_1 + E_2$:
\begin{equation*}
  E_1 + E_2 \geq -\epsilon \log(d)  + \errort,
\end{equation*}
which is negative for large $d$ and any fixed $\epsilon > 0$. Therefore, by boosting $E_3$ enough to satisfy the other constraints of eq.~(\ref{eq:condDistCompress2}), there exists a multiparty state merging protocol with error $\epsilon$ and entanglement cost-tuple $(E_1,E_2,E_3)$ such that both $E_1$ and $E_2$ are negative.

Turning to the entanglement costs provided by Proposition \ref{prop:epsilonCost}, let's attempt to achieve negative costs $E'_1$ and $E'_2$ for the senders $C_1$ and $C_2$. This obviously requires the sender $C_3$ to be the first to transfer his system. Thus, it restricts the possible ordering of the senders to either $\{C_3, C_2, C_1\}$ or $\{C_3, C_1, C_2\}$. Assuming $C_2$ is the next sender to transfer his system, the cost $E'_2$ must be at least
\newcommand{\errotA}{4\log(6/\epsilon)+ 2\log(13)}
\begin{equation*}
E'_2 \geq -H^{\delta}_{\min}(C_2|C_1R)_{\psi} + \errotA,
\end{equation*}
where $\delta := \epsilon^2/468$. By Lemma 5 of Renes et al. \cite{Renes02}, which bounds the smooth min-entropy by the conditional von Neumann entropy and a function of the smoothing parameter $\delta$, we have
\begin{equation*}
\begin{split}
H^{\delta}_{\min}(C_2|C_1R)_{\psi} &\leq S(C_2|C_1R)_{\psi} + 8\delta(\epsilon+1)\log(d) + 2h_2(2\delta) \\
&= -\log(d) + \epsilon \log(d) + 8\delta(\epsilon+1) \log(d) + 2h_2(2\delta) \\
\end{split}
\end{equation*}
where $h_2(x)$ is the binary entropy function $ h_2(x)= -x\log(x) -(1-x)\log(1-x)$. Hence, the entanglement cost $E'_2$ is at least
\begin{equation*}
\begin{split}
E'_2  &\geq  (1-\frac{8\epsilon^3}{468} - \frac{8\epsilon^2}{468} - \epsilon)\log(d)  - 2h_2(2\delta) + \errotA, \\
&\geq (1-\frac{4\epsilon^2}{117}-\epsilon)\log(d) + 4\log(6/\epsilon) + 5
\end{split}
\end{equation*}
which is positive for any $\epsilon \leq 0.9$.

If instead the sender $C_1$ is the next to transfer his system to the receiver, the cost $E'_1$ must be at least
\begin{equation*}
E'_1 \geq  -H^{\delta}_{\min}(C_1|C_2R)_{\psi}  + \errotA.
\end{equation*}
By strong subadditivity of the smooth min entropy, we have
\begin{equation*}
\begin{split}
  H^{\delta}_{\min}(C_1|C_2R)_{\psi} &= H^{\delta}_{\min}(C_1|C_2^1C_2^2R)_{\psi} \\
  &\leq H^{\delta}_{\min}(C_1|C^1_2)_{\psi} \\
  &=-H^{\delta}_{\max}(\tau_d^{C_1}),
\end{split}
\end{equation*}
where $\tau_d^{C_1}$ is the maximally mixed state of dimension $d$ and the last line is obtained by using the duality between the smooth min and max-entropies, eq.~(\ref{eq:duality}). Using Lemma \ref{lem:hmax-smoothing} with $k = \lceil d(1- 2\delta) \rceil$, we have
\begin{equation*}
 \begin{split}
 H^{\delta}_{\max}(\psi^{C_1}) &\geq 2\log \frac{k-1}{\sqrt{d}} \\
 & = 2\log (k-1) - \log(d) \\
 & \geq \log(d) - 2.
\end{split}
\end{equation*}
The last line follows from this very weak lower bound on $k-1$, which is nonetheless sufficient for our purposes:
\begin{equation*}
  \begin{split}
    k-1&= \lceil d(1 - 2\delta)\rceil-1  \\
  &= \lfloor d(1 - 2\delta) \rfloor \\
  &= \lfloor d(1 - \frac{\epsilon^2}{234}) \rfloor \\
    &\geq \frac{d}{2}
   \end{split}
\end{equation*}
for any $\epsilon \leq 2$.  Hence, the cost $E'_1$ for merging $C_1$ is bounded by
\begin{equation}\label{eq:weakd}
\begin{split}
E'_1 &\geq  -H^{\delta}_{\min}(C_1|C_2R)_{\psi}  + \errotA\\
&\geq H^{\delta}_{\max}(\tau_d^{C_1}) + \errotA \\
&\geq \log(d) + 4\log(6/\epsilon) + 5,
\end{split}
\end{equation}
which is positive for any $d$ and $\epsilon > 0$. Thus, Proposition \ref{prop:epsilonCost} provides no entanglement cost tuple $(E'_1, E'_2, E'_3)$ allowing negative costs for the senders $C_1$ and $C_2$. The interpretation from Section \ref{sec:iid} remains valid in the one-shot regime for the multiparty merging protocols provided by Theorem \ref{thm:cost1}: the excess of entanglement distributed between the receiver and the sender $C_3$ is used by the decoder to transfer the other shares $C_1$ and $C_2$ for free, with possible extra entanglement shared with the receiver.
\subsection{Example II}
Our second example for illustrating the benefits of Theorem \ref{thm:cost1} over Proposition \ref{prop:epsilonCost} considers the same kind of state $\psi^{C_1C_2C_3R}$, but with a less artificial state $\vartheta$ which does not depend on the ``merging'' error $\epsilon$:
\begin{equation*}
\ket{\vartheta^{C_2^2C_3^2}} := \frac{1}{\sqrt{H_d}} \sum^d_{j=1}\frac{1}{\sqrt{j}} \ket{jj}^{C^2_2C^2_3},
\end{equation*}
where $H_d = \sum_{j=1}^d 1/j$ is the $d$th harmonic number. These states are known as \textit{embezzled} states and were introduced in van Dam et al. \cite{Hayden04}. They are useful resources for channel simulation and other tasks~\cite{Harrow,Berta02,Reverse}.

First, to simplify calculations, we use the following bound on the $d$th harmonic number:
\[ \ln(d+1) \leq H_d \leq \ln(d) + 1.\]
Let's compute a bound on the von Neumman entropy of the reduced density operator~$\vartheta^{C_2^2}$:
\begin{equation}\label{eq:bounden}
\begin{split}
  S(C_2^2)_{\vartheta} &= -\sum^d_{j=1} \frac{1}{jH_d} \log \left (\frac{1}{jH_d} \right ) \\
  &= \frac{1}{H_d}\sum^d_{j=2} \frac{\log(j)}{j} + \log(H_d) \\
  & \leq \frac{1}{H_d}\int^d_{j=1} \frac{\log(x)}{x} dx  + \log(H_d) \\
  &= \frac{1}{H_d} \frac{\log(d)}{2} \ln(d) + \log(H_d) \\
  &\leq \frac{\log(d)\ln(d)}{2 \ln(d+1)} + \log \left (\frac{4\log d}{5} \right ) \\
  &\leq \frac{\log(d)}{2} + \log\log(d).
\end{split}
\end{equation}
The fifth line was obtained using the bound $\ln(d)+1 \leq \frac{4 \log(d)}{5}$, holding for sufficiently large $d$. The maximum eigenvalue for the reduced state $\vartheta^{C_2^2}$ is $1/H_d$, giving us the following constraint for the cost sum $E_1 + E_2$:
\begin{equation*}
E_1+E_2 \geq -\log(H_d) + \errort \\
\end{equation*}
which is satisfied by choosing $E_1$ and $E_2$ such that
\begin{equation*}
 \begin{split}
 E_1 &\geq -\log(d) +  4\log(1/\epsilon) + 14 \\
 E_2 &\geq -\log\log(d) -\log(d)  + 4\log(1/\epsilon) + 15  \\
 E_1+E_2 &\geq -\log\log(d) + 4\log(1/\epsilon) + 15. \\
\end{split}
\end{equation*}
Thus, although $\log \log (d)$ grows very slowly with $d$, we can nonetheless choose costs $E_1 < 0$ and $E_2 < 0$ for sufficiently large $d$ and any fixed value of $\epsilon > 0$.

Now, for the protocols of Proposition \ref{prop:epsilonCost}, with the ordering $\{C_3, C_2, C_1\}$ selected for the senders, we can substitute the entropy $S(C^2_2)_{\vartheta}$ by the bound of eq.~(\ref{eq:bounden}), and obtain a lower bound on the cost $E'_2$:
\begin{equation*}
\begin{split}
E'_2 &\geq  (1-16\delta)\log(d) -S(C^2_2)_{\vartheta}  - 2h_2(2\delta) + \errotA, \\
& \geq (1/2 -\frac{4\epsilon^2}{117})\log(d) - \log \log(d) + 4\log(6/\epsilon) + 5.
\end{split}
\end{equation*}
which is positive for any $\epsilon \leq 2$. For the other possible ordering $\{C_3, C_1, C_2\}$ of the senders, the $E'_1$ cost is independent of the entropy of the state $\vartheta^{C_2^2}$ (see eq.~(\ref{eq:weakd})). So once again, we have an example where the protocols of Theorem \ref{thm:cost1} allow the senders $C_1$ and $C_2$ the possibility to send their systems for free, although the dimension $d$ required for achieving this is unrealistic for this example. The protocols provided by Proposition \ref{prop:epsilonCost}, on the other hand, allow no such feat for any value of~$d$.

\subsection{Example III}
Suppose two senders $C_1$ and $C_2$ have systems $C_1$ and $C_2$ of dimensions $d$ in a state of the form
\begin{equation*}
|\psi\rangle^{C_1 C_2 R} := \frac{1}{\sqrt{H_d}} \sum_{j=1}^d \frac{1}{\sqrt{j}} |j\rangle^{C_1} |\psi_j\rangle^{C_2} |j\rangle^{R}.
\end{equation*}
These states are close relatives of the embezzling states introduced in the previous example. They make interesting examples because they have sufficient variation in their Schmidt coefficients that the i.i.d. state merging rates of Theorem \ref{thm:statemerging} are not achievable in the one-shot regime.
Simple teleportation of the two systems $C_1$ and $C_2$ to the receiver would require $\log{d}$ EPR pairs by sender \cite{GenTel}. Our protocols will yield nontrivial one-shot rates that are significantly better than teleportation.

We assume that $| \langle \psi_i | \psi_j \rangle | \leq \alpha$ for $i \neq j$ and try to express the lower bounds of eqs.~(\ref{eq:cost1}) and (\ref{eq:entanglementcost}) in terms of $\alpha$.

\subsubsection{Protocols from Theorem \ref{thm:cost1}}
Let $(E_1,E_2)$ be a pair of entanglement costs achievable according to Theorem~\ref{thm:cost1}. The only constraints on the costs (aside from needing to be the logarithms of integers) are
\begin{eqnarray}
E_1 &\geq& -H_{\min}(\psi^{C_1 R} | \psi^R ) + 4 \log ( 1 / \epsilon ) + 12 \label{eqn:ach-e1} \\
E_2 &\geq&  -H_{\min}(\psi^{C_2 R} | \psi^R ) + 4 \log ( 1 / \epsilon ) + 12 \label{eqn:ach-e2} \\
E_1 + E_2  &\geq&  -H_\min(\psi^{C_1 C_2 R} | \psi^R ) + 4 \log ( 1 / \epsilon ) + 12. \label{eqn:ach-e1+e2}
\end{eqnarray}
To begin, we will find a sufficient condition for the $E_1$ constraint to be satisfied,
so we need to evaluate $H_\min(\psi^{C_1R} | \psi^R )$. Let $\lambda_{\min}$ be the smallest real number such that $\lambda_{\min}( I^{C_1} \otimes \psi^R ) - \psi^{C_1 R} \geq 0$. Expanding the operators, the condition is the same as
\begin{equation} \label{eqn:cond-ent-operator}
\sum_{ij} \frac{\lambda_{\min} - \delta_{ij}}{j} |ij\rangle\langle ij|^{C_1 R}
	-  \sum_i \sum_{ j \neq i} \frac{1}{\sqrt{ij}} |ii\rangle\langle jj|^{C_1 R}
			\langle \psi_j | \psi_i \rangle \geq 0,
\end{equation}
where $\delta_{ij}$ is the Kronecker delta function. Let $\lambda > 0$ be any real number. By the Gershgorin Circle Theorem \cite{Horn,wiki}, the operator $\lambda( I^{C_1} \otimes \psi^R ) - \psi^{C_1 R}$ is positive if each diagonal entry dominates the sum of the absolute values of the off-diagonal entries in the corresponding row. That condition reduces to
\begin{equation}
\frac{\lambda - 1}{i}
	\geq
	\sum_{j \neq i} \frac{1}{\sqrt{ij}} | \langle \psi_j | \psi_i \rangle |
\end{equation}
holding for all $i$, which is true provided
$\lambda - 1 \geq \alpha \sum_{j=1}^d \sqrt{d/j}$.
But
\begin{equation}
\sum_{j=1}^d \frac{1}{\sqrt{j}}
	\leq
	\int_0^{d} \frac{1}{\sqrt{x}} \, dx
	= 2 \sqrt{d}.
\end{equation}
Therefore, if $\lambda \geq 2 \alpha d + 1$, the operator $\lambda( I^{C_1} \otimes \psi^R ) - \psi^{C_1 R}$ is positive. Hence, we have
\begin{equation*}
-H_{\min}(\psi^{C_1R}|\psi^R) = \log{\lambda_{\min}} \leq \log(2\alpha d + 1) \leq \log(\alpha d) + 2,
\end{equation*}
provided $\alpha \geq \frac{1}{2d}$.
The lower bound of eq.~(\ref{eqn:ach-e1}) is satisfied if we set
\begin{equation*}
E_1 \geq \log( \alpha d ) + 4\log(1/\epsilon) + 14.
\end{equation*}
The interpretation is that if the states $\{ |\psi_j\rangle \}$ are indistinguishable, then $C_1$ holds the whole purification of $R$ and must therefore be responsible for the full cost of merging. As the states $\{ | \psi_j \rangle \}$ become more distinguishable, the purification of $R$ becomes shared between the systems $C_1$ and $C_2$, allowing the entanglement cost to be more distributed between the two senders. Indeed, if $\alpha=O(1/d)$, the lower bound on $E_1$ becomes a constant, independent of the size of the input state $|\psi\rangle^{C_1C_2R}$.

Moving on to the $E_2$ constraint, eq.~(\ref{eqn:ach-e2}), the state $\psi^{C_2R}$ is a classical quantum state with classical system $R$:
\begin{equation*}
 \psi^{C_2R} = \sum^d_j \frac{1}{j H_d } \braket{j}^R \otimes \braket{\psi_j}^{C_2}.
\end{equation*}
Since the state $\braket{\psi_j}^{C_2}$ is pure for all $j$, conditioning on the classical system reduces the min-entropy to zero. A formal proof of this fact is obtained by applying Lemma 3.1.8 of Renner \cite{Renner02}. Thus, the lower bound of eq.~(\ref{eqn:ach-e2}) is satisfied if we set
\begin{equation*}
E_2 \geq 4\log(1/\epsilon) + 12.
\end{equation*}
For the sum rate $E_1 + E_2$, it is necessary to evaluate $H_\min(\psi^{C_1 C_2 R} | \psi^R )$. Since the state $\psi^{C_1C_2R}$ is pure, we can apply Proposition 3.11 in Berta \cite{Berta} and obtain
\begin{equation*}
 -H_{\min}(\psi^{C_1C_2R} | \psi^{R}) = H_0(\psi^{C_1C_2}) := \log \mathrm{rank} (\psi^{C_1C_2}) = \log d.
\end{equation*}
 Hence, by Theorem \ref{thm:cost1}, there exists a multiparty state merging protocol for the state $\psi^{C_1C_2R}$, with error $\epsilon$ and entanglement cost pair $(E_1,E_2)$ satisfying
\begin{eqnarray}
E_1 &\geq&  \log( \alpha d ) + 4\log(1/\epsilon) + 14 \\
E_2 &\geq& 4\log(1/\epsilon) + 12 \\
E_1 + E_2 &\geq& \log(d) + 4 \log(1/\epsilon) + 12. \label{eqn:cor2}
\end{eqnarray}
The total entanglement cost $E_1 + E_2$ must be at least $\log d$ plus terms independent of the dimensions of the systems, and we can distribute that cost between the senders $C_1$ and $C_2$. The lower bound on $E_2$ varies independently with $d$ and can be regarded as a small ``overhead'' for the protocol. There \emph{is} a minimal $d$-dependent cost for $E_1$, however, which encodes the fact that if the sender $C_2$ does not carry enough of the purification of $R$ by virtue of the nonorthogonality of the $\{|\psi_j\rangle\}$, then more of the burden will fall to the sender $C_1$.

\subsubsection{Protocols from Proposition \ref{prop:epsilonCost}}
Now let us consider the lower bound of eq.~(\ref{eq:entanglementcost}) (Proposition \ref{prop:epsilonCost}). For fixed $\epsilon$, the proposition provides two cost pairs, plus others that are simply degraded versions of those two arising from the wasteful consumption of unnecessary entanglement. Proposition \ref{prop:epsilonCost} does not permit interpolation between the two points, as compared to Theorem \ref{thm:cost1}.
It might be the case, however, that Proposition \ref{prop:epsilonCost}'s freedom to smooth the entropy and vary the operator being conditioned upon could result in those two cost pairs being much better than any of those provided by Theorem \ref{thm:cost1}. On the contrary, for the states of the example, the improvement achieved with the extra freedom is minimal.

Let $(E_1',E_2')$ be a cost pair achievable by Proposition \ref{prop:epsilonCost}. For the purposes of illustration, consider the point with the smallest possible value of $E_2'$. Letting $\delta = \epsilon^2/208$, that point will satisfy
\begin{eqnarray}
E_1' &\geq&  -H_\min^\delta (\psi^{C_1 C_2 R} | C_2 R)
		+ 4 \log\left( 4 / \epsilon \right) + 2\log(13)  \\
E_2' &\geq&  -H_\min^\delta (\psi^{C_2 R} | R)
		+   4 \log\left( 4 / \epsilon \right) + 2\log(13).
\end{eqnarray}
First, to bound the entanglement cost $E_2$, we use Lemma 5 of Renes \textit{et al.} \cite{Renes02}:
\begin{equation*}
 \begin{split}
-H_{\min}^{\delta}(\psi^{C_2R}|R) &\geq -S(C_2|R)_{\psi} - 8\delta \log(d) - 2h_2(2\delta) \\
&=-8\delta \log(d) - 2h_2(2\delta),
\end{split}
\end{equation*}
where $S(C_2|R)_{\psi}=0$ for the state $\psi^{C_2R}$. Thus, we have
\begin{equation*}
E_2' \geq  -8\delta \log(d) -2h_2(2\delta) +   4 \log\left( 4 / \epsilon \right) + 2\log(13).
\end{equation*}
Before introducing the extra complication of smoothing, consider first $H_\min(\psi^{C_1 C_2 R} | C_2 R )$. By duality of the smooth min and max-entropies, eq.~(\ref{eq:duality}), we have
\begin{eqnarray*}
   -H_{\min}(\psi^{C_1C_2R}|C_2R) &=& H_{\max}(\psi^{C_1}) \\
    &=& 2\log \sum^d_{j=1} \frac{1}{\sqrt{j H_d}} \\
	&\geq& \log \left( \frac{1}{\sqrt{H_d}} \int_1^d \frac{1}{\sqrt{x}} \, dx \right)^2 \\
	&\geq& \log \frac{4d}{H_d} \left( 1 - O\left(\frac{1}{\sqrt{d}} \right) \right)  \\
    &\geq& \log \left (\frac{5d}{\log d}\right ) \\
	&=& \log d - \log \log d+ \log 5,
\end{eqnarray*}
where $4/(\ln d + 1) \geq 5.7 / \log d$ for sufficiently large $d$.
Therefore, ignoring smoothing, the total entanglement cost for Proposition \ref{prop:epsilonCost} satisfies
\begin{equation}
E_1' + E_2'
	\geq   (1-\frac{\epsilon^2}{26})\log(d) - \log \log(d) + 14 + 8 \log\left( \frac{4}{\epsilon} \right).
\end{equation}
for sufficiently large $d$, which has worse constants than the sum cost (\ref{eqn:cor2}) for Theorem \ref{thm:cost1}.
Now let us introduce some smoothing. By duality of the min- and max- entropies,
\begin{equation}
-H^\delta_\min(\psi^{C_1 C_2 R} | R C_2 ) = H^\delta_\max(\psi^{C_1}).
\end{equation}
Lemma \ref{lem:hmax-smoothing} of Appendix B gives that
\begin{equation}
H^\delta_\max(\psi^{C_1})
	\geq 2 \log \min \left\{ \sum_{j=1}^{k-1} \frac{1}{\sqrt{j \cdot H_d}}: k \mbox{ such that }
		\sum_{j=k+1}^d \frac{1}{j \cdot H_d} \leq 2\delta \right\}.
\end{equation}
Getting a lower bound on this expression requires finding large $k$ that nonetheless fail to satisfy the tail condition. That restriction on $k$ is equivalent to $1-H_k/H_d \leq 2\delta$, which will not be met by any $k$ small enough to obey
\begin{equation} \label{eqn:kbound}
k \leq (d+1)^{1-2\delta} / e
\end{equation}
for sufficiently large $d$.
\begin{eqnarray}
2 \log \sum_{j=1}^{k-1} \frac{1}{\sqrt{j\cdot H_d}}
	&\geq& \log \left( \frac{1}{\sqrt{H_d}} \int_1^k \frac{1}{\sqrt{x}} \, dx \right)^2 \\
	&\geq& \log \frac{4k}{H_d} \left( 1 - O\left(\frac{1}{\sqrt{k}} \right) \right)  \\
	&\geq& \log k - \log \log d+ \log 5
\end{eqnarray}
for sufficiently large $k$.
Substituting in the largest possible $k$ consistent with eq.~(\ref{eqn:kbound}) and $\delta = \epsilon^2/208$ gives
\begin{equation}
E_1' + E_2' \geq \left( 1 - \frac{\epsilon^2}{26} - \frac{\epsilon^2}{104} \right) \log(d) -\log \log d + 12 + 8 \log\left(\frac{4}{\epsilon}\right),
\end{equation}
for sufficiently large $d$. The additional savings from smoothing are only about $\frac{\epsilon^2 \log(d)}{104}$ ebits, which is insignificant for small $\epsilon$. These tiny savings also come at the expense of being able to interpolate between achievable costs. To be fair, these states were chosen specifically because they are known to maintain their essential character even after smoothing, as was observed in \cite{Hayden05}. The freedom to smooth is certainly more beneficial for some other classes of states, most notably i.i.d. states. Indeed, since $S(C_1 C_2)_{\psi} = (\log d)/2 + O( \log \log d)$, merging many copies of $|\psi\rangle^{C_1 C_2 R}$ can be done at a rate roughly half the cost required for one-shot merging.

\chapter{Assisted Entanglement Distillation}
\section{Introduction}
 The protocols discussed in the previous chapters are based on a random coding strategy: the senders apply randomly chosen unitaries on their systems and perform projective measurements in a fixed basis. The decoder, conditioned on the measurement outcomes, applies an isometry on his systems and recovers the original state with arbitrarily good fidelity. A similar approach is used in \cite{merge} for solving the multipartite entanglement of assistance problem when the state shared between the parties is pure. Recall that the formula is given by the min-cut entanglement of the state $\psi^{C_M AB}$:
\begin{equation}
\label{eq:mincutEofA}
D_A^{\infty}(\psi^{C_MAB}) = \min_{{\cal T}} \{ S(A{\cal T})_{\psi}\},
\end{equation}
where the minimum is taken over all bipartite cuts $\cal{T}$. (Recall that a bipartite cut consists of a partition of the helpers ${C_1,\ldots, C_{m}}$
into a set $\cal{T}$ and its complement $\overline{\cal T }=\{C_1,\ldots, C_{m}\} \backslash \cal{T}$.)

Consider a one-dimensional chain with $m$ repeater stations $C_1, C_2, \ldots, C_m$ separating the two endpoints (Alice and Bob). If many copies of the state $\ket{\psi}=\sqrt{\lambda_1} \ket{00} + \sqrt{\lambda_2}\ket{11}$ are prepared and distributed across the network, so that the global state of the network is given by
 \begin{equation*}
  \begin{split}
  (\Psi^{AC_1C_2\ldots C_m B})^{\otimes n} &:= (\psi_{AC_1^1} \otimes \psi_{C_1^2C_2^1} \otimes \ldots \otimes \psi_{C_m^2 B})^{\otimes n} \\
  &=\psi^{\otimes n}_{AC_1^1} \otimes \psi^{\otimes n}_{C_1^2C_2^1} \otimes \ldots \otimes \psi^{\otimes n}_{C_m^2 B}, \\
  \end{split}
 \end{equation*}
then the previous formula applied to the state $\Psi^{AC_1C_2\ldots C_mB}$ reduces to the entropy of entanglement $S(A)_{\psi}$ of the state $\ket{\psi}$. If the fiber optic transmitting the quantum information is perfect up to distances of roughly 100 kilometers, then we can establish close to $nS(A)_{\psi}$ ebits between Alice and Bob no matter how far they are to each other by introducing repeater stations at approximately every 100 kilometers. Of course, as the number of repeater stations increases between Alice and Bob, more copies of the state $\ket{\psi}$ must be distributed between the nodes of the network for the assisted distillation protocol to continue producing high quality entanglement at the rate $S(A)_{\psi}$.

As discussed in the introductory chapter, an implementation of the previous strategy will have to deal with the accumulation of errors during the various phases of preparation, distribution, storage and local operations of the quantum information. In this chapter, we continue the theoretical analysis of this problem by extending the models previously studied in \cite{entangleDistr,cirac} to allow for an arbitrary mixed state between adjacent nodes. This is an initial step towards handling more complex and realistic situations. First, we will consider a network consisting of two receiving nodes (Alice and Bob), separated by a repeater node (Charlie), whose global state is a mixed state $\psi^{ABC}$. This is a realistic assumption, as we recall imperfections in local operations \cite{repeaters} and decoherence in the quantum memories will most likely introduce noise in the stored qubits. We study the optimal distillable rate achievable for Alice and Bob when assistance from Charlie is available. This problem reduces to the two-way distillable entanglement for states in a product form $\psi^{C} \otimes \psi^{AB}$. There is currently no simple formula for computing the two-way distillable entanglement of a bipartite state $\psi^{AB}$, which has been studied extensively by Bennett et al. and others in \cite{Bennett,PurNoisy,interpolation,adaptive}. We do not attempt to solve this problem here, and turn our attention instead to good computable lower bounds for assisted distillation of mixed states. We provide a bound which exceeds the hashing inequality for states $\psi^{ABC}$ which do not saturate the strong subadditivity of the von Neumann entropy and allow the recovery of the $C$ system if Alice and Bob can perform joint operations on their systems.

\section{Assisted distillation for mixed states}

 \subsection{The task}
In this section, we extend the entanglement of assistance to the case of a general mixed state $\psi^{ABC}$: a measurement of Charlie's system followed by an entanglement distillation protocol between Alice and Bob. The problem is illustrated in Figure~\ref{fig:b1s}.

\begin{problem}[\textbf{Broadcast, Assisted Distillation}] \label{pbl:oneway}
Given many copies of a tripartite \textit{mixed} state $\psi^{ABC}$ shared between two recipients (Alice and Bob) and a helper (Charlie), find the optimal distillable rate between Alice and Bob with the help of Charlie if no feedback communication is allowed: Charlie performs a POVM and broadcasts the measurement outcome to Alice and Bob. The optimal rate is denoted by $D_A^{\infty}(\psi^{ABC})$. It is the asymptotic entanglement of assistance.
\end{problem}
We call a protocol which satisfies the constraint of Problem \ref{pbl:oneway} a \emph{broadcast assisted distillation protocol}. More
formally, it consists of
\newcommand{\LOCC}{{\cal V}_x}
\begin{packed_enum}
 \item A POVM $E=(E_x)_{x=1}^X$ for Charlie. Without loss of generality, we can assume that the operators $E_x$ are all of rank one.
 \item For each $x$, an LOCC operation $\LOCC: A^nB^n \rightarrow A_1B_1$, where $A_1$ and $B_1$ are subspaces of $A^n$ and $B^n$ of equal dimensions, implemented by Alice and Bob.
\end{packed_enum}
\begin{figure}
\begin{center}
\includegraphics{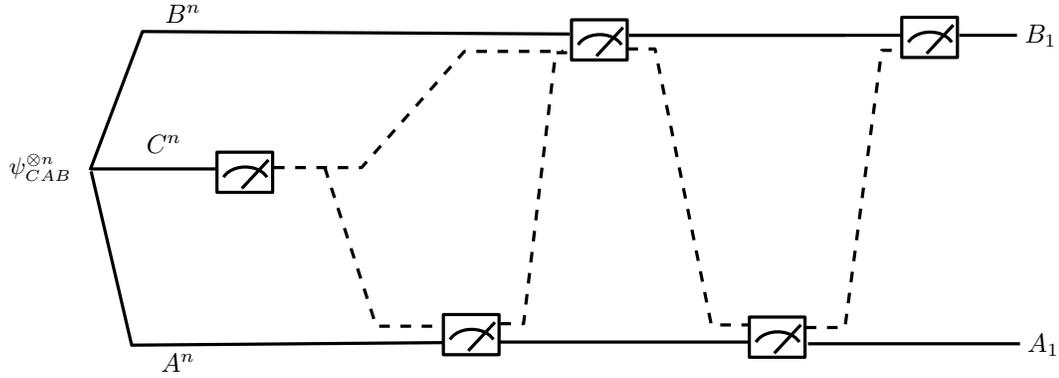}
\end{center}
\caption{Quantum circuit representing a broadcast assisted entanglement distillation protocol. Solid lines indicate quantum information and dashed lines classical information. Charlie first performs a measurement, sending copies of the classical outcome to Alice and Bob. Alice and Bob then implement an LOCC operation, conditioned on that classical outcome. }  \label{fig:b1s}
\end{figure}
We refer to a broadcast assisted protocol as an $(n,\epsilon)$-protocol if it acts on
$n$ copies of the state $\psi^{ABC}$ and produces a maximally
entangled state of dimension $M_n:=d_{A_1}$ \begin{equation*}
\ket{\Phi^{M_n}} = \frac{1}{\sqrt{{M_n}}}
\sum_{m=1}^{M_n} \ket{m}^{A_1} \otimes\ket{m}^{B_1} \end{equation*} up to fidelity $1-
\epsilon$: \begin{equation*} F^2 \biggr (\Phi^{M_n}, \sum_{x=1}^X p(x) \LOCC(\psi_x^{A^nB^n}) \biggl )
\geq 1-\epsilon, \end{equation*} where
\begin{equation*}
\psi_x^{A^nB^n} = \frac{1}{\Tr_{C^n}[E_x(\psi^C)^{\otimes
n}]}\Tr_{C^n} \biggr [ (E_x \otimes I^{AB}) (\psi^{ABC})^{\otimes
n}  \biggl ].
\end{equation*}

A real number $R \geq 0$ is said to be an achievable rate if there exists, for every $n$ sufficiently large, an
$(n,\epsilon)$-protocol with $\epsilon \rightarrow 0$ and $\frac{1}{n} \log{M_n} \rightarrow R$ as $n \rightarrow
\infty$. Lastly, we have
\begin{equation*} D_A^{\infty}(\psi^{ABC}) := \sup \{R: R \text{ is achievable} \}.
\end{equation*}

The restriction to POVMs with rank one operators in the preceding definition can be justified as follows: any POVM $F$ containing positive operators with rank higher than one that Charlie would wish to perform can be simulated by a POVM $E$ with rank one operators on Charlie's system followed by some processing by Alice and Bob. More precisely, suppose Charlie wants to perform a POVM $F=\{F_x\}$ on his state with some operators having rank greater than one. Consider the spectral decomposition of each operator: \[ F_x = \sum_i \lambda^x_i \braket{\alpha^x_i}, \] where $\{\ket{\alpha^x_i}\}$ are eigenvectors of $F_x$ with eigenvalues $\{\lambda^x_i\}$. Then $E = \{\lambda^x_i\braket{\alpha^x_i}\}_{x,i}$ is a POVM with rank one operators. Instead of performing the POVM $F$, Charlie does a measurement corresponding to the POVM $E$. After Alice and Bob receive the measurement outcome, the state is given by \[ \psi^{AA_1A_2BB_1B_2} = \sum_{x,i} q_{x,i} \psi^{AB}_{x,i} \otimes \braket{xx}^{A_1B_1} \otimes \braket{ii}^{A_2B_2}. \] To simulate $F$ being performed by Charlie, Alice and Bob can trace out the $A_2$ and $B_2$ systems. The state becomes
\[ \psi^{AA_1BB_1} = \sum_{x} p_x \psi^{AB}_{x} \otimes \braket{xx}^{A_1B_1}, \] with $\psi^{AB}_x := \smfrac{1}{p_x}\sum_i q_{x,i} \psi^{AB}_{x,i}$ and $p_x = \sum_i q_{x,i}$. Observe that this preprocessing can be embedded within the LOCC operation $\LOCC$. Hence, there is no loss of generality in assuming POVMs with rank one operators in step 1 of the protocol.

For pure states, $D^{\infty}_A(\psi)$ reduces to the asymptotic entanglement of assistance $E^{\infty}_A(\psi)$. For product states of the form $\psi^{C} \otimes \psi^{AB}$, $D^{\infty}_A(\psi)$ is equivalent to the two-way distillable entanglement $D(\psi^{AB})$. A formula is known for the two-way distillable entanglement (see Theorem 15 in Devetak and Winter \cite{DW}), but its calculation is intractable for most states. We will instead use the hashing bound to the one-way distillable entanglement $D_{\rightarrow}(\psi^{AB})$ \cite{DW}, which is much easier to evaluate. We remind the reader of the result for convenience:

\begin{lemma}[Hashing inequality \cite{Bennett,DW}]
\label{hashinginequality} Let $\psi^{AB}$ be an arbitrary bipartite mixed state. Then,
\begin{equation}\label{eq:hashing}
D_{\rightarrow}(\psi^{AB}) \geq S(B)_{\psi} - S(AB)_{\psi} =: I(A\rangle B)_{\psi}.
\end{equation}
\end{lemma}

\subsection{Entanglement of assistance} \label{sec:EoA}

As mentioned before, it was shown in \cite{SVW} that for pure states, the operationally defined quantity $E_A^{\infty}$ corresponds to the regularization of the one-shot entanglement of assistance $E_A$. In a similar fashion, we define the
\emph{one-shot entanglement of assistance} $D_A(\psi^{ABC})$ of a tripartite mixed state $\psi^{ABC}$ and show
that its regularization is equal to $D^{\infty}_A(\psi^{ABC})$. We then look at some of the properties of $D_A(\psi^{ABC})$.

\begin{definition}  \label{def:A1D}
For an arbitrary state $\psi^{ABC}$, define
\begin{equation}
 \begin{split}\label{eq:A1D}
  D_A(\psi^{ABC}) &:= \sup_{\substack{E=\{E_x\}}} \left\{ \sum_x p_x D(\psi_x^{AB}) \bigg | \psi_x^{AB} = \frac{1}{p_x} \mathrm{Tr}_C[(E_x \otimes I_{AB})\psi^{ABC}] \right\},
 \end{split}
\end{equation}
where $p_x = \mathrm{Tr}[E_x \psi^C]$ and the supremum is taken over all POVMs $E=\{E_x\}$ with rank one operators on Charlie's system $C$.
\end{definition}

The quantity $D_A(\psi^{ABC})$ can also be characterized using a maximization over all pure state decompositions $\{p_i,
\psi_i^{ABR}\}$ of the purified state $\psi^{ABCR}$:
\begin{proposition}\label{prop:equiv}
Let $\psi^{ABC}$ be an arbitrary state, with purification $\psi^{ABCR}$, then
  \begin{equation}
   \label{eq:equiv}
     D_A(\psi^{ABC}) = \sup_{ \{p_i, \psi_i^{ABR}\}}   \sum_i p_i D(\psi_i^{AB}),
 \end{equation}
where the supremum is taken over all ensembles of pure states
$\{p_i, \psi_i^{ABR}\}$ satisfying $\sum_i p_i \psi_i^{ABR} =
\mathrm{Tr}_C \psi^{ABCR}$.
\end{proposition}
\begin{proof} Any rank one POVM on $C$ induces an ensemble of pure states on $ABR$ with average state $\psi^{ABR}$ and for every such ensemble there exists a corresponding POVM~\cite{Hughston}. Applying this observation to the definition of the one-shot entanglement of assistance yields the result. \end{proof}

We can interpret eq.~(\ref{eq:equiv}) as follows: by varying a
POVM on his state, Charlie can collapse the purified state
$\psi^{ABCR}$ into any pure state ensemble decomposition
$\{p_i,\psi_i^{ABR}\}$ for the $A$,$B$, and $R$ systems. Since we
don't have access to the purifying system $R$, the quantity $D_A(\psi^{ABC})$
maximizes the average amount of distillable entanglement
between Alice and Bob. The next result shows that the regularized
version of $D_A$ is in fact equal to the asymptotic entanglement of assistance $D_A^{\infty}(\psi^{ABC})$.

\subsection{Basic properties} \label{subsec:properties}

\begin{theorem}[Equivalence]
\label{thm:equivalence} Let $\psi^{ABC}$ be an arbitrary tripartite state. Then the following equality holds:
\begin{equation}\label{eq:ensEq}
D_A^{\infty}(\psi^{ABC}) = \lim_{n \rightarrow \infty} \frac{1}{n} D_A\biggl ((\psi^{ABC})^{\otimes n} \biggr ).
\end{equation}
\end{theorem}
\begin{proof} We demonstrate the ``$\leq$'' first. Consider any achievable rate $R$ for a broadcast assisted protocol. By definition, there exists, for every $n$ sufficiently large, an $(n, \epsilon)$-protocol with $\epsilon \rightarrow 0$ and $\smfrac{1}{n}\log(M_n)\rightarrow R$ as $n \rightarrow \infty$. For a protocol working on $n$ copies of the state $\psi^{ABC}$, denote Charlie's POVM by $E=(E_x)_{x=1}^X$, and for each outcome $x$, the LOCC operation implemented by Alice and Bob by $\LOCC$. Write  \begin{equation*} \begin{split}
\Omega^{A_1B_1} &:= \sum_{x=1}^X p_x \LOCC (\psi_x^{A^nB^n}) \\
&= \sum_{x=1}^X p_x \Omega_x^{A_1B_1}, \\
\end{split}
\end{equation*} where $p_x =\mathrm{Tr}[E_x (\psi^{C})^{\otimes n}]$ and $\psi_x^{A^nB^n} =\frac{1}{p_x} \Tr_{C^n}[ (E_x \otimes I_{AB})\psi_{ABC}^{\otimes n}]$. The state $\Omega_x^{A_1B_1}$ is the output state of $\LOCC(\psi_x^{A^nB^n})$. By hypothesis, we have \beu F^2(\Phi^{M_n}, \Omega^{A_1B_1}) \geq 1-\epsilon, \eeu which, shifting to the trace norm, implies \be \label{eq:dist1} \biggr \|
\Phi^{M_n} - \Omega^{A_1B_1} \biggl \|_1 \leq 2\sqrt{\epsilon} := \epsilon'. \ee
The trace distance is non-increasing under the partial trace, and so tracing out the $A_1$ system, we have \be \label{eq:dist2} \biggr \|
\Phi^{M_n}_{B_1} - \Omega^{B_1} \biggl \|_1 \leq \epsilon', \ee
where $\Phi^{M_n}_{B_1} = \frac{1}{M_n}\sum^{M_n}_{m=1} \braket{m}^{B_1}$.

We can apply the Fannes inequality (Lemma \ref{lem:Fannes}) on eqs. (\ref{eq:dist1}) and (\ref{eq:dist2}) to get a bound on $\log(M_n)$ in terms of the coherent information of the state $\Omega_{A_1B_1}$:
  \begin{equation*}
    \begin{split}
      \log{M_n} &\leq S(B_1)_{\Omega} - S(A_1B_1)_{\Omega} + 3\log(M_n)\eta(\epsilon')\\
      &= I(A_1 \rangle B_1)_{\Omega} +3\log(M_n)\eta(\epsilon'), \\
    \end{split}
  \end{equation*}
where $\eta(\epsilon')$ is a function which converges to zero for sufficiently small $\epsilon'$. (The definition of $\eta(\epsilon')$ can be found in Lemma $\ref{lem:Fannes}$.)
Using the convexity of the coherent information~\cite{bns}, the hashing inequality, and the definitions of $D$ and $D_A$, we get the following series of inequalities:
\begin{equation*}
    \begin{split}
      \log{M_n}&\leq I(A_1\rangle B_1)_{\Omega} + 3\log(M_n)\eta(\epsilon')\\
      &\leq \sum_x p_x I(A_1\rangle B_1)_{\Omega_x} + 3\log(M_n)\eta(\epsilon') \\
      &\leq \sum_x p_x D(\Omega^{A_1B_1}_x) + 3\log(M_n)\eta(\epsilon') \\
      &\leq \sum_x p_x D(\psi_x^{A^nB^n}) +3\log(M_n)\eta(\epsilon')  \\
      &\leq D_A((\psi^{ABC})^{\otimes n}) + 3n\log(d_A)\eta(\epsilon').
    \end{split}
  \end{equation*}
Since $\epsilon \rightarrow 0$ and $\smfrac{1}{n}\log(M_n) \rightarrow R$ as $n\rightarrow \infty$, the achievable rate $R$ is at most $\lim_{n \rightarrow \infty} \smfrac{1}{n}D_A((\psi^{ABC})^{\otimes n})$, which proves the ``$\leq$'' part since $R$ was arbitrarily chosen.

To show the ``$\geq$'' part, suppose Charlie performs any POVM $E=(E_x)$ on one copy of the state $\psi^{ABC}$ and broadcasts the
result to Alice and Bob. They now share the state \beu \tilde{\psi}^{A'ABB'} = \sum_x p_x \braket{x}^{A'} \otimes
\psi_x^{AB} \otimes \braket{x}^{B'}. \eeu Since Alice and Bob know the outcome of Charlie's POVM, the distillable entanglement of $\tilde{\psi}^{A'ABB'}$ is at least \beu D(\tilde{\psi}^{A'ABB'}) \geq \sum_x p_x
D(\psi_x^{AB}). \eeu To see this, consider many copies of ${\psi}^{A'ABB'}$ and let Alice and Bob perform projective measurements on the systems $A'$ and $B'$ for each copy of the state. Group the outcome states into blocks, where each block corresponds to a specific measurement outcome. For each of these blocks, there exist LOCC operations $\LOCC$ which will distill arbitrarily close to the rate $D(\psi^{AB}_x)$. Thus, there is a protocol achieving the rate $\sum_x p_x D(\psi_x^{AB})$, which proves the ``$\geq$'' part.
\end{proof}

Finding a formula for the one-shot quantity $D_A(\psi^{ABC})$ appears to be a difficult problem, and so we look for upper bounds which are attained for a subset of all possible states. For the remainder of this section, we look at two upper bounds and give examples of states attaining them. 

\begin{proposition}
\label{thm:upperbound} Let $\psi^{ABC}$ be an arbitrary tripartite state. We have the following upper bound for
$D_A(\psi^{ABC})$:
\begin{equation*}
D_A(\psi^{ABC}) \leq \inf_{\cal E} \sum_i p_i E_A (\psi^{ABC}_i) ,
\end{equation*}
where the infimum is taken over all ensembles of pure states
$\{p_i,\psi^{ABC}_i\}$ such that $\psi^{ABC} = \sum_i p_i
\psi_i^{ABC} $.
\end{proposition}

%

\begin{proof} Let $\psi^{ABC} = \sum_i p_i \psi_i^{ABC}$, where the states $\psi_i^{ABC}$ are pure. Consider the following classical-quantum
state $\phi^{ABCX} = \sum_i p_i \psi_i^{ABC} \otimes \braket{i}^X$. If Charlie is in possession of the $X$
system, then
\beu D_A (\psi^{ABC}) = D_A(\sum_i p_i \psi_i^{ABC}) \leq D_A(\phi^{ABCX}) \eeu
by the definition of $D_A$. By the convexity of $D_A(\phi^{ABCX})$ on the ensemble $\{p_i, \psi^{ABC}_i \otimes \braket{i}^X \}$ (see Proposition \ref{lemma:convex}) and the fact that $D_A(\psi_i^{ABC} \otimes \braket{i}^X) = E_A(\psi^{ABC}_i)$, we have
\begin{equation}
D_A(\psi^{ABC}) \leq D_A(\phi^{ABCX}) \leq \sum_i p_i E_A(\psi_i^{ABC}).
\end{equation}
Since this holds for any pure state ensemble $\{p_i, \psi_i^{ABC}\}$, we arrive at the statement of the proposition.
\end{proof}


With this result in hand, we now exhibit a set of states for which we can compute the value of $D_A$ exactly.
\begin{example}
Consider the following family of classical-quantum states, with classical system $C$:
\[\psi^{ABC} = \sum_{i=1}^{d_C} p_i \psi_i^{AB} \otimes \ket{i} \bra{i}^C,\] where
$\psi_i^{AB}$ are pure states. Since $D_A$ is convex on pure ensembles $\{p_i,\psi^{ABC}_i\}$, the quantity
$D_A(\psi^{ABC})$ is upper bounded by $\sum_i p_i D_A(\psi_i^{AB} \otimes \braket{i}^C)$. Since assistance is not helpful for a product state $\psi^{AB} \otimes \phi^{C}$, we have that $D_A(\psi_i^{AB}
\otimes \braket{i}^C) = D(\psi_i^{AB}) = S(A)_{\psi_i}$. By considering the POVM $E=\{\braket{i}^C\}_{i=1}^{d_C}$, we also have $D_A(\psi^{ABC}) \geq \sum_i p_i D(\psi_i^{AB}) = \sum_i p_i S(A)_{\psi^i}$. Hence, for this special class of classical-quantum states, the upper bound is attained and $D_A$ is just the average entropy of the $A$ system for the ensemble $\{p_i,\psi_i^{AB}\}$.
\end{example}

\begin{proposition}
\label{thm:upperboundEofA} Let $\psi^{ABC}$ be an arbitrary tripartite state. Then
 \begin{equation*}
   D_A(\psi^{ABC}) \leq E_A(\psi^{AB}).
 \end{equation*}
\end{proposition}
\begin{proof} From Proposition \ref{thm:upperbound} and the concavity of the entanglement of assistance quantity $E_A$ (see \cite{dfm} for a proof), we have
\begin{equation*}
  \begin{split}
  D_A(\psi^{ABC}) &\leq \inf_{{\cal E}} \sum_i p_i E_A(\psi_i^{ABC}) \\
    & = \inf_{{\cal E}} \sum_i p_i E_A(\psi_i^{AB}) \\
     & \leq \inf_{{\cal E}} E_A(\sum_i p_i \psi_i^{AB}) \\
     & = E_A(\psi^{AB})
  \end{split}
\end{equation*}
where the minimization is taken over all pure state ensembles $\{p_i,
\psi_i^{ABC}\}$ of the state $\psi^{ABC}$.
\end{proof}

The previous bound on $D_A$ is better understood by imagining the following scenario. The $A'$ system of a pure state $\psi^{AA'}$ is sent to a receiver (i.e Bob) via a noisy channel ${\cal N}$, which can be expressed in its Stinespring form as ${\cal N}(\psi) = \Tr_E U\rho U^{\dag}$, where $U: A' \rightarrow BE$ is an isometry. Another player, Charlie,  tries to help Alice and Bob by measuring the environment and sending its measurement outcome to Alice and Bob. Two cases can occur. If Charlie has complete access to the environment, the best rate Alice and Bob can achieve is given by the entanglement of assistance $E_A(\psi^{ABC})$. More likely, however, is the case where Charlie will only be able to measure a subsystem $C_1$ of the environment $E=C_1C_2$. In this situation, the optimal rate is given by the one-shot entanglement of assistance $D_A(\psi^{ABC_1})$, where $\psi^{ABC_1} = \Tr_{C_2} \psi^{ABE}$. Since this case is more restrictive to Charlie in terms of measuring possibilities, it makes sense that $D_A(\psi^{ABC}) \leq E_A(\psi^{ABC})$ for any tripartite mixed state $\psi^{ABC}$.  This bound will be attained for all pure states $\psi^{ABC}$ since $D_A$ reduces to $E_A$ in this case.

\section{Achievable rates for assisted distillation} \label{sec:coding}

In this section, we find the rates achieved by a random coding
strategy for assisted entanglement distillation. The helper Charlie will simply perform
a random measurement in his typical subspace. In light
of the equivalence demonstrated in the previous section, eq.~(\ref{eq:ensEq}), we will prove a lower bound on the asymptotic entanglement of assistance by bounding the regularized entanglement of assistance quantity. We will use a much simpler form of Proposition \ref{prop:isometry}:
\begin{proposition}\cite{merge} \label{thm:random} Suppose we have $n$ copies of a tripartite pure state $\psi^{CBR}$, where $S(R)_{\psi} < S(B)_{\psi}$.
Let $\psi^{\tilde{C}\tilde{B}\tilde{R}}$  be the normalized state obtained by projecting
$C^n, B^n, R^n$ into their respective typical subspaces $\tilde{C},\tilde{B},\tilde{R}$. Charlie performs a projective measurement using an orthonormal basis $\{\ket{e_i}^{\tilde{C}}\}$ of $\tilde{C}$ chosen at random according to the Haar measure. Denote by $p_i$ the probability of obtaining outcome $i$. Then, for any $\epsilon > 0$, and large enough $n$, we have
\begin{equation*}
\int_{\mathbb{U}(\tilde{C})}  \sum_i p_i \bigl \| \psi^{\tilde{R}}_i - \psi^{\tilde{R}} \bigr \|_1 dU \leq \epsilon,
\end{equation*}
where $\psi^{\tilde{R}}_i$ is the state of the system $R$ upon
obtaining outcome $i$. The average is taken over the unitary group $\mathbb{U}(\tilde{C})$ using the Haar measure.
\end{proposition}
\begin{proof}
To choose a random orthonormal basis $\{\ket{e_i}^{\tilde{C}}\}^{d_{\tilde{C}}}_{i=1}$, let $\ket{e_i}^{\tilde{C}} = U\ket{i}^{\tilde{C}}$, where $U$ is a Haar distributed unitary on $\mathbb{U}(\tilde{C})$ and $\{\ket{i}\}^{d_{\tilde{C}}}_{i=1}$ is the computational basis on $\tilde{C}$. A measurement in the basis $\{\ket{e_i}^{\tilde{C}}\}$ is equivalent to applying the unitary $U^{\dag}$ on the $\tilde{C}$ system, followed by a projective measurement in the computational basis. Hence, we can apply Proposition \ref{prop:isometry} with $m=1$, $L=1$ , and $K=1$. This gives us
\begin{equation*}
\int_{\mathbb{U}(\tilde{C})} \sum_i p_i \bigl \| \psi^{\tilde{R}}_i - \psi^{\tilde{R}} \bigr \|_1 dU  \leq \frac{2}{d_{\tilde{C}}} + 2\sqrt{d_{\tilde{R}} \Tr \bigg [ \psi^2_{\tilde{R}\tilde{C}} \bigg ]}.
\end{equation*}
Since $\Tr[\psi^2_{\tilde{R}\tilde{C}}]=\Tr[\psi^2_{\tilde{B}}]$, we can use the properties of typicality found in eqs.~(\ref{eq:deux}) and (\ref{eq:trois}) to simplify the last equation:
\begin{equation}\label{eq:zero}
 \begin{split}
\int_{\mathbb{U}(\tilde{C})} \sum_i p_i \bigl \| \psi^{\tilde{R}}_i - \psi^{\tilde{R}} \bigr \|_1 dU  &\leq \frac{2^{1-n(S(\tilde{C})_{\psi}-\delta)}}{(1-\zeta)} + 2\frac{\sqrt{2^{-n(S(B)_{\psi} - S(R)_{\psi} - 4\delta)}}}{1-\zeta}, \\
\end{split}
\end{equation}
where $\delta$ and $\zeta$ can be made arbitrarily small by choosing $n$ large enough. If the condition $S(R)_{\psi} < S(B)_{\psi}$ is satisfied, the right hand side of eq.~(\ref{eq:zero}) vanishes (i.e can be made less than any $\epsilon > 0$) as $n$ grows larger.
\end{proof}

In Horodecki et al. \cite{merge}, the previous proposition was used to study assisted distillation of pure states. The following theorem generalizes the reasoning used there to the mixed state case. The lower bound on the rate at which ebits are distilled, involving the minimum of $I(AC\rangle B)_\psi$ and $I(A\rangle BC)_\psi$, suggests that $C$ is merged either to Alice or Bob, at which point they engage in an entanglement distillation protocol achieving the hashing bound. This need not be the case, however. In the discussion following the proof of the theorem, we will exhibit an example where merging is impossible but the rates are nonetheless achieved.

\begin{theorem}
\label{thm:lowerbound} Let $\psi^{ABC}$ be an arbitrary tripartite state shared by two recipients (Alice and Bob) and a helper (Charlie). Then the asymptotic entanglement of
assistance is bounded below as follows:
\begin{equation}\label{eq:L}
D_A^{\infty}(\psi^{ABC}) \geq \max \{ I(A\rangle B)_{\psi} , L(\psi) \},
\end{equation}
where $L(\psi):=\min \{I(AC\rangle B)_{\psi},I(A\rangle BC)_{\psi}\}$.
\end{theorem}

\begin{proof} That $D_A^{\infty}(\psi^{ABC})$ is always greater than or equal to the coherent
information $I(A\rangle B)_{\psi}$ follows from the hashing inequality and the fact that Charlie's worst measurement is
no worse than throwing away his system and letting Alice and Bob perform a two-way distillation protocol without
outside help. Hence, it remains to show that $D_A^{\infty}(\psi^{ABC}) \geq \min \{I(AC\rangle B)_{\psi},I(A\rangle
BC)_{\psi}\}$.

Since $D^{\infty}_A(\psi^{ABC})$ is equal to the regularization of $D_A(\psi^{ABC})$, we only need to show the existence of a
measurement for Charlie for which the average distillable entanglement is asymptotically close to
$L(\psi)$. We prove this fact via a protocol which uses a random coding strategy. The state $\psi^{ABC}$ and its purifying system $R$ can be regarded as:
\begin{packed_enum}
\item a tripartite system composed of $C,AB$, and $R$.
\item a tripartite system composed of $C,AR$, and $B$.
\end{packed_enum}
 Let's consider $n$ copies of $\psi^{ABC}$, and furthermore, let's assume that $S(AB)_{\psi}$ (resp. $S(AR)_{\psi}$) and $S(R)_{\psi}$ (resp. $S(B)_{\psi}$) are different. This can be enforced by using only a sub-linear amount of entanglement shared between chosen parties in the limit of large $n$. After Schumacher compressing his share of the state $\psi_{C}^{\otimes n}$, Charlie performs a random measurement of his system $\tilde{C}$. Let $J$ be the random variable associated with the measurement outcome and let $\psi^{A^nB^nR^n}_J$ be the state of the systems $A^n, B^n$ and $R^n$ after Charlie's measurement. By Lemma \ref{lem:unionbound} and the Fannes inequality, there exists a measurement of Charlie's system which will produce a state $\psi^{A^nB^nR^n}_J$ satisfying, with arbitrarily high probability:
\begin{equation}\label{eq:property}
 \begin{split}
  S(A^nB^n)_{\psi_J}=S(R^n)_{\psi_J} &= n(\min\{S(AB)_{\psi},S(R)_{\psi}\}
\pm \delta) \\
  S(A^nB^n)_{\psi_J}=S(B^n)_{\psi_J} &= n(\min\{S(AR)_{\psi},S(B)_{\psi}\}
\pm \delta), \\
 \end{split}
\end{equation}
where $\delta$ can be made arbitrarily small by choosing $n$ large enough.
Applying the hashing inequality to such a state will give:
\begin{equation*}
\begin{split}
D(\psi^{A^nB^n}_J) & \geq S(B^n)_{\psi_J} - S(A^nB^n)_{\psi_J} \\
&=  n(\min \{S(B)_{\psi},S(AR)_{\psi}\} \pm \delta) -
n(\min\{S(AB)_{\psi},S(R)_{\psi}\} \pm \delta) \\
&\geq n (\min \{S(B)_{\psi},S(AR)_{\psi}\} - S(R)_{\psi} - 2\delta) \\
& = n (\min \{ I(AC \rangle B)_{\psi},I(A \rangle BC)_{\psi} \} - 2\delta).
\end{split}
\end{equation*}
For each outcome $j$, define $X_j$ to be the variable taking the value zero if $\psi_j^{A^nB^nR^n}$ satisfies eq~(\ref{eq:property}), or one otherwise.
The average two-way distillable entanglement for this measurement will be at least
 \begin{equation*}
   \begin{split}
    \sum_j p_j D(\psi^{A^nB^n}_j) &= \sum_{X_j=0}p_j D(\psi^{A^nB^n}_j) + \sum_{X_j=1} p_j D(\psi^{A^nB^n}_j) \\
    &\geq P(X_J=0)n (\min \{ I(AC \rangle B)_{\psi},I(A \rangle BC)_{\psi} \} - 2\delta) + \sum_{X_j=1} p_j D(\psi^{A^nB^n}_j)\\
    &\geq (1-\alpha)n  \left[\min \{I(AC \rangle B)_{\psi},I(A \rangle
    BC)_ {\psi} \} - 2\delta\right],
 \end{split}
 \end{equation*}
where $\alpha$ can be made arbitrarily small by taking sufficiently large values of $n$.  Finally, we have
 \begin{equation*}
  \begin{split}
   \frac{1}{n} D^A((\psi^{ABC})^{\otimes n}) &\geq \sum_j p_j D(\psi^{A^nB^n}_j) \\
 &\geq (1-\alpha)\left[\min \{I(AC \rangle B)_{\psi},I(A \rangle BC)_ {\psi} \} -2\delta\right].  \\
  \end{split}
 \end{equation*}
Since $\alpha$ and $\delta$ can be chosen to be arbitrarily small, we are done.
\end{proof}

\begin{corollary} \label{cor:lowerbound}
Let $\psi^{ABC}$ be an arbitrary tripartite state shared by two recipients (Alice and Bob) and a helper (Charlie).
Then the asymptotic entanglement of assistance is bounded below as follows:
\begin{equation}\label{eq:L1}
D_A^{\infty}(\psi^{ABC}) \geq \lim_{n \rightarrow \infty}\frac{1}{n} \sup_{{\cal I}} \sum_i p_i L(\sigma^{A^nB^n\bar{C}}_i),
\end{equation}
where the supremum is over all instruments ${\cal I} := \{{\cal E}_i\}$ performed by Charlie, with $\sigma^{A^nB^n\bar{C}}_i = \frac{1}{p_i}(\mathrm{id}^{A^nB^n} \otimes {\cal E}_i)(\psi_{ABC}^{\otimes n})$ and $p_i = \Tr [ {\cal E}_i \psi_C^{\otimes n} ]$.
\end{corollary}
\begin{proof}
First, to see that the maximization of eq.~(\ref{eq:L}) can be removed, consider an instrument ${\cal J}$ which traces out the $C$ system: $\sigma^{AB} = \Tr_C \psi^{ABC}$. Then, both coherent information quantities in $L(\sigma)$ reduce to the coherent information $I(A\rangle B)_{\psi}$. Achievability of the rate $\sum_i p_i L(\sigma_i^{ABC})$, for any instrument ${\cal I}$ performed on $n$ copies of $\psi^{ABC}$, follows by considering a blocking strategy.
\end{proof}

Let's look into some of the peculiarities of the previous results.
First, observe that the right hand side of eq.~(\ref{eq:L}) is bounded from
above by the coherent information $I(A\rangle BC)_{\psi}$. This
follows from the definition of $L(\psi)$, and the strong
subadditivity of the von Neumann entropy, expressed in terms of coherent information quantities as: \[ I(A\rangle
BC)_{\psi} \geq I(A\rangle B)_{\psi}.\] When the lower bound of eq.~(\ref{eq:L}) is equal to
$I(A\rangle BC)_{\psi}$, we have $I(A\rangle BC)_{\psi} \leq
I(AC\rangle B)_{\psi}$, which implies by further calculation that
$I(C\rangle B)_{\psi} \geq 0$. Suppose that $I(C \rangle B)_{\psi} > 0$ and consider $n$ copies of the purified state $(\psi^{ABCR})^{\otimes n}$, written
as $\psi^{C^nB^nR^n_1}$ where $R_1:=AR$ is the relative reference for the helper $C$. State merging (see Chapter 3) tells us that a random measurement on
the typical subspace $\tilde{C}$, as described in our protocol, will decouple the system from its relative reference $R^n_1$, allowing recovery of $C^n$ by Bob up to arbitrarily high fidelity. Our assisted distillation protocol can be improved for this case by recovering the $C^n$ system at Bob's location
before engaging in a two-way distillation protocol, which will now act on the state
$(\psi^{AB\tilde{B}})^{\otimes n}$, where $\tilde{B}$ is an ancilla of the same dimension as the $C$ system. Since the distillable entanglement
across the cut $A$ vs $BC$ cannot increase by local operations and classical communication, the previous strategy is in fact optimal. A
small amount of initial entanglement between $C^n$ and $B^n$
may be needed if $I(C\rangle B)_{\psi}=0$ (see \cite{merge}).

The previous analysis may lead us to believe that when the lower bound of eq.~(\ref{eq:L}) is equal to $I(AC\rangle B)_{\psi}$, a similar strategy of
transferring the system $C$ to Alice could be
applied. However, the following counterexample will show that this
is not always true. Let
\begin{equation}\label{eq:stateexample}
\begin{split}
\ket{\psi}^{BC_2R}
    &= \frac{1}{\sqrt{2}} \ket{000}^{BC_2R}
            +\frac{1}{2}\ket{110}^{BC_2R}+\frac{1}{2}\ket{111}^{BC_2R} \quad \mbox{and} \\
\ket{\psi}^{AC_1}
    &=  \frac{1}{2}\ket{00}^{AC_1} + \sqrt{\smfrac{3}{4}}\ket{11}^{AC_1}.
\end{split}
\end{equation}
Alice and Bob are to perform assisted distillation on $n$ copies
of $\psi^{ABC_1C_2} = \psi^{AC_1} \otimes \psi^{BC_2}$ with the
help of a single Charlie holding both the $C_1$ and $C_2$ systems.
Such a situation could arise in practice if Alice had a
high-quality quantum channel to Charlie but Charlie's channel
to Bob were noisy. The system $R$ would represent the environment of the
noisy channel.
%
In this case, $L(\psi)$ is equal to $I(AC\rangle B)_\psi$, which
is easily calculated to be approximately 0.40, since $I(A\rangle
BC)_\psi \approx 0.81$ and $I(A\rangle B)_\psi$ is negative.
For this example, the achievable rate of our random coding
protocol is therefore at least the rate that could have been
obtained by a strategy of first transferring the state of the $C$
system to Alice, followed by entanglement distillation between
Alice and Bob at the hashing bound rate. However, the coherent
information $I(C\rangle A)_{\psi}$ is negative for the state
$\psi^{ABC_1C_2R}$. By the optimality of state merging, the state
transfer from Charlie to Alice cannot be accomplished without the
injection of additional entanglement between them. Therefore, the
protocol achieves the rate $I(AC\rangle B)_\psi$ \emph{without}
performing the Charlie to Alice state transfer.

This example also illustrates a general relationship between
hierarchical distillation strategies and the random measurement
strategy proposed in this chapter. A hierarchical strategy for a
state $\psi^{ABC_1C_2} = \psi^{AC_1} \otimes \psi^{C_2B}$ would
consist of first distilling entanglement between $A$ and $C_1$ as
well as between $C_2$ and $B$, followed by entanglement swapping
to establish ebits between Alice and Bob. If the first level
distillations are performed at the hashing rate, then this
strategy will establish $\min [ I(A\rangle C_1)_\psi, I(C_2
\rangle B)_\psi ]$ ebits between Alice and Bob per copy of the input state. On the other hand,
the random measurement strategy will establish at least $L(\psi)$,
which in the case of the example is the minimum of
\begin{eqnarray*}
I(AC \rangle B)_\psi
    &=& I(C_2 \rangle B)_\psi - S(AC_1)_\psi  = I(C_2 \rangle B)_\psi \mbox{  and} \\
I(A\rangle BC)_\psi
    &=& I(A\rangle C_1)_\psi,
\end{eqnarray*}
yielding exactly the same rate as the hierarchical strategy. (The
first line uses the fact that $\psi^{AC_1}$ is pure.) So, for the
random measurement strategy to beat the hierarchical strategy, it
is necessary that the state not factor into the form $\psi^{AC_1}
\otimes \psi^{C_2B}$. As an example, consider modifying the state of eq.~(\ref{eq:stateexample}) by applying a $CNOT$ operation between
the systems $C_1$ and $C_2$ held by Charlie:
\begin{equation*}
CNOT = \left( \begin{array}{cccc}
1 & 0 & 0 & 0 \\
0 & 1 & 0 & 0 \\
0 & 0 & 0 & 1 \\
0 & 0 & 1 & 0 \\
 \end{array} \right)
 \end{equation*}
A CNOT operation can be used to model phase dampening effects between an input state (i.e the $C_2$ system) and its environment (i.e the $C_1$ system).
If the control qubit is the system $C_1$ and the target qubit is $C_2$, the previous state transforms to:
\begin{equation*}
\ket{\phi}^{C_1C_2ABR}
    :=\frac{\ket{00}^{AC_1}\ket{\psi}}{2}^{BC_2R} + \sqrt{\smfrac{3}{4}}\ket{11}^{AC_1} \bigg ( \frac{\ket{010}}{\sqrt{2}}^{BC_2R}
            +\frac{\ket{100}}{2}^{BC_2R}+\frac{\ket{101}}{2}^{BC_2R} \bigg ).
\end{equation*}
The reduced state $\phi^{C_1}$ of the $C_1$ system is equal to:
\begin{equation*}
\phi^{C_1}:=\frac{1}{4}\braket{0}^{C_1} + \frac{3}{4}\braket{1}^{C_1},
\end{equation*}
and the reduced state $\phi^{AC_1}$ of the system $AC_1$ is given by
\begin{equation*}
\phi^{AC_1}:=\frac{1}{4}\braket{00}^{AC_1} + \frac{3}{4}\braket{11}^{AC_1}.
\end{equation*}
Hence, the coherent information $I(A \rangle C_1)_{\phi}$ is zero, yielding a null rate for the hierarchical strategy. On the other hand, the quantities $I(AC \rangle B)_{\phi}$ and $I(A \rangle BC)_{\phi}$ are equal to the coherent informations $I(AC \rangle B)_{\psi}$ and $I(A \rangle BC)_{\psi}$. Thus, the random measurement strategy is tolerant against a CNOT ``error'' on Charlie's systems, as opposed to the hierarchical strategy, which fails to recover from this error.

Finally, it is easy to determine conditions under which
the random measurement strategy for assisted entanglement distillation will yield a
higher rate than the hashing bound between Alice and Bob. As the
next result shows, a state $\psi^{ABC}$ is a good
candidate for the random measurement strategy if it does not
saturate the strong subadditivity inequality of the von Neumann
entropy, and if the $C$ system can be redistributed to Alice and
Bob provided they are allowed to perform joint operations on their
systems.

\begin{proposition}[Beating the Hashing Inequality]
\label{prop:abovehashing} For any state $\psi^{ABC}$, the value of
$L(\psi)$ is positive and strictly greater than the coherent
information $I(A \rangle B)_{\psi}$ if \begin{equation*} I(C
\rangle AB)_{\psi}
> 0 \mbox{ and } S(A|BC)_{\psi} < S(A|B)_{\psi}.\end{equation*}
\end{proposition}

\begin{proof}
  The inequality $S(A|BC)_{\psi} < S(A|B)_{\psi}$ can be rewritten
  as \begin{equation*}
    I(A\rangle BC)_{\psi} > I(A\rangle B)_{\psi},
    \end{equation*} and the condition $I(C \rangle AB)_{\psi} >
  0$ as \begin{equation*}
    S(AB)_{\psi} > S(ABC)_{\psi}.
    \end{equation*} By
  negating and adding $S(B)_{\psi}$ on both sides of the previous inequality, we
  get back
  \begin{equation*}
   \begin{split}
      I(A\rangle B)_{\psi} := S(B)_{\psi} - S(AB)_{\psi} < S(B)_{\psi} - S(ABC)_{\psi} =: I(AC \rangle B)_{\psi}. \\
   \end{split}
  \end{equation*}
\end{proof}

\section{Multipartite entanglement of assistance} \label{sec:multi}

In this section, we look at the optimal rate achievable when many spatially separated parties are assisting
Alice and Bob in distilling entanglement. First, we extend the one-shot entanglement of assistance $D_A$ (Definition \ref{def:A1D}) to arbitrary multipartite states $\psi^{C_1C_2\ldots AB}$, henceforth written simply as $\psi^{C_MAB}$. The type of protocols involved is depicted in Figure~\ref{fig:many}.
\begin{definition}
For a general multipartite state $\psi^{C_MAB}$, consider POVMs $E_1, \ldots, E_m$ performed by
$\{C_1,C_2,\ldots,C_{m}\}$ respectively which lead to a (possibly mixed) bipartite state $\psi^{AB}_{k_1k_2\ldots k_m}$ for POVM outcomes $\overline{k}:=k_1k_2\ldots k_m$. We define the multipartite entanglement of assistance as
\begin{equation*}
D_A(\psi^{C_MAB}) := \sup \sum_{\overline{k}} p_{\overline{k}} D(\psi^{AB}_{\overline{k}}),
\end{equation*}
where the supremum is taken over the above measurements. The asymptotic multipartite entanglement of
assistance $D^{\infty}_A(\psi^{C_MAB})$ is obtained by regularization of the above quantity
$D_A^{\infty}(\psi^{C_MAB})=\lim_{n\rightarrow\infty} \frac{1}{n} D_A(\psi^{\otimes n})$.
\end{definition}

\begin{figure}
\begin{center}
\includegraphics{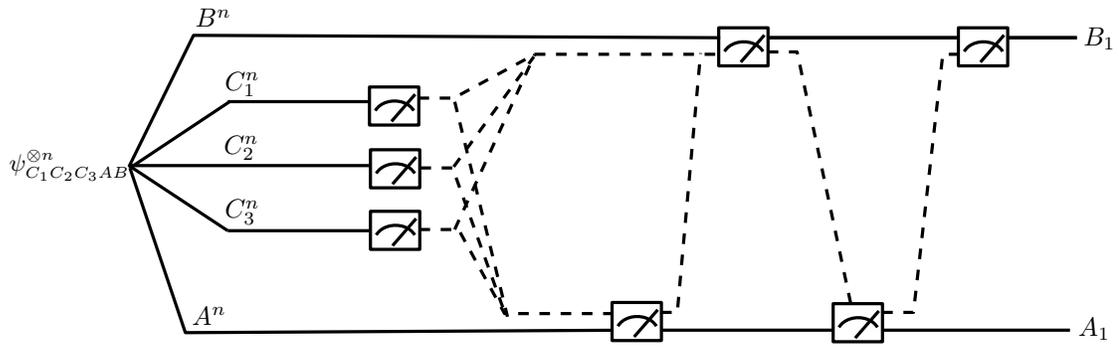}
\end{center}
\caption{Quantum circuit representing a broadcast assisted entanglement distillation protocol involving three helpers. The three helpers perform their measurements, sending copies of the classical outcomes to Alice and Bob. Alice and Bob then implement an LOCC operation, based on that outcome.}
\label{fig:many}
\end{figure}

For a pure state $\psi^{C_MAB}$, it is immediate that the maximization in the preceding definition is attained for POVMs of rank one, leading to an ensemble of pure states $\{q_{\overline{k}}, \psi^{AB}_{\overline{k}}\}$. And so,  $D^{\infty}_A(\psi^{C_MAB})$ reduces to the asymptotic multipartite entanglement of assistance \cite{merge} for pure states.

\begin{proposition}\label{prop:upper}
Let $\psi^{C_MAB}$ be an arbitrary multipartite state. The quantity $D^{\infty}_A(\psi^{C_MAB})$ is bounded
from above by the following quantity:
\begin{equation*}
  D^{\infty}_A(\psi^{C_MAB}) \leq \min_{{\cal T}}  D(\psi^{A{\cal T}|B\overline{{\cal T}}}),
\end{equation*}
where the minimum is over all bipartite cuts $\cal{T}$ and $\psi^{A{\cal T}|B\overline{{\cal T}}}$ is a bipartite state with Alice holding the systems $A{\cal T}$ and Bob holding the systems $B\overline{\cal T}$.
\end{proposition}

\begin{proof} Consider any cut ${\cal T}$ of the helpers $\{C_1,C_2,\ldots, C_{m}\}$ and suppose Alice (resp. Bob) is allowed to perform joint operations on the systems $A{\cal T}$ (resp. $B{\overline{\cal T}}$). Any protocol achieving $D^{\infty}_A(\psi^{C_MAB})$ consists of: 1) POVMs on the helpers followed by a transmission of the outcomes to Alice and Bob 2) local operations and classical communication between the systems $A$ and $B$. This kind of protocol is contained in protocols allowing local operations on the systems $A{\cal T}$ and $B{\overline{{\cal T}}}$  and classical communication between the cut $A{\cal T}$ vs $B{\cal \overline{T}}$. Since the distillable entanglement across the cut $A{\cal T}$ vs $B{\cal {\overline{T}}}$ cannot increase under local operations and classical communication, the optimal achievable rate for these protocols is given by $D(\psi^{A{\cal T}|B{\overline{{\cal T}}}})$. Since this holds for any cut ${\cal T}$ of the helpers, we are done. \end{proof}

\begin{definition}
\label{minhashcut} For an arbitrary multipartite state $\psi^{C_MAB}$, we
define the minimum cut coherent information as:
\[ I^c_{min}(\psi,A:B) := \min_{{\cal T}} I(A {\cal T} \rangle B\overline{{\cal T}})_{\psi},
\] where the minimization is over all bipartite cuts ${\cal T} \subseteq \{C_1,C_2,\ldots, C_{m}\}$.
\end{definition}

\begin{theorem}\label{thm:gen}
Let $\psi^{C_MAB}$ be an arbitrary multipartite state. The asymptotic
multipartite entanglement of assistance $D_A^{\infty}(\psi^{C_MAB})$ is bounded below by:
\begin{equation}
\label{eq:lowerbound} D_A^{\infty}(\psi^{C_MAB}) \geq \max\{I(A\rangle B)_{\psi},I^c_{min}(\psi, A:B)\}.
\end{equation}
\end{theorem}
Before giving a proof of Theorem \ref{thm:gen}, we need the following lemma, which states that the minimum cut coherent information of the original state is preserved, up to a vanishingly small perturbation, after a helper has finished performing a random measurement on his system. The arguments needed for  demonstrating this lemma are similar to those used to in \cite{merge} to prove eq.~(\ref{eq:mincutEofA}).
\begin{lemma}\label{lem:mincut}
Given $n$ copies of a state $\psi^{C_MAB}$, let $C_{m}$ perform a random measurement on his typical subspace $\tilde{C}_{m}$ as in Proposition \ref{thm:random}. For any $\delta > 0$ and $n$ large enough, there exists a measurement performed by the helper $C_m$ such that, with arbitrarily high probability, the outcome state $\psi^{A^nB^nC_1^n\ldots C_{m-1}^n}_J$ satisfies the following
inequality:
\begin{equation*}
I^c_{min}(\psi_J,A^n:B^n) \geq n (I^c_{min}(\psi,A:B) - \delta),
\end{equation*}
where $J$ is the random variable associated with the measurement outcome.
\end{lemma}

\begin{proof} The minimum cut coherent information $I^c_{min}(\psi,A:B)$ of the state $\psi^{C_MAB}$ can be rewritten as
   \beu
      I^c_{min}(\psi, A:B) = \min_{{\cal T} \subseteq \{C_1,C_2,\ldots, C_{m}\}} \bigl \{ S(B{\cal T})_{\psi} \bigr \} - S(R)_{\psi},
   \eeu
where $R$ is the purifying system for the state $\psi^{C_MAB}$.

Let $\cal T$ be a bipartite cut of the helpers $\{C_1,C_2,\ldots, C_{m}\}$ such that $C_{m} \notin
{\cal T}$. We define its relative complement as ${\cal T'}=\{C_1,\ldots,C_{m-1}\} \backslash \cal T$.
For any such cut $\cal T$, the state $\psi^{C_MAB}$ and its purifying system $R$ can be regarded as a tripartite system composed of $C_{m}$, $AR{\cal T}$ and $B{\cal T'}$. Assuming $S(AR{\cal T})_{\psi}$ and $S(B\cal T')_{\psi}$ to be distinct, the helper $C_{m}$ performs a random measurement on his typical subspace $\tilde{C}_m$. By Proposition \ref{thm:random} and the Fannes inequality, there exists a measurement for Charlie's system for which the outcome state $\psi^{C_1^n\ldots C_{m-1}^nA^nB^n}_J$ satisfies, with arbitrarily high probability:
\begin{equation}
 \label{eq:entropy1}
\min\{S(AR{\cal T})_{\psi},S(B{\cal T'})_{\psi}\} - \delta' \leq \frac{1}{n} S(B^n{\cal T'}^n)_{\psi_J} \leq \min\{S(AR{\cal T})_{\psi},S(B{\cal T'})_{\psi}\} + \delta',
\end{equation} where $\delta'$ can be made arbitrarily small by taking sufficiently large values for $n$. Hence, the reduced state entropies stay distinct by taking a sufficiently small value of $\delta'$. Since $I^c_{min}(\psi_J,A^n:B^n)$ can be re-expressed as
\begin{equation}
\label{mincut} I^c_{min}(\psi_J,A^n:B^n) = \min_{{\cal T}\subseteq \{C_1,C_2,\ldots, C_{m-1}\}} \{S(B^n{\cal T'}^n)_{\psi_J} \} - S(R^n)_{\psi_{J}},
\end{equation}
we can substitute the lower bound for $S(B^n{\cal T'}^n)_{\psi_J}$ into eq.~(\ref{mincut}) and obtain
\begin{subequations}
\begin{align}
I^c_{min}(\psi_J,A^n:B^n) & \geq n \min_{\cal T}  ( \min\{S(AR{\cal
T})_{\psi},S(B{\cal T'})_{\psi}    \} - \delta' ) - S(R^n)_{\psi_{J}} \\
 & = n (\min_{\cal T} \{ S(B{\cal T'}C_{m})_{\psi},S(B{\cal T'})_{\psi} \} - \delta') -
S(R^n)_{\psi_J} \\
 & = n (\min_{\cal T} \{ S(B{\cal T})\} - \delta')  -
S(R^n)_{\psi_J}. \label{psijeq}
\end{align}
\end{subequations}

To finish the proof, the last fact we need concerns the entropy of the purifying system $R$. If we consider the purified state
$\psi^{C_MABR}$ as a tripartite system composed of $C_{m}$, $R$ and $ABC_1,\ldots,C_{m-1}$, we
can apply Proposition \ref{thm:random} and obtain, w.h.p:
\begin{equation}
 \label{eq:entropyreference}
   \begin{split}
S(R^n)_{\psi_J} & =  n
(\min\{S(R)_{\psi},S(C_1,\ldots,C_{m-1}AB)_{\psi}\} \pm \delta'' ) \\
& \leq  n(S(R)_{\psi} + \delta''), \\
 \end{split}
\end{equation}
where $\delta''$ can be made arbitrarily small. This tells us that for large values of $n$, the entropy of the purifying system will not significantly increase as a result of the helper $C_m$ performing a measurement on his typical subspace $\tilde{C}_{m}$. Note that, as in Theorem \ref{thm:lowerbound}, we can use the union bound and Markov's inequality (see Lemma \ref{lem:unionbound}) to show the existence of a measurement on $\tilde{C}_{m}$ which produces states such that, w.h.p, eqs.(\ref{eq:entropy1}) and (\ref{eq:entropyreference}) are both satisfied.
Combining the last equation with eq.~(\ref{psijeq}) and choosing values
for $\delta',\delta''$ small enough that $\delta'+\delta'' < \delta$, we get the desired result. \end{proof}

\begin{proof+}{of Theorem \ref{thm:gen}} The right hand side of eq.~(\ref{eq:lowerbound}) is just the coherent information when $m=0$, and is equal to $\max\{I(A \rangle B)_{\psi}, L(\psi)\}$ for $m=1$. Eq.~(\ref{eq:lowerbound}) holds for these base cases by the hashing inequality and Theorem \ref{thm:lowerbound}. So, from here on, assume $m \geq 2$. Moreover, that $D^{\infty}_A(\psi^{C_MAB})$ is at least $I(A\rangle B)_{\psi}$ follows again from the hashing inequality. Hence, we can focus on proving that $D^{\infty}_A(\psi^{C_MAB})$ is bounded below by the minimum cut coherent information.

By Lemma \ref{lem:mincut}, there exists a measurement $E_m$ for the helper $C_m$ which produces an outcome state $\psi^{C_1^n\ldots C_{m-1}^nA^nB^n}_{J}$ satisfying w.h.p. the following inequality:
\begin{equation*}
 \begin{split}
I^c_{min}(\psi_J,A^n:B^n) \geq n (I^c_{min}(\psi,A:B)-\delta), \\
 \end{split}
\end{equation*}
for an arbitrary small $\delta$ and sufficiently large $n$. If we have at our disposal $n^{m}$ copies of the state $\psi^{C_MAB}$ and perform the measurement $E_m$ for each block of $n$ copies of the state, we expect to obtain approximately $n^{m-1}p_j$ copies of the state $\psi_j$, with $p_j$ being the probability of obtaining the state $\psi^j$ after the measurement $E_m$ is performed by $C_{m}$.

For each block consisting of many copies of the state $\psi_j$, we repeat the previous procedure in a recursive manner. We continue this process until all $m$ helpers have performed measurements on their systems. In the end, Alice and Bob will obtain a number of bipartite states $\psi^{AB}_{J_1J_2\ldots J_{m}}$ each satisfying w.h.p.
\begin{equation*}
I^c_{min}(\psi^{AB}_{J_1J_2\ldots J_{m}},A^{n^{m}}:B^{n^{m}}) \geq n^{m} (I^c_{min}(\psi,A:B)-\delta'), \\
\end{equation*}
where $\delta'$ can be made arbitrarily small. Observe that the term on the left hand side of the inequality is the coherent information $I(A\rangle B)_{\psi^{AB}_{J_1J_2\ldots J_{m}}}$, which is bounded above by the distillable entanglement $D(\psi^{AB}_{J_1J_2\ldots J_{m}})$. Since $D_A(\psi^{n^{m}})$ is a supremum over all LOCC measurements performed by the helpers, we have
 \begin{equation*}
  \begin{split}
  \frac{1}{n^{m}} D_A(\psi^{\otimes (n^{m})}) &\geq \frac{1}{n^{m}} \sum_{j_1j_2\ldots j_m} p_{j_1j_2\ldots j_m} D(\psi^{AB}_{j_1j_2\ldots j_{m}}) \\
  &\geq \frac{1}{n^{m}}\sum_{j_1j_2\ldots j_m} p_{j_1j_2\ldots j_m} I(A\rangle B)_{\psi^{AB}_{j_1j_2\ldots j_{m}}} \\
  &\geq (1-\epsilon)(I^c_{min}(\psi,A:B)-\delta'), \\
  \end{split}
\end{equation*}
where $\epsilon$ and $\delta'$ can be both be made arbitrarily small by the arguments of the previous paragraphs. This concludes the proof.  \end{proof+}

Before closing this section, let us say a few words on assisted
distillation when the two recipients are separated by a
one-dimensional chain of repeater nodes, as depicted in figure
\ref{fig:chain}. Applying a hierarchical distillation strategy on
$\psi^{ABCD}$ will achieve a rate of ebits corresponding to
\begin{equation*}
 R(\psi):=\min\{I(A\rangle C_1)_{\psi_1},I(C_2 \rangle
D_1)_{\psi_2}, I(D_2 \rangle B)_{\psi_3}\}.
\end{equation*}
\begin{figure}
\begin{center}
\includegraphics{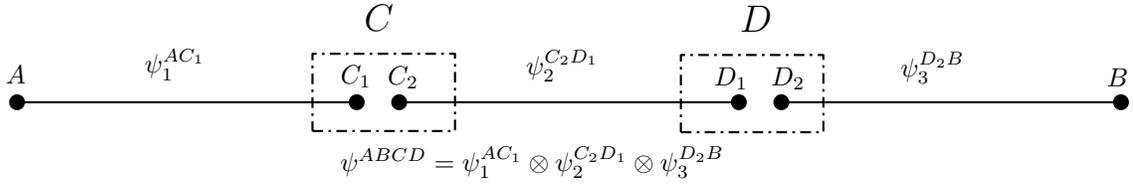}
\end{center}
\caption{A 1-dimensional chain with two repeater stations
separating the two recipients $A$ and $B$.} \label{fig:chain}
\end{figure}
If we consider the cut ${\cal T}_1:=\{C\}$ of the helpers $C$ and
$D$, the coherent information $I(A{\cal T}_1\rangle
B{\overline{\cal T}_1})_{\psi}$ can be simplified to
\begin{equation*} I(A{\cal T}_1\rangle B{\overline{\cal T}_1})_{\psi}
= I(C_2 \rangle D_1)_{\psi_2} - S(AC_1)_{\psi_1} \leq I(C_2
\rangle D_1)_{\psi_2},
\end{equation*} and similarly for the cuts
${\cal T}_2=\{CD\}$ and ${\cal T}_3=\emptyset$, we have
\begin{equation*}
 \begin{split}
I(A{\cal T}_2\rangle B{\overline{\cal T}_2})_{\psi} &= I(D_2
\rangle B)_{\psi_3} - S(AC_1)_{\psi_1} - S(C_2D_1)_{\psi_2} \leq
I(D_2\rangle B)_{\psi_3}, \\
I(A{\cal T}_3\rangle B{\overline{\cal T}_3})_{\psi} &= I(A \rangle C_1)_{\psi_1}. \\
\end{split}
\end{equation*} Thus, the minimum cut coherent information
$I^c_{min}(\psi^{ABCD},A:B)$ is not greater than $R(\psi)$, and so
a hierarchical strategy might be better suited for the case of a
chain state than a random measurement strategy, provided the information stored in the repeaters is not subject to errors (see the example of the previous section). It is easy to
generalize the previous arguments to a chain of arbitrary length
(i.e $m \geq 3$), and to other network configurations.

\chapter{Conclusion}
  \section{Summary}
Following a brief review of quantum information theory in Chapter 2, we studied the problem of multiparty state merging. The main contribution of that chapter was to perform a direct technical analysis of the task of multiparty state merging. The main technical challenge to overcome was to adapt the decoupling lemma of \cite{merge} and the upper bound on the quantum decoupling error (Proposition 4 in \cite{merge}) for the case of $m$ senders and a single receiver sharing a multipartite state. Our upper bound is derived using a random measurement strategy for the senders. Our random coding strategy allowed us to use a well-known result from representation theory to simplify the calculations and thereby obtain a nice simple bound for the decoupling error. Using this calculation in the i.i.d setting, we proved the existence of protocols which achieve multiparty state merging without the need for time-sharing for the case of two senders and no side information at the decoder. We conjectured that time-sharing is not essential also for an arbitrary number of senders and discussed the main difficulty for proving this.

We also introduced the split-transfer problem, a variation on the state merging task, and applied it in the context of assisted distillation. The main technical difficulty here was to formally prove that the decoding operations, implemented by two receivers $A$ and $B$, can be done at the same time following simultaneous measurements by the senders. The essential ingredients for showing this were the commutativity of the Kraus operators $P^{\cK}_{\jK}$ and $P^{\cKbar}_{\jKbar}$ and the triangle inequality. The rate region for a split-transfer is composed of two sub-regions, each
corresponding to rates which would be achievable for a merging operation from $\cK$ (resp. $\cKbar$) to $A$ (resp. $B$). In the context of assisted distillation, our split-transfer protocol was used to redistribute a pure multipartite state to two decoders $A$ and $B$ in such a way that it preserves the min-cut entanglement of the input state, provided the multiparty typicality conjecture holds. Under this assumption, we gave a non-recursive proof that the optimal assisted distillation rate is equal to the min-cut entanglement of the input state.

In Chapter 4, an emphasis on how to accomplish merging when the participants have access only to a single copy of a quantum state was considered. The one-shot analysis was performed using the quantum min- and max-entropy formalism of \cite{Renner02}, and presented other difficulties
than in the asymptotic setting. Most notably, because time-sharing is impossible with only a single copy of a quantum state, our intrinsically multiparty protocol provides the first method to interpolate between achievable costs in the multiparty setting. The technical challenge was to derive an upper bound on the decoupling error for a random coding strategy in terms of the min-entropies. We suspect that it might be possible to further improve our bound by replacing the
min-entropies with their smooth variations, but it is unclear how to proceed in order to show this. We leave it as an open problem. To illustrate the advantages of intrinsic multiparty merging over iterated two-party merging, we have considered three different examples of one-shot distributed compression. Two of those examples demonstrate clearly the advantages of our protocol by allowing some of the senders to transfer their systems for free, something which is impossible for a protocol relying on two-party state merging. The last example considers a state which does not reduce to bipartite entanglement between various subsystems, and thus provided more of a challenge, but yielded greater rewards. Tractable computations for min- and max-entropies of non-trivial states are difficult to perform, and we feel this example is a useful contribution to the ``one-shot'' literature, which contains very few explicit examples.

In the last chapter, we generalized the entanglement of assistance problem by allowing the parties to share a multipartite mixed state. For the case of three parties holding a mixed tripartite state, the optimal assisted rate was proven to be equal to the regularization of the one-shot entanglement of assistance, a quantity which maximizes the average distillable entanglement over all measurements performed by the helper. Two upper bounds for this quantity were established and examples of classes of states attaining them were given. Additionally, the one-shot entanglement of assistance was proven to be a convex quantity for pure ensembles. We also presented new protocols for assisted entanglement distillation, based on a random coding strategy, which are proven to distill entanglement at a rate no less than the minimum cut coherent information, defined as the minimum coherent information over all possible bipartite cuts of the helpers. For states not saturating strong subadditivity, and recoverable by Alice and Bob if they can implement joint operations, we proved that our random coding strategy achieves rates surpassing the hashing inequality. Moreover, the rates formally resemble those achievable if the helper system were merged to either Alice or Bob even when such merging is impossible. Finally, we compared our protocol to a hierarchical strategy in the context of quantum repeaters. We identified a major weakness of the hierarchical strategy by analyzing the effect of a CNOT error on the rates achievable for such strategy. We found that the rate, which can be as good as the rate of our random measurement, becomes null when such error occurs at the repeater node. On the other hand, our protocol is completely fault tolerant and yields the same rate even if this error goes undetected by the helper holding the systems at the repeater node.

\section{Future research directions}

Our proposed protocol in Chapter 5 for assisted distillation of an arbitrary multipartite state involved a measurement on a long block of states, and then a measurement on blocks of these blocks, and so on. It seems likely that a strategy where all the helpers measure in a random basis of their respective typical subspaces and broadcast the results to Alice and Bob would still produce states preserving the minimum cut coherent information. For pure multipartite states, we showed in Chapter 3 that a split-transfer, followed by a distillation protocol between Alice and Bob, removes the multiple blocking required in \cite{merge} for distilling an optimal number of ebits between Alice and Bob. For an arbitrary multipartite state, we are still unsure if our split-transfer can be applied to remove the multiple blocking argument needed to show the lower bound. More generally, it would be interesting to come up with other potential applications for the split-transfer protocol. State merging was used as a building block for solving various communication tasks, and we believe split-transfer could be useful in other multipartite scenarios than the assisted distillation context. Alternatively, it could also simplify some of the existing protocols which rely on multiple applications of the state merging primitive.

We could extend our assisted entanglement scenario in several ways. For instance, we can consider other forms of pure entanglement such as GHZ states and look at the optimal achievable rates under LOCC operations. Another interesting question is to analyze whether general LOCC operations between the parties give more power to the helpers. In the pure multipartite case, we saw that such a strategy is not required to achieve optimal assisted rates. For multipartite mixed state, we should expect a difference in achievable rates when allowing more communication freedom to the helpers. We just have to consider bipartite distillation protocols to see this: the hashing protocol is impossible without communication between the parties. Finally another potential line of research is to analyze our assisted protocol using smooth min- and max-entropies. Recent work by Buscemi and Datta \cite{Datta2} analyzed the one-way distillable entanglement of a bipartite mixed state $\psi^{AB}$ in the one-shot regime using one-shot entropic quantities similar to the quantum min- and max-entropy of \cite{Renner02}. The entanglement of assistance for pure states $\psi^{ABC}$ was also analyzed under this framework \cite{Datta}, and it is another natural progression of our work to analyze the entanglement of assistance using the quantum min- and max-entropy formalism.

Quantum min- and max-entropies are quantities which require an optimization over an infinite set of objects, and thus, do not lend easily to computation. In \cite{Renner03}, an approach for computing these quantities was suggested in terms of semidefinite programming. In Chapter 4, we gave three examples where computation of min-entropies are tractable and have a relatively simple form. Finding other examples of classes of states for which their min- and max-entropies can be evaluated or characterized in more simple terms would be a welcomed addition to the literature on this subject. 

To conclude, theoretical investigations of multiparty protocols are difficult to perform. In this thesis, we suggested new multiparty quantum communication protocols and performed a rigorous analysis of their properties when such analysis could be carried. The multiparty typicality conjecture is the most important open problem we leave on the table. We suspect that a proof technique for answering this conjecture in the affirmative will also be beneficial for several other multiuser information processing tasks~\cite{int1,int2,int3}. As the design of these protocols become more complex, better techniques will be required for analyzing these protocols.

\renewcommand{\appendixname}{Appendix}
\fancyhead[RO]{\emph{Appendix A}}
\begin{appendices}
\chapter{Various Technical Results}
 \renewcommand{\min}{{\operatorname{min}}}
\renewcommand{\max}{{\operatorname{max}}}
\newcommand{\tr}{{\operatorname{tr}}}
\section{Trace norm}

\begin{lemma}\cite{merge}\label{lem:relHS}
For an operator $X$ acting on a space $A$, the Hilbert-Schmidt norm $\|X\|_2$ is related to the trace norm $\|X\|_1$ as follows:
\begin{equation}\label{eq:relationT}
\|X\|^2_1 \leq d \|X\|^2_2,
\end{equation}
where $d$ is the dimension of the support of $X$.
\end{lemma}
\begin{proof}
This follows straightforwardly from the convexity of the $x^2$ function by taking probabilities $1/d$.
\end{proof}

 \section{The swap operator}
Let $A$ and $\tilde{A}$ be any two isomorphic vector spaces of dimensions $d_A$. Consider the computational bases $\{\ket{i}\}^{d_A}_{i=1}$ and $\{\ket{j}\}^{d_A}_{i=1}$ of the spaces $A$ and $\tilde{A}$. The swap operator $F^{A\tilde{A}}$ is defined by the following action on the basis elements $\{\ket{i}^A\otimes \ket{j}^{\tilde{A}}\}$ of $A \otimes \tilde{A}$:
\begin{equation*}
  F^{A\tilde{A}}\ket{i}^A\ket{j}^{\tilde{A}} := \ket{j}^{A} \ket{i}^{\tilde{A}} \quad 1\leq i,j \leq d_A.
\end{equation*}
For any two arbitrary vectors $\ket{\psi}^A = \sum^{d_A}_{i=1} \alpha_i \ket{i}^A$ and $\ket{\phi}^{\tilde{A}}=\sum^{d_A}_{j=1} \beta_j \ket{j}^{\tilde{A}}$, we have
\begin{equation*}
 F^{A\tilde{A}}\ket{\psi}^A\ket{\phi}^{\tilde{A}} = \ket{\phi}^{A} \ket{\psi}^{\tilde{A}},
\end{equation*}
where $\ket{\phi}^A :=\sum^{d_A}_{i=1} \beta_i \ket{i}^A$ and $\ket{\psi}^{\tilde{A}}:=\sum^{d_A}_{j=1} \alpha_j \ket{j}^{\tilde{A}}$.

The swap operator has many interesting properties. It is a unitary operator which is also hermitian. To see this, consider two basis elements $\ket{ij}^{A\tilde{A}}$ and $\ket{kl}^{A\tilde{A}}$. We have
\begin{equation}
\begin{split}
   \bra{ij}F^{\dag}F\ket{kl} &= \langle ji | lk \rangle  = \delta_{j,l} \delta_{i,k}, \\
   F^2\ket{ij}^{A\tilde{A}} &= F \ket{ji}^{A\tilde{A}} = \ket{ij}^{A\tilde{A}}. \\
\end{split}
\end{equation}
From the last line, it follows that the eigenvalues of $F^{A\tilde{A}}$ are 1 and -1. Denote by $\Pi^{A\tilde{A}}_{sym}$ the subspace spanned by the eigenvectors $\{\ket{e^{sym}_i}^{A\tilde{A}}\}$ of $F^{A\tilde{A}}$ with eigenvalues equal to 1. We call $\Pi^{A\tilde{A}}_{sym}$ the \textit{symmetric} subspace of $A\tilde{A}$. We sometimes use the symbol $\Pi^{A\tilde{A}}_{sym}$ to denote the projector onto this subspace. It will be clear from the context which definition applies. The subspace spanned by the eigenvectors $\{\ket{e^{anti}_i}^{A\tilde{A}}\}$ of $F^{A\tilde{A}}$ with eigenvalues equal to -1 is written as $\Pi^{A\tilde{A}}_{anti}$ and is called the \textit{anti-symmetric} subspace of $A\tilde{A}$. This decomposes the space $A\tilde{A}$ into two orthogonal subspaces.
\begin{lemma}
Let $A$ and $\tilde{A}$ be any two isomorphic vector spaces of dimensions $d_A$. Then, we have
\begin{equation*}
  \Pi^{A\tilde{A}}_{sym} = \mathrm{span}\bigg \{ \ket{\psi}^A\ket{\psi}^{\tilde{A}} \bigg | \ket{\psi}^A \in A\bigg \},
\end{equation*}
with $\Tr [\Pi^{A\tilde{A}}_{sym}] = d_A(d_A+1)/2$.
\end{lemma}
\begin{proof}
   That the right hand side is contained in the symmetric subspace follows from the definition. It remains to show that any vector $\ket{\psi}^{A\tilde{A}} \in \Pi^{A\tilde{A}}_{sym}$ can be written as a linear combination of vectors of the form $\ket{\phi}^A\ket{\phi}^{\tilde{A}}$. Write $\ket{\psi}^{A\tilde{A}}$ in the computational basis:
  \begin{equation*}
   \begin{split}
     \sum_{ij} \alpha_{ij} \ket{ij}^{A\tilde{A}} = \ket{\psi}^{A\tilde{A}} = F\ket{\psi}^{A\tilde{A}} = \sum_{ij} \alpha_{ij} \ket{ji}^{A\tilde{A}}. \\
   \end{split}
  \end{equation*}
Hence, we have $\alpha_{ij} = \alpha_{ji}$ for all $1\leq i,j\leq d_A$. We rewrite the state $\ket{\psi}^{A\tilde{A}}$ as
\begin{equation}\label{eq:proofSym}
 \begin{split}
  \ket{\psi}^{A\tilde{A}} &= \sum_{1\leq i < j \leq d_A} \alpha_{ij} (\ket{ij}^{A\tilde{A}}+\ket{ji}^{A\tilde{A}}) + \sum^{d_A}_{i=1} \alpha_{ii} \ket{ii}^{A\tilde{A}} \\
  &= \sum_{1\leq i < j\leq d_A} \alpha_{ij}[ (\ket{i}^A+\ket{j}^A)(\ket{i}^{\tilde{A}}+\ket{j}^{\tilde{A}})-\ket{ii}^{A\tilde{A}}-\ket{jj}^{A\tilde{A}} ] + \sum^{d_A}_{i=1} \alpha_{ii} \ket{ii}^{A\tilde{A}}. \\
 \end{split}
\end{equation}
This proves the first statement. (Observe that the previous set of vectors $(\ket{i}^A+\ket{j}^A)(\ket{i}^{\tilde{A}}+\ket{j}^{\tilde{A}})$ and $\ket{ii}^{A\tilde{A}}$ are also linearly independent.) To obtain the second result, notice that in the first line of the previous equation, the state is written as a linear combination of orthogonal vectors $\{\ket{ij}^{A\tilde{A}}+\ket{ji}^{A\tilde{A}}\}_{1\leq i<j\leq d_A}$ and $\{\ket{ii}\}^{d_A}_{i=1}$. Since the state was arbitrarily chosen, these vectors must generate the symmetric subspace. We have ${{d_A}\choose{2}}+d_A = d_A(d_A+1)/2$ of these vectors.
\end{proof}
Observe that the preceding equation implies this orthonormal basis for the anti-symmetric subspace: $\{ \frac{1}{\sqrt{2}}(\ket{ij}^{A\tilde{A}} -\ket{ji}^{A\tilde{A}})\}_{1\leq i < j \leq d_A}$.

\begin{corollary}\label{cor:invariance}
 Let $U$ be a unitary operator acting on a space $A$ of dimension $d_A$. For any two vectors $\ket{\psi}^{A\tilde{A}}$ in $\Pi^{A\tilde{A}}_{sym}$ and $\ket{\phi}^{A\tilde{A}} \in \Pi^{A\tilde{A}}_{anti}$, we have $(U \otimes \tilde{U}) \ket{\psi}^{A\tilde{A}} \in \Pi^{A\tilde{A}}_{sym}$ and $(U \otimes \tilde{U}) \ket{\phi}^{A\tilde{A}} \in \Pi^{A\tilde{A}}_{anti}$. Here, the unitary $\tilde{U}$ is a ``copy'' version of $U$: if $U\ket{i}^A = \ket{\psi_i}^A$, then $\tilde{U}\ket{i}^{\tilde{A}} = \ket{\psi_i}^{\tilde{A}}$.
\end{corollary}
\begin{proof}
  From the preceding lemma, we write the vector $\ket{\psi}^{A\tilde{A}}$ as
  \begin{equation*}
  \begin{split}
  \ket{\psi}^{A\tilde{A}} &= \sum_{1\leq i < j\leq d_A} \alpha_{ij}[ (\ket{i}^A+\ket{j}^A)(\ket{i}^{\tilde{A}}+\ket{j}^{\tilde{A}})-\ket{ii}^{A\tilde{A}}-\ket{jj}^{A\tilde{A}} ] + \sum^{d_A}_{i=1} \alpha_{ii} \ket{ii}^{A\tilde{A}} \\
  & =: \sum_i \beta_i \ket{\psi_i}^{A}\ket{\psi_i}^{\tilde{A}}.
  \end{split}
  \end{equation*}
 Applying the unitary $(U \otimes \tilde{U})$ to this vector, followed by the swap operator $F^{A\tilde{A}}$, we have
  \begin{equation*}
  F^{A\tilde{A}}(U \otimes \tilde{U}) \ket{\psi}^{A\tilde{A}} = F (\sum_i \beta_i U\ket{\psi_i}^{A}\tilde{U}\ket{\psi_i}^{\tilde{A}}) = \sum_i \beta_i U\ket{\psi_i}^{A} \tilde{U}\ket{\psi_i}^{\tilde{A}}.
  \end{equation*}
Hence, the vector $(U \otimes \tilde{U}) \ket{\psi}^{A\tilde{A}}$ is in the symmetric subspace. To prove the second statement, we proceed similarly by writing the vector $\ket{\phi}^{A\tilde{A}} \in \Pi^{A\tilde{A}}_{anti}$ as
 \begin{equation*}
  \ket{\phi}^{A\tilde{A}} = \sum_{1\leq i < j\leq d_A} \alpha'_{ij} (\ket{ij}^{A\tilde{A}} - \ket{ji}^{A\tilde{A}}).
  \end{equation*}
Applying the swap operator to $(U \otimes \tilde{U})\ket{\phi}^{A\tilde{A}}$, we get
\begin{equation*}
\begin{split}
  F(U \otimes \tilde{U})\ket{\phi}^{A\tilde{A}} &= F (\sum_{1\leq i < j\leq d_A} \alpha'_{ij} (U\ket{i}^A\tilde{U}\ket{j}^{\tilde{A}} - U\ket{j}^A\tilde{U}\ket{i}^{\tilde{A}}) \\
  &= \sum_{1\leq i < j\leq d_A} \alpha'_{ij} (U\ket{j}^A\tilde{U}\ket{i}^{\tilde{A}} - U\ket{i}^A\tilde{U}\ket{j}^{\tilde{A}} \\
  &= - (U \otimes \tilde{U})\ket{\phi}^{A\tilde{A}}.
 \end{split}
\end{equation*}
Hence, the anti-symmetric subspace is also \textit{invariant} under unitaries of the form $(U \otimes \tilde{U})$.
\end{proof}

\begin{lemma}[The swap trick]\label{lem:swaptrick}
Let $X^A$ be any operator acting on a vector space $A$ of dimension $d_A$. Then, we have
\begin{equation}\label{eq:swaptrick}
 \Tr [ (X^A)^2] = \Tr [ (X^A \otimes X^{\tilde{A}}) F^{A\tilde{A}}],
\end{equation}
where $X^{\tilde{A}}$ is a ``copy'' of $X$ acting on an isomorphic space $\tilde{A}$ of $A$.
\end{lemma}
\begin{proof}
Let's expand $(X^A)^2$ (we remove some of the notation for clarity):
\begin{equation*}
\begin{split}
  (X^A)^2 &= \sum_{ij} \lambda_{ij} \ket{i}\bra{j} \sum_{pq}\lambda_{pq} \ket{p}\bra{q} \\
     &= \sum_{ijpq} \lambda_{ij}\lambda_{pq} \delta_{j,p} \ket{i}\bra{q}. \\
     &= \sum_{ipq} \lambda_{ip}\lambda_{pq} \ket{i}\bra{q}. \\
 \end{split}
\end{equation*}
Hence, $\Tr[ (X^A)^2] = \sum^{d_A}_{p=1}\sum^{d_A}_{q=1} \lambda_{qp}\lambda_{pq}$. Now, expand the right hand side of eq.~(\ref{eq:swaptrick}):
\begin{equation*}
\begin{split}
  \Tr [ (X^A \otimes X^{\tilde{A}}) F^{A\tilde{A}}] &= \Tr \bigg [ \sum_{ijpq} \lambda_{ij}\lambda_{pq} \ket{ip}\bra{jq} F^{A\tilde{A}}\bigg ] \\
  &=  \Tr \bigg [ \sum_{ijpq} \lambda_{ij}\lambda_{pq} \ket{pi}\bra{jq} \bigg ]\\
  &= \sum_{ijpq} \lambda_{ij}\lambda_{pq} \delta_{p,j} \delta_{i,q}\\
  &= \sum^{d_A}_{p=1}\sum^{d_A}_{q=1} \lambda_{qp}\lambda_{pq}\\
  &= \Tr[ (X^A)^2]. \\
\end{split}
\end{equation*}
\end{proof}
\begin{lemma}\label{lem:swaptensor}
Let $A$ and $R$ be any two arbitrary vector spaces of dimensions $d_A$ and $d_R$ respectively. Then, we have
\begin{equation}
  F^{AR,\widetilde{AR}} = F^{A\tilde{A}} \otimes F^{R\tilde{R}},
\end{equation}
where $\widetilde{AR}:=\ti{A}\otimes\ti{R}$. Here, $\ti{A}$ and $\ti{R}$ are isomorphic vector spaces of $A$ and $R$ respectively.
\end{lemma}
\begin{proof}
The elements of the computational basis $\{\ket{k}^{AR}\}^{d_Ad_R}_{k=1}$ can be rewritten as $\{\ket{ij}^{AR}\}^{d_A,d_R}_{i=1,j=1}$ by identifying $k$ with $d_A(j-1)+i$. We have
\begin{equation*}
\begin{split}
 F^{AR,\widetilde{AR}}\ket{ij}^{AR}\ket{mn}^{\widetilde{AR}} &= \ket{mn}^{AR}\ket{ij}^{\widetilde{AR}}\\
 & = \ket{m}^A\ket{n}^R \ket{i}^{\tilde{A}}\ket{j}^{\tilde{R}} \\
 &= \ket{m}^A\ket{i}^{\tilde{A}}\ket{n}^R\ket{j}^{\tilde{R}} \\
 &= F^{A\tilde{A}}\ket{i}^A\ket{m}^{\tilde{A}} \otimes F^{R\tilde{R}}\ket{j}^R\ket{n}^{\tilde{R}} \\
 &= (F^{A\tilde{A}} \otimes F^{R\tilde{R}})\ket{ij}^{AR}\ket{mn}^{\tilde{A}\tilde{R}}.\\
\end{split}
\end{equation*}
Since this holds for any $1\leq i,m \leq d_A$ and $1\leq j,n\leq d_R$, we are done.
\end{proof}
 \section{Averages over the unitary group}

 \begin{lemma}\label{lem:Haar1}
  Let $\{\ket{e_k}\}^{d_A}_{k=1}$ be an orthonormal basis of a space $A$. For all $i,j$ with $1 \leq i,j \leq d_A$, we have
  \begin{equation}
   \label{eq:Haar1}
    \int_{\mathbb{U}(A)}U \ket{e_i}\bra{e_j}^A U^{\dag} dU = \delta_{i,j} \frac{I^A}{d_A},
  \end{equation}
  where the average is taken over the unitary group $\mathbb{U}(A)$ using the Haar measure.
 \end{lemma}
 \begin{proof}
 Let $P(\ket{e_i}\bra{e_j}) =  \int_{\mathbb{U}(A)}U \ket{e_i}\bra{e_j}^A U^{\dag} dU$. Using the linearity of the trace, we have
\begin{equation*}
  \begin{split}
 \Tr \bigg [ P(\ket{e_i}\bra{e_j}) \bigg] &= \int_{\mathbb{U}(A)} \Tr \bigg [ U \ket{e_i}\bra{e_j}^A U^{\dag} \bigg ] dU \\
 &= \int_{\mathbb{U}(A)} \Tr \bigg [ U^{\dag} U \ket{e_i}\bra{e_j}^A \bigg ] dU \\
  &= \delta_{i,j},
  \end{split}
 \end{equation*}
 where the last step follows since $U^{\dag}U = I^A$ and $\int_{\mathbb{U}(A)} dU = 1$. (The Haar measure on a topological compact group is also a probability measure).
 The operator $P(\ket{e_i}\bra{e_j})$ is hermitian. To see this, choose any unitary operator $V$ such that $V\ket{e_i}^A = \ket{e_j}^A$ and $V\ket{e_j}^A=\ket{e_i}^A$ and use the right invariance of the Haar measure on the unitary group $\mathbb{U}(A)$ \footnote{The Haar measure is left and right invariant for compact topological groups.}:
 \begin{equation*}
   \begin{split}
     P(\ket{e_i}\bra{e_j}) = P(V\ket{e_i}\bra{e_j}V^{\dag}) &= \int_{\mathbb{U}(A)}UV \ket{e_i}\bra{e_j}^A V^{\dag}U^{\dag} dU \\
     &= \int_{\mathbb{U}(A)}U \ket{e_j}\bra{e_i}^A U^{\dag} dU \\
     &= P(\ket{e_i}\bra{e_j})^{\dag}. \\
   \end{split}
 \end{equation*}
Finally, write $P(\ket{e_i}\bra{e_j})$ in its spectral decomposition as $\sum^{d_A}_{k=1} \lambda_k \braket{\mu_k}^A$ and let $\pi : \{1,2,\ldots,d_A\} \rightarrow \{1,2,\ldots, d_A\}$ be any permutation of the set $\{1,2,\ldots,d_A\}$. Define the unitary $V_{\pi}$ such that $V_{\pi} \ket{\mu_k}^A = \ket{\mu_{\pi(k)}}^A$ for all $1 \leq k \leq d_A$. Using the left-invariance of the Haar measure, we have
\begin{equation*}
  \begin{split}
     \sum_k \lambda_k \braket{\mu_{\pi(k)}} &= V_{\pi}P(\ket{e_i}\bra{e_j})V_{\pi}^{\dag}  \\
      & = P(\ket{e_i}\bra{e_j}) = \sum_k \lambda_k \braket{\mu_k}. \\
  \end{split}
\end{equation*}
Since the eigenvectors $\ket{\mu_k}$ are orthogonal, we have
\begin{equation*}
  \lambda_{\pi^{-1}(k)} = \Tr [ V_{\pi}P(\ket{e_i}\bra{e_j})V_{\pi}^{\dag} \braket{\mu_k}] = \Tr [ P(\ket{e_i}\bra{e_j}) \braket{\mu_k} ] = \lambda_k
\end{equation*}
for any permutation $\pi$. Hence, all the eigenvalues of $P(\ket{e_i}\bra{e_j})$ must be the same, and since $\Tr[P(\ket{e_i}\bra{e_j})] = \delta_{i,j}$, we either have $\lambda_k = 0$ for all $k$ when $i\neq j$ or $\lambda_k = \frac{1}{d_A}$ for all $k$ when $i=j$.
\end{proof}

\begin{lemma}\label{lem:Haar2}
   Let $\psi^{AR}$ be any mixed bipartite state of the system $AR$. We have,
  \begin{equation}
   \label{eq:Haar2}
    \int_{\mathbb{U}(A)}(U \otimes I^R)\psi^{AR} (U \otimes I^R)^{\dag} dU = \frac{I^A}{d_A} \otimes \psi^R,
  \end{equation}
where the average is taken over the unitary group $\mathbb{U}(A)$ using the Haar measure.
\end{lemma}
\begin{proof}
Given orthonormal bases $\{\ket{e_i}\}^{d_A}_{i=1}$ and $\{\ket{f_i}\}^{d_R}_{i=1}$ of the systems $A$ and $R$, write the state $\psi^{AR}$ as $\sum_{ijkl} \lambda_{ijkl} \ket{e_if_j}\bra{e_kf_l}^{AR}$. Using the linearity of the integral, we have
\begin{equation*}
  \begin{split}
    \int_{\mathbb{U}(A)}(U \otimes I^R)\psi^{AR} &(U \otimes I^R)^{\dag} dU \\
    &=\sum_{ijkl} \lambda_{ijkl} \int_{\mathbb{U}(A)}(U \otimes I^R)\ket{e_i}\bra{e_k}^A \otimes \ket{f_j}\bra{f_l}^R (U \otimes I^R)^{\dag} dU \\
    &=\sum_{ijkl} \lambda_{ijkl} \bigg [ \int_{\mathbb{U}(A)}U \ket{e_i}\bra{e_k}^A U^{\dag} dU \bigg ] \otimes \ket{f_j}\bra{f_l}^R  \\
    &=\sum_{ijkl} \lambda_{ijkl} \delta_{i,k} \frac{I^A}{d_A} \otimes \ket{f_j}\bra{f_l}^R \\
    &= \frac{I^A}{d_A} \otimes \sum_{ijl} \lambda_{ijil} \ket{f_j}\bra{f_l}^R \\
    &= \frac{I^A}{d_A} \otimes \psi^{R},\\
  \end{split}
\end{equation*}
where we have used Lemma \ref{lem:Haar1} to get the fourth line.
\end{proof}

\begin{lemma}\label{lem:tensorlemma}
 Consider any arbitrary mixed state $\psi^{C_MR}$ of the systems $C_M$ and $R$, where $C_M = C_1 \otimes C_2 \otimes \ldots \otimes C_m$. Let $Q_i$ be a projector of rank $L_i$ onto a subspace $C^1_i$ of $C_i$ and $U_i$ a unitary acting on $C_i$. Define the sub-normalized density operator
\begin{equation*}
 \omeg := (Q_1U_1 \otimes Q_2U_2 \otimes \ldots \otimes Q_mU_m \otimes I^{R}) \psi^{C_MR}
(Q_1U_1 \otimes Q_2U_2  \otimes \ldots \otimes Q_mU_m \otimes I^{R})^{\dag},
\end{equation*}
where $U_M:=U_1 \otimes U_2 \otimes \ldots \otimes U_m$. Then, we have
  \begin{equation}
   \label{eq:Haar3}
 \int_{\mathbb{U}(C_1)}\int_{\mathbb{U}(C_2)}\cdots\int_{\mathbb{U}(C_m)} \omeg dU_M = \frac{L_M}{d_{C_M}} \tau^{C^1_M} \otimes \psi^{R},
  \end{equation}
where the average is taken over the unitary groups $\mathbb{U}(C_1), \mathbb{U}(C_2), \ldots, \mathbb{U}(C_m)$ using the Haar measure. Here $dU_M = dU_1 dU_2 \ldots dU_m$, with $\int_{\mathbb{U}(C_i)} dU_i = 1$ and the average is over unitaries of the form $U_1 \otimes U_2 \otimes \ldots \otimes U_m$.
\end{lemma}
\begin{proof}
Let \[D(\psi^{C_MR}) := \int_{\mathbb{U}(C_m)} (U_m \otimes I^{C_{M-1}R}) \psi^{C_MR} (U_m \otimes I^{C_{M-1}R})^{\dag} dU_m,\]
and write the integral $\int_{\mathbb{U}(C_1)}\int_{\mathbb{U}(C_2)}\cdots\int_{\mathbb{U}(C_m)}dU_M$ as $\int_{\mathbb{U}(C_M)}dU_M$. Additionally, define $Q_M := Q_1 \otimes Q_2 \otimes \ldots \otimes Q_m$.  Using Lemma \ref{lem:Haar2}, we can simplify each of the integrals, starting with the inner most one:
\begin{equation*}
\begin{split}
&\int_{\mathbb{U}(C_M)} \omeg dU_M \\
&=Q_M \bigg [ \int_{\mathbb{U}(C_M)} (U_M \otimes I^R) \psi^{C_MR} (U_M \otimes I^R)^{\dag} dU_M \bigg ] Q_M \\
&=Q_M \bigg [ \int_{\mathbb{U}(C_{M-1})} (U_{M-1} \otimes I^{C_mR}) D(\psi^{C_MR}) (U_{M-1} \otimes I^{C_mR})^{\dag} dU_{M-1} \bigg ] Q_M \\
&=Q_M \bigg [ \int_{\mathbb{U}(C_{M-1})} (U_{M-1} \otimes I^{C_mR}) (\frac{I^{C_m}}{d_{C_m}} \otimes \psi^{C_{M-1}R}) (U_{M-1} \otimes I^{C_mR})^{\dag} dU_{M-1} \bigg ] Q_M \\
&=\frac{L_m}{d_{C_m}}\tau^{C^1_m}\otimes Q_{M-1}\bigg[  \int_{\mathbb{U}(C_{M-1})} (U_{M-1} \otimes I^{R}) \psi^{C_{M-1}R} (U_{M-1} \otimes I^{R})^{\dag} dU_{M-1}\bigg ] Q_{M-1}, \\
\end{split}
\end{equation*}
where in the last line we have used the fact that
\[Q_m \frac{I^{C_m}}{d_{C_m}} Q_m = \frac{I^{C^1_m}}{d_{C_m}} = \frac{L_m}{d_{C_m}}\tau^{C^1_m}.\]
Continuing in this way for the other integrals, the left hand side of eq.~(\ref{eq:Haar3}) is eventually equal to the state
\[\frac{L_1}{d_{C_1}}\tau^{C^1_1} \otimes \frac{L_2}{d_{C_2}}\tau^{C^1_2} \otimes \ldots \frac{L_m}{d_{C_m}}\tau^{C^1_m} \otimes \psi^{R} = \frac{L_M}{d_{C_M}}\tau^{C^1_M}\otimes \psi^R, \]
where we recall that $L_M := \prod^m_{i=1} L_i$ and $d_{C_M} := \prod^m_{i=1} d_{C_i}$.
\end{proof}

For the following lemmas and the main proposition, we denote by $D(X)$ the average $\int_{\mathbb{U}(A)} (U^{\dag} \otimes U^{\dag}) X (U \otimes U) dU$. We will sometimes drop the superscript notation for clarity.
\begin{lemma}\label{lem:zeroanti1}
Let $1 \leq i < j \leq d_A$ and $ 1 \leq k < l \leq d_A$. If $i \neq k$ or $j \neq l$, we have
\begin{equation*}
\int_{\mathbb{U}(A)} (U^{\dag} \otimes \tilde{U}^{\dag})\bigg ( \ket{ij}\pm\ket{ji}\bigg )\bigg (\bra{kl} \pm \bra{lk}\bigg ) (U \otimes \tilde{U}) dU = 0
\end{equation*}
\end{lemma}
\begin{proof}
 If $i \neq k$ and $i \neq l$, consider the unitary $V$ which flips the sign of $\ket{i}$ and fixes all other basis vectors. Using the left invariance of the Haar measure, we quickly see that these averages must be zero. If $i \neq k$ , but $i = l$, use the left invariance with the unitary $W$ which flips the signs of the vectors $\ket{i}$ and $\ket{k}$. Since $j > i = l > k$, the averages must once again be zero. The other cases are treated similarly.
\end{proof}

\begin{lemma}\label{lem:zeroanti2}
Let $1 \leq i \leq d_A$ and $ 1 \leq k < l \leq d_A$. We have
\begin{equation*}
\begin{split}
\int_{\mathbb{U}(A)} (U^{\dag} \otimes \tilde{U}^{\dag})\ket{ii}\bigg (\bra{kl}\pm\bra{lk}\bigg )(U \otimes \tilde{U}) dU &=0 \\
\int_{\mathbb{U}(A)} (U^{\dag} \otimes \tilde{U}^{\dag})\bigg (\ket{kl}\pm \ket{lk}\bigg )\bra{ii}(U \otimes \tilde{U}) dU &=0 \\
\end{split}
\end{equation*}
\end{lemma}
\begin{proof}
 If $i \neq k$, then apply the left invariance property of the Haar measure with a unitary which flips the sign of the vector $\ket{k}$ and fixes all other basis vectors. If $i \neq l$, use the left invariance property of the Haar measure with a unitary which flips the sign of $\ket{l}$ and fixes all other basis vectors.
\end{proof}

The following proposition was first proven in \cite{merge}. We give a different proof in this appendix which does not rely on Schur's lemma (see pg.37 in \cite{Tung}). Our proof suggests that the following result may hold for other measures than the Haar measure. That is, any measure which is invariant under permutations, sign-flip operators (i.e $V\ket{i}=-\ket{i}$) and Hadamard operations will satisfy the following proposition:
\begin{proposition}\cite{merge}\label{prop:SchurApplication}
  Let $A$ and $\tilde{A}$ be any two isomorphic vector spaces of dimensions $d_A$, and let $X$ be any operator acting on the tensor space $A\tilde{A}$. Then, we have
  \begin{equation}\label{eq:UU}
     \int_{\mathbb{U}(A)} (U^{\dag} \otimes \tilde{U}^{\dag}) X (U \otimes \tilde{U}) dU = \frac{\Tr[X \Pi^{A\tilde{A}}_{sym}]}{\Tr[\Pi^{A\tilde{A}}_{sym}]} \Pi^{A\tilde{A}}_{sym}  + \frac{\Tr[X \Pi^{A\tilde{A}}_{anti}]}{\Tr[\Pi^{A\tilde{A}}_{anti}]} \Pi^{A\tilde{A}}_{anti},
  \end{equation}
where $\tilde{U}$ is a ``copy'' version of the unitary $U$ acting on the space $\tilde{A}$.
\end{proposition}

\begin{proof}
Consider the representation of $X$ in the orthonormal basis with vectors $\{\ket{ii}\}^{d_A}_{i=1}$ , $\{\frac{1}{\sqrt{2}}(\ket{ij}+\ket{ji})\}_{1\leq i < j \leq d_A}$ and $\{\frac{1}{\sqrt{2}}(\ket{ij}-\ket{ji})\}_{1 \leq i < j\leq d_A}$. To simplify the notation, we impose some ordering on the basis elements and write the vectors in $\Pi^{A\tilde{A}}_{sym}$ as $\ket{e^{sym}_i}$ and the vectors in $\Pi^{A\tilde{A}}_{anti}$ as $\ket{e^{ant}_i}$. We have
\begin{equation*}
     X =\sum_{ij} \alpha_{ij} \ket{e^{sym}_i}\bra{e^{sym}_j} + \sum_{ij} \alpha_{ij} \ket{e^{ant}_{i}}\bra{e^{sym}_j} + \sum_{ij} \alpha_{ij} \ket{e^{sym}_i}\bra{e^{ant}_{j}} + \sum_{ij} \alpha_{ij} \ket{e^{ant}_{i}}\bra{e^{ant}_{j}}.
\end{equation*}
Using linearity of the integral, we have
\begin{equation*}
 \begin{split}
    D(X) &=\sum_{ij} \alpha_{ij} D(\ket{e^{sym}_i}\bra{e^{sym}_j}) + \sum_{ij} \alpha_{ij} D(\ket{e^{ant}_{i}}\bra{e^{sym}_j}) \\
    &+ \sum_{ij} \alpha_{ij} D(\ket{e^{sym}_i}\bra{e^{ant}_{j}}) + \sum_{ij} \alpha_{ij} D(\ket{e^{ant}_{i}}\bra{e^{ant}_{j}})
  \end{split}
\end{equation*}
Using Lemmas \ref{lem:zeroanti1} and \ref{lem:zeroanti2}, the previous equation simplifies to
\begin{equation}\label{eq:invariance}
  \begin{split}
  D(X)&= D(\Pi^{A\tilde{A}}_{sym}X\Pi^{A\tilde{A}}_{sym}) + D(\Pi^{A\tilde{A}}_{anti}X\Pi^{A\tilde{A}}_{anti}). \\
  \end{split}
\end{equation}
Write $D(\Pi^{A\tilde{A}}_{sym}X\Pi^{A\tilde{A}}_{sym})$ as $\sum_{ij} \mu_{ij} \ket{e^{sym}_i}\bra{e^{sym}_j}.$
Using the right invariance of the Haar measure, we have
\begin{equation*}
V\otimes \tilde{V} D(\Pi^{A\tilde{A}}_{sym} X \Pi^{A\tilde{A}}_{sym}) = D(\Pi^{A\tilde{A}}_{sym} X \Pi^{A\tilde{A}}_{sym}) V \otimes \tilde{V}
\end{equation*}
for all unitaries $V$ acting on $A$, with $\tilde{V}$ being a copy version of $V$ acting on $\tilde{A}$. As in Lemmas $\ref{lem:zeroanti1}$ and $\ref{lem:zeroanti2}$, we can choose carefully our unitary , flipping some of the vectors in the appropriate way, so that the previous equation implies
\begin{equation*}
  \mu_{ij} = -\mu_{ij} \quad \text{for all $1 \leq i ,j \leq d_A(d_A+1)/2$ such that $i \neq j.$}
\end{equation*}
Hence, the operator $D(\Pi^{A\tilde{A}}_{sym} X \Pi^{A\tilde{A}}_{sym})$ is diagonal. To show that it is also a multiple of the projector onto the symmetric subspace, consider the unitary $V$ which permutes the basis elements $\frac{1}{\sqrt{2}}(\ket{12}+\ket{21})$ and $\frac{1}{\sqrt{2}}(\ket{kl}+\ket{lk})$. Then, the right invariance of the Haar measure implies that
\begin{equation*}
\begin{split}
\Tr \bigg [ D(\Pi^{A\tilde{A}}_{sym} X \Pi^{A\tilde{A}}_{sym})&\bigg (\frac{\ket{12}+\ket{12}}{\sqrt{2}}\bigg )\bigg(\frac{\bra{12}+\bra{21}}{\sqrt{2}}\bigg)\bigg ] \\
&= \\
\Tr \bigg[D(\Pi^{A\tilde{A}}_{sym} X \Pi^{A\tilde{A}}_{sym})&\bigg(\frac{\ket{kl}+\ket{lk}}{\sqrt{2}}\bigg)\bigg(\frac{\bra{kl}+\bra{lk}}{\sqrt{2}}\bigg) \bigg ]
\end{split}
\end{equation*}
for all $1 \leq k < l \leq d_A$. We can proceed similarly and show that
\[\Tr\bigg [D(\Pi^{A\tilde{A}}_{sym} X \Pi^{A\tilde{A}}_{sym})\braket{ii}\bigg ]=\Tr\bigg [D(\Pi^{A\tilde{A}}_{sym} X \Pi^{A\tilde{A}}_{sym})\braket{jj}\bigg ] \]
for all $1\leq i,j \leq d_A$. The last thing we need to prove is that
\[\Tr\bigg [D(\Pi^{A\tilde{A}}_{sym} X \Pi^{A\tilde{A}}_{sym})\braket{11}\bigg ]=\Tr\bigg [D(\Pi^{A\tilde{A}}_{sym} X \Pi^{A\tilde{A}}_{sym})\bigg(\frac{\ket{12}+\ket{21}}{\sqrt{2}}\bigg)\bigg(\frac{\bra{12}+\bra{21}}{\sqrt{2}}\bigg)\bigg ]. \]
This is proven using the right invariance of the Haar measure with a unitary $V$, which transforms the vectors $\ket{1}$ and $\ket{2}$ to the vectors $\frac{1}{\sqrt{2}}(\ket{1}+\ket{2})$ and $\frac{1}{\sqrt{2}}(\ket{1}-\ket{2})$, while fixing all other basis vectors.
Combining the previous facts, the operator $D(\Pi^{A\tilde{A}}_{sym} X \Pi^{A\tilde{A}}_{sym})$ is a multiple of $\Pi^{A\tilde{A}}_{sym}$:
 \begin{equation*}
 D(\Pi^{A\tilde{A}}_{sym} X \Pi^{A\tilde{A}}_{sym}) = \lambda \Pi^{A\tilde{A}}_{sym}.
 \end{equation*}
 Using the linearity and the cyclic property of the trace, we have $\Tr [D(\Pi^{A\tilde{A}}_{sym} X \Pi^{A\tilde{A}}_{sym})] = \Tr [ \Pi^{A\tilde{A}}_{sym} X]$, and so
 \begin{equation*}
   \lambda = \frac{\Tr [ \Pi^{A\tilde{A}}_{sym} X]}{\Tr [\Pi^{A\tilde{A}}_{sym}]}.
 \end{equation*}
By a similar argumentation, we have
 \begin{equation*}
 D(\Pi^{A\tilde{A}}_{anti} X \Pi^{A\tilde{A}}_{anti}) =  \frac{\Tr [ \Pi^{A\tilde{A}}_{anti} X]}{\Tr [\Pi^{A\tilde{A}}_{anti}]} \Pi^{A\tilde{A}}_{anti},
 \end{equation*}
and so we are done.
\end{proof}

\begin{proposition} \label{prop:twirl}
Let $A$ be a vector space of dimension $d_A$ and consider a subspace $A_1$ of dimension $L$. For the swap operator $F^{A_1\tilde{A}_1}$, the previous proposition evaluates to
 \begin{equation}\label{eq:UU2}
     \int_{\mathbb{U}(A)} (U^{\dag} \otimes \tilde{U}^{\dag}) F^{A_1\tilde{A}_1} (U \otimes \tilde{U}) dU = \frac{L(d_A-L)}{d_A(d_A^2-1)} I^{A\tilde{A}}  + \frac{L(Ld_A-1)}{d_A(d_A^2-1)} F^{A\tilde{A}}
  \end{equation}
\end{proposition}
\begin{proof}
The swap operator $F^{A_1\tilde{A}_1}$ can be expressed as $\Pi^{A_1\tilde{A}_1}_{sym} - \Pi^{A_1\tilde{A}_1}_{anti}$. We have $\Pi^{A_1\tilde{A}_1}_{sym} \in \Pi^{A\tilde{A}}_{sym}$ and $\Pi^{A_1\tilde{A}_1}_{anti} \in \Pi^{A\tilde{A}}_{anti}$. Hence,
\begin{equation*}
\begin{split} \Tr[F^{A_1\tilde{A}_1}\Pi^{A\tilde{A}}_{sym}]&=L(L+1)/2,\\
\Tr[F^{A_1\tilde{A}_1}\Pi^{A\tilde{A}}_{anti}]&=(L-L^2)/2. \\
\end{split}
\end{equation*}
Using the previous proposition, we have
 \begin{equation}\label{eq:UU2}
 \begin{split}
     \int_{\mathbb{U}(A)} (U^{\dag} \otimes \tilde{U}^{\dag}) F^{A_1\tilde{A}_1} (U \otimes \tilde{U}) dU &=
     \frac{L(L+1)}{d_A(d_A+1)} \Pi^{A\tilde{A}}_{sym}  - \frac{L(L-1)}{d_A(d_A-1)} \Pi^{A\tilde{A}}_{anti} \\
     &=\frac{L(L+1)}{d_A(d_A+1)}\frac{I^{A\tilde{A}}+F^{A\tilde{A}}_{sym}}{2}  - \frac{L(L-1)}{d_A(d_A-1)}\frac{I^{A\tilde{A}}-F^{A\tilde{A}}}{2} \\
     &= \frac{L(d_A-L)}{d_A(d_A^2-1)} I^{A\tilde{A}}  + \frac{L(Ld_A-1)}{d_A(d_A^2-1)} F^{A\tilde{A}} \\
  \end{split}
 \end{equation}
\end{proof}

\subsection{Convexity of $D_A$ for pure ensembles}
\begin{lemma}
For a state $\psi^{ABC}=\sum_i p_i \psi^{ABC}_i$, where $\psi^{ABC}_i$ are pure states, let $F =\{F_x\}_{x=1}^X$ be a POVM of rank one operators on the system $C$. Then, we have
\begin{equation}
   \sum_x q_x D(\psi^{AB}_x) \leq \sum_{x,i} p_i \Tr[F_x\psi_i^{C}] D(\tilde{\psi}^{AB}_{i,x}) ,
\end{equation}
where $\psi^{AB}_x = \frac{1}{q_x} \Tr_C[ (F_x \otimes I^{AB})\psi^{ABC}]$, $q_x = \Tr[F_x \psi^C]$ and $\tilde{\psi}^{AB}_{i,x} = \frac{1}{\Tr[F_x \psi_i^C]} \Tr_C[ (F_x \otimes I^{AB}) \psi_i^{ABC}]$.
\end{lemma}

\begin{proof} The state we get after applying $F_x$ on system $C$ is given by:
 \begin{align}\label{eq:pureD}
   \frac{1}{q_x} \Tr_{C} \biggl [ \biggl (F_x \otimes I^{AB} \biggr ) \sum_i p_i \psi_i^{ABC} \biggr ] &= \frac{1}{q_x} \sum_i p_i \Tr_C[(F_x \otimes I^{AB}) \psi_i^{ABC}] \notag\\
    &= \sum_i \frac{p_i \Tr_C[F_x \psi_i^C]}{q_x} \tilde{\psi}^{AB}_{i,x}.
 \end{align}
Since $q_x = \Tr_C[F_x \psi^{C}]$ and $\sum_i p_i \psi^C_i = \psi^C$, we have a well-defined ensemble of pure states on the right hand side of Eq.~(\ref{eq:pureD}). Since the distillable entanglement is bounded from above by the entanglement of formation, we get
 \begin{align}
   \sum_x q_x D( \frac{1}{q_x} \Tr_{C} \biggl [ \biggl (F_x \otimes I^{AB} \biggr ) \sum_i p_i \psi_i^{ABC} \biggr ]) &\leq \sum_x q_x \sum_i \frac{p_i \Tr_C[F_x \psi_i^C]}{q_x} S(A)_{\tilde{\psi}_{i,x}} \notag \\
   &= \sum_{x,i} p_i \Tr_C[F_x \psi_i^C] D(\tilde{\psi}^{AB}_{i,x}).
 \end{align}
\end{proof}

\begin{proposition}[Convexity of $D_A$ for Pure Ensembles]\label{lemma:convex}
Let $\psi^{ABC}$ be an arbitrary tripartite state. Then, for any convex decomposition $\{p_i,
\psi_i^{ABC}\}$ of $\psi^{ABC}$ into pure states,
  \begin{equation}
    D_A(\psi^{ABC}) \leq \sum_i p_i D_A(\psi_i^{ABC}).
  \end{equation}
\end{proposition}
\begin{proof} For any $\nu > 0$, there exists a POVM $E=\{E_x\}^X_{x=1}$ of rank one operators such that
\begin{equation}
  \sum_x q_x D(\psi^{AB}_x) \geq D_A(\psi^{ABC}) - \nu,
\end{equation}
where $\psi^{AB}_x = \frac{1}{q_x} \Tr_{C}[(E_x \otimes I^{AB}) \psi^{ABC}]$.
From the previous lemma, we have
 \begin{align}
      \sum_x q_x D(\psi^{AB}_x) &\leq \sum_{x,i} p_i \Tr[E_x\psi_i^{C}] D(\tilde{\psi}^{AB}_{i,x}) \notag \\
        &\leq \sum_i p_i \sum_x  \Tr[E_x\psi_i^{C}] D(\tilde{\psi}^{AB}_{i,x}) \notag \\
        &\leq \sum_i p_i D_A(\psi_i^{ABC}),
 \end{align}
where $\tilde{\psi}_{i,x} = \frac{1}{\Tr_C[E_x \psi_i^{C}]}\Tr_{C}[(E_x \otimes I^{AB})\psi^{ABC}_i]$ is the state obtained after performing the POVM $E$ on the state $\psi^{ABC}_i$. Since $\nu$ was arbitrarily chosen, we get back the statement of the proof.
\end{proof}

\section{Typicality}
\begin{Lemma}\label{Lem:operatorineq}
For $n$ copies of a state $\initstate$, let $\Pi_{\ti{B}},\Pi_{\ti{C}_1},\Pi_{\ti{C}_2},\ldots,\Pi_{\ti{C}_m},\Pi_{\ti{R}}$ be the projectors onto the $\delta-$typical subspaces $\ti{B},\ti{C}_1,\ti{C}_2,\ldots,\ti{C}_m$ and $\ti{R}$ respectively. Then, we have
\begin{equation}\label{eq:operatorineq1}
    \Pi_{\ti{B}\ti{C}_M\ti{R}} := \Pi_{\ti{B}} \otimes \Pi_{\ti{C_1}} \otimes \ldots \otimes \Pi_{\ti{C_m}} \otimes \Pi_{\ti{R}} \geq \Pi_{\ti{B}} +
   \Pi_{\ti{C_1}} + \ldots + \Pi_{\ti{C_m}} + \Pi_{\ti{R}} - (m+1) I_{B C_M R}, \\
 \end{equation}
 where $\Pi_{\ti{B}}$ is a shorthand for $\Pi_{\ti{B}} \otimes I^{C_1C_2\ldots C_mR}$, and similarly for $\Pi_{\ti{C}_1},\Pi_{\ti{C}_2},\ldots,\Pi_{\ti{C}_m}$ and $\Pi_{\ti{R}}$.
\end{Lemma}
\begin{proof}
 The projection operators involved in the proof statement pairwise commute, and thus, are simultaneously diagonalizable. Let $\{\ket{e_i}\}$ be a common eigenbasis for these projectors. Then any eigenvector $\ket{e_i}$ with $\Pi_{\ti{B}\ti{C}_M\ti{R}}\ket{e_i} = \ket{e_i}$ satisfies
 \[
 \bigg (\Pi_{\ti{B}} + \Pi_{\ti{C_1}} + \ldots + \Pi_{\ti{C_m}} + \Pi_{\ti{R}} - (m+1) I_{BC_MR} \bigg ) \ket{e_i} = \ket{e_i}.
 \]
 If $\ket{e_i}$ is any eigenvector with $\Pi_{\ti{B}\ti{C}_M\ti{R}}\ket{e_i} = 0$, then it must be in the kernel of at least one of the projection operators $\Pi_{\ti{B}},\Pi_{\ti{C}_1},\Pi_{\ti{C}_2},\ldots,\Pi_{\ti{C}_m}$ and $\Pi_{\ti{R}}$, which implies that
\[
\bigg (\Pi_{\ti{B}} +
   \Pi_{\ti{C_1}} + \ldots + \Pi_{\ti{C_m}} + \Pi_{\ti{R}} - (m+1) I_{BC_MR} \bigg )\ket{e_i} = \lambda_i \ket{e_i},
\]
where $\lambda_i \leq 0$.
Using both of these observations, we have
\begin{equation}
 \begin{split}
   \Pi_{\ti{B}\ti{C}_M\ti{R}} = \sum_{ \Pi_{\ti{B}\ti{C}_M\ti{R}}\ket{e_i}=\ket{e_i}} \braket{e_i} & \geq \sum_{ \Pi_{\ti{B}\ti{C}_M\ti{R}}\ket{e_i}=\ket{e_i}} \braket{e_i} + \sum_{ \Pi_{\ti{B}\ti{C}_M\ti{R}}\ket{e_i}=0} \lambda_i \braket{e_i} \\
   &= \Pi_{\ti{B}} +
   \Pi_{\ti{C_1}} + \ldots + \Pi_{\ti{C_m}} + \Pi_{\ti{R}} - (m+1) I_{BC_MR} \\
 \end{split}
\end{equation}
\end{proof}

\begin{lemma}(Hoeffding's inequality)
For i.i.d random variables $X_1, X_2, \ldots, X_n$, with $X_i \in [a_i,b_i]$, we have
\begin{equation*}
P( \bigg |\sum_i X_i - \sum_i E X_i \bigg | \geq t) \leq 2\exp\bigg (\frac{-2t^2}{\sum^n_{i=1} (b_i - a_i)^2}\bigg )
\end{equation*}
\end{lemma}

\begin{lemma}
Let $\psi^{A} = \sum^d_{i=1} p(x_i)\braket{x_i}$. For $n$ copies of this state, and any $\delta > 0$, we have
\begin{equation*}
\Tr [\psi^{\otimes n}_A \Pi^{n,\delta}_A] \geq 1 - 2d \exp(-2n \delta^2).
\end{equation*}
where $\Pi^{n,\delta}_A$ is the projector for the $\delta-$typical subspace for $\psi^{\otimes n}_A$.
\end{lemma}
\begin{proof}
This follows from the previous lemma by using the union bound on $\Tr [\psi^{\otimes n}_A \Pi^{n,\delta}_A] = P(x_n \in {\cal T}^n_{\delta,p})$.
\end{proof}

\begin{lemma}\label{lem:mergeTyp}
Let $\rho$ be a state on $X \otimes Y$ and let both $\Pi_X$ and $\Pi_Y$ be orthogonal projectors acting on $X$ and $Y$, respectively. Let $\Omega = (\Pi_X \otimes \Pi_Y) \rho (\Pi_X \otimes \Pi_Y)$. Then $\Omega_X \leq \Pi_X \rho_X \Pi_X$.
\end{lemma}

\begin{proof}
The statement is equivalent to demonstrating that for all $|\psi\rangle$,
$\langle \psi | \Omega_X | \psi \rangle \leq \langle \psi | \Pi_X \rho_X \Pi_X | \psi \rangle$, which can be seen by direct calculation. Let $\Pi_Y^c = I - \Pi_Y$.
\begin{eqnarray*}
 \langle \psi | \Pi_X \rho_X \Pi_X | \psi \rangle
 &=& \Tr \left[ (|\psi\rangle\langle\psi | \otimes I) (\Pi_X \otimes I) \rho (\Pi_X \otimes I) \right] \\
 &=& \Tr \left[ \left(|\psi\rangle\langle\psi | \otimes (\Pi_Y + \Pi_Y^c) \right)
 		(\Pi_X \otimes I) \rho (\Pi_X \otimes I) \right] \\
&=& \langle \psi | \Omega_X | \psi \rangle +
	 \Tr \left[ \left(|\psi\rangle\langle\psi | \otimes \Pi_Y^c \right)
 		(\Pi_X \otimes I) \rho (\Pi_X \otimes I) \right] \\
&\geq& \langle \psi | \Omega_X | \psi \rangle.
\end{eqnarray*}
The inequality follows from the fact that $\Tr[AB] \geq 0$ whenever $A, B \geq 0$.
\end{proof}
\newcommand{\da}{\delta_1}
\newcommand{\db}{\delta_2}
\newcommand{\sigAP}{\sigma_1^{\ti{C}_1\ti{C}_2}}
\newcommand{\sigBP}{\sigma_2^{\ti{C}_1\ti{C}_2}}

\begin{proposition}\label{prop:mixedstate}
Let $\psi^{C_1C_2}$ be an arbitrary mixed state and consider $n$ copies of it. For any $\epsilon > 0$ and $n$ large enough, there exists a state $\Psi^{\ti{C}_1\ti{C}_2}$ which satisfies
\begin{equation*}
\begin{split}
\|\Psi^{\ti{C}_1\ti{C}_2} - \psi^{\otimes n}\|_1 &\leq \nu(\epsilon) \\
\Tr[(\Psi^{\ti{C}_1\ti{C}_2})^2] &\leq (1-\epsilon)^{-2}2^{-n(S(C_1C_2)_{\psi} - \upsilon)} \\
\Tr[(\Psi^{\ti{C}_1})^2] &\leq (1-\epsilon)^{-2} 2^{-n(S(C_1)_{\psi} - 3\da)} \\
\Tr[(\Psi^{\ti{C}_2})^2]  &\leq (1-\epsilon)^{-2}2^{-n(S(C_2)_{\psi} - 3\da)}\\
\mbox{rank }\Psi^{\ti{C}_i} &\leq 2^{n(S(C_i)_{\psi}+\da)},
\end{split}
\end{equation*}
where $\da$ and $\upsilon$ can be made arbitrarily small by taking sufficiently large values of $n$. Here, $\nu(\epsilon)$ is a function of $\epsilon$ which vanishes as $\epsilon \rightarrow 0$.
\end{proposition}
\begin{proof}
For any $\delta_i > 0$, with $1 \leq i \leq 2$, define the states:
\begin{equation*}
\begin{split}
\sigAP &:= \Pi^{s_1,\da}_{\ti{C}_1} \otimes \Pi^{s_1,\da}_{\ti{C}_2} \psi_{C_1C_2}^{\otimes s_1} \Pi^{s_1,\da}_{\ti{C}_1} \otimes \Pi^{s_1,\da}_{\ti{C}_2}  \\
\sigBP &:= \Pi^{s_2,\db}_{\ti{C}_1\ti{C}_2} (\sigma^{\ti{C}_1\ti{C}_2}_{1})^{\otimes s_2} \Pi^{s_2,\db}_{\ti{C}_1\ti{C}_2} \\
\end{split}
\end{equation*}
where $\Pi^{s_2,\db}_{\ti{C}_1\ti{C}_2}$ is the projector onto the $\db$-typical subspace for the state $(\sigma^{\ti{C}_1\ti{C}_2}_1)^{\otimes s_2}$. Using Lemma \ref{Lem:operatorineq} and Hoeffding's inequality, we have
\begin{equation*}
\begin{split}
\Tr [ \Pi^{s_1,\da}_{\ti{C}_1} \otimes\Pi^{s_1,\da}_{\ti{C}_2}\psi_{C_1C_2}^{\otimes s_1}] & \geq 1- 2(d_{C_1}+d_{C_2}) \exp(-2s_1 \da^2) \\
 &\geq 1- c \exp(-2s_1 \da^2) \\
\Tr [ \Pi^{s_2,\db}_{\ti{C}_1\ti{C}_2} \sigma^{\otimes s_2}_1] & \geq 1- 2(d_{C_1}^{s_1}d_{C_2}^{s_1}) \exp(-2s_2 \db^2) \\
 &\geq 1- c^{s_1} \exp(-2s_2 \db^2) \\
 \end{split}
\end{equation*}
\newcommand{\sigEnd}{\sigma^{\ti{C}_1\ti{C}_2}}
for some constant $c > 0$. The distance between $\sigBP$ and $\psi^{\otimes n}$ is bounded using the triangle inequality and the Gentle Measurement Lemma:
\begin{equation} \label{eq:t}
\begin{split}
 \| \sigBP - \psi^{\otimes n}\|_1 &\leq \|\sigBP - (\sigma^{\ti{C}_1\ti{C}_2}_{1})^{\otimes s_2}\|_1 + \|(\sigma^{\ti{C}_1\ti{C}_2}_{1})^{\otimes s_2} - \psi^{\otimes n}\|_1 \\
 &\leq 2\sqrt{c^{s_1}} \exp(-s_2 \db^2) + s_2 \|\sigma^{\ti{C}_1\ti{C}_2}_{1} - \psi^{\otimes s_1}\|_1\\
 &\leq 2\sqrt{c^{s_1}} \exp(-s_2 \db^2) + s_2 2\sqrt{c} \exp(-s_1 \da^2) \\
 \end{split}
\end{equation}
For a fixed value of $s_1$, choose $s_2$ to be such that
\begin{equation}\label{eq:xx}
 \epsilon:= c\exp(-2s_{1} \da^2) = c^{s_1} \exp(-2s_2 \db^2)
\end{equation}
Taking the logarithm on both sides, this is equivalent to:
\begin{equation*}
\begin{split}
s_2 &= \frac{1}{2\db^2}(  2s_1 \da^2 + s_1\ln(c) -\ln(c)) \\
s_2 &= \frac{1}{2\db^2}( s_1 ( 2\da^2 + \ln(c) ) -\ln(c))  \\
\end{split}
\end{equation*}
Replacing $s_2$ into eq.~(\ref{eq:t}) and using eq.~(\ref{eq:xx}), we have
\begin{equation*} \label{eq:triangle}
\begin{split}
 \| \sigBP - \psi^{\otimes n}\|_1 &\leq 2\sqrt{c^{s_1}} \exp(-s_2 \db^2) + s_2 2\sqrt{c} \exp(-s_1 \da^2) \\
 &=   2 \sqrt{c}\exp(-s_{1} \da^2) + \frac{1}{\db^2}( s_1 ( 2\da^2 + \ln(c) ) -\ln(c))\sqrt{c} \exp(-s_1 \da^2)  \\
 \end{split}
\end{equation*}
which vanishes for sufficiently large values of $s_1$. Let $\Psi^{\ti{C}_1\ti{C}_2}$ be the normalized state of $\sigBP$. Using the triangle inequality, we have
\begin{equation*} \label{eq:triangle}
\begin{split}
 \| \Psi^{\ti{C}_1\ti{C}_2} - \psi^{\otimes n}\|_1 &\leq 4 \sqrt{c}\exp(-s_{1} \da^2) + \frac{1}{\db^2}( s_1 ( 2\da^2 + \ln(c) ) -\ln(c))\sqrt{c} \exp(-s_1 \da^2)\\
\end{split}
\end{equation*}
which also vanishes as $s_1 \rightarrow \infty$. To bound the quantity $\Tr[ (\Psi^{\ti{C}_1\ti{C}_2})^2]$, we have, using typicality (see eq.~(\ref{eq:typicXX})),
\begin{equation}\label{eq:tracesq}
\Tr[ (\Psi^{\ti{C}_1\ti{C}_2})^2] \leq (1-\epsilon)^{-2}2^{-s_2(S(\ti{C}_1\ti{C}_2)_{\sigma_1}-3\db)}.
\end{equation}
Since the state $\sigma^{\ti{C}_1\ti{C}_2}_1$ is close in the trace distance to the state $\psi^{\otimes s_1}$, we can apply the Fannes inequality, obtaining
\begin{equation*}
S(\ti{C}_1\ti{C}_2)_{\sigma_1} \geq s_1S(C_1C_2)_{\psi} - s_1\eta(\epsilon)\log(d_{C_1}d_{C_2})
\end{equation*}
Substituting into eq.~(\ref{eq:tracesq}), we have
\begin{equation}
\begin{split}
\Tr[ (\Psi^{\ti{C}_1\ti{C}_2})^2] &\leq (1-\epsilon)^{-2}2^{-s_2(S(\ti{C}_1\ti{C}_2)_{\sigma_1}-3\db)} \\
&\leq (1-\epsilon)^{-2}2^{-s_2 s_1 (S(C_1C_2)_{\psi} - \eta(\epsilon)\log(d_{C_1}d_{C_2}) - 3\frac{\db}{s_1})} \\
&\leq (1-\epsilon)^{-2}2^{-n(S(C_1C_2)_{\psi} -\upsilon)},
\end{split}
\end{equation}
where $\upsilon$ will vanish as $s_1$ increases. To obtain a bound on $\Tr[ (\Psi^{\ti{C}_1})^2]$, we have
\begin{equation}
 \sigma^{\ti{C}_1\ti{C}_2}_2 \leq (\sigma^{\ti{C}_1\ti{C}_2}_1)^{\otimes s_2}.
\end{equation}
Since the partial trace preserves the previous ordering and $\Tr[AB] \geq 0$ for two positive operators $A$ and $B$, we have
\begin{equation}
\begin{split}
 \Tr[(\sigma^{\ti{C}_1}_2)^2] &\leq \Tr[((\sigma^{\ti{C}_1}_1)^{\otimes s_2})^2] \\
&= \Tr[(\sigma^{\ti{C}_1}_1)^2]^{s_2} \\
&\leq (2^{-s_1(S(C_1)_{\psi} -3\da)})^{s_2}\\
&= 2^{-n(S(C_1)_{\psi} - 3\da)},
\end{split}
\end{equation}
where the third line follows from Lemma \ref{lem:mergeTyp}. Finally, we have
\begin{equation}
\begin{split}
 \Tr[(\Psi^{\ti{C}_1})^2] &\leq (1-\epsilon)^{-2}\Tr[((\sigma^{\ti{C}_1}_1)^{\otimes s_2})^2] \\
&\leq (1-\epsilon)^{-2} 2^{-n(S(C_1)_{\psi} - 3\da)}.
\end{split}
\end{equation}
We can apply a similar set of inequalities for the quantity $\Tr[(\sigma^{\ti{C}_2}_2)^2]$ and obtain
\begin{equation}
\begin{split}
 \Tr[(\Psi^{\ti{C}_2})^2] &\leq (1-\epsilon)^{-2}\Tr[((\sigma^{\ti{C}_2}_1)^{\otimes s_2})^2] \\
&\leq (1-\epsilon)^{-2} 2^{-n(S(C_2)_{\psi} - 3\da)}.
\end{split}
\end{equation}
Since $\sigma^{\ti{C}_1}_2 \leq (\sigma^{\ti{C}_1}_1)^{\otimes s_2}$, the rank of $\Psi^{\ti{C}_1}$ can be bounded as follows:
\begin{equation*}
\begin{split}
\mbox{rank }\Psi^{\ti{C}_1} &\leq  (\mbox{rank }\sigma^{\ti{C}_1}_1)^{s_2} \\
&\leq  (2^{-s_1(S(C_1)_{\psi}+\da)})^{s_2}\\
&=  2^{-n(S(C_1)_{\psi}+\da)}.
\end{split}
\end{equation*}
The second line follows from Lemma \ref{lem:mergeTyp}. We have a similar calculation for $\mbox{rank }\Psi^{\ti{C}_2}$, and so we are done.
\end{proof}
\section{Smooth max entropy}
\begin{Lemma} \label{lem:hmax-smoothing}
Suppose the density operator $\rho$ has eigenvalues $r = (r_1, \ldots, r_d)$ with $r_j \geq r_{j+1}$. Then
\begin{equation}
H^\epsilon_\max(\rho) \geq 2 \log \min \left\{ \sum_{j=1}^{k-1} \sqrt{r_j}:
	k \mbox{ such that } \sum_{j=k+1}^d r_j \leq 2\epsilon \right\}.
\end{equation}
\end{Lemma}
\begin{proof}
By eq.~(\ref{eq:smoothmax}) (see Chapter 2), $H^\epsilon_\max(\rho)$ is equal to the minimum of $H_\max(\overline\rho)$ over all sub-normalized density operators $\overline\rho$ such that $P(\rho,\bar{\rho}) \leq \epsilon$. From eq.~(\ref{eq:purified}), we can bound the purified distance from below by the trace distance:
\begin{equation*}
 P(\rho, \bar{\rho}) \geq \frac{1}{2}\| \rho - \bar{\rho} \|_1.
\end{equation*}
Since the smooth max entropy $H^{\epsilon}_{\max}(\rho)$ is the minimization of $H_{\max}(\bar{\rho})$ over all sub-normalized density operator $\bar{\rho}$ with $P(\rho,\bar{\rho}) \leq \epsilon$, the previous bound implies
\begin{eqnarray}
H^\epsilon_\max(\rho)
	&\geq& \min \left\{
	H_\max(\overline\rho) : \| \rho - \bar{\rho} \|_1 \leq 2\epsilon
	\right\} =: \bar{H}^{\epsilon}_{\max}(\rho),
\end{eqnarray}
\begin{figure}[t]
  \centering
    \includegraphics{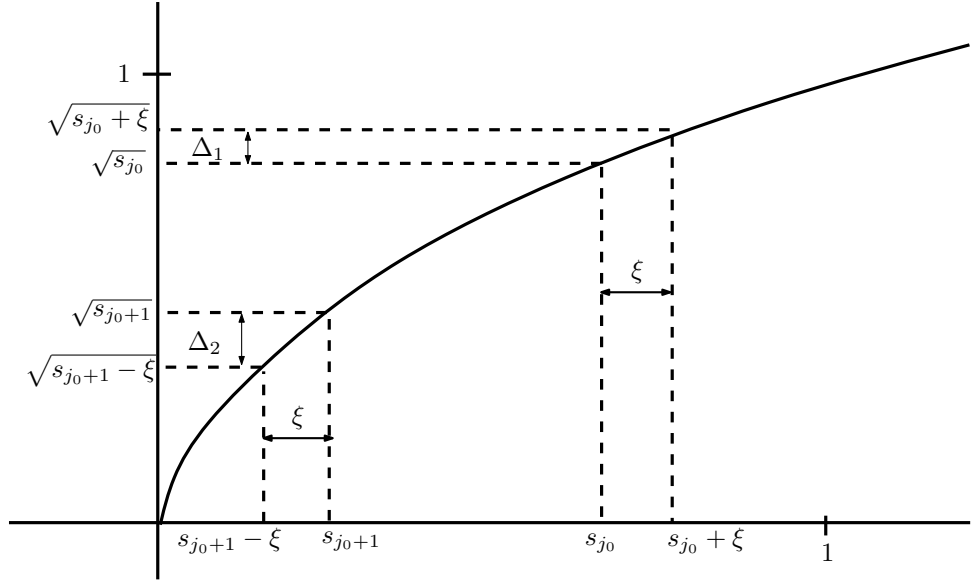}
\caption{The square root function evaluated at four different points. The derivative of $\sqrt{x}$ is a non-increasing function for $x> 0$. Since we have $s_{j_0}\geq s_{j_0+1}$ by hypothesis, the finite difference $\Delta_1$ cannot be greater than the finite difference $\Delta_2$. }\label{fig:sqrt}
\end{figure}
 where $\bar{H}^{\epsilon}_{\max}(\rho)$ the minimization of the max entropy $H_{\max}(\bar{\rho})$ over all sub-normalized density operators $\bar{\rho}$ such that $\|\rho - \bar{\rho}\|_1 \leq 2\epsilon$. Let $\overline{\rho}$ be a sub-normalized density operator such that $\bar{H}^{\epsilon}_{\max}(\rho) = H_{\max}(\overline{\rho})$. Let $\overline{r} = ( \overline{r}_1,\ldots,\overline{r}_d )$ be the eigenvalues of $\overline\rho$, ordered such that $\overline{r}_j \geq \overline{r}_{j+1}$. We will identify $\overline{r}$ and $r$ with their corresponding diagonal matrices. Then, we have (see the proof of the Fannes inequality in \cite{Nielsen})
\begin{equation*}
\| \overline{r} - r \|_1 \leq \| \overline{\rho} - \rho \|_1 \leq 2\epsilon.
\end{equation*}
Thus, without loss of generality, we can assume that $\overline{\rho}$ and $\rho$ are diagonal in the same basis. We therefore dispense with $\rho$ and $\overline\rho$, discussing only $r$ and $\overline{r}$ from now on.

By eq.~(\ref{eq:hmax}), we have $H_\max(r) = 2 \log \sum_j \sqrt{r_j}$, which is monotonically decreasing in each $r_j$. This implies that a minimizing $\overline{r}$ must satisfy $\overline{r}_j \leq r_j$. If not, redefining $\overline{r}_j = r_j$ decreases $\| \overline{r} - r \|_1$ and $H_\max(\overline{r})$ at the same time.

We will now argue that there is a minimizing $\overline{r}$ with the following property: there is a $j_0$ for which $r_j = \overline{r}_j$ for all $j < j_0$ and $\overline{r}_j = 0$ for all $j > j_0$. Let $s = (s_1,\ldots,s_d)$ be any vector such that $H_{\max}(s) = \bar{H}^{\epsilon}_{\max}(\rho)$, with $s_j \geq s_{j+1} \geq 0$ and $s_j \leq r_j$. That such a vector exists follows from the arguments in the previous paragraphs. Suppose that $s$ does not have the prescribed form. That is, there is a $j_0$ such that $s_{j_0} < r_{j_0}$ but $s_{j_0+1} > 0$. Note that this implies that $s_{j_0} > 0$ since $s_{j_0} \geq s_{j_0+1}$. Let $\xi = \min\{ r_{j_0} - s_{j_0}, s_{j_0+1}\}$ be the minimum between $s_{j_0+1}$ and the difference between the eigenvalues $r_{j_0}$ and $s_{j_0}$. Define the vector $s'$ such that $s'_{j_0} = s_{j_0}+\xi$, $s'_{j_0+1} = s_{j_0+1}-\xi$ and $s'_j= s_j$ for $j \not\in \{ j_0, j_0+1 \}$. If $\xi = s_{j_0+1}$, the trace norm $\| r - s'\|_1$ is equal to
\begin{equation*}
\begin{split}
\|r - s'\|_1 &= \sum^{j_0-1}_{j=1} |r_j - s_j| + |r_{j_0} - (s_{j_0} + s_{j_0+1})| + r_{j_0+1} \\
&=  \sum^{j_0-1}_{j=1} (r_j - s_j) + (r_{j_0}-s_{j_0}) + r_{j_0+1} - s_{j_0+1} \\
&= \|r-s\|_1.
\end{split}
\end{equation*}
If $\xi = r_{j_0}-s_{j_0}$, the trace norm $\| r - s'\|_1$ is equal to
\begin{equation*}
\begin{split}
\|r - s'\|_1 &= \sum^{j_0-1}_{j=1} |r_j - s_j| + |r_{j_0} - (s_{j_0}+\xi)| + |r_{j_0+1}-(s_{j_0+1}-\xi)| \\
&=  \sum^{j_0-1}_{j=1} (r_j - s_j)  + r_{j_0+1} - s_{j_0+1} + r_{j_0}-s_{j_0} \\
&= \|r-s\|_1.
\end{split}
\end{equation*}
Hence, our new vector $s'$ preserves the trace norm $\|r-s\|_1$. The derivative of the function $\sqrt{x}$ is given by
\begin{equation}
\frac{df}{dx} \sqrt{x} = \frac{1}{2\sqrt{x}},
\end{equation}
which is well-defined for any $x > 0$ and is also a non-increasing function of $x$. Since $s_{j_0+1} \leq s_{j_0}$, we have (see Figure \ref{fig:sqrt}):
\begin{equation*}
\Delta_1 := \sqrt{s_{j_0}+\xi}-\sqrt{s_{j_0}} \leq \sqrt{s_{j_0+1}} - \sqrt{s_{j_0+1}-\xi} =: \Delta_2
\end{equation*}
and so $H_{\max}(s') \leq H_{\max}(s)$. If $H_{\max}(s')$ is less than $H_{\max}(s)$, we have a contradiction and $s$ must have the prescribed form. Thus, assume from now on the max entropies are equal. If $\xi = s_{j_0+1}$, the new vector $s'$ has almost the prescribed form. We have $s'_{j_0+1}=0$, but it could be that $s'_{j} > 0$ for any $j > j_0+1$. Let $s''$ be the vector such that $s^{''}_{j_0+1} = s'_{j_0+2}$,$s^{''}_{j_0+2}=0$ and $s^{''}_j = s'_j$ for all other values of $j$. Then, the trace norm $\|r-s''\|_1$ is equal to $\|r-s'\|_1$, as can be seen from the following equations:
\begin{equation*}
\begin{split}
\|r-s''\|_1 &= \sum^d_{j \neq \{j_0+1,j_0+2\}} |r_{j} - s^{''}_{j}| + |r_{j_0+1} -s^{''}_{j_0+1}| + |r_{j_0+2} -s^{''}_{j_0+2}| \\
&=  \sum^d_{j \neq \{j_0+1,j_0+2\}} (r_{j} - s'_{j}) + |r_{j_0+1} -s^{'}_{j_0+2}| + r_{j_0+2} \\
&=  \sum^d_{j \neq \{j_0+1,j_0+2\}} (r_{j} - s'_{j}) + r_{j_0+1} + (r_{j_0+2} -s^{'}_{j_0+2}) \\
&= \|r-s'\|.
\end{split}
\end{equation*}
The third line is obtained using $s'_{j_0+2} \leq r_{j_0+2} \leq r_{j_0+1}$. Thus, we can push back $s'_{j_0+1}$ until another zero value is encountered. That is, there exists a vector $\overline{s}$ such that $\overline{s}_j \geq \overline{s}_{j+1}$, with $\overline{s}_j =r_j$ for all $j< j_0$, $\overline{s}_{j_0}=s'_{j_0}$ and $\overline{s}_{j}=0$ for all $j \geq k$ , where $j_0 < k \leq d$. If $\overline{s}_{j_0+1} > 0$, we apply the previous argument until the prescribed form $\overline{r}$ is obtained.

If $s_{j_0+1} \geq \xi$, we have $s'_{j_0} = r_{j_0}$ and we can repeat the previous argumentation with $j'_0 = j_0+1$. Since $j_0$ is at most $d$, we will eventually find a vector $\overline{r}$ of the prescribed form. The statement follows by evaluating the max entropy for this vector $\overline{r}$:
\begin{equation*}
\begin{split}
H_{\max}^{\epsilon}(\rho) &\geq \bar{H}^{\epsilon}_{\max}(\rho) = H_{\max}(\overline{r}) \\
& \geq 2 \log \min \left\{ \sum_{j=1}^{k-1} \sqrt{r_j}: k \mbox{ such that } \sum_{j=k+1}^d r_j \leq 2\epsilon \right\}
\end{split}
\end{equation*}
\end{proof}

\section{Assisted distillation}

\begin{lemma}[Markov's Inequality] \label{lem:Markov}
 If $X$ is a random variable with probability distribution $p(x)$ and expectation $E(X)$, then, for any positive number $a$, we have: \[ P(|X| \geq a) \leq \frac{E(|X|)}{a}. \]
  \end{lemma}

\begin{lemma}\label{lem:unionbound}
Suppose we have $n$ copies of the pure state $\psi^{CABR}$ with $S(R)_{\psi} < S(AB)_{\psi}$ and $S(B)_{\psi} < S(AR)_{\psi}$. Let $\psi^{\tilde{C}A^nB^nR^n}$ be be the normalized state obtained after projecting the space $C^n$ into its $\delta-$typical subspace $\tilde{C}$. If Charlie performs a (rank one) random measurement of his system $\tilde{C}$, we have, for any fixed $\xi_1 > 0$ and $\xi_2 > 0$,
\begin{equation}\label{eq:fanness}
\begin{split}
\int_{\mathbb{U}(\tilde{C})} P\left (\| \psi^{R^n}_J - (\psi^R)^{\otimes n} \|_1 < \xi_1 \bigcap \| \psi^{B^n}_J - (\psi^B)^{\otimes n}\|_1 < \xi_2\right ) dU \geq 1-\alpha,
\end{split}
\end{equation}
 where $\alpha$ can be made arbitrarily small by taking sufficiently large values of $n$. Here, $J$ is the random variable associated with the measurement outcome and $\psi_J^{A^nB^nR^n}$ is the pure state of the systems $A^n, B^n$ and $R^n$ after Charlie's measurement.
\end{lemma}
\begin{proof}
The proof of this statement is obtained by combining Proposition \ref{thm:random} with Markov's inequality and Boole's inequality (the union bound). For any $\xi_1 > 0$ and $\xi_2 > 0$, consider a projective measurement of Charlie with rank one projectors $U\braket{i}U^{\dag}$ and let $J$ be the measurement outcome. We want to bound the following probability from below:
\begin{equation}\label{eq:probR}
P_J:=P(\| \psi^{R^n}_J - (\psi^R)^{\otimes n} \|_1 < \xi_1 \bigcap \| \psi^{B^n}_J - (\psi^B)^{\otimes n}\|_1 < \xi_2)  \geq 1 - \alpha
\end{equation}
for any $\alpha > 0$. Applying the union bound and Markov's inequality to such probability, we have
\begin{equation}\label{eq:probS}
\begin{split}
P_J &\geq 1 - P(\| \psi^{R^n}_J - (\psi^R)^{\otimes n} \|_1 \geq \xi_1) - P(\| \psi^{B^n}_J - (\psi^B)^{\otimes n}\|_1 \geq \xi_2)  \\
&\geq 1- \frac{\sum_j p_j \| \psi^{R^n}_j - (\psi^R)^{\otimes n} \|_1} {\xi_1} - \frac{\sum_j p_j \| \psi^{B^n}_j - (\psi^B)^{\otimes n} \|_1}{\xi_2}.
\end{split}
\end{equation}
Averaging over all unitaries, using the Haar measure, we get
\begin{equation}\label{eq:probS}
\begin{split}
\int_{\mathbb{U}(\tilde{C})} P_J dU &\geq 1- \frac{\int_{\mathbb{U}(\tilde{C})} \sum_j p_j \| \psi^{R^n}_j - (\psi^R)^{\otimes n} \|_1 dU} {\xi_1} - \frac{\int_{\mathbb{U}(\tilde{C})}\sum_j p_j \| \psi^{B^n}_j - (\psi^B)^{\otimes n} \|_1 dU}{\xi_2}.
\end{split}
\end{equation}
The averages are not quite of the desired form to apply Proposition \ref{thm:random} directly. Define the state
\begin{equation*}
\ket{\Omega}^{\tilde{C}\tilde{A}\tilde{B}\tilde{R}} := (\Pi_{\tilde{A}} \otimes \Pi_{\tilde{B}} \otimes \Pi_{\tilde{C}} \otimes \Pi_{\tilde{R}}) \ket{\psi}^{\otimes n},
\end{equation*}
and let $\ket{\Psi}^{\tilde{C}\tilde{A}\tilde{B}\tilde{R}}$ be the normalized version of $\ket{\Omega}^{\tilde{C}\tilde{A}\tilde{B}\tilde{R}}$. If Charlie were to perform his measurement on the state $\Psi$, the properties of typicality tell us that the trace norms $\|\psi^{R^n}_j - \Psi^{\tilde{R}}_j\|_1$  and $\|\psi^{B^n}_j - \Psi^{\tilde{B}}_j\|_1$ should be arbitrarily close. This is verified by using the bounds between the trace distance and the purified distance, eq.~(\ref{eq:purified}), and the monotonicity of the purified distance under trace non-increasing quantum operations (see \cite{Renner01} for a proof of this fact):
\begin{equation*}
\begin{split}
 \|\psi^{R^n}_j -  \Psi^{\tilde{R}}_j\|_1 &\leq \|\psi_j^{\tilde{C}A^nB^nR^n} - \Psi_j^{\tilde{C}\tilde{A}\tilde{B}\tilde{R}}\|_1 \\
 &\leq 2P(\psi^{\tilde{C}A^nB^nR^n}_j, \Psi^{\tilde{C}\tilde{A}\tilde{B}\tilde{R}}_j) \\
&\leq 2P(\psi^{\tilde{C}A^nB^nR^n}, \Psi^{\tilde{C}\tilde{A}\tilde{B}\tilde{R}})\\
&\leq \epsilon,
\end{split}
\end{equation*}
for any $\epsilon > 0$ by choosing sufficiently large values of $n$. The last line follows from typicality and the triangle inequality. A similar statement holds for the trace norm $\|\psi^{B^n}_j - \Psi^{\tilde{B}}_j\|_1$.
Applying the triangle inequality twice on each average of eq.~(\ref{eq:probS}), we have
\begin{equation}\label{eq:probSD}
\begin{split}
\int_{\mathbb{U}(\tilde{C})} P_J dU &\geq 1- f(\epsilon) - \frac{\int_{\mathbb{U}(\tilde{C})} \sum_j p_j \| \Psi^{\tilde{R}}_j - \Psi^{\tilde{R}} \|_1 dU}{\xi_1} - \frac{\int_{\mathbb{U}(\tilde{C})} \sum_j p_j \| \Psi^{\tilde{B}}_j - \Psi^{\tilde{B}}\|_1 dU}{\xi_2},
\end{split}
\end{equation}
where $f(\epsilon)$ is a function of various trace norms which vanish, by typicality, for sufficiently large values of $n$. Applying Proposition \ref{thm:random} on the averages of eq.~(\ref{eq:probSD}), we can make the right hand side bigger than $1-\alpha$ for any $\alpha > 0$ by choosing $n$ sufficiently large.
\end{proof}
\end{appendices}

\fancyhead[RO]{\emph{References}}
\bibliographystyle{unsrt}
\bibliography{EofABib}
\end{document}